\documentclass[12pt]{article}
\pdfoutput=1
\usepackage[utf8]{inputenc}
\usepackage[margin=2.5cm]{geometry}

\usepackage{enumitem}
\usepackage{booktabs}
\usepackage{dsfont}
\usepackage{amsthm}
\usepackage{multirow}
\usepackage{dirtytalk}
\usepackage{multicol}
\usepackage[nodisplayskipstretch]{setspace}

\usepackage{dcolumn}
\usepackage{graphicx}
\usepackage{commath}
\usepackage{subcaption}
\usepackage{bbm}
\usepackage{hyperref}
\usepackage[authoryear]{natbib}
\hypersetup{
    colorlinks=true,       
    linkcolor=red,          
    citecolor=blue!70!black,        
    urlcolor=blue!70!black           
}
\usepackage{stmaryrd}
\usepackage{mathrsfs}

\usepackage{amsmath, amsthm, amssymb, amsfonts}
\usepackage{econometrics}
\newtheorem{theorem}{Theorem}
\newtheorem{lemma}{Lemma}
\newtheorem{corollary}{Corollary}[lemma]
\newtheorem{proposition}{Proposition}
\theoremstyle{definition}

\newtheorem{example}{Example}
\newtheorem{remark}{Remark}
  {
      \theoremstyle{plain}
      \newtheorem{assumption}{Assumption}
  }

\usepackage{physics}
\usepackage{tikz}
\usepackage{array}

\usepackage{dsfont}
\newcommand{\cleansum}{\sum\limits} 
\DeclareMathOperator*{\argmax}{arg\,max}
\DeclareMathOperator*{\argmin}{arg\,min}
\DeclareMathOperator{\spn}{span}
\DeclareMathOperator{\Ima}{Im}

\usepackage{amsmath}

\title{Transition Probabilities and Moment Restrictions in Dynamic Fixed Effects Logit Models \vspace{0.25cm}}
\author{\parbox{\linewidth}{
        \centering 
	\large Job Market Paper \\
\vspace{0.5cm}
\normalsize \setstretch{1}
\large Kevin Dano
\thanks{Department of Economics, University of California Berkeley. E-mail: kdano@berkeley.edu.  
\newline I am very grateful to  my advisors Bryan Graham, Stéphane Bonhomme, Demian Pouzo and Jim Powell for their generous support and advice. I also thank Chris Muris,  Bocar Ba, Yassine Sbai Sassi, Nick Gebbia, Tahsin Saffat and the audiences at the 2023 IAAE Annual Conference, the 2023 California Econometrics Conference and the 2023 Causal Panel Data Conference at the Stanford Graduate School of Business for valuable comments and discussions. Financial support from the 2023 IAAE Conference is gratefully acknowledged. All errors are my own.}
}}

 \date{ \parbox{\linewidth}{
 \centering \normalsize \setstretch{1}
 \hfill \\ 
 This version:
  \today . \\
Newest version \href{https://kevindano.github.io/assets/files/JMP_KevinDano.pdf}{here}. \endgraf}} 

\begin{document}
\bibliographystyle{apalike}
\setcitestyle{authoryear,open={(},close={)}}
\maketitle

\begin{abstract}\setstretch{1.25}\noindent
Dynamic logit models are popular tools in economics to measure state dependence. This paper introduces a new method to derive moment restrictions in a large class of such models with strictly exogenous regressors and fixed effects.  We exploit the common structure of logit-type transition probabilities and elementary properties of rational fractions, to formulate a systematic procedure that scales naturally with model complexity (e.g the lag order or the number of observed time periods). We detail the construction of moment restrictions in binary response models of arbitrary lag order as well as first-order panel vector autoregressions and dynamic multinomial logit models. Identification of common parameters and average marginal effects is also discussed for the binary response case.
Finally, we illustrate our results by studying the dynamics of drug consumption amongst young people inspired by \cite{deza2015there}.
\end{abstract}



\newpage
\section{Introduction} \label{Section_1}
The analysis of state dependence is a classic and important topic in many areas of economics. Several discrete processes such as welfare and labor force participation manifest strong serial persistence, and economists have sought various methods to unravel the underlying factors. In this paper, we reexamine the estimation of one notable set of models employed for this purpose:  discrete choice models with lagged dependent variables, strictly exogenous regressors, fixed effects and logistic errors.  We shall refer to this class of models as dynamic fixed effects logit models (DFEL) throughout. Specifications of this kind are used to discriminate between \say{structural} state dependence, i.e the causal effect of past choices on current outcomes, and heterogeneity, i.e the serial correlation induced by unobserved individual attributes (\cite{heckman1981heterogeneity}). An example of this approach is the analysis of welfare participation in \cite{chay1999non}. There has been considerable interest in this family of panel data models in econometrics, with a recent surge in attention following new developments reported in \cite{honore2020moment}. One general reason is that DFEL models stand out as a rare case of nonlinear dynamic panel data models for which solutions to the \textit{incidental parameters problem} (\cite{neyman1948consistent}) and \textit{initial conditions problem} (e.g \cite{heckman1981heterogeneity}) have been known to exist in short panels\footnote{The incidental parameters problem refers to the general inconsistency of maximum likelihood in short panels. The initial conditions problem refers to the general difficulty of formulating a correct conditional distribution for the initial observations given the fixed effects and covariates.}.   \\
\indent  In the \say{pure} version of the basic model which abstracts from  covariates other than a first order lag, \cite{cox1958regression}, \cite{chamberlain_1985} and \cite{magnac2000subsidised} showed that the autoregressive parameter can be consistently estimated by conditional likelihood. This approach relies on the existence of a sufficient statistic linked to the logistic assumption to eliminate the fixed effect. In an important subsequent paper, \cite{honore2000panel} extended this idea to a setting with strictly exogenous regressors and showed that the conditional likelihood approach remains viable if one can further condition on the regressors being equal in specific periods. This strategy was also found to be effective in dynamic multinomial logit models (\cite{honore2000panel}), panel vector autoregressions (\cite{honore2019panel}) and dynamic ordered logit models (\cite{muris2020dynamic}. At the same time, it has also been noted 
that the necessity to be able to \say{match} the covariates imposes two limitations for the conditional likelihood approach: it inherently rules out time effects and implies rates of convergence  slower than $\sqrt{N}$ for continuous explanatory variables. Furthermore,  calculations from \cite{honore2000panel} suggested that it does not easily extend to models with a higher lag order.  These shortcomings have motivated the search for alternative methods of estimation. \\
\indent Recently, \cite{kitazawa2013exploration,kitazawa2016root} and
 \cite{Kitazawa_JOE2021} revisited the AR(1) logit model - autoregressive of order one - of \cite{honore2000panel} and proposed a transformation approach that deals with the fixed effects without restricting the nature of the covariates besides the conventional assumption of strict exogeneity. Their methodology leads to moment restrictions that can serve as a basis to estimate the model parameters at $\sqrt{N}$-rate by GMM; even with continuous regressors. In parallel work, \cite{honore2020moment} also derived moment conditions for the AR(1), AR(2) and AR(3) logit models in panels of specific length using the functional differencing technique of \cite{Bonhomme_EM12}. Their approach is partly numerical and relies on symbolic computing (e.g  Mathematica) to obtain analytical expressions but has a wider scope of potential applications, e.g  dynamic ordered logit specifications (\cite{honore2021dynamic}).  In another recent paper, \cite{dobronyi2021identification}, the authors analyze the full likelihood of AR(1) and AR(2) logit models with discrete covariates under a new angle that reveals a connection to the \textit{truncated moment problem} in mathematics. Drawing on well established results in that literature, they derive moment equality and new moment inequality restrictions that fully characterize the  sharp identified set. \\
 \indent In this paper, 
 we introduce a new systematic approach to construct moment restrictions in  DFEL models with additive fixed effects, i.e when fixed effects are heterogeneous \say{intercepts}. This class of models encompasses most specifications studied in prior work but excludes models with heterogeneous coefficients on lagged outcomes and/or regressors as in \cite{chamberlain_1985} and \cite{browning2014dynamic}.  Unlike some recent competing approaches, we do not require numerical experimentation nor symbolic computing. Rather, as we shall see in examples, we exploit the common structure of logit-type transition probabilities and elementary properties of rational fractions, to obtain  analytic expressions for the identifying moments. We shall focus our attention on deriving valid moment functions for AR($p$) models with arbitrary lag order $p\geq 1$ as well as first-order panel vector autoregressions and dynamic multinomial logit models (\cite{magnac2000subsidised}). \\
   \indent Our methodology exploits two key observations.  First, the transition probabilities of logit-type models can often be expressed as conditional expectations of functions of observables and common parameters given the initial condition, the regressors and the fixed effects. We shall refer to these moment functions as \textit{transition functions}.  They have the important feature of not depending on individual fixed effects. Second, as soon as $T\geq p+2$, where $T$ denotes the number of observations post initial condition, many transition probabilities in periods $t\in \{p+1,\ldots,T-1\}$ admit at least two distinct transition functions.  The combination of these two features motivates a two-step approach to obtain moment restrictions in panels of adequate length. In the first step, we shall compute the model transition functions. Then, the second step will simply consist in differencing two transition functions associated to the same transition probability. We show that a careful application of this procedure delivers all the moment equality restrictions available in the binary response case. We shall further elaborate on these steps in examples and use the resulting moment functions to derive new identification results. At a high level, the approach we advocate in this paper consists in solving a sequence of problems with identical structure period by period instead of solving directly a large system of equations based on the model full likelihood as in \cite{honore2020moment} and \cite{dobronyi2021identification}. As a consequence, our procedure remains tractable when the number of time periods increases and in models with higher order lags. \\
\indent Besides the aforementioned papers, our work also connects to a line of research studying the identification of features of the distribution of fixed effects in discrete choice models. One branch in this literature has focused on developing general optimization tools to compute sharp numerical bounds on average marginal effects. This includes most notably the linear programming method of \cite{honore2006bounds}, recently adapted by \cite{bonhomme2023identification} to the case of sequentially exogenous covariates, and the quadratic programming method of \cite{chernozhukov2013average}. A second branch in this literature has sought instead to harness the specificities of logit models to obtain simple analytical bounds. In static logit models,  \cite{davezies2021identification} exploit mathematical results on the \textit{moment problem} to formulate sharp bounds on the average partial effects of regressors on outcomes.  In DFEL models, \cite{aguirregabiria2021identification} are the first to prove the point identification of average marginal effects  in the baseline AR(1) logit model when $T\geq 3$. In related work, \cite{dobronyi2021identification} make use of their moment equality and moment inequality restrictions to establish sharp bounds on functionals of the fixed effects such as average marginal effects and average posterior means in AR(1) and AR(2) specifications. We complement these results as a byproduct of our methodology:  average marginal effects and their variants in AR($p$) models, with arbitrary $p\geq 1$ are merely differences of average transition functions. \\
\indent The remainder of the paper is organized as follows. Section \ref{Section_2} presents the setting and our main objective. Section \ref{Section_3} introduces some terminology and gives an outline of our procedure to construct moment restrictions. Section \ref{Section_4}  implements our approach in AR($p$) logit models with $p\geq 1$ and discusses identification of model parameters and average marginal effects. The semiparametric efficiency bound for the AR(1) is also presented for the base case of four waves of data. Section \ref{Section_5} discusses extensions to the VAR(1) and the dynamic multinomial logit model with one lag,  MAR(1) for short. In Section \ref{Section_6}, we present an empirical illustration on the dynamics of drug consumption amongst young people and Section \ref{Section_7} offers concluding remarks. A complementary set of Monte Carlo simulations showing the small sample performance of GMM estimators based on our moment restrictions is available in Appendix  Section \ref{Section_MonteCarlo}.  Proofs are gathered in the Appendix.

\section{Setting, assumptions and objective} \label{Section_2}
Let $i=1,\ldots,N$ denote a population index and $t=0,\ldots,T$ be an index for time. We study DFEL models which may be viewed as threshold-crossing econometric specifications describing a discrete outcome $Y_{it}$ through a latent index involving  lagged outcomes (e.g $Y_{it-1}$), strictly exogenous regressors $X_{it}$, an individual-specific time-invariant unobservable $A_i$ and an error term $\epsilon_{it}$. The canonical example is the AR(1) model:
\begin{align*}
    Y_{it}=\mathds{1}\{\gamma_0 Y_{it-1}+X_{it}'\beta_0+A_i-\epsilon_{it}\geq 0\}, \quad t= 1,\ldots,T
\end{align*}
 and we shall concentrate more broadly on cases where $A_i$ is additively separable from the other explanatory variables.
 An initial condition that we will generically denote $Y_{i}^{0}$ completes such models to enable dynamics. The common parameter $\theta_0$ is one target of interest and governs the influence of lagged outcomes and the regressors on the contemporaneous outcome.  Other quantities of interest include counterfactual parameters such as average marginal effects.\\
 \indent Throughout, we leave the joint distribution of $(Y_{i}^{0},X_i,A_i)$ unrestricted where $X_i=(X_{i1},\ldots,X_{iT})$ and thus refer to $A_i$ as a fixed effect in common with the literature. The schocks $\epsilon_{it}$ are assumed to be serially independent logistically distributed, independent of $(Y_{i}^{0},X_i,A_i)$, except for the MAR(1) model where they are instead extreme value distributed. Finally, we shall assume that $(Y_i,Y_{i}^{0},X_i,A_i)$ are jointly i.i.d across individuals.

\indent The data available to the econometrician consists of the initial condition $Y_{i}^0$, the outcome vector $Y_{i}=(Y_{i1},\ldots,Y_{iT})$, and the covariates $X_i$ for all $N$ individuals. Interest centers primarily on the identification and estimation of $\theta_0$ in short panels, i.e for fixed $T$. To this end,  the chief objective of this paper is to show how to construct moment functions $\psi_{\theta}(Y_i,Y_{i}^0,X_{i})$ free of the fixed effect parameter that are valid in the sense that:
\begin{align}\label{moment_function_final}
    \mathbb{E}\left[\psi_{\theta_0}(Y_i,Y_{i}^{0},X_{i})\,|\,Y_{i}^{0},X_i,A_i\right]=0
\end{align}
When this is possible, the law of iterated expectations implies the conditional moment: $$    \mathbb{E}\left[\psi_{\theta_0}(Y_i,Y_{i}^{0},X_{i})\,|\,Y_{i}^0,X_{i}\right]=0$$ which can in turn be leveraged to assess the identifiability of $\theta_0$ and form the basis of a GMM estimation strategy. This is the central idea underlying functional differencing (\cite{Bonhomme_EM12}) and was applied by \cite{honore2020moment} to derive valid moment conditions for a class of dynamic logit models with scalar fixed effects. We borrow the same insight but instead of searching for solutions numerically on a case-by-case basis, we propose a complementary systematic algebraic procedure to recover the model's valid moments \footnote{\cite{dobronyi2021identification} and \cite{Kitazawa_JOE2021} also have an algebraic approach but our methodologies are very different. The first paper uses the full likelihood of the model and focuses on the AR(1) and special instances of the AR(2) model. The second paper has a transformation approach adapted to the AR(1) model. Our emphasis here is primarily on developing an approach that is tractable for a large class of models.}.  In doing so, we flesh out the mechanics implied by the logistic assumption which in turn suggest a blueprint to deal with estimation of general DFEL models. For example, we are able to characterize the expressions of valid moment functions in AR($p$) models for arbitrary $p>1$ which to the best of our knowledge is a new result in the literature. Furthermore, our approach carries over to multidimensional fixed effect specifications: VAR(1), dynamic network formation models and the MAR(1) in which searching for moments numerically is cumbersome or intractable. \\
\indent In what follows, we shall use the shorthand $Y_{it_{1}}^{t_2}=(Y_{it_1},\ldots,Y_{it_2})$ to denote a collection of random variables over periods $t_1$ to $t_2$ with the convention that $Y_{it_{1}}^{t_2}=\emptyset$ if $t_1>t_2$. Likewise, we may use the notation $y_{t_1}^{t_2}=(y_{t_1},\ldots,y_{t_2})$ to denote any $(t_2-t_1)$-dimensional vector of reals with the convention $y_{t_{1}}^{t_{2}}=\emptyset$ for $t_1>t_2$.  Elements $1_{n}$ and $0_{n}$ shall refer to the $n$-dimensional vectors of ones and zeros respectively. The support of the outcome variable $Y_{it}$ shall be denoted $\mathcal{Y}$. We let $\Delta$ denote the first-differencing operator so that $\Delta Z_{it}=Z_{it}-Z_{it-1}$ for any random variable $Z_{it}$ and make use of the notation $Z_{its}=Z_{it}-Z_{is}$ for $s\neq t$ to accommodate long differences. We use $\mathds{1}\{.\}$ for the indicator function;  $\Ima(f)$, $\ker(f)$, $\rank(f)$ to denote the image, the nullspace and the rank of a linear map $f$. 

\section{Outline of the procedure to derive valid moment functions} \label{Section_3}
Let $T\geq 1$. Given an initial condition $y^0\in \mathcal{Y}^{p}$, $p\geq 1$ being the lag order of the model, and strictly exogenous regressors $X_i\in \mathbb{R}^{K_{x}\times T}$, we denote the (one-period ahead) transition probability in period $t\geq 1$ from state $(l_{1}^t,y^
0)\in \mathcal{Y}^{t}\times \mathcal{Y}^{p}$ to state $k\in \mathcal{Y}$ as:
\begin{align*}
    \pi^{k|l_{1}^t,y^0}_{t}(A_i,X_{i})=\pi^{k|l_{1}^t,y^0}_{t}(A_i,X_{i};\theta_0)\equiv P(Y_{it+1}=k \,|\ Y_{i}^0=y^0,Y_{i1}^t=l_{1}^t,X_i,A_i)
\end{align*}
With $p$ lags, the markovian nature of the models considered in this paper imply that $  \pi^{k|l_{1}^t,y^0}_{t}(A_i,X_{i})$ will not depend on the entire path of past outcomes but only on the value of the most recent $p$ outcomes. For instance, in an AR(1) model where $p=1$, we have:
\begin{align*}
    \pi^{k|l_{1}^t,y^0}_{t}(A_i,X_{i})=P(Y_{it+1}=k \,|\ Y_{i}^0=y^0,Y_{i1}^t=l_{1}^t,X_i,A_i)=P(Y_{it+1}=k \,|\ Y_{it}=l_{t},X_i,A_i)
\end{align*}
and thus we will suppress the dependence on $(y^0,l_1,\ldots,l_{t-1})$ and write $\pi^{k|l_t}_{t}(A_i,X_{i})$. We shall proceed analogously for the more general case $p\geq 1$. \\
\indent We call a \textit{transition function} associated to a transition probability  $  \pi^{k|l_{1}^t,y^0}_{t}(A_i,X_{i})$
any moment function $\phi^{k|l_{1}^t,y^0}_{\theta}(Y_i,Y_{i}^{0},X_i)$ of the data and the common parameters verifying:
\begin{align} \label{trans_functions}
    \mathbb{E}\left[\phi^{k|l_{1}^t,y_0}_{\theta_0}(Y_i,Y_{i}^{0},X_i)\,|\,Y_{i}^0,X_{i},A_i\right]=\pi^{k|l_{1}^t,y_0}_{t}(A_i,X_{i})
\end{align}
With these notions in hand, we are ready to describe our two-step approach to derive valid moment functions in the sense of equation (\ref{moment_function_final}). In \textbf{Step 1)}, we begin by computing the model's transition functions. Our procedure requires a minimum  of $T=p+1$ periods of observations to accommodate arbitrary regressors and initial condition. In this case, we can get analytical formulas for the transition functions associated to the transition probabilities in period $t=p$  and Theorem \ref{theorem_nummoments_AR1} and Theorem \ref{theorem_nummoments_ARp}  below imply that they are unique. However, this is not immediately helpful to get moment (equality) restrictions on $\theta_0$. We require one more period. 
As soon as $T\geq p+2$, we explain how to construct distinct transition functions associated to the same transition probabilities in periods $t\in\{p+1,\ldots,T-1\}$. The key ingredient is the use of \textit{partial fraction decompositions} for \textit{rational fractions} adapted to the structure of the transition probabilities. It is then a matter of taking differences of two transition functions associated to the same transition probability to obtain valid moment functions; we refer to this last step as \textbf{Step 2)}.  The ensuing sections demonstrate this procedure in scalar and multidimensional fixed effect models.

\section{Scalar fixed effect models} \label{Section_4}

\subsection{Moment restrictions for the AR(1) logit model} 
For exposition, we begin with the baseline AR(1) logit model with fixed effects introduced above:
\begin{align}\label{AR1_logit_general}
    Y_{it}=\mathds{1}\{\gamma_0 Y_{it-1}+X_{it}'\beta_0+A_i-\epsilon_{it}\geq 0\}, \quad t= 1,\ldots,T
\end{align}
Here, $\mathcal{Y}=\{0,1\}$, $\theta_0=(\gamma_0,\beta_0')\in \mathbb{R}\times \mathbb{R}^{K_{x}}$, the initial condition $Y_{i}^0$ consists of the binary-valued random variable $Y_{i0}$ and $A_i\in \mathbb{R}$. 

\subsubsection{The number of moment restrictions in the AR(1)}
We start out by enumerating the moment restrictions implied by the model. This will provide a means to assess the exhaustiveness of our approach. To this end, let $\mathcal{E}_{y_0,x}$ denote the conditional expectation operator mapping any function of the outcome variable $Y_i$ to its conditional expectation given $Y_{i0}=y_0,X_{i}=x$ and the fixed effect $A_i$, i.e
\begin{align*}
  \mathcal{E}_{y_0,x}\colon \mathbb{R}^{\mathcal{Y}^T}& \longrightarrow \mathbb{R}^\mathbb{R} \\[-1ex]
  \phi(.;y_0,x) & \longmapsto \mathbb{E}\left[\phi(Y_i,y_0,x)|Y_{i0}=y_0,X_{i}=x,A_{i}=.\right]
\end{align*}
For example, for any $y\in \mathcal{Y}^T$, $\mathcal{E}_{y_0,x}\left[\mathds{1}\{.=y\}\right]$ yields the conditional probability of observing history $y$ for all possible values of the fixed effect, i.e:
\begin{align*}
    \mathcal{E}_{y_0,x}\left[\mathds{1}\{.=y\}\right]&=P(Y_i=y|Y_{i0}=y_0,X_{i}=x,A_{i}=.)
\end{align*}
where  $P(Y_i=y|Y_{i0}=y_0,X_{i}=x,A_{i}=a)=\prod \limits_{t=1}^T \frac{e^{y_t(\gamma_0y_{t-1}+x_t'\beta_0+a)}}{1+e^{\gamma_0y_{t-1}+x_t'\beta_0+a}}, \quad \forall a \in \mathbb{R}$ . Then, we have the following result,
\begin{theorem}\label{theorem_nummoments_AR1}
     Consider model (\ref{AR1_logit_general}) with $T\geq 1$ and initial condition $y_0\in \mathcal{Y}$. Suppose that for any   $t,s\in \{1,\ldots,T-1\}$ and $y,\tilde{y}\in \mathcal{Y}$, $\gamma_0y+x_{t}'\beta_0\neq \gamma_0\tilde{y}+x_{s}'\beta_0$ if $t\neq s$ or $y\neq\tilde{y}$. Then, the family $ \mathcal{F}_{y_0,T}=\left\{1,\pi_{0}^{y_0|y_0}(.,x),(\pi_{t}^{0|0}(.,x),\pi_{t}^{1|1}(.,x))_{t=1}^{T-1}\right\}$ of size $2T$ forms a basis of $\Ima(\mathcal{E}_{y_0,x})$ and $\dim\left(\ker(\mathcal{E}_{y_0,x})\right)=2^T-2T$. 
\end{theorem}
\noindent Theorem \ref{theorem_nummoments_AR1} formalizes the intuition that the transition probabilities summarize the parametric component of the model:  $2^T$ histories are possible yet only $2T$ basis elements are necessary to fully characterize their conditional probabilities. 
This follows from the observation that when the covariate index \footnote{We refer to the quantity $\gamma_0y_{t-1}+x_{t}'\beta_0$ for a given period $t$.} of each transition probability differ, the conditional probability of each history $y\in \mathcal{Y}^{T}$ is a ratio of polynomials in $e^{a}$, where the numerator has lower degree than the denominator, and the later is a product of distinct irreducible terms. A sufficient condition for this is that $\gamma_0\neq 0$ and that one regressor is continuously distributed with non-zero slope. In turn, standard results on \textit{partial fraction decompositions} ensure that this ratio can be expressed as a unique linear combination of transition probabilities. To finally conclude that $\mathcal{F}_{y_0,T}$ is a basis of $\Ima(\mathcal{E}_{y_0,x})$, we leverage upcoming results demonstrating that the transition probabilities live in $\Ima(\mathcal{E}_{y_0,x})$ as expectations of transition functions.  \\
\indent Importantly, since $\ker(\mathcal{E}_{y_0,x})$ is the set of valid moment functions verifying equation (\ref{moment_function_final}),  Theorem \ref{theorem_nummoments_AR1} tells us that the AR(1) model features $2^T-2T$ linearly independent moment restrictions in general. This is a  consequence of the \textit{rank nullity theorem} for linear maps with finite dimensional domains. The fact that $2^T-2T$ moment conditions are available for the AR(1) appeared initially as a conjecture in \cite{honore2020moment} and  was later established by \cite{dobronyi2021identification} using different arguments from here. They do not emphasize the role of the transition probabilities. Our ideas extend naturally to the case of arbitrary lags which was hitherto an open problem. We discuss this extension in Section \ref{Section_ARp_impossibility}.

\begin{remark}[Counting moments in logit models]
The idea of decomposing the conditional probabilities of all choice histories in a basis provides a useful device to infer a lower bound on the number of moment restrictions in logit models. If one can further prove that elements of this basis belong to the image of the conditional expectation operator, then this lower bound coincides with the exact number of moment restrictions.
\begin{itemize}
    \item In the static panel logit model of \cite{rasch1960studies}, $\gamma_0=0$ and we have $\pi_{t}^{1|1}(.,x)=1-\pi_{t}^{0|0}(.,x)$. Thus, provided that $x_t'\beta_0\neq x_s'\beta_0$ for all $t\neq s$, $\mathcal{F}_{T}=\left\{1,(\pi_{t}^{0|0}(.,x))_{t=0}^{T-1}\right\}$ spans the image of the conditional expectation operator. This implies at least $2^T-(T+1)$ moment restrictions. It turns out that  $2^T-(T+1)$ is precisely the total number of  moment restrictions for this model. This follows from Remark \ref{static_logit_moments} below which characterizes  the transition functions associated to each element of  $\mathcal{F}_{T}$. 
    \item In the \cite{cox1958regression} model, $\gamma_0\neq 0$ and $\beta_0=0$ and the transition probabilities are: $\pi^{0|0}(a)=\frac{1}{1+e^{a}}$ and $\pi^{1|1}(a)=\frac{e^{\gamma_0+a}}{1+e^{\gamma_0+a}}$ (or equivalently $\pi^{0|1}(a)=\frac{1}{1+e^{\gamma_0+a}}$). See the next section for further details. In this case, the family  $\mathcal{F}_{y_0,T}=\left\{1,\left(\pi^{0|0}(.)^j,\pi^{0|1}(.)^j\right)_{j=1}^{T-1},\pi^{0|y_0}(.)^T \right\}$ which consists of powers of the time-invariant transition probabilities spans the image of the conditional expectation operator. Since $|\mathcal{F}_{y_0,T}|=2T$, the model produces at least $2^T-2T$ linearly independent moment restrictions.
\end{itemize}
\end{remark}

\begin{remark}[A matrix perspective] Since $\mathcal{E}_{y_0,x}$ is a linear map, it admits a unique $2^T\times 2T $ matrix representation $\Lambda_{y_0,x}$ where each row translates the conditional probability of a choice history $y\in \mathcal{Y}^T$ in terms of the transition probabilities of $\mathcal{F}_{y_0,T}$\footnote{ Entries of this matrix may be found using for example the identities in Appendix Lemma \ref{tech_lemma_1} or any other standard textbook tools for \textit{rational fractions}.}. From this point of view, valid moments correspond to $2^T$-vectors $\psi$  in the left nullspace of $\Lambda_{y_0,x}$, meaning  $\psi'\Lambda_{y_0,x}=0$. Constructing $\Lambda_{y_0,x}$ and then solving this $2T$ linear system of equations in $2^T$ unknowns directly is straightforward using symbolic tools when $T$ is \say{small} (e.g \cite{dobronyi2021identification}, \cite{honore2020moment}) but is computationally impractical otherwise. Instead, we propose a constructive approach to back out analytic expressions of the valid moment functions that is tractable for arbitrary values of $T$.
\end{remark}
\noindent Having clarified the total count of moment restrictions in the AR(1) logit model, we next discuss how to construct them with our two-step procedure.
 
\subsubsection{Construction of valid moment functions for the pure model}
In the absence of exogenous regressors, model (\ref{AR1_logit_general}) simplifies to:
\begin{align}\label{AR1_logit_pure}
    Y_{it}=\mathds{1}\{\gamma_0 Y_{it-1}+A_i-\epsilon_{it}\geq 0\}, \quad t= 1,\ldots,T
\end{align}
which was first introduced by \cite{cox1958regression} and then revisited in \cite{chamberlain_1985}, \cite{magnac2000subsidised}. These papers established the identification of $\gamma_0$ for $T\geq 3$ via conditional likelihood  based on the insight that $(Y_{i0},\sum_{t=1}^{T-1} Y_{it},Y_{iT})$ are  sufficient statistics for the fixed effect. Our methodology is conceptually different as we seek to directly construct moment functions verifying equation (\ref{moment_function_final}). \\
\indent For what follows, it is helpful to remember that the individual-specific transition probability from state $l$ to state $k$ is time-invariant and given by:
\begin{align*}
    \pi^{k|l}(A_i)&=P(Y_{it+1}=k|Y_{it}=l,A_i)=\frac{e^{k(\gamma_0l+A_i)}}{1+e^{\gamma_0l+A_i}}, \quad \forall (l,k)\in \mathcal{Y} 
\end{align*}
\textbf{Step 1)}.  We shall begin by deriving the transition functions for $\pi^{0|0}(A_i)$ and $\pi^{1|1}(A_i)$. Observe that $\pi^{1|0}(A_i)$ and $\pi^{0|1}(A_i)$ are effectively redundant since probabilities sum to one. A natural starting place is to investigate the case $T=2$, i.e 2 periods of observations after the initial condition. Recalling definition (\ref{trans_functions}), we search for $\phi_{\theta}^{0|0}(Y_{i2},Y_{i1},Y_{i0})$, respectively $\phi_{\theta}^{1|1}(Y_{i2},Y_{i1},Y_{i0})$, whose conditional expectation given $(Y_{i0},A_i)$ yields $\pi^{0|0}(A_i)$, respectively $\pi^{1|1}(A_i)$. For the purposes of illustration and to show the kind of calculations arising broadly in DFEL models, 
let us derive $\phi_{\theta}^{0|0}(Y_{i2},Y_{i1},Y_{i0})$. By Bayes's rule:
\begin{multline*}
    \mathbb{E}\left[\phi_{\theta}^{0|0}(Y_{i2},Y_{i1},Y_{i0}) \,|\ Y_{i0}=y_0,A_i=a\right]\\
    =\cleansum_{y_{2}=0}^1\cleansum_{y_{1}=0}^1P(Y_{i2}=y_{2}|Y_{i1}=y_{1},A_i=a)P(Y_{i1}=y_{1}|Y_{i0}=y_{0},A_i=a)\phi_{\theta}^{0|0}(y_{2},y_{1},y_{0}) \\
     =\frac{e^{\gamma_0 y_{0}+a}}{1+e^{\gamma_0 y_{0}+a}}\left(\frac{e^{\gamma_0+a}}{1+e^{\gamma_0+a}}\phi_{\theta}^{0|0}(1,1,y_{0})+\frac{1}{1+e^{\gamma_0+a}}\phi_{\theta}^{0|0}(0,1,y_{0})\right)\\
     +\frac{1}{1+e^{\gamma_0 y_{0}+a}}\left(\frac{e^{a}}{1+e^{a}}\phi_{\theta}^{0|0}(1,0,y_{0})+\frac{1}{1+e^{a}}\phi_{\theta}^{0|0}(0,0,y_{0})\right)
\end{multline*}

where the second equality uses the logistic hypothesis. By quick inspection, we see that the terms in the first parenthesis have $(1+e^{\gamma_0+a})$ in their denominator unlike $\pi^{0|0}(A_i)$. Because $-e^{-\gamma_0}$ is not a \textit{pole} of $\pi^{0|0}(A_i)$\footnote{A \textit{pole} of a rational function is a root of its denominator. Formally, we are substituting $u=e^{a}$ and we are extending $\pi^{0|0}(u)$ to the real line.}, we conclude that $\phi_{\theta}^{0|0}(1,1,y_{0})=\phi_{\theta}^{0|0}(0,1,y_{0})=0$. This first deduction leaves us with
\begin{align*}
     &\mathbb{E}\left[\phi_{\theta}^{0|0}(Y_{i2},Y_{i1},Y_{i0}) \,|\ Y_{i0}=y_0,A_i=a\right]=
     \frac{1}{1+e^{\gamma_0 y_{0}+a}}\left(\frac{e^{a}}{1+e^{a}}\phi_{\theta}^{0|0}(1,0,y_{0})+\frac{1}{1+e^{a}}\phi_{\theta}^{0|0}(0,0,y_{0})\right)
\end{align*}
Now, since $\pi^{0|0}(A_i)$ does not depend on $y_0$, we must cancel the denominator $(1+e^{\gamma_0 y_{0}+a})$.  To achieve this, we must set: $\phi_{\theta_0}^{0|0}(1,0,y_{t-1})=C_0 e^{\gamma_0 y_{0}}, \phi_{\theta_0}^{0|0}(0,0,y_{t-1})=C_0$ for some constant $C_0\in \mathbb{R}\setminus\{0\}$. Then,
\begin{align*}
     \mathbb{E}\left[\phi_{\theta_0}^{0|0}(Y_{i2},Y_{i1},Y_{i0})|Y_{i0}=y_0,A_i=a\right]&=C_0\frac{1}{1+e^{a}}
\end{align*}
and $C_0=1$ is the appropriate normalization to obtain the desired transition function. Of course, the exact same logic applies for $\phi_{\theta_0}^{1|1}(Y_{i2},Y_{i1},Y_{i0})$ and $\pi^{1|1}(A_i)$. \\
\indent This short calculation provides a useful recipe for the general case $T\geq 2$. We learned that we can search for functions of three consecutive outcomes $\phi_{\theta}^{k|k}(Y_{it+1},Y_{it},Y_{it-1})$ such that:
\begin{align*}
    &\phi_{\theta}^{k|k}(Y_{it+1},Y_{it},Y_{it-1})=\mathds{1}\{Y_{it}=k\}\phi_{\theta}^{k|k}(Y_{it+1},k,Y_{it-1}) \\
     &\mathbb{E}\left[\phi_{\theta_0}^{k|k}(Y_{it+1},Y_{it},Y_{it-1}) \,|\ Y_{i0},Y_{i1}^{t-1},A_i\right]= \pi^{k|k}(A_i)
\end{align*}
The first restriction is a functional form that eliminates terms with inadequate \textit{poles} after taking expectations. The second restriction is a normalization condition to match the desired transition probability. Following this argument, we arrive at the expressions in Lemma \ref{lemma_1}.
\begin{lemma} \label{lemma_1} 
In model (\ref{AR1_logit_pure}) with $T\geq 2$ and $t\in\{1,\ldots,T-1\}$, let
\begin{align*}
   \phi_{\theta}^{0|0}(Y_{it+1},Y_{it},Y_{it-1})&=(1-Y_{it})e^{\gamma Y_{it+1}Y_{it-1}} \\
    \phi_{\theta}^{1|1}(Y_{it+1},Y_{it},Y_{it-1})&=Y_{it}e^{\gamma (1-Y_{it+1})(1-Y_{it-1})}
\end{align*}
Then:
\begin{align*}
    \mathbb{E}\left[\phi_{\theta_0}^{0|0}(Y_{it+1},Y_{it},Y_{it-1})|Y_{i0},Y_{i1}^{t-1},A_i\right]&=\pi^{0|0}(A_i)=\frac{1}{1+e^{A_i}} \\
    \mathbb{E}\left[\phi_{\theta_0}^{1|1}(Y_{it+1},Y_{it},Y_{it-1})|Y_{i0},Y_{i1}^{t-1},A_i\right]&=\pi^{1|1}(A_i)=\frac{e^{\gamma_0+A_i}}{1+e^{\gamma_0+A_i}}
\end{align*}
\end{lemma}
\vspace{0.5cm}
\begin{remark}[Connection to Kitazawa]
Interestingly, Lemma \ref{lemma_1} is a reformulation of results first shown by  \cite{kitazawa2013exploration,kitazawa2016root}, \cite{Kitazawa_JOE2021}, albeit with a very different logic than the calculations displayed above. We set out the connection between our respective approaches in Section \ref{Section_4_3} where we also discuss the case with exogenous regressors.
\end{remark}
\indent \textbf{Step 2)}. The second step in the agenda is the construction of valid moment functions.  Because the transition probability of the model are time-invariant, one trivial way to achieve this is to consider the pairwise difference of  $\phi_{\theta}^{k|k}(Y_{it+1},Y_{it},Y_{it-1})$ and $\phi_{\theta}^{k|k}(Y_{is+1},Y_{is},Y_{is-1})$ for any feasible $s\neq t$. This is the content of Proposition \ref{proposition_1}. We will need a minimum of four total periods of observations, which coincides with the requirements of the conditional likelihood approach. 

\begin{proposition} \label{proposition_1}
In model (\ref{AR1_logit_pure}) with $T\geq 3$, let
\begin{align*}
    \psi_{\theta}^{k|k}(Y_{it-1}^{t+1},Y_{is-1}^{s+1})&=\phi_{\theta}^{k|k}(Y_{it+1},Y_{it},Y_{it-1})-\phi_{\theta}^{k|k}(Y_{is+1},Y_{is},Y_{is-1})
\end{align*}
 for all $ k \in \mathcal{Y} $,  $t \in \{2,\ldots,T-1\}$ and  $s  \in \{1,\ldots,t-1\}$.
Then,
\begin{align*}
   \mathbb{E}\left[\psi_{\theta_0}^{k|k}(Y_{it-1}^{t+1},Y_{is-1}^{s+1})|Y_{i0},Y_{i1}^{s-1},A_i\right]&=0 
\end{align*}
\end{proposition}

\begin{remark}[Efficient GMM] \label{remark_2} Given that the conditional likelihood is semi-parametrically efficient for $T=3$ (\cite{gu2023information}, \cite{hahn2001information}), it is natural to ask whether the approach advocated here accounts for all the information in the model in that case. It turns out that it does. Specifically, letting $s_i^{c}(\theta)$ denote the conditional scores when $y_0=0$ as in \cite{hahn2001information}, we have:
\begin{align*}
    s_{i}^{c}(\gamma_0)&=\frac{1}{(1+e^{\gamma_0})(e^{-\gamma_0}-1)}\left(\psi_{\theta}^{0|0}(Y_{i1}^{3},Y_{i1}^{2},0)+\psi_{\theta}^{1|1}(Y_{i1}^{3},Y_{i1}^{2},0)\right)
\end{align*} 
where the right-hand side corresponds to the efficient moment for the moment restriction $\mathbb{E}\left[\psi_{\theta}(Y_{i1}^{3},Y_{i0}^{2})|Y_{i0}=0\right]=0$,  $\psi_{\theta}(Y_{i1}^{3},Y_{i1}^{2},0)=(\psi_{\theta}^{0|0}(Y_{i1}^{3},Y_{i1}^{2},0),\psi_{\theta}^{1|1}(Y_{i1}^{3},Y_{i1}^{2},0))'$. 
 
\end{remark}

\subsubsection{Construction of valid moment functions with strictly exogenous regressors} \label{Section_4_1_1_2}

In this subsection, we move on to the AR(1) logit model with strictly exogenous covariates characterized by equation (\ref{AR1_logit_general}). \\
\indent \textbf{Step 1)}. We employ the same shortcut recipe as in the \say{pure} case and
begin by looking for moment functions  $\phi_{\theta}^{0|0}(.)$ and $\phi_{\theta}^{1|1}(.)$ verifying:
\begin{align*}
    &\phi_{\theta}^{k|k}(Y_{it+1},Y_{it},Y_{it-1},X_i)=\mathds{1}\{Y_{it}=k\}\phi_{\theta}^{k|k}(Y_{it+1},k,Y_{it-1},X_i) \\
     &\mathbb{E}\left[\phi_{\theta_0}^{k|k}(Y_{it+1},Y_{it},Y_{it-1},X_i)|Y_{i0},Y_{i1}^{t-1},X_i,A_i\right]= \pi^{k|k}_{t}(A_i,X_{i}), \quad k\in\mathcal{Y}
\end{align*}
where this time
\begin{align*}
    \pi^{k|l}_{t}(A_i,X_{i})=P(Y_{it+1}=k|Y_{it}=l,X_i,A_i)=\frac{e^{k(\gamma_0l+X_{it+1}'\beta_0+A_i)}}{1+e^{\gamma_0l+X_{it+1}'\beta_0+A_i}}, \quad \forall (k,l)\in \mathcal{Y}^2
\end{align*}
 The same simple calculations described just above lead to the expressions in Lemma \ref{lemma_2}. The only (expected) change is the appearance of a new term $+/-\Delta X_{it+1}'\beta$ which accounts for the presence of covariates in the model.
\begin{lemma} \label{lemma_2} 
 In model (\ref{AR1_logit_general}) with $T\geq 2$ and $t\in\{1,\ldots,T-1\}$, let
\begin{align*}
   \phi_{\theta}^{0|0}(Y_{it+1},Y_{it},Y_{it-1},X_i)&=(1-Y_{it})e^{ Y_{it+1}\left(\gamma Y_{it-1}-\Delta X_{it+1}'\beta \right)} \\
    \phi_{\theta}^{1|1}(Y_{it+1},Y_{it},Y_{it-1},X_i)&=Y_{it}e^{ (1-Y_{it+1})\left(\gamma(1-Y_{it-1})+ \Delta X_{it+1}'\beta \right)}
\end{align*}
Then:
\begin{align*}
    \mathbb{E}\left[\phi_{\theta_0}^{0|0}(Y_{it+1},Y_{it},Y_{it-1},X_i)|Y_{i0},Y_{i1}^{t-1},X_i,A_i\right]&= \pi^{0|0}_{t}(A_i,X_{i})=\frac{1}{1+e^{A_i+X_{it+1}'\beta_0}} \\
    \mathbb{E}\left[\phi_{\theta_0}^{1|1}(Y_{it+1},Y_{it},Y_{it-1},X_i)|Y_{i0},Y_{i1}^{t-1},X_i,A_i\right]&=\pi^{1|1}_{t}(A_i,X_{i})=\frac{e^{\gamma_0+X_{it+1}'\beta_0+A_i}}{1+e^{\gamma_0+X_{it+1}'\beta_0+A_i}}
\end{align*}
\end{lemma}
\vspace{0.25cm}

\noindent At this point, it is important to highlight that unlike previously, the transition probabilities are covariate-dependent. The upshot is that the naive difference of $\phi_{\theta}^{k|k}(Y_{it+1},Y_{it},Y_{it-1},X_i)$ and $\phi_{\theta}^{k|k}(Y_{is+1},Y_{is},Y_{is-1},X_i)$ for  $s \neq t$ no longer leads  to valid moment functions in general. Indeed, while Lemma \ref{lemma_2} ensures that  \begin{align*}
    \mathbb{E}\left[\phi_{\theta}^{k|k}(Y_{it+1},Y_{it},Y_{it-1},X_i)-\phi_{\theta}^{k|k}(Y_{is+1},Y_{is},Y_{is-1},X_i)|Y_{i0},X_i,A_i\right]= \pi_{t}^{k|k}(A_i,X_i)-\pi_{s}^{k|k}(A_i,X_i)
\end{align*}
clearly, $\pi_{t}^{k|k}(A_i,X_i)-\pi_{s}^{k|k}(A_i,X_i)\neq 0$ when $X_{it+1}'\beta_0\neq X_{is+1}'\beta_0$  \footnote{A matching strategy in the spirit of \cite{honore2000panel} may still be applicable when in our example $X_{it+1}= X_{is+1}$. However, this is known to lead to estimators converging at rate less than $\sqrt{N}$ for continuous covariates and it rules out certain regressors such as time dummies and time trends.}. Thus, a different logic is required in the presence of explanatory variables other than a first order lag.  \\
\indent The key, as foreshadowed in Section \ref{Section_3} is that as soon as $T\geq 3$, it is possible to construct transition functions other than 
$\phi_{\theta}^{k|k}(Y_{it-1}^{t+1},X_i)$ also mapping to $\pi^{k|k}_{t}(A_i,X_{i})$ in time periods $t\in\{2,\ldots,T-1\}$.  These new transition functions that we denote $\zeta_{\theta}^{k|k}(.)$ to emphasize their difference have a particular form. They consist of a weighted combination of past outcome $\mathds{1}(Y_{is}=k)$,  $1 \leq s<t$,  and the interaction of $\mathds{1}(Y_{is}\neq k)$ with any transition function  associated to $\pi^{k|k}_{t}(A_i,X_{i})$ having no dependence on outcomes prior to period $s$, e.g $\phi_{\theta}^{k|k}(Y_{it-1}^{t+1},X_i)$. This property follows from a \textit{partial fraction decomposition} presented in Lemma \ref{tech_lemma_1} that exploits the structure of the model probabilities under the logistic assumption. It relates to the hyperbolic transformations ideas of \cite{Kitazawa_JOE2021}.  In the sequel, we shall see that this insight carries over to the AR($p$) logit model with $p>1$. Lemma \ref{lemma_3} below gives the \say{simplest} additional transition functions that one can construct when $T\geq 3$ for the AR(1) model with exogenous regressors (the only ones when $T=3$).
\begin{lemma} \label{lemma_3}
In model (\ref{AR1_logit_general}) with $T\geq 3$,  for all $t,s$ such that $T-1\geq t> s\geq 1$, let: 
\begin{align*}
    \mu_{s}(\theta)&=\gamma Y_{is-1}+X_{is}'\beta \\
    \kappa_{t}^{0|0}(\theta)&=X_{it+1}'\beta, \quad \kappa_{t}^{1|1}(\theta)=\gamma+X_{it+1}'\beta \\
    \omega_{t,s}^{0|0}(\theta)&=1-e^{(\kappa_{t}^{0|0}(\theta)-\mu_{s}(\theta))}, \quad \omega_{t,s}^{1|1}(\theta)=1-e^{-(\kappa_{t}^{1|1}(\theta)-\mu_{s}(\theta))}
\end{align*}
and define the moment functions:
\begin{align*}
    \zeta_{\theta}^{0|0}(Y_{it-1}^{t+1},Y_{is-1}^s,X_i)&=(1-Y_{is})+\omega_{t,s}^{0|0}(\theta)Y_{is}\phi_{\theta}^{0|0}(Y_{it+1},Y_{it},Y_{it-1},X_i) \\
    \zeta_{\theta}^{1|1}(Y_{it-1}^{t+1},Y_{is-1}^s,X_i)&=Y_{is}+\omega_{t,s}^{1|1}(\theta)(1-Y_{is})\phi_{\theta}^{1|1}(Y_{it+1},Y_{it},Y_{it-1},X_i)
\end{align*}
Then,
\begin{align*}
    \mathbb{E}\left[ \zeta_{\theta_0}^{0|0}(Y_{it-1}^{t+1},Y_{is-1}^s,X_i)|Y_{i0},Y_{i1}^{s-1},X_i,A_i\right]&= \pi^{0|0}_{t}(A_i,X_{i})=\frac{1}{1+e^{X_{it+1}'\beta_0+A_i}} \\
    \mathbb{E}\left[\zeta_{\theta_0}^{1|1}(Y_{it-1}^{t+1},Y_{is-1}^s,X_i)|Y_{i0},Y_{i1}^{s-1},X_i,A_i\right]&=\pi^{1|1}_{t}(A_i,X_{i})=\frac{e^{\gamma_0+X_{it+1}'\beta_0+A_i}}{1+e^{\gamma_0+X_{it+1}'\beta_0+A_i}}
\end{align*}
\end{lemma}
\noindent When $T\geq 4$, it turns out that we can build even more transition functions from those given in Lemma \ref{lemma_3} by repeating the same type of logic based on \textit{partial fraction expansions}; Corollary \ref{corollary_2} provides a recursive formulation.
\begin{corollary}  \label{corollary_2}
In model (\ref{AR1_logit_general}) with $T\geq 4$, for any $t$ and ordered collection of indices $s_1^J$, $J\geq 2$,  satisfying $T-1\geq t> s_1>\ldots>s_J\geq 1$, let
\begin{align*}
    \zeta_{\theta}^{0|0}(Y_{it-1}^{t+1},Y_{is_1-1}^{s_1},\ldots,Y_{is_{J}-1}^{s_J},X_i)&=(1-Y_{is_J})+\omega_{t,s_J}^{0|0}(\theta)Y_{is_J} \zeta_{\theta}^{0|0}(Y_{it-1}^{t+1},Y_{is_1-1}^{s_1},\ldots,Y_{is_{J-1}-1}^{s_{J-1}},X_i) \\
    \zeta_{\theta}^{1|1}(Y_{it-1}^{t+1},Y_{is_1-1}^{s_1},\ldots,Y_{is_{J}-1}^{s_J},X_i)&=Y_{is_J}+\omega_{t,s_J}^{1|1}(\theta)(1-Y_{is_J}) \zeta_{\theta}^{1|1}(Y_{it-1}^{t+1},Y_{is_1-1}^{s_1},\ldots,Y_{is_{J-1}-1}^{s_{J-1}},X_i)
\end{align*}
with weights $\omega_{t,s_J}^{0|0}(\theta), \omega_{t,s_J}^{1|1}(\theta)$ defined as in Lemma \ref{lemma_3}. 
Then,
\begin{align*}
    \mathbb{E}\left[ \zeta_{\theta_0}^{k|k}(Y_{it-1}^{t+1},Y_{is_1-1}^{s_1},\ldots,Y_{is_{J}-1}^{s_J},X_i)|Y_{i0},Y_{i1}^{s_J-1},X_i,A_i\right]&= \pi^{k|k}_{t}(A_i,X_{i}),  \quad \forall k \in \mathcal{Y}
\end{align*}
\end{corollary}
\vspace{0.5cm}
\textbf{Step 2)}. Provided $T\geq 3$, the difference between any transition functions associated to the same transition probabilities in periods $t\in\{2,\ldots,T-1\}$ constitutes a valid candidate for (\ref{moment_function_final}). One particularly relevant set of valid moment functions for reasons explained below is presented in Proposition \ref{proposition_2}.
\begin{proposition}\label{proposition_2}
In model (\ref{AR1_logit_general}), for all $k \in \mathcal{Y}$, \\
if $T\geq 3$, for all $t,s$ such that $T-1\geq t> s\geq 1$ , let
\begin{align*}
    \psi_{\theta}^{k|k}(Y_{it-1}^{t+1},Y_{is-1}^s,X_i)&=\phi_{\theta}^{k|k}(Y_{it-1}^{t+1},X_i)-\zeta_{\theta}^{k|k}(Y_{it-1}^{t+1},Y_{is-1}^s,X_i), 
\end{align*}
if $T\geq 4$, for any $t$ and ordered collection of indices $s_1^J$, $J \geq 2$, satisfying $T-1\geq t> s_1>\ldots>s_J\geq 1$, let
\begin{align*}
    \psi_{\theta}^{k|k}(Y_{it-1}^{t+1},Y_{is_1-1}^{s_1},\ldots,Y_{is_{J}-1}^{s_J},X_i)&=\phi_{\theta}^{k|k}(Y_{it-1}^{t+1},X_i)-\zeta_{\theta}^{k|k}(Y_{it-1}^{t+1},Y_{is_1-1}^{s_1},\ldots,Y_{is_{J}-1}^{s_J},X_i), 
\end{align*}

Then,
\begin{align*}
   &\mathbb{E}\left[\psi_{\theta_0}^{k|k}(Y_{it-1}^{t+1},Y_{is-1}^s,X_i)|Y_{i0},Y_{i1}^{s-1},X_i,A_i\right]=0 \\
    &\mathbb{E}\left[\psi_{\theta_0}^{k|k}(Y_{it-1}^{t+1},Y_{is_1-1}^{s_1},\ldots,Y_{is_{J}-1}^{s_J},X_i)|Y_{i0},Y_{i1}^{s_J-1},X_i,A_i\right]=0
\end{align*}
\end{proposition}

\indent This family of moment functions has cardinality $2^T-2T$ which by Theorem \ref{theorem_nummoments_AR1} is precisely the number of linearly independent moment conditions available for the AR(1). To see this, notice that for fixed $(k,Y_{i0})\in \mathcal{Y}^2$,  and a given time period $ t\in\{ 2,\ldots,T-1\}$, Proposition \ref{proposition_2} gives a total of:
\begin{align*}
    \sum_{l=1}^{t-1} \binom{t-1}{l}=2^{t-1}-1
\end{align*}
valid moment functions. This follows from a simple counting argument. First, we get $\binom{t-1}{1}$ possibilities from choosing any $s$ in $\{1,\ldots,t-1\}$ to form $\psi_{\theta}^{k|k}(Y_{it-1}^{t+1},Y_{is-1}^s,X_i)$. To that, we must add another $\sum_{l=2}^{t-1} \binom{t-1}{l}$ possibilities from choosing all feasible sequences $s_1^J$ with $t-1\geq s_1>s_2>\ldots>s_J\geq 1$ to form  $\psi_{\theta}^{k|k}(Y_{it-1}^{t+1},Y_{is_1-1}^{s_1},\ldots,Y_{is_{J}-1}^{s_J},X_i)$. Summing over $t=2,\ldots, T-1$ and multiplying by 2 to account for the two possible values for $k$ delivers the result:
\begin{align*}
    2 \times \sum_{t=2}^{T-1}\sum_{l=1}^{t-1}\binom{t-1}{l}= 2 \times \sum_{t=2}^{T-1}(2^{t-1}-1)=2^T-2T
\end{align*}
Furthermore, there is evidence that the family is linearly independent. It is readily verified for $T=3$ since the two valid moment functions produced by the model depend on two distinct sets of choice histories. This can be seen from their unpacked expressions in equations (\ref{moment_b_honoreweidner}) and (\ref{moment_a_honoreweidner}) in the Appendix. Unfortunately, this argument does not carry over to longer panels but we have verified numerically that the linear independence property of this family continues to hold for several different values of $T\geq 4$. This suggests that our approach delivers all the \textit{moment equality} restrictions available in the AR(1) model with $T$ periods post initial condition \footnote{This is not all the identifying content of the AR(1) specification since we know from \cite{dobronyi2021identification} that the model also implies moment inequality conditions.}.
\begin{remark}[Symmetry] \label{remark_symmetry}
    The transition functions and valid moment functions of the AR(1) model share a special symmetry property. Indeed, by inspection the transition functions of Lemma \ref{lemma_2} verify
    \begin{align*}
    \phi_{\theta}^{0|0}(1-Y_{it+1},1-Y_{it},1-Y_{it-1} ,-X_i)=\phi_{\theta}^{1|1}(Y_{it+1},Y_{it},Y_{it-1},X_i)
    \end{align*}
It is not difficult to see that this symmetry, i.e substituting $Y_{it}$ by $(1-Y_{it})$ and $X_{it}$ by $-X_{it}$ to obtain $\phi_{\theta}^{1|1}(Y_{it-1}^{t+1},X_i)$ from $\phi_{\theta}^{0|0}(Y_{it-1}^{t+1},X_i)$  transfers to the other transition functions of Lemma \ref{lemma_3}, Corollary \ref{corollary_2} and ultimately to the valid moment functions of Proposition \ref{proposition_2}. This symmetry can be useful for computational purposes.
\end{remark}

\begin{remark}[Static logit]\label{static_logit_moments}
If $\gamma_0=0$, model (\ref{AR1_logit_general}) specializes to the static panel logit model of \cite{rasch1960studies} and our two-step approach is still applicable.  For that case, Lemma \ref{lemma_2} gives two moment functions for $T=2$:
\begin{align*}
   &\phi_{\theta}^{0|0}(Y_{i2},Y_{i1},X_{i})=(1-Y_{1})e^{ -Y_{i2}\Delta X_{2}'\beta} \\
    &\phi_{\theta}^{1|1}(Y_{i2},Y_{i1},X_{i})=Y_{i1}e^{ (1-Y_{2})\Delta X_{i2}'\beta}
\end{align*}
such that $\mathbb{E}\left[\phi_{\theta_0}^{0|0}(Y_{i1}^2,,X_{i})|X_i,A_i\right]=\frac{1}{1+e^{X_{i2}'\beta_0+A_i}}$ and $\mathbb{E}\left[\phi_{\theta_0}^{1|1}(Y_{i1}^2,X_{i})|X_i,A_i\right]=\frac{e^{X_{i2}'\beta_0+A_i}}{1+e^{X_{i2}'\beta_0+A_i}}$. It follows that a valid moment function with two periods of observation is 
\begin{align*}
    \psi_{\theta}(Y_{i2},Y_{i1},X_{i})&=\phi_{\theta}^{1|1}(Y_{i2},Y_{i1},X_{i})-(1-\phi_{\theta}^{0|0}(Y_{i2},Y_{i1},X_{i})) \\
    &=(1-e^{-\Delta X_{i2}'\beta})\left(Y_{i1}(1-Y_{i2})e^{\Delta X_{i2}'\beta}-(1-Y_{i1})Y_{i2}\right)
\end{align*}
which is proportional to the score of the conditional likelihood based on the sufficient statistic $Y_{i1}+Y_{i2}$ (\cite{rasch1960studies}, \cite{andersen1970asymptotic}, \cite{chamberlain1980analysis}).
\end{remark}

\subsection{Semiparametric efficiency bound for the AR(1) with regressors}
\cite{honore2020moment} gave sufficient conditions to identify $\theta_0=(\gamma_0,\beta_0')'$ in the AR(1) model with $T=3$. A natural follow-up question is to ask how accurately can $\theta_0$ be estimated in that case, or equivalently what is the semi-parametric information bound. In a corrigendum to \cite{hahn2001information}, \cite{gu2023information} confirmed that the conditional likelihood estimator is semiparametrically efficient for $T=3$ in the \say{pure} AR(1) model.  However, the characterization of the semiparametric efficiency bound and the question of what estimator attains it remain unclear with covariates. \\
\indent To answer these questions, let $\psi_{\theta}(Y_{i1}^{3},Y_{i0}^1,X_i)=(\psi_{\theta}^{0|0}(Y_{i1}^{3},Y_{i0}^1,X_i),\psi_{\theta}^{1|1}(Y_{i1}^{3},Y_{i0}^1,X_i))'$ where the two components correspond to the valid moment functions of Proposition \ref{proposition_2} for $T=3$. Additionally, let $D(X_i,y_0)=\mathbb{E}\left[\pdv{ \psi_{\theta_0}(Y_{i1}^{3},Y_{i0}^1,X_i)}{\theta'}|Y_{i0}=y_0,X_i\right]$ and let $\Sigma(X_i,y_0)=\mathbb{E}\left[\psi_{\theta_0}(Y_{i1}^{3},Y_{i0}^1,X_i)\psi_{\theta_0}(Y_{i1}^{3},Y_{i0}^1,X_i)'|Y_{i0}=y_0,X_i\right]$.

\begin{assumption}\label{assumption_effbound}
   In model (\ref{AR1_logit_general}) with $T=3$ and initial condition $y_0\in \{0,1\}$, the matrix $\mathbb{E}\left[D(X_i,y_0)\Sigma(X_i,y_0)^{-1}D(X_i,y_0)'|Y_{i0}=y_0\right]$ exists and is nonsingular.
\end{assumption}
\noindent With these notations in hand and under the mild conditions of Assumption  \ref{assumption_effbound}, Theorem \ref{thm_effbound} clarifies that the \textit{efficient score} coincides with the efficient moment for the conditional moment problem: $\mathbb{E}\left[\psi_{\theta}(Y_{i1}^{3},Y_{i0}^1,X_i)|Y_{i0}=y_0,X_i\right]=0$. Put differently, the maximal efficiency with which $\theta_0$ can be estimated is $V_{0}(y_0)=\mathbb{E}[D(X_i,y_0)\Sigma(X_i,y_0)^{-1}D(X_i,y_0)'|Y_{i0}=y_0]^{-1}$. This result is in accordance with Remark \ref{remark_2} which noted that the score of the conditional likelihood without covariates is precisely the efficient moment implied by our conditional moment restrictions in this case. 

\begin{theorem}\label{thm_effbound}
   Consider model (\ref{AR1_logit_general}) with $T=3$. Fix an initial condition $y_0\in \{0,1\}$ and suppose that Assumption \ref{assumption_effbound} holds. Then, the semiparametric efficiency bound of $\theta_0$ is finite and given by $V_0(y_0) = \mathbb{E}[D(X_i,y_0)\Sigma(X_i,y_0)^{-1}D(X_i,y_0)'|Y_{i0}=y_0]^{-1}$. 
\end{theorem}
\noindent The proof of Theorem \ref{thm_effbound} only involves careful bookkeeping of some tedious algebra and an application of Theorem 3.2 in \cite{newey1990semiparametric}. Interestingly, \cite{davezies2023fixed} presented analogous results in the static panel data case with three periods of observations. 

\subsection{Connections to other works on the AR(1) logit model} \label{Section_4_3}
As indicated previously, there is a connection between our methodology and that of \cite{Kitazawa_JOE2021} for the AR(1) model. Indeed, after some algebraic manipulation, we can re-express the transition functions of Lemma \ref{lemma_2} (or Lemma \ref{lemma_1} without covariates) as:
\begin{align*}
\phi_{\theta}^{0|0}(Y_{it-1}^{t+1},X_i)&= 1-Y_{it}-(1-Y_{it})Y_{it+1}+(1-Y_{it})Y_{it+1} e^{- \Delta X_{it+1}'\beta}+\delta Y_{it-1}(1-Y_{it+1})Y_{it+1}e^{- \Delta X_{it+1}'\beta}\\
\phi_{\theta}^{1|1}(Y_{it-1}^{t+1},X_i)&= Y_{it}Y_{it+1}+ Y_{it}(1-Y_{it+1})e^{\Delta X_{it+1}'\beta}+\delta (1-Y_{it-1})Y_{it}(1-Y_{it+1})e^{\Delta X_{it+1}'\beta}
\end{align*}
where $\delta=(e^{\gamma}-1)$. Thus, the 
moment conditions of Lemma \ref{lemma_2} imply that we can write:
\begin{align*}
    &Y_{it}+(1-Y_{it})Y_{it+1}-(1-Y_{it})Y_{it+1} e^{- \Delta X_{it+1}'\beta_0}-\delta_0 Y_{it-1}(1-Y_{it+1})Y_{it+1}e^{- \Delta X_{it+1}'\beta_0}=\frac{e^{X_{it+1}'\beta_0+A_i}}{1+e^{X_{it+1}'\beta_0+A_i}}+\epsilon_{it}^{0|0} \\
    &Y_{it}Y_{it+1}+ Y_{it}(1-Y_{it+1})e^{\Delta X_{it+1}'\beta_0}+\delta_0 (1-Y_{it-1})Y_{it}(1-Y_{it+1})e^{\Delta X_{it+1}'\beta_0}=\frac{e^{\gamma_0+X_{it+1}'\beta_0+A_i}}{1+e^{\gamma_0+X_{it+1}'\beta_0+A_i}}+\epsilon_{it}^{1|1}
\end{align*}
where $\mathbb{E}\left[\epsilon_{it}^{0|0}|Y_{i0},Y_{i1}^{t-1},X_i,A_i\right]=0$ and $\mathbb{E}\left[\epsilon_{it}^{1|1}|Y_{i0},Y_{i1}^{t-1},X_i,A_i\right]=0$. These expressions are the so-called \textit{h-form} and \textit{g-form} of \cite{Kitazawa_JOE2021} for model (\ref{AR1_logit_general}) and were originally obtained through an ingenious usage of the mathematical properties of the hyperbolic tangent function. The evident connection between the transition functions and the \textit{h-form} and \textit{g-form} offers an interesting new perspective on the transformation approach of \cite{Kitazawa_JOE2021} for the AR(1) model. If we further define 
\begin{align*}
    &U_{it}=Y_{it}+(1-Y_{it})Y_{it+1}-(1-Y_{it})Y_{it+1} e^{- \Delta X_{it+1}'\beta}-\delta Y_{it-1}(1-Y_{it+1})Y_{it+1})e^{- \Delta X_{it+1}'\beta} \\
    &\Upsilon_{it}=Y_{it}Y_{it+1}+ Y_{it}(1-Y_{it+1})e^{\Delta X_{it+1}'\beta}+\delta (1-Y_{it-1})Y_{it}(1-Y_{it+1})e^{\Delta X_{it+1}'\beta}
\end{align*}
 the two moment functions of \cite{Kitazawa_JOE2021} for the AR(1) model write
\begin{align*}
    &\hbar U_{it}=  U_{it}-Y_{it-1}-\tanh \left( \frac{-\gamma Y_{it-2}+(\Delta X_{it}+\Delta X_{it+1})'\beta}{2}\right)\left(U_{it}+Y_{it-1}-2U_{it}Y_{it-1}\right)\\
    &\hbar \Upsilon_{it}= \Upsilon_{it}-Y_{it-1}-\tanh \left( \frac{\gamma(1-Y_{it-2})+(\Delta X_{it}+\Delta X_{it+1})'\beta}{2}\right)\left(\Upsilon_{it}+Y_{it-1}-2\Upsilon_{it}Y_{it-1}\right)
\end{align*}
which can be formulated in terms of our own moment functions as
\begin{align*}
    \hbar U_{it}&=-\frac{2}{2-\omega_{t,t-1}^{0|0}(\theta)} \psi_{\theta}^{0|0}(Y_{it-1}^{t+1},Y_{it-2}^{t-1},X_i) \\
    \hbar \Upsilon_{it}&=\frac{2}{2-\omega_{t,t-1}^{1|1}(\theta)} \psi_{\theta}^{1|1}(Y_{it-1}^{t+1},Y_{it-2}^{t-1},X_i)
\end{align*}
Appendix Section \ref{mapping_to_kitazawa_honoreweidner} provides detailed derivations for the mapping between our two approaches. This last result indicates that our moment conditions essentially match those of \cite{Kitazawa_JOE2021} when $T=3$. 
However, for $T\geq 4$, Proposition \ref{proposition_2} imply that there are further identifying moments than those based solely on $\hbar U_{it}$ and $\hbar \Upsilon_{it}$ for the AR(1) model. Interestingly, it turns out as we demonstrate in Appendix Section \ref{mapping_to_kitazawa_honoreweidner} that our moment functions coincide exactly with those derived by \cite{honore2020moment} for the special case $T=3$. \\ 
\indent To the best of our knowledge, besides the AR(1) model and a few specific examples, the structure of moment conditions in models with arbitrary lag order is not fully understood in the literature. Building on \cite{Bonhomme_EM12}, \cite{honore2020moment} propose moment functions for the AR(2) model up to $T=4$ and the AR(3) model with $T=5$ but no results are offered beyond these special instances. Yet, this is of general interest not only to better understand the properties of DFEL models but also for practical modelling and estimation purposes. For example, \cite{card2005estimating} argue in favor of using higher order logit specifications to better fit the behavior of a control group  in the context of a welfare experiment.  Relatedly, there are few results available for multivariate fixed effect models and existing methods developed for the scalar case are likely to be difficult to adapt in practice due to computational barriers.  In the remaining sections, we show that our two-step approach addresses these issues by providing closed form expressions for the moment equality conditions of these more complex models.

\subsection{Moment restrictions for the AR(\texorpdfstring{$p$}{}) logit model, \texorpdfstring{$p>1$}{}} \label{Section_ARp}
Allowing for more than one lag is often desirable in empirical work to model persistent stochastic processes and to better fit the data (e.g, \cite{magnac2000subsidised} on labour market histories, \cite{chay1999non} and \cite{card2005estimating} on welfare recipiency).  To this end, we now discuss how to extend our identification scheme to general univariate autoregressive models. We consider
\begin{align}\label{ARp_logit_general}
    Y_{it}=\mathds{1}\left\{\sum_{r=1}^{p}\gamma_{0r} Y_{it-r}+X_{it}'\beta_0+A_i-\epsilon_{it}\geq 0\right\}, \quad t= 1,\ldots, T
\end{align}
for known autoregressive order $p>1$ and vector of initial values $Y_{i}^{0}=(Y_{i-(p-1)},\ldots,Y_{i-1},Y_{i0})'\in \mathcal{Y}^{p}$, with $A_i\in \mathbb{R}$. Here, we let $\theta_0=(\gamma_0',\beta_0')'\in\mathbb{R}^{p+K_{x}}$. The corresponding transition probabilities are:
\begin{align*}
     \pi^{k|l_{1}^p}_{t}(A_i,X_{i})&=P(Y_{it+1}=k|Y_{it}=l_1,\ldots,Y_{it-(p-1)}=l_p,X_i,A_i)=\frac{e^{k(\sum_{r=1}^{p} \gamma_{0r} l_r+X_{it+1}'\beta_0+A_i)}}{1+e^{\sum_{r=1}^{p} \gamma_{0r} l_r+X_{it+1}'\beta_0+A_i}}
\end{align*}
and there will be moment restrictions attached to each of the $2^p$ (non-redundant) transition probabilities. Before detailing the specifics of their construction, we  enumerate the moment restrictions for this model as we did for the AR(1). This provides a way to ensure that we are not leaving any information on the table.

\subsubsection{Impossibility results and number of moment restrictions when $p\geq 1$ } \label{Section_ARp_impossibility}
Based on simulation evidence, \cite{honore2020moment} conjectured that AR($p$) models possess $2^T-(T+p-1)2^p$ linearly independent moment conditions in panels of sufficient length. We prove this claim in Theorem \ref{theorem_nummoments_ARp} and establish that no moment restrictions for the common parameters exist when $T\leq p+1$; that is with less than $2p+1$ periods of observations per individual. To introduce the result formally, 
it is again convenient to consider the conditional expectation operator mapping functions of histories $Y_i$ to their conditional expectation given $Y_{i}^0=y^0,X_{i}=x$ and the fixed effect, i.e
\begin{align*}
  \mathcal{E}_{y^0,x}^{(p)}\colon \mathbb{R}^{\mathcal{Y}^T}& \longrightarrow \mathbb{R}^\mathbb{R} \\[-1ex]
  \phi(.,y^0,x) & \longmapsto \mathbb{E}\left[\phi(Y_i,y^0,x)|Y_{i}^0=y^0,X_{i}=x,A_{i}=.\right]
\end{align*}
so that for any $y\in \mathcal{Y}^T$, $ \mathcal{E}_{y^0,x}^{(p)}\left[\mathds{1}\{.=y\}\right]$ yields the conditional likelihood of history $y$ for all possible values of $A_i$ in the AR($p$) model. That is,
\begin{align*}
   \mathcal{E}_{y^0,x}^{(p)}\left[\mathds{1}\{.=y\}\right]&=P(Y_i=y|Y_{i}^0=y^0,X_{i}=x,A_{i}=.)=a \mapsto \prod_{t=1}^T \frac{e^{y_t\left(\sum_{r=1}^{p}\gamma_{0r}y_{t-r}+x_t'\beta_0+a\right)}}{1+e^{\sum_{r=1}^{p}\gamma_{0r}y_{t-r}+x_t'\beta_0+a}}
\end{align*}
 Then the following result holds:
\begin{theorem} \label{theorem_nummoments_ARp}
     Consider model (\ref{ARp_logit_general}) with $T\geq 1$ and initial condition $y^0\in \mathcal{Y}^p$. Suppose that for any   $t,s\in \{1,\ldots,T-1\}$ and $y,\tilde{y}\in \mathcal{Y}^{p}$, $\gamma_0'y+x_{t}'\beta_0\neq \gamma_0'\tilde{y}+x_{s}'\beta_0$ if $t\neq s$ or $y\neq\tilde{y}$. Then, the family 
\begin{align*}
    \mathcal{F}_{y^0,p,T}=\left\{1,\pi_{0}^{y_0|y^0}(.,x),\left\{\left(\pi_{t-1}^{y_{1}|y_{1}^{t-1},y_{0},\ldots,y_{-(p-t)}}(.,x))\right)_{y_{1}^{t-1}\in \mathcal{Y}^{t-1}}\right\}_{t=2}^p,\left\{\left(\pi_{t-1}^{y_1|y_{1}^{p}}(.,x)\right)_{y_{1}^p\in \mathcal{Y}^p}\right\}_{t=p+1}^T\right\}
\end{align*}
forms a basis of $\Ima\left(\mathcal{E}_{y^0,x}^{(p)}\right)$ and therefore
\begin{enumerate}
    \item If $T\leq p+1$, $\rank\left(\mathcal{E}_{y^0,x}^{(p)}\right)=2^T$ and $\dim\left(\ker\left(\mathcal{E}_{y^0,x}^{(p)}\right)\right)=0$
    \item If $T\geq p+2$, $\rank\left(\mathcal{E}_{y^0,x}^{(p)}\right)=(T-p+1)2^{p}$ and $\dim\left(\ker\left(\mathcal{E}_{y^0,x}^{(p)}\right)\right)=2^T-(T-p+1)2^{p}$
\end{enumerate}
\end{theorem}
 \noindent Theorem \ref{theorem_nummoments_ARp} generalizes Theorem \ref{theorem_nummoments_AR1} for AR($p$) logit models with $p>1$. It confirms the basic intuition that all the parametric content lies in the transition probabilities, no matter the lag order. Specifically,  the conditional probabilities of all choice histories are spanned by the transition probabilities. In the basis $\mathcal{F}_{y^0,p,T}$, elements $\pi_{0}^{y_0|y^0}(.,x)$ and $\left\{\left(\pi_{t-1}^{y_{1}|y_{1}^{t-1},y_{0},\ldots,y_{-(p-t)}}(.,x))\right)_{y_{1}^{t-1}\in \mathcal{Y}^{t-1}}\right\}_{t=2}^p$ correspond to transition probabilities that are affected by the initial condition $y^0$. In the AR(1) case, it reduces to  $\pi_{0}^{y_0|y_0}(.,x)$ (see Theorem \ref{theorem_nummoments_AR1}). The remaining basis elements are free from the initial condition and correspond to the collection of all transition probabilities in each period starting from $t=p$. \\
 \indent Theorem \ref{theorem_nummoments_ARp} is an implication of \textit{partial fraction decompositions}  and of the fact that the transition probabilities of AR($p$) models admit transition functions. This property is set out in the following section. If $T\leq p+1$, $\mathcal{E}_{y^0,x}^{(p)}$ is injective and no non-trivial moment conditions can be found. Beyond this threshold, the \textit{rank nullity theorem} which connects image and nullspace of linear maps tells us that $2^T-(T-p+1)2^{p}$ moment restrictions exist. Under weaker conditions on the parameters or regressors then those of the theorem,  the model may admit additional moment conditions even with $T\leq p+1$.

\subsubsection{Construction of transition probabilities with $p>1$}\label{Section_ARp_transifunc}
Having clarified that $T=p+2$ is the minimum number of periods required for the existence of identifying moments, we are now ready to address the issue of their construction. The blueprint generalizes that of the AR(1) model and can be summarized as follows:
\begin{enumerate}
    \item \textbf{Step 1)}
    \begin{enumerate}
        \item Start by obtaining analytical expressions of the unique transition functions for the transition probability in period $t=p$ when $T=p+1$ \footnote{The fact that the transition functions in period $t=p$ are unique when $T=p+1$ is a direct corollary of Theorem \ref{theorem_nummoments_ARp}. Otherwise, the difference of two distinct transition functions mapping to the same transition probability would yield a valid moment which is a contradiction.}. Shift these expressions by one period, two periods, three periods etc to get a set of transition functions for period $t\in\{p+1,\ldots,T-1\}$ when $T\geq p+2$.
        \item Apply \textit{partial fraction decompositions} to the expressions obtained in (a) for $t\in\{p+1,\ldots,T-1\}$ to generate other transition functions mapping to the same transition probabilities.
    \end{enumerate}
    \item \textbf{Step 2)}. Take \say{adequate} differences of transition functions associated to the same transition probability in periods $t\in\{p+1,\ldots,T-1\}$ to obtain valid moments that are linearly independent.
\end{enumerate}
\textbf{Step 1)} (a) is akin to how we started by getting closed form expressions for the transition functions in period $t=1$ for $T=2$ in the one lag case and then deducted a general principle for $t\geq 2$ (see Section \ref{Section_4}). From a technical perspective, this is the only part of the two-step procedure that differs from the baseline AR(1). Indeed,  \textbf{Step 2)} is fundamentally identical and \textbf{Step 1)} (b) is also unchanged for the simple reason that the transition probabilities keep the same functional form as before. That is, a logistic transformation of a linear index composed of common parameters, the regressors and the fixed effect only. Hence, the same \textit{partial fraction expansions} apply. In light of those close similarities with the AR(1) and in order to focus on the primary issues, we defer a discussion of \textbf{Step 1)}(b) and  \textbf{Step 2)} to Appendix Section \ref{Section_ARp_othersteps}. \\
\indent Theorem \ref{theorem_ARp_transit} provides the algorithm to compute the transition functions for \textbf{Step 1)} (a) for arbitrary lag order greater than one. It is based on the insight that we can leverage the transition functions of an AR($p-1$) and  \textit{partial fraction decompositions} to generate the transition functions of an AR($p$).
A simple example is helpful to illustrate those ideas. Consider  an AR(2) with $T=3$ (i.e 5 observations in total) and suppose that we seek a transition function associated to, say, the transition probability
\begin{align*}
    \pi^{0|0,1}_{2}(A_i,X_{i})=\frac{1}{1+e^{\gamma_{02}+X_{i3}'\beta_0+A_i}}
\end{align*}
The first ingredient of the theorem is to view the AR(2) model as an AR(1) model where we treat the second order lag as an additional strictly exogenous regressor. This change of perspective is advantageous since we already know how to deal with the single lag case. In particular, Lemma \ref{lemma_2} readily gives the transition function $\phi_{\theta_0}^{0|0}(Y_{i3},Y_{i2},Y_{i1},Y_{i0},X_i)$ for the transition probability $\pi^{0|0,Y_{i1}}_{2}(A_i,X_{i})=P(Y_{i3}=0|Y_{i2}=0,Y_{i1},X_i,A_i)$ in the sense that it verifies:
\begin{align*}
    \mathbb{E}\left[\phi_{\theta_0}^{0|0}(Y_{i3},Y_{i2},Y_{i1},Y_{i0},X_i)|Y_i^0,Y_{i1},X_i,A_i\right]&= \pi^{0|0,Y_{i1}}_{2}(A_i,X_{i})
\end{align*}
This is an intermediate stage since $\phi_{\theta_0}^{0|0}(Y_{i3},Y_{i2},Y_{i1},Y_{i0},X_i)$ does not quite map to the target of interest; indeed $\pi^{0|0,Y_{i1}}_{2}(A_i,X_{i})$ depends on the random variable $Y_{i1}$ unlike  $\pi^{0|0,1}_{2}(A_i,X_{i})$. To make further progress, one would intuitively need to \say{set} $Y_{i1}$ to unity to make the two transition probabilities coincide. We operationalize this idea by interacting $\phi_{\theta_0}^{0|0}(Y_{i3},Y_{i2},Y_{i1},Y_{i0},X_i)$ and $Y_{i1}$ to achieve the desired effect in expectation:
\begin{align*}
        \mathbb{E}\left[Y_{i1}\phi_{\theta_0}^{0|0}(Y_{i3},Y_{i2},Y_{i1},Y_{i0},X_i)|Y_i^0,X_i,A_i\right]&=\mathbb{E}\left[Y_{i1}\pi^{0|0,1}_{2}(A_i,X_{i})|Y_i^0,X_i,A_i\right] \\
        &=\frac{1}{1+e^{\gamma_{02}+X_{i3}'\beta+A_i}}\frac{e^{\gamma_{01}Y_{i0}+\gamma_{02}Y_{i-1}+X_{i1}'\beta_0+A_i}}{1+e^{\gamma_{01}Y_{i0}+\gamma_{02}Y_{i-1}+X_{i1}'\beta_0+A_i}}
\end{align*}
Here, the first equality follows from the law of iterated expectations. Then, the second ingredient of the theorem is a \textit{partial fraction expansion} (Appendix Lemma \ref{tech_lemma_1}) to turn this product of logistic indices into $\pi^{0|0,1}_{2}(A_i,X_{i})$. This last operation is analogous to how we constructed sequences of transition functions in the AR(1) model. It ultimately tells us that the solution is a weighted sum of 
$(1-Y_{i1})$ and $Y_{i1}\phi_{\theta_0}^{0|0}(Y_{i3},Y_{i2},Y_{i1},Y_{i0},X_i)$. Theorem \ref{theorem_ARp_transit} turns this procedure into a recursive algorithm that computes the transition functions for any lag order $p>1$. \\
 \begin{theorem} \label{theorem_ARp_transit}
In model (\ref{ARp_logit_general}) with $T\geq p+1$, for all $t\in\{p,\ldots,T-1\}$ and $ y_{1}^p\in \mathcal{Y}^p$ , let
{\allowdisplaybreaks
\begin{align*}
    &k^{y_1|y_{1}^{p}}_{t}(\theta)=\sum_{r=1}^{p}\gamma_{r}y_r+X_{it+1}'\beta \\
    &k^{y_1|y_{1}^{k+1}}_{t}(\theta)=\sum_{r=1}^{k+1}\gamma_{r}y_r+\sum_{r=k+2}^{p} \gamma_{r}Y_{it-(r-1)}+X_{it+1}'\beta, \quad k=1,\ldots,p-2, \text{ if } p>2  \\
    &u_{t-k}(\theta)=\sum_{r=1}^{p} \gamma_{r}Y_{it-(r+k)}+X_{it-k}'\beta, \quad k=1,\ldots,p-1 \\
    &w^{y_1|y_{1}^{k+1}}_t(\theta)=\left[1-e^{(k^{y_1|y_{1}^{k+1}}_{t}(\theta)-u_{t-k}(\theta))}\right]^{y_{k+1}}\left[1-e^{-(k^{y_1|y_{1}^{k+1}}_{t}(\theta)-u_{t-k}(\theta))}\right]^{1-y_{k+1}}, \quad k=1,\ldots,p-1 
\end{align*}
}
and 
\begin{align*}
    &\phi_{\theta}^{y_1|y_{1}^{k+1}}(Y_{it+1},Y_{it},Y^{t-1}_{it-(p+k)},X_i)=\\
    &\left[(1-Y_{it-k})+w^{y_1|y_{1}^{k+1}}_t(\theta)\phi_{\theta}^{y_1|y_{1}^{k}}(Y_{it+1},Y_{it},Y^{t-1}_{it-(p+k-1)},X_i)Y_{it-k}\right]^{(1-y_1)y_{k+1}}\times \\
    &\left[1-Y_{it-k}-w^{y_1|y_{1}^{k+1}}_t(\theta)\left(1-\phi_{\theta}^{y_1|y_{1}^{k}}(Y_{it+1},Y_{it},Y^{t-1}_{it-(p+k-1)},X_i)\right)(1-Y_{it-k})\right]^{(1-y_1)(1-y_{k+1})}\times \\
    &\left[Y_{it-k}+w^{y_1|y_{1}^{k+1}}_t(\theta)\phi_{\theta}^{y_1|y_{1}^{k}}(Y_{it+1},Y_{it},Y^{t-1}_{it-(p+k-1)},X_i)(1-Y_{it-k})\right]^{y_1(1-y_{k+1})}\times \\
    &\left[1-(1-Y_{it-k})-w^{y_1|y_{1}^{k+1}}_t(\theta)\left(1-\phi_{\theta}^{y_1|y_{1}^{k}}(Y_{it+1},Y_{it},Y^{t-1}_{it-(p+k-1)},X_i)\right)Y_{it-k}\right]^{y_1y_{k+1}}, \quad k=1,\ldots,p-1 
\end{align*}
where 
\begin{align*}
    &\phi_{\theta}^{0|0}(Y_{it+1},Y_{it},Y_{it-p}^{t-1},X_i)=(1-Y_{it})e^{Y_{it+1}(\gamma_{1}Y_{it-1}-\sum_{l=2}^{p} \gamma_{l}\Delta Y_{it+1-l}-\Delta X_{it+1}'\beta)} \\
    &\phi_{\theta}^{1|1}(Y_{it+1},Y_{it},Y^{t-1}_{it-p},X_i)=Y_{it}e^{ (1-Y_{it+1})\left(\gamma_{1}(1-Y_{it-1})+\sum_{l=2}^{p} \gamma_{l}\Delta Y_{it+1-l}+\Delta X_{it+1}'\beta\right)}
\end{align*}
Then,
\begin{align*}
    &\mathbb{E}\left[\phi_{\theta_0}^{y_1|y_{1}^{p}}(Y_{it+1},Y_{it},Y^{t-1}_{it-(2p-1)},X_i)\,|\,Y_i^0,Y_{i1}^{t-p},X_i,A_i\right]=\pi^{y_1|y_{1}^{p}}_{t}(A_i,X_i) 
\end{align*}
and for $k=0,\ldots,p-2$
\begin{align*}
       &\mathbb{E}\left[\phi_{\theta_0}^{y_1|y_{1}^{k+1}}(Y_{it+1},Y_{it},Y^{t-1}_{it-(p+k)},X_i)\,|\,Y_i^0,Y_{i1}^{t-(k+1)},X_i,A_i\right]=\pi^{y_1|y_{1}^{k+1},Y_{it-(k+1)},\ldots,Y_{it-(p-1)}}_{t}(A_i,X_i), 
\end{align*}
\end{theorem}
\vspace{0.5cm}

\noindent The remaining steps to complete the construction of valid moment functions are described at length in Appendix Section \ref{Section_ARp_othersteps}. The end product is a family of (numerically) linearly independent moment functions of size $2^T-(T+1-p)2^{p}$. By Theorem \ref{theorem_nummoments_ARp}, this implies that our two-step approach recovers all \textit{moment equality} conditions in the model.

\begin{remark}(Extensions) \label{remark_extensionARp}
  While the exposition emphasized model (\ref{ARp_logit_general}), our methodology applies more broadly to models of the form
    \begin{align*}
    Y_{it}=\mathds{1}\left\{g(Y_{it-1},\ldots,Y_{it-p},X_{it},\theta_0)+A_i-\epsilon_{it}\geq 0\right\}, \quad t= 1,\ldots, T
    \end{align*}
    where the lag order $p>1$ is known and $g(.)$ is known up to the finite dimensional parameter $\theta_0$. We can thus incorporate interaction effects which are often of interest in applied work. For instance, \cite{card2005estimating} model welfare participation as a random effect AR(2) logit process of the form
   \begin{align*}
    Y_{it}=\mathds{1}\left\{\gamma_{01} Y_{it-1}+\gamma_{02} Y_{it-2}+\delta_{0} Y_{it-1}Y_{it-2}+X_{it}'\beta_0+A_i-\epsilon_{it}\geq 0\right\}, \quad t= 1,\ldots, T
    \end{align*}
    where $A_i$ either follows a normal distribution or a discrete distribution with few support points. In this case, minor modifications of the results in this section will deliver moment conditions for $\theta_0=(\gamma_{01},\gamma_{02},\delta_{0},\beta_0')'$ that are robust to misspecifications of individual unobserved heterogeneity. The key is that $A_i$ enters additivity in order to leverage the rational fraction identities of Lemma \ref{tech_lemma_1}.
\end{remark}

\subsection{Identification with more than one lag} \label{Section_AR2_identif}
This section discusses ways to leverage our methodology and moment restrictions to assess the identifiability of common parameters. For ease of exposition, we concentrate on the AR(2) logit model. \\
\indent We start by briefly reexamining an identification result due to \cite{honore2020moment}.
Using functional differencing, they proved (under some regularity conditions) that $\theta_0$ is identified with $T=3$ provided $X_{i2}=X_{i3}$ and that the initial condition $Y_{i}^0=(Y_{i-1},Y_{i0})$ varies in the population. Notice that this is not in contradiction to Theorem \ref{theorem_nummoments_ARp} since $X_{i2}=X_{i3}$ and $Y_{i}^0$ \say{varying} constitute two violations of its key assumptions. It is therefore not unsurprising that identifying moment exist in that case despite $T<4$. To understand why, note that imposing $X_{i2}=X_{i3}$ effectively amounts to equate the transition probabilities in period $t=2$ and in period $t=1$ for adequate choices of the initial condition; e.g  $\pi_{1}^{0|0,Y_{i0}}(A_i,X_{i})=\pi_{2}^{0|0,0}(A_i,X_{i})$ provided that $Y_{i0}=0$ and $X_{i2}=X_{i3}$. In turn, this implies that differences of the corresponding transition functions in periods $t=2$ and $t=1$ deliver valid moment functions to estimate $\theta_0$ in certain subpopulations. In Appendix Section \ref{appendix_identif_AR2_logit_general_subsec1}, we show that this is an interpretation of the moment conditions that \cite{honore2020moment} use to show point identification. \\
\indent Because this identification argument hinges on matching covariates as in \cite{honore2000panel}, it breaks down in the presence of certain types of regressors like an age variable or a time trend. In fact, \cite{dobronyi2021identification} showed that there are actually no moment equality conditions available in the model with such regressors. This finding is consistent with the intuition that we cannot match the transition probabilities in periods $t=1$ and $t=2$ in that case. However, with one additional period, i.e $T=4$, we can leverage the moment restrictions of Proposition \ref{proposition_3} which are valid for free-varying regressors and any initial condition. This leads to two possible approaches to inference. The first is to consider the \say{identified set} $\Theta^{I}$ of $\theta_0$ based on the four conditional moment restrictions implied by the model:
\begin{align*}
    \Theta^{I}=\left\{\theta \in \mathbb{R}^{2+K_x}: \mathbb{E}_{\theta_0}\left[\psi_{\theta}^{y_1|y_{1},y_2}(Y^{4}_{i0},Y_{i-1}^1,X_i)|Y_i^0,X_i\right]=0, \quad \forall (y_1,y_2)\in \{0,1\}^2\right\}
\end{align*}
and construct confidence sets for $\theta_0$ following e.g \cite{andrews2013inference}. Instead, the sharp identified set may be computed following the approach of \cite{dobronyi2021identification} if the covariates $X_i$ are discrete with finite support. Alternatively, a second approach which we develop further here is to formulate sensible restrictions on covariates that secure point identification in the spirit of \cite{honore2000panel}. Specifically, we consider the case where a continuous scalar component $W_{i2}$ of $X_{i2}$ has unbounded positive support conditional on $Y_{i}^0$, the other regressors, $A_i$ and has a non-trivial effect $\beta_{0W}$ of known sign to the econometrician.  This is the content of Assumption \ref{assumption_1} in which $Z_i=(R_{i}',W_{i1},W_{i3},W_{i4})$, and $X_{it}=(W_{it},R_{it}')\in \mathbb{R}^{K_x}$ for all $t\in \{1,2,3,4\}$. \cite{dobronyi2023revisiting} used a similar device to develop an alternative distribution-free semiparametric estimator to that of \cite{honore2000panel} that can accommodate time effects in the baseline one lag model.
\begin{assumption}\label{assumption_1} 
    (i) The covariate $W_{i2}$ is continuously distributed with unbounded support on $\mathbb{R}_{+}$ conditional on $Y_{i}^0,Z_{i},A_i$ and (ii) $\beta_{0W}$ is known to be strictly negative.
\end{assumption}
\noindent Besides being a technical convenience, Assumption \ref{assumption_1}  may be reasonable in some situations, e.g in the context of our empirical application, the econometrician may have a confident prior that drug prices affect individual drug consumption negatively. We point out that nothing in the discussion that follows hinges critically on $\beta_{W}<0$ and or $W_{i2}$ having support on the positive reals. A set of perfectly symmetric arguments will deliver the same conclusions if instead $\beta_{W}>0$ and $W_{i2}$ has unbounded support on $\mathbb{R}_{-}$.
\begin{assumption}
\label{assumption_2}
(i) $\theta_0=(\gamma_{01},\gamma_{02},\beta_0')'\in \mathbb{G}_1\times \mathbb{G}_2\times \mathbb{B}=\Theta$, $\mathbb{G}_1,\mathbb{G}_2,\mathbb{B}$ compact. The conditional densities of $A_i$ and  $Z_i$ verify:
\begin{enumerate}[label=(\roman*)]
\setcounter{enumi}{1}
    \item $\lim \limits_{w_2\to\infty}p(a|y^0,z,w_2)=q(a|y^0,z)$,  $\lim \limits_{w_2\to\infty}p(z|y^0,w_2)=q(z|y^0)$ 
     \item There exists positive integrable functions $d_0(a),d_1(z),d_2(z)$ such that $p(a|y^{0},z,w_2)\leq d_{0}(a)$ for all $a\in \mathbb{R}$,  $d_{1}(z) \leq p(z|y^0,w_2)\leq d_{2}(z)$ for all $z\in \mathbb{R}^{K_{x}-1}$
    \item $w_2\mapsto p(a|y^0,z,w_2), w_2\mapsto p(z|y^0,w_2) $ are continuous in $w_2$. 
\end{enumerate}
\end{assumption}
\noindent Assumption \ref{assumption_2} are standard regularity conditions for an application of the dominated convergence theorem that once paired with Assumption \ref{assumption_1} are sufficient to establish that $\theta_{0}$ is \textit{identified at infinity}. The outline of the argument is as follows. Under these assumptions, by sending $W_{i2}$ to $\infty$, the valid moment function $\psi_{\theta}^{0|0,0}(Y_{i4},Y_{i3},Y^{2}_{i-1},X_i)$ of Proposition \ref{proposition_3} reduces to 
\begin{align} \label{limit_valid_moment_funcition}
\begin{split}
     \psi_{\theta,\infty}^{0|0,0}(Y_{i4},Y_{i3},Y^{2}_{i-1},Z_i)&=-(1-Y_{i1})(1-Y_{i2})Y_{i3} \\
    &+\left[e^{X_{i34}'\beta}-1\right](1-Y_{i1})(1-Y_{i2})(1-Y_{i3})Y_{i4} \\
    &+e^{-\gamma_1Y_{i0}+\gamma_{2}(1-Y_{i-1})+X_{i31}'\beta}Y_{i1}(1-Y_{i2})(1-Y_{i3})Y_{i4} \\
    &+e^{-\gamma_1Y_{i0}-\gamma_{2}Y_{i-1}+X_{i41}'\beta}Y_{i1}(1-Y_{i2})(1-Y_{i3})(1-Y_{i4})
\end{split}
\end{align} 
which occurs because $\lim_{w_2\to\infty} e^{w_2\beta_W}=0$ and $Y_{i2}=0$ with probability one conditional on the regressors and the fixed effects. The key observation is that this \say{limiting} moment function has a similar functional form to the valid moment functions of the AR(1) model with $T=3$. In turn, this implies monotonicity properties on certain regions of the covariate space that we can exploit to point identify $\theta_0$ in the spirit of \cite{honore2020moment}. To this end,  let $(\bar{x},\underline{x})\in \mathbb{R}^2$, such that $\bar{x}>\underline{x}$ and
define the sets
\begin{align*}
    \mathcal{X}_{k,+}&=\{x\in \mathbb{R}^{4K_{x}}|  \bar{x} \geq x_{k,3}\geq x_{k,4} > x_{k,1}\geq \underline{x} \text{  or  } \bar{x}\geq x_{k,3}> x_{k,4} \geq x_{k,1}\geq \underline{x} \} \\
    \mathcal{X}_{k,-}&=\{x\in \mathbb{R}^{4K_{x}}|\underline{x}\leq x_{k,3}\leq x_{k,4} < x_{k,1}\leq \bar{x} \text{ or } \underline{x} \leq x_{k,3}< x_{k,4} \leq x_{k,1}\leq \bar{x} \}
\end{align*}
for all $k \in \{1,\ldots,K_{x}\}$. In words, $\mathcal{X}_{k,+}$ is the region of the covariate space in which values of the $k$-th regressor in periods $t\in\{1,3,4\}$ belong to $[\underline{x},\bar{x}]$ and verify $x_{k,3}\geq x_{k,4} \geq x_{k,1}$  with at least one strict inequality. Instead, $\mathcal{X}_{k,-}$ is the region of the covariate space where realizations of the $k$-th regressor obey the reverse ranking. With these notations in hands, we have the following theorem,
\begin{theorem} \label{theorem_AR2_identif}
For $T=4$, suppose that outcomes $(Y_{i1},Y_{i2},Y_{i3},Y_{i4})$ are generated from model (\ref{ARp_logit_general}) with $p=2$, initial condition $y^{0}\in\mathcal{Y}^2$, common parameters $\theta_0=(\gamma_{0}',\beta_0')\in \mathbb{R}^{2+K_{x}}$ and that Assumptions \ref{assumption_1} and \ref{assumption_2} hold. Further,  for all $s\in\{-,+\}^{K_{x}}$, let $\mathcal{X}_s=\bigcap \limits _{k=1}^{K_{x}} \mathcal{X}_{k,s_k}$ and suppose that for all $y^{0}\in \mathcal{Y}^2$
\begin{align*}
    \lim_{w_{2}\to \infty} P\left(Y_{i}^0=y^{0},\quad X_i\in \mathcal{X}_s \,|\,  W_{i2}=w_2 \right)>0 
\end{align*}
Let 
\begin{align*}
     \Psi_{s,y^0}^{0|0,0}(\theta)
    &=\lim_{w_{2}\to \infty} \mathbb{E}\left[ \psi_{\theta,\infty}^{0|0,0}(Y_{i4},Y_{i3},Y^{2}_{i-1},X_i)\,|\,Y_{i}^0=y^0,X_i\in \mathcal{X}_s,W_{i2}=w_2\right]
\end{align*}
Then, $\theta_0$ is the unique solution to the system of equations
\begin{align*}
    \Psi_{s,y^0}^{0|0,0}(\theta)&=0, \quad \forall s\in\{-,+\}^{K_{x}},\quad  \forall y^0\in \mathcal{Y}^2
\end{align*}
\end{theorem}
\noindent Theorem \ref{theorem_AR2_identif} shows that point identification of $\theta_0$ is achievable in higher-order dynamic logit models in short panels. The main cost for this guarantee is Assumption \ref{assumption_1} which presumes knowledge of the data generating process beyond the baseline setup. Additionally, there should be sufficient variation in the regressors $X_{it}$  as $W_{i2}\mapsto \infty$ to ensure that $\lim_{w_{2}\to \infty} P\left(Y_{i}^0=y^{0},\quad X_i\in \mathcal{X}_s \,|\, W_{i2}=w_2 \right)>0$ for all $s\in\{-,+\}^{K_{x}}$. Our arguments are easily generalizable to AR($p$) models with lag order $p\geq 3$. Under natural extensions of Assumptions \ref{assumption_1} and \ref{assumption_2},  the model parameters $\theta_0=(\gamma_{01},\ldots,\gamma_{0p},\beta_0')$ are \textit{identified at infinity} provided $T\geq 2+p$.

\begin{remark}[Identification with time effects]
Theorem \ref{theorem_AR2_identif} does not readily deals with time effects but it is straightforward to adapt the argument for this case. Suppose for concreteness that one covariate is a time trend. By further sending $W_{i3}$ to infinity, the limiting moment function of equation (\ref{limit_valid_moment_funcition}) reduces to 
\begin{multline*}
     \psi_{\theta,\infty}^{0|0,0}(Y_{i4},Y_{i3},Y^{2}_{i-1},Z_i)=-(1-Y_{i1})(1-Y_{i2})(1-Y_{i3})Y_{i4} \\ 
     +e^{-\gamma_1Y_{i0}-\gamma_{2}Y_{i-1}+X_{i41}'\beta}Y_{i1}(1-Y_{i2})(1-Y_{i3})(1-Y_{i4})
\end{multline*}
For  $(Y_{i0},Y_{i-1})=(0,0)$, this valid moment function only depends on $\beta$ and arguments analogous to those in Theorem \ref{theorem_AR2_identif} will point identify $\beta_0$. Varying the initial condition is then sufficient to point identify $\gamma_0$ given the monotonicity of the moment function in $(\gamma_1,\gamma_2)$.

\end{remark}

\subsection{Average Marginal Effects in  AR($p$) logit models}

In discrete choice settings, interest often centers on certain functionals of unobserved heterogeneity rather than on the value of the model parameters per se. One particular family of such functionals that are of interest from a policy perspective are average marginal effects (AMEs) which capture mean response to a counterfactual change in past outcomes. It turns out that these key quantities are simply expectations of our transition functions. To see this, consider first the baseline AR(1) model with discrete covariates $X_{it}$. We can define the average transition probability from state $l$ to state $k$ in period $t$ for a subpopulation of individuals with covariate  $x_{1}^{t+1}=(x_1,\ldots,x_{t+1})$ and initial condition $y_0$ as
\begin{align*}
    \Pi_{t}^{k|l}(y_0,x_{1}^{t+1})=\mathbb{E}\left[\underbrace{\pi_{t}^{k|l}(X_{it+1},A_i)}_{\equiv \pi_{t}^{k|l}(X_{i},A_i)}\,|\,Y_{i0}=y_0,X_{i1}^{t+1}=x_{1}^{t+1}\right]=\int \pi_{t}^{k|l}(x_{t+1},a)p(a|y_0,x_{1}^{t+1})da
\end{align*}
where $p(a|y_0,x_{1}^{t+1})$ denotes the conditional density of the fixed effect $A$ given $(y_0,x_{1}^{t+1})$. The AME is defined as the following contrast of average transition probabilities:
\begin{align*}
    AME_{t}(y_0,x_{1}^{t+1})= \Pi^{1|1}_{t}(y_0,x_{1}^{t+1}) - \Pi^{1|0}_{t}(y_0,x_{1}^{t+1})=\Pi^{1|1}_{t}(y_0,x_{1}^{t+1})-(1-\Pi^{0|0}_{t}(y_0,x_{1}^{t+1}))
\end{align*}
It is interpreted as the population average causal effect on $Y_{it+1}$ of a change from 0 to 1 of $Y_{it}$ given $(y_0,x_{1}^{t+1})$. By Lemma \ref{lemma_2} and the law of iterated expectations, we have that for $T\geq 2$ and $t\geq 1$:
\begin{align*}
\Pi_{t}^{0|0}(y_0,x_{1}^{t+1})&=\mathbb{E}\left[\phi_{\theta_0}^{0|0}(Y_{it+1},Y_{it},Y_{it-1},X_i)\,|\,Y_{i0}=y_0,X_{i1}^{t+1}=x_{1}^{t+1}\right] \\
\Pi_{t}^{1|1}(y_0,x_{1}^{t+1})&=\mathbb{E}\left[\phi_{\theta_0}^{1|1}(Y_{it+1},Y_{it},Y_{it-1},X_i)\,|\,Y_{i0}=y_0,X_{i1}^{t+1}=x_{1}^{t+1}\right]
\end{align*}
which implies that $AME_{t}(y_0,x_{1}^{t+1})$ is identified so long as $\theta_0$ is identified. A sufficient condition for that is $T\geq 3$ and $X_{i3}-X_{i2}$ having support in a neighborhood of zero (\cite{honore2000panel}). \cite{aguirregabiria2021identification} were the first to highlight that AMEs can be point identified in the AR(1) model.
When the lag order $p$ is greater than one - which seems to be the case for persistent variables such as unemployment (e.g \cite{magnac2000subsidised}) and welfare recipiency (e.g \cite{chay1999non}) -  we can analogously define average transition probabilities from states $l_1^p\in\mathcal{Y}^{p}$ to state $k\in\mathcal{Y}$ as:
\begin{align*}
\Pi_{t}^{k|l_{1}^p}(y^0,x_{1}^{t+1})&=\mathbb{E}\left[\underbrace{\pi_{t}^{k|l_{1}^p}(X_{it+1},A_i)}_{\equiv \pi_{t}^{k|l}(X_{i},A_i)}\,|\,Y_{i}^0=y^0,X_{i1}^{t+1}=x_{1}^{t+1}\right]=\int \pi_{t}^{k|l_{1}^p}(x_{t+1},a)p(a|y_0,x_{1}^{t+1})da 
\end{align*}
This permits the consideration of more nuanced counterfactual parameters compared to the AR(1). In the context of studies on long term unemployment, contrasts of the form $\Pi_{t}^{k|l_{1}^p}(y^0,x_{1}^{t+1})-\Pi_{t}^{k|v_{1}^p}(y^0,x_{1}^{t+1})$ may be especially relevant to measure more accurately the relative effects of work histories spanning multiple periods. Again, these counterfactuals are simply expectations of transition functions by Theorem \ref{theorem_ARp_transit} and will be identified whenever $\theta_0$ is identified (see Section \ref{Section_AR2_identif} for examples of sufficient conditions). \\
\indent Multiperiod analogs of average transition probabilities in AR($p$) models
\begin{align*}
&\Pi_{t}^{k_{1}^s|l_1^{p}}(y^0,x_{1}^{t+s})=\\
&\mathbb{E}\left[P(Y_{it+s}=k_s,\ldots, Y_{it+1}=k_1\,|\,Y_{it}=l_1,\ldots,Y_{it-(p-1)}=l_p,X_{i1}^{t+s}=x_{1}^{t+s},A_i)\,|\,Y_{i}^0=y^0,X_{i1}^{t+s}=x_{1}^{t+s}\right] 
\end{align*}
may also be of interest to assess state-dependence.
These quantities give the average probability of moving from states $l_1^p\in\mathcal{Y}^{p}$ to future states $k_{1}^s\in \mathcal{Y}^s$, where $s\geq 1$ and the average is taken with respect to the distribution of $A_i$ conditional on $(y_0,x_{1}^{t+1})$.  The special case  $k_1=k_2=\ldots=k_{s}$ delivers a discrete version of the survivor function employed in duration analysis, i.e  the average likelihood to survive $s$ consecutive periods in the same state after experiencing a given choice history. Proposition \ref{proposition_4} shows that they are also identified when $\theta_0$ is identified under certain conditions.

\begin{proposition} \label{proposition_4}
    Consider model (\ref{ARp_logit_general}) with $T\geq p+2$, and initial condition $y^0\in \mathcal{Y}^p$. Suppose that $\theta_0$ is identified and that for any   $t\in\{p,\ldots,T-2\}$,
    $s\in\{1,\ldots,T-1-t\}$ and $y,\tilde{y}\in \mathcal{Y}^{p}$, $\gamma_0'y+x_{t}'\beta_0\neq \gamma_0'\tilde{y}+x_{t+s}'\beta_0$ . Then, for $t\in\{p,\ldots,T-2\}$,
    $s\in\{1,\ldots,T-1-t\}$, and any $l_1^p\in\mathcal{Y}^{p}$, $k_{1}^s\in \mathcal{Y}^s$, the quantity $\Pi_{t}^{k_{1}^s|l_1^{p}}(y^0,x_{1}^{t+s})$ is identified.
\end{proposition}
\noindent  The source of this result is the fact that the integrand of $\Pi_{t}^{k_{1}^s|l_1^{p}}(y^0,x_{1}^{t+s})$ is a product of transition probabilities. This entails  that under appropriate conditions on the regressors and common parameters, we can turn this integrand into a unique linear combination of transition probabilities by means of a \textit{partial fraction decomposition}. It is then a matter of taking expectations and invoking the fact that average transition probabilities are identified from our transition functions. 

\begin{example}[Survivor function for an AR(2)] 
To illustrate Proposition \ref{proposition_4}, and in the spirit of our upcoming empirical application,  suppose that $Y_{it}$ is an indicator for drug consumption at time $t$ obeying an AR(2) logit process. Fix $y^0\in \mathcal{Y}^2$ and assume $T=5$. One might be interested in 
\begin{align*}
    \Pi_{3}^{0,0|1,1}(y^0,x)&=\mathbb{E}\left[P(Y_{i5}=0,Y_{i4}=0\,|\,Y_{i3}=1,Y_{i2}=1,X_{i}=x,A_i)\,|\,Y_{i}^0=y^0,X_{i}=x\right] \\
    &=\mathbb{E}\left[\pi_{4}^{0|0,1}(A_i,x)\pi_{3}^{0|1,1}(A_i,x)\,|\,Y_{i}^0=y^0,X_{i}=x\right]
\end{align*}
which gives the average propensity of individuals with characteristics $(y^0,x)$ who consumed drugs in $t=2,3$ to stay drug-free over the next two time periods. A simple calculation using for instance the identities of Appendix Lemma \ref{tech_lemma_1} gives
\begin{align*}
    \pi_{4}^{0|0,1}(A_i,x)\pi_{3}^{0|1,1}(A_i,x)&=\frac{1}{1+e^{\gamma_{02}+x_5'\beta_0+A_i}}\frac{1}{1+e^{\gamma_{01}+\gamma_{02}+x_4'\beta_0+A_i}} \\
    &=\frac{1}{1-e^{\gamma_{01}+x_{45}'\beta_0}}\pi_{4}^{0|0,1}(A_i,x)-\frac{e^{\gamma_{01}+x_{45}'\beta_0}}{1-e^{\gamma_{01}+x_{45}'\beta_0}}\pi_{3}^{0|1,1}(A_i,x)
\end{align*}
and since Theorem \ref{theorem_ARp_transit} implies
$\mathbb{E}\left[\phi_{\theta_0}^{0|0,1}(Y_{i1}^5,x)\,|\,Y_i^0=y^0,Y_{i1}^{2},X_i=x,A_i\right]=\pi^{0|0,1}_{4}(A_i,x)$ and $\mathbb{E}\left[\phi_{\theta_0}^{0|1,1}(Y_{i0}^4,x)\,|\,Y_i^0=y^0,Y_{i1},X_i=x,A_i\right]=\pi^{0|0,1}_{3}(A_i,x)$, we obtain
\begin{align*}
\Pi_{3}^{0,0|1,1}(y^0,x)=\mathbb{E}\left[\frac{1}{1-e^{\gamma_{01}+x_{45}'\beta_0}}\phi_{\theta_0}^{0|0,1}(Y_{i1}^5,x)-\frac{e^{\gamma_{01}+x_{45}'\beta_0}}{1-e^{\gamma_{01}+x_{45}'\beta_0}}\phi_{\theta_0}^{0|1,1}(Y_{i0}^4,x)\,|\,Y_i^0=y^0,X_i=x\right]
\end{align*}
\end{example}

\section{Multi-dimensional fixed effects models}  \label{Section_5}
We now turn our attention to multi-dimensional fixed effects models. We show that the general blueprint developed in the scalar case to derive valid moment functions carries over to VAR(1) and  MAR(1) models.  We make no attempt at showing that our approach is exhaustive in those cases and do not claim that it is. We leave these important questions for future work. Readers uninterested in the details of the multivariate extensions can skip directly to Section \ref{Section_6} where we discuss the empirical application.

\subsection{Moment restrictions for the VAR(1) logit model} \label{Section_5_VAR1}
\noindent We begin with the analysis of VAR(1) logit models, variants of which have been successfully used to study the relationship between sickness and unemployment (\cite{narendranthan1985investigation}), the progression from softer drug use to harder drug use among teenagers (\cite{deza2015there}), transitivity in networks (\cite{graham2013comment}, \cite{graham2016homophily}) and more recently the employment of couples (\cite{honore2022simultaneity}).  For a given $M\geq 2$, the model reads:
\begin{align}\label{VAR1_logit_general}
    Y_{m,it}&=\mathds{1}\left\{\sum_{j=1}^{M} \gamma_{0mj}Y_{j,it-1} +X_{m,it}'\beta_{0m}+A_{m,i}-\epsilon_{m,it}\geq 0\right\}, \quad m=1,\ldots,M, \quad t=1,\ldots,T
\end{align}
We let $Y_{it}=(Y_{1,it},\ldots,Y_{M,it})'$ denote the outcome vector in period $t$ with support $\mathcal{Y}=\{0,1\}^M$ of cardinality $2^M$. We let $X_{it}=(X_{1,it}',\ldots ,X_{M,it}')'\in \mathbb{R}^{K_1}\times \ldots \times \mathbb{R}^{K_M}$ denote the vector of exogenous covariates in period $t$ and $A_i=(A_{1,i},\ldots,A_{M,i})'\in \mathbb{R}^M$ . The initial condition is now given by $Y_{i0}=(Y_{1,i0},\ldots,Y_{M,i0})'\in \mathcal{Y}$ and the model transition probabilities are given by:
\begin{align*}
    \pi^{k|l}_{t}(A_i,X_{i})=P(Y_{it+1}=k|Y_{it}=l,X_i,A_i)=\prod_{m=1}^M \frac{e^{k_m(\sum_{j=1}^{M} \gamma_{0mj}l_j+X_{m,it+1}'\beta_{0m}+A_{m,i})}}{1+e^{\sum_{j=1}^{M} \gamma_{0mj}l_j+X_{m,it+1}'\beta_{0m}+A_{m,i}}}
\end{align*}
for all $(k,l)\in \mathcal{Y}\times \mathcal{Y}$. \\ 
\indent  Building on \cite{honore2000panel}, \cite{honore2019panel} use a conditional likelihood approach to prove the identification $\theta_0=(\gamma_{011},\gamma_{012},\gamma_{021},\gamma_{022},\beta_{01},\beta_{02})$ for the bivariate specification when $T=3$ and the regressors do not vary over the last two periods. As in scalar models, we show hereinafter that this strong restriction which can yield undesirable rates of convergence is unnecessary to obtain valid moment conditions. \\
\indent \textbf{Step 1)} in the VAR(1) logit model has a nuance relative to its scalar counterpart in that the only transition functions that appear to exist are those associated to $\pi^{k|k}_{t}(A_i,X_{i})$, for $k \in \mathcal{Y}$, i.e the probabilities of staying in the same state. We can use the same heuristic as in the baseline AR(1) model to derive their expressions, especially in the bivariate case. Once all four transition functions are obtained for the case $M=2$, it becomes clear that the general functional form is as per Lemma \ref{lemma_5}. It is then a matter of brute force calculation to verify that this is indeed correct. 
\begin{lemma} \label{lemma_5} 
 In model (\ref{VAR1_logit_general}) with $T\geq 2$ and $t \in \{1,\ldots,T-1\}$, let for all $k \in \mathcal{Y}$
 \begin{align*}
     \phi_{\theta}^{k|k}(Y_{it+1},Y_{it},Y_{it-1},X_i)&=\mathds{1}\{Y_{it}=k\}e^{\sum_{m=1}^M (Y_{m,it+1}-k_m)\left(\sum_{j=1}^M \gamma_{mj}(Y_{j,it-1}-k_j)-\Delta X_{m,it+1}'\beta_m\right)}
 \end{align*}
Then:
\begin{align*}
    \mathbb{E}\left[ \phi_{\theta_0}^{k|k}(Y_{it+1},Y_{it},Y_{it-1},X_i)|Y_{i0},Y_{i1}^{t-1},X_i,A_i\right]&= \pi^{k|k}_{t}(A_i,X_{i})=\prod_{m=1}^M \frac{e^{k_m(\sum_{j=1}^{M} \gamma_{0mj}k_{j} +X_{m,it+1}'\beta_{0m}+A_{m,i})}}{1+e^{\sum_{j=1}^{M} \gamma_{0mj}k_{j} +X_{m,it+1}'\beta_{0m}+A_{m,i}}}
\end{align*}
\end{lemma}

\noindent Next, we can appeal to the second \textit{partial fraction decomposition} formula in Appendix Lemma \ref{tech_lemma_2} to guide the construction of another set of transition functions when $T\geq 3$. These identities may be regarded as a generalization of \cite{Kitazawa_JOE2021}'s hyperbolic transformations to the multivariate case. As is clear from Lemma \ref{lemma_6}, the resulting transition functions have a special structure that generalizes those found in the AR(1) model.
\begin{lemma} \label{lemma_6}
In model (\ref{VAR1_logit_general}) with $T\geq 3$, for all $t,s$ such that $T-1\geq t> s\geq 1$, let for all $ m\in \{1,\ldots,M\}$ and $(k,l)\in \mathcal{Y}^2$
\begin{align*}
    &\mu_{m,s}(\theta)=\sum_{j=1}^M \gamma_{mj} Y_{j,is-1}+X_{m,is}'\beta_{m} \\
    &\kappa_{m,t}^{ k|k}(\theta)=\sum_{j=1}^M \gamma_{mj} k_j+X_{m,it+1}'\beta_m \\
    &\omega_{t,s,l}^{k|k}(\theta)=1-e^{\sum_{j=1}^M (l_j-k_j)\left[\kappa_{j,t}^{ k|k}(\theta)-\mu_{j,s}(\theta)\right]}
\end{align*}
and define the moment functions
\begin{align*}
     \zeta_{\theta}^{k|k}(Y_{it-1}^{t+1},Y_{is-1}^s,X_i)=\mathds{1}\{Y_{is}=k\}+ \cleansum_{l\in \mathcal{Y}\setminus \{k\}} \omega_{t,s,l}^{k|k}(\theta) \mathds{1}\{Y_{is}=l\} \phi_{\theta}^{k|k}(Y_{it-1}^{t+1},X_i)
\end{align*}
Then,
\begin{align*}
    \mathbb{E}\left[\zeta_{\theta_0}^{k|k}(Y_{it-1}^{t+1},Y_{is-1}^s,X_i)|Y_{i0},Y_{i1}^{s-1},X_i,A_i\right]&= \pi^{k|k}_{t}(A_i,X_{i})
\end{align*}
\end{lemma}
\noindent Beyond $T=4$, more transition functions are available and can be derived sequentially from those of Lemma \ref{lemma_6}. 
See Corollary \ref{corollary_6} for their expressions. 
\begin{corollary}\label{corollary_6}
In model (\ref{VAR1_logit_general}) with $T\geq 4$,  for any $t$ and ordered collection of indices $s_{1}^J$, $J\geq 2$, satisfying $T-1\geq t> s_1>\ldots>s_J\geq 1$, let for all $k\in \mathcal{Y}$
\begin{multline*}
    \zeta_{\theta}^{k|k}(Y_{it-1}^{t+1},Y_{is_1-1}^{s_1},\ldots,Y_{is_{J}-1}^{s_J},X_i)=\mathds{1}\{Y_{is_J}=k\} \\
    +\sum_{l\in \mathcal{Y}\setminus \{k\}}\omega_{t,s_J,l}^{k|k}(\theta)\mathds{1}\{Y_{is_J}=l\} \zeta_{\theta}^{k|k}(Y_{it-1}^{t+1},Y_{is_1-1}^{s_1},\ldots,Y_{is_{J-1}-1}^{s_{J-1}},X_i)
\end{multline*}

with weights $\omega_{t,s_J,l}^{k|k}(\theta)$ defined as in Lemma \ref{lemma_6}. Then,
\begin{align*}
    \mathbb{E}\left[\zeta_{\theta_0}^{k|k}(Y_{it-1}^{t+1},Y_{is_1-1}^{s_1},\ldots,Y_{is_{J}-1}^{s_J},X_i)|Y_{i0},Y_{i1}^{s_{J}-1},X_i,A_i\right]=\pi^{k|k}_{t}(A_i,X_{i})
\end{align*}
\end{corollary}
\indent \textbf{Step 2)}. One can obtain a family of valid moment functions by adequately repurposing the statement of Proposition \ref{proposition_2} to the VAR(1) case, i.e by updating the expressions of  $\phi_{\theta}^{k|k}(.)$ and $\zeta_{\theta}^{k|k}$ according to Lemma \ref{lemma_5} and Corollary \ref{corollary_6}. To conserve on space and avoid repetition, we leave this simple exercise to the reader. 
\begin{remark}[Network Extension]
    Similarly to Remarks \ref{remark_extensionARp}, we emphasize that the tools developed here can be modified to handle other interesting variants featuring more complex interdependencies across the different layers of the model indexed by $m=1,\ldots,M$ . To illustrate the wider applicability of our two-step method, we show in Appendix \ref{network_model} how one can derive moment restrictions in the dynamic network formation model of \cite{graham2013comment} and extensions thereof incorporating exogenous covariates.
\end{remark}

\subsection{Moment restrictions for the dynamic multinomial logit model} \label{Section_5_MAR1}
Last, we cover dynamic multinomial logit models which have been utilized to measure state-dependence in a range of economic contexts including: employment history in the French labor market (\cite{magnac2000subsidised}), the impact of international trade on the transition matrix of employment across sectors (\cite{egger2003sectoral}) and consumer product choice (\cite{dube2010state}) amongst others. 
\\
\indent We focus on the the baseline MAR(1) logit model with fixed effects.

The model assumes a fixed number of alternatives $C+1$ with $C\geq 1$ and is characterized by the following transition probabilities:
\begin{align} \label{MAR1_logit_general}
    \pi_{t}^{k|l}(A_i,X_i)=P(Y_{it+1}=k|Y_{it}=l,X_i,A_i)=\frac{e^{\gamma_{kl}+X_{ikt+1}'\beta_k+A_{ik}}}{\sum \limits_{c=0}^C e^{\gamma_{cl}+X_{ict+1}'\beta_j+A_{ic}}}, \quad  t=1,\ldots,T
\end{align}
with $(k,l)\in \mathcal{Y}=\{0,1,\ldots,C\}$. Here, $Y_{it}\in \mathcal{Y}$ indicates the  choice of individual $i$ in period $t$, $X_{ijt}$ denotes a vector of individual-alternative specific exogenous covariates and $A_{ij}\in \mathbb{R}$ is the fixed effect attached to alternative $j$ for individual $i$. The initial condition is $Y_{i0}\in \mathcal{Y}$ and in keeping with the fixed effect assumption, its conditional distribution given unobserved heterogeneity and the regressors, $\left(P(Y_{i0}=k|X_i,A_i)\right)_{k=1}^C$, is left fully unrestricted. Following \cite{magnac2000subsidised}, we normalize the transition parameters and fixed effect of the reference alternative \say{$0$} to zero \footnote{ The transition parameters of the reference state cannot be identified so a normalization constraint must be imposed. Setting $A_{i0}=0$ is also without loss of generality since we can always redefine the fixed effect as $A_{ik}^*=A_{ik}-A_{i0}$.}. That is $\gamma_{j0}=\gamma_{0j}=0, A_{0,j}=0$ for all $j\in \mathcal{Y}$ leaving $\theta=\left((\gamma_{kl})_{k,l\geq 1},(\beta_{l})_{l\geq 0}\right)$ as the unknown model parameters.\\
\indent This specification can be motivated by assuming that agents rank options according to random latent utility indices with disturbances independent over time and across alternatives. In this context, equation (\ref{MAR1_logit_general}) is obtained if the best alternative is selected and the error terms are Type 1 extreme value distributed conditional on $Y_{i0},A_i,X_i$. \cite{magnac2000subsidised} studies the \say{pure} case without covariates and shows that an extension of the conditional likelihood approach proposed by \cite{chamberlain_1985} can be used to identify and estimate the state-dependence parameters. \cite{honore2000panel} show that this argument carries over to  the case with exogenous explanatory variables if one matches the regressors across specific time periods. Here, we offer an alternative estimation strategy that circumvents the need for matching. \\
\indent \textbf{Step 1)}. Similarly to the VAR(1) model the MAR(1) appears to admit transition functions only for the probabilities of staying in the same state, namely $\pi^{k|k}_{t}(A_i,X_i)$ for $k\in \mathcal{Y}$. This feature appears to be a common trait of multidimensional fixed effects specifications. To facilitate the derivation of the relevant transition functions, we follow our usual heuristic of looking for  $\phi_{\theta}^{k|k}(.), k \in \mathcal{Y}$ satisfying:
\begin{align*}
    &\phi_{\theta}^{k|k}(Y_{it+1},Y_{it},Y_{it-1},X_i)=\mathds{1}\{Y_{it}=k\}\phi_{\theta}^{k|k}(Y_{it+1},k,Y_{it-1}) \\
     &\mathbb{E}\left[\phi_{\theta_0}^{k|k}(Y_{it+1},Y_{it},Y_{it-1},X_i) \,|\ Y_{i0},Y_{i1}^{t-1},X_i,A_i\right]= \pi^{k|k}_{t}(A_i,X_i)
\end{align*}
Upon obtaining their exact expressions for the simplest case with $C=2$, it is easy to conjecture and verify by direct calculations that the general expressions of the  $C+1$ transition functions of the MAR(1) model are as displayed in Lemma \ref{lemma_7}.

\begin{lemma} \label{lemma_7} 
 In model (\ref{MAR1_logit_general}) with $T\geq 2$ and  $t \in \{1,\ldots,T-1\}$, let for all $k \in \mathcal{Y}$
\begin{align*}
    \phi_{\theta}^{k|k}(Y_{it-1}^{t+1},X_i)&=\mathds{1}\{Y_{it}=k\} e^{\sum_{c\in \mathcal{Y}\setminus\{k\}} \mathds{1}\{Y_{it+1}=c\} \left(\sum_{j\in \mathcal{Y}} (\gamma_{cj}-\gamma_{kj})\mathds{1}(Y_{it-1}=j)+\gamma_{kk}-\gamma_{ck}+\Delta X_{ikt+1}'\beta_k-\Delta X_{ict+1}'\beta_c\right)}
\end{align*}

Then:
\begin{align*}
    \mathbb{E}\left[ \phi_{\theta_0}^{k|k}(Y_{it+1},Y_{it},Y_{it-1},X_i)|Y_{i0},Y_{i1}^{t-1},X_i,A_i\right]&= \pi^{k|k}_{t}(A_i,X_{i})
\end{align*}
\end{lemma}

\noindent Unsurprisingly, given the similarities shared between the MAR(1) and all other specifications discussed in the paper, so long as $T\geq 3$, one can again derive transition functions other than $\phi_{\theta}^{k|k}(Y_{it-1}^{t+1},X_i)$ also associated to $ \pi^{k|k}_{t}(A_i,X_i)$ for $k \in \mathcal{Y}$ in periods $t\in\{1,\ldots,T-1\}$. The simple logistic identities of Appendix Lemma \ref{tech_lemma_1} imply that these transition functions, that we keep denoting $\zeta_{\theta}^{k|k}(.)$ have a similar form to those of the VAR(1) model as shown in  Lemma \ref{lemma_8}.

\begin{lemma} \label{lemma_8}
In model (\ref{MAR1_logit_general}) with $T\geq 3$, for all $t,s$ such that $T-1\geq t> s\geq 1$, let for all  $(c,k)\in \mathcal{Y}^2$
\begin{align*}
    \mu_{c,s}(\theta)&=\sum_{j=1}^C \gamma_{cj}\mathds{1}(Y_{is-1}=j)+X_{ics}'\beta_c-X_{i0s}'\beta_0 \\
    \kappa_{c,t}^{k|k}(\theta)&=\gamma_{ck}+X_{ict+1}'\beta_c-X_{i0t+1}'\beta_0 \\
    \omega_{t,s,c}^{k|k}(\theta)&=1-e^{(\kappa_{c,t}^{k|k}(\theta)-\mu_{c,s}(\theta))-( \kappa_{k,t}^{k|k}(\theta)-\mu_{k,s}(\theta))}
\end{align*}
and define the moment functions
\begin{align*}
     \zeta_{\theta}^{k|k}(Y_{it-1}^{t+1},Y_{is-1}^s,X_i)=\mathds{1}\{Y_{is}=k\}+ \cleansum_{l\in \mathcal{Y}\setminus \{k\}} \omega_{t,s,l}^{k|k}(\theta) \mathds{1}\{Y_{is}=l\} \phi_{\theta}^{k|k}(Y_{it-1}^{t+1},X_i)
\end{align*}
Then,
\begin{align*}
    \mathbb{E}\left[\zeta_{\theta_0}^{k|k}(Y_{it-1}^{t+1},Y_{is-1}^s,X_i)|Y_{i0},Y_{i1}^{s-1},X_i,A_i\right]&= \pi^{k|k}_{t}(A_i,X_{i})
\end{align*}
\end{lemma}
\noindent Additionally, if the econometrician has access to a dataset with more than four observations per sampling unit - counting the initial condition - then, more transition functions associated to the same transition probabilities are available per Corollary \ref{corollary_7}.
\begin{corollary}\label{corollary_7}
In model (\ref{MAR1_logit_general}) with $T\geq 4$,  for any $t$ and ordered collection of indices $s_{1}^J$, $J\geq 2$, satisfying $T-1\geq t> s_1>\ldots>s_J\geq 1$, let for all $k\in \mathcal{Y}$
\begin{multline*}
    \zeta_{\theta}^{k|k}(Y_{it-1}^{t+1},Y_{is_1-1}^{s_1},\ldots,Y_{is_{J}-1}^{s_J},X_i)=\mathds{1}\{Y_{is_J}=k\} \\
    +\sum_{l\in \mathcal{Y}\setminus \{k\}}\omega_{t,s_J,l}^{k|k}(\theta)\mathds{1}\{Y_{is_J}=l\} \zeta_{\theta}^{k|k}(Y_{it-1}^{t+1},Y_{is_1-1}^{s_1},\ldots,Y_{is_{j-1}-1}^{s_{J-1}},X_i)
\end{multline*}
with weigts $\omega_{t,s_J,l}^{k|k}(\theta)$ defined as in Lemma \ref{lemma_8}. Then,
\begin{align*}
    \mathbb{E}\left[\zeta_{\theta_0}^{k|k}(Y_{it-1}^{t+1},Y_{is_1-1}^{s_1},\ldots,Y_{is_{J}-1}^{s_J},X_i)|Y_{i0},Y_{i1}^{s_{J}-1},X_i,A_i\right]=\pi^{k|k}_{t}(A_i,X_{i})
\end{align*}
\end{corollary}
\indent This completes \textbf{Step 1)} for the MAR(1) logit model. For \textbf{Step 2)}, we recommend a family of valid moment functions mirroring those of Proposition \ref{proposition_2} for the AR(1) case to ensure the linear independence of its elements. 

\section{Empirical Illustration} \label{Section_6}
In this last section, we illustrate the usefulness of our methodology by revisiting the analysis of \cite{deza2015there} on the dynamics of drug consumption amongst young adults in the United States.\footnote{This research was conducted with restricted access to Bureau of Labor Statistics (BLS) data. The views expressed here are those of the author and do not reflect the views of the BLS.}  \\
\indent To provide context, multiple studies have documented that young individuals who experiment with soft drugs have a tendency to continue using them and are at a higher risk of transitioning to hard drugs. Such correlations are certainly concerning. However, the empirical evidence of genuine causal links, in particular from softer drugs to harder drugs, remains limited with \cite{deza2015there} standing as a notable exception.  Fundamentally, these empirical regularities may be attributed to a causal effect (i.e. state dependence within and between drugs) or alternatively to latent traits that make individuals more prone to using illicit substances in general. Our primary concern is to untangle these two explanations to inform the design of policies aiming to mitigate drug addiction \footnote{See \cite{heckman1981heterogeneity} for insights on the implications of state dependence for the design of labor market policies.}. For example, if marijuana consumption indeed serves as a gateway to later cocaine use, early educational interventions cautioning against casual marijuana usage could potentially have enduring effects on the population of heavy drug users. \\
\indent To investigate these issues, we employ the restricted version of the National Longitudinal Survey of Youth 1997 (NLSY97). This is a panel dataset of 8984 individuals surveyed on a diverse range of subjects, including drug-related matters from 1997 to 2019. 
We concentrate on a subsample of four waves, spanning from 2001 to 2004. This subsample provides insight into the behavior of young adults between the age of 16 and 20 in 2001 to 19 and 24 in 2004. We shall examine the statistical association between three binary outcome variables, namely the consumption of alcohol, marijuana and hard drugs, derived from respondents answers' during annual interviews. Upon retaining those providing answers in all four waves as well as a valid state of residence, our cross section ultimately consists of $N=6317$ individuals \footnote{We adapt the sample selection procedure described in \cite{deza2015there} for the period 2001-2004.}.  Following \cite{deza2015there}, we then consider the trivariate VAR(1) logit model
\begin{align*}
    Y_{m,it}&=\mathds{1}\left\{\sum_{j=1}^{3} \gamma_{0mj}Y_{j,it-1} +\beta_{0m}age_{it}+\rho_{0m} TEDS_{m,it}+\nu_{01} \mathds{1}\{age_{it}\geq 21\}\mathds{1}\{m=1\}+A_{m,i}-\epsilon_{m,it}\geq 0\right\} 
\end{align*}
$m\in\{1,2,3\}$ ($1$=\say{alcohol}, $2$=\say{marijuana}, $3$=\say{hard drugs}), $t=1,2,3$ where $t=0$ corresponds to the year $2001$. The state-dependence coefficients $\gamma_{0mm}$ (within) and $\gamma_{0mj},m\neq j$ (between) are the principal coefficients of interest in the 16-dimensional vector of common parameters $\theta_{0}$. We are most particularly concerned about the sign and the statistical significance of $\gamma_{032}$, i.e the so called \say{stepping-stone} effect of marijuana on hard drugs. The covariate $age_{it}$ denotes the age of respondent $i$ at time $t$. The regressors $TEDS_{m,it}$ measure state-level deviations from national trends in treatment admissions for substance abuse caused by drug $m$ in year $t$ in the state of residence of $i$\footnote{ The variables $TEDS_{m,it}$ are constructed from the Treatment Episode Data Set-Admissions which records admissions to substance abuse treatment facilities in the United States.}. They are computed as the ratio of
 the share of admissions to treatment centers due to drug $m$ in the state of $i$ in year $t$ against the country wide analog in year $t$. Intuitively, this may be interpreted as a measure of exposure to substance $m$ for each respondent in our sample.\\
\indent \cite{deza2015there} parameterizes both the latent permanent heterogeneity $(A_{m,i})_{m=1}^3$ and the initial condition $Y_i^{0}$ to estimate the model by maximum likelihood. We leave these components unrestricted and exploit the valid moment functions presented in Section \ref{Section_5_VAR1}. We specifically use six of the eight valid moment functions available: $\psi_{\theta}^{k|k}(Y_{i1}^{3},Y_{i0}^{1},X_i)$ for $k\in \{(0,0,0),(0,1,0),(1,1,1),(1,1,0),(1,0,1),(1,0,0)\}$. The other two corresponding to states $k\in\{(0,0,1),(0,1,1)\}$ are null for over $99.5\%$ of our sample and were dropped to mitigate noise in estimation. Next, we (arbitrarily) select a constant, the initial condition $Y_{i}^{0}$, $age_{it}$ and the covariates $TEDS_{m,it}$ in all periods $t=1,2,3$ as instruments to form the $96\times 1$ moment vector \\
\begin{align*}
    m_{\theta}(Y_{i},Y_{i}^{0},X_i)=
    \begin{pmatrix}
    &\psi_{\theta}^{(0,0,0)|(0,0,0)}(Y_{i1}^{3},Y_{i0}^{1},X_i) \\
    &\psi_{\theta}^{(0,1,0)|(0,1,0)}(Y_{i1}^{3},Y_{i0}^{1},X_i) \\
    &\psi_{\theta}^{(1,1,1)|(1,1,1)}(Y_{i1}^{3},Y_{i0}^{1},X_i) \\
    &\psi_{\theta}^{(1,1,0)|(1,1,0)}(Y_{i1}^{3},Y_{i0}^{1},X_i) \\
    &\psi_{\theta}^{(1,0,1)|(1,0,1)}(Y_{i1}^{3},Y_{i0}^{1},X_i) \\
    &\psi_{\theta}^{(1,0,0)|(1,0,0)}(Y_{i1}^{3},Y_{i0}^{1},X_i) 
    \end{pmatrix} \otimes 
    \begin{pmatrix}
    &1 \\
    &Y_{i}^{0'} \\
    &age_{i1}^{3'} \\
    &TEDS_{1,i1}^{3'}\\
    &TEDS_{2,i1}^{3'} \\
    &TEDS_{3,i1}^{3'}
\end{pmatrix}
\end{align*}
With $m_{\theta}(Y_{i},Y_{i}^{0},X_i)$ in hand, we then consider the iterated GMM estimator of \cite{hansen1996finite}. Starting from an initial candidate $\hat{\theta}_0$\footnote{In practice, we used the GMM estimator putting equal weights on each moment as our starting candidate.}, it can be described as  
\begin{align*}
    \hat{\theta}&=\lim_{s\to\infty} \hat{\theta}_{s} \\
    \hat{\theta}_{s}&=\argmin_{\theta} \overline{m}_{N}(\theta)'\overline{W}_{N}( \hat{\theta}_{s-1})^{-1}\overline{m}_{N}(\theta) 
\end{align*}
where $\overline{m}_{N}(\theta)=\frac{1}{N}\sum_{i=1}^N m_{\theta}(Y_{i},Y_{i}^{0},X_i)$ and $\overline{W}_{N}(\theta)=\frac{1}{N}\sum_{i=1}^N m_{\theta}(Y_{i},Y_{i}^{0},X_i)m_{\theta}(Y_{i},Y_{i}^{0},X_i)'$. Under some regularity conditions (\cite{hansen2021inference}), this estimator is well defined and asymptotically normally distributed with 
\begin{align*}
    \sqrt{N}(\hat{\theta}-\theta_0) \dto \mathcal{N}(0,(M_0'W_0^{-1}M_0)^{-1})
\end{align*}
where $M_0 = \mathbb{E}\left[\pdv{m_{\theta_0}(Y_{i},Y_{i}^{0},X_i)}{\theta}\right]$ and $W_0=\mathbb{E}\left[m_{\theta_0}(Y_{i},Y_{i}^{0},X_i)m_{\theta_0}(Y_{i},Y_{i}^{0},X_i)'\right]$. Our motivation for focusing on this specific estimator originates mainly from \cite{hansen2021inference} who advocate its use for two practical reasons. First, for a given set of moments, it eliminates the arbitrariness in the choice of the initial weight matrix of 2-step GMM estimators (see also \cite{imbens2002generalized}). Second, because the iteration sequence is a contraction, each iteration is approximately variance reducing in the sense that: $Var(\hat{\theta}_{s}) \approx  c^2 Var(\hat{\theta}_{s-1})$ for some constant $c<1$ \footnote{Note that the limiting variance of the iterated GMM estimator and a 2-step GMM estimator will be identical.}. Empirically, we also found in Monte Carlo simulations that the iterated GMM estimator performs relatively well for this type of specification (see Appendix \ref{Section_MonteCarlo}).
\\ 
\indent Table \ref{table_4} presents the iterated GMM estimates for the trivariate VAR(1) logit model in columns (I), (II), (III). For comparison, columns (IV), (V), (VI) report a random effect (RE) estimator akin to \cite{deza2015there} \footnote{We borrow the specification presented in \cite{deza2015there}. The heterogeneity distribution is discrete with 3 mass points and is independent of the regressors. The initial condition relates to the covariates through a logistic regression.} while columns (VII), (VIII), (IV) display the \say{naive} logit maximum likelihood estimator (MLE) neglecting the presence of fixed effects. \\
The first observation is that, in line with conventional wisdom, GMM estimates for the state-dependence parameters within drug, $\gamma_{11},\gamma_{22},\gamma_{33}$, are all positive. As is apparent from columns (I)-(III), they are statistically significant for alcohol and marijuana but surprisingly not for hard drugs. In other words, there is no statistical evidence of a direct effect from past consumption of hard drug to future usage of hard drugs once we account for unobserved heterogeneity and the effects of other substances, at least in our four-wave sample\footnote{The transition parameters for hard drugs are expected to be noisier given that a smaller fraction of individuals consume these more lethal substances: approximately 15\% of the respondents indicate having consumed hard drugs at least once from 2001-2004. This contrasts with 86\% for alcohol and 40\% for marijuana.}. Notice that the magnitude of the estimates for $\gamma_{11},\gamma_{22},\gamma_{33}$ sharply contrast with the other two estimators. The naive MLE largely overestimates the amount of within state-dependence, yielding coefficients that are comparatively four to eight times larger. Intuitively, this can be rationalized by the fact that this estimator misinterprets any serial correlation produced by $A_i$ as evidence of state dependence. The RE estimator borrowed from \cite{deza2015there} (see also \cite{card2005estimating}, \cite{chay1998identification}) acts as an intermediate estimator between the other two as can be seen in columns (IV)-(VI). This behavior is expected to the extent that the additional parametric structure of this methodology will account to some degree for the presence of unobserved heterogeneity. We note that the role of within state dependence in the dynamics of drug consumption is nevertheless overstated by this approach.

Second and importantly, we observe in column (III) a positive and statistically significant effect of marijuana on hard drugs. This supports the view that marijuana usage can be a gateway to the consumption of harder drugs and accords with the key findings of \cite{deza2015there}. From a practical standpoint, this result corroborates that there may be scope for policies on marijuana usage to indirectly curb the consumption of more lethal substances by teenagers and young adults. The efficacy of such policies in the short and long run are important questions that will intuitively depend on the distribution of heterogeneity in the population.  We do not explore those questions here but further research in this direction would be of interest \footnote{A natural idea to gauge the effectiveness of policy interventions would be to compute average marginal effects. However, as mentioned in Section \ref{Section_5_VAR1}, we were unable to find transition functions for the transition probabilities where the state switches in VAR(1) models. This leads us to believe that only the average transition probabilities where the state remains unchanged are identified. In turn,  this would imply that average marginal effects are generally partially identified in VAR(1) models. In this case, it is possible that ideas analogous to those in \cite{dobronyi2021identification} and \cite{davezies2021identification} could be used to characterize and compute the identified set of average marginal effects; albeit some difficulties might arise due to the fact that the fixed effects are now multidimensional. Computing outer bounds as in \cite{pakel2023bounds} could be another plausible option.}. The other two estimators also agree on a positive influence of marijuana on the consumption of harder drugs, albeit it is statistically insignificant in the RE case. 

\begin{table}[!htb] 
\caption{Parameter estimates of the trivariate VAR(1) logit }
\label{table_4}
\centerline{
\scalebox{0.9}{
\begin{tabular}{lccccccccccccc} \toprule  \toprule 
  & \\
  & \multicolumn{3}{c}{Iterated GMM} & & & \multicolumn{3}{c}{Random Effects} & & & \multicolumn{3}{c}{Naive MLE}\\
    \cline{2-4}  \cline{7-9} \cline{12-14}\\
  & A & M & HD &  &  &A & M & HD & & &A & M & HD\\ 
   & (I) & (II) & (III) &  &  &(IV) & (V) & (VI) &  &  &(VII) & (VIII) & (IV)\\ 
  & \\
  $\gamma_{m1}$ & \textbf{0.30} & -0.04 & -0.02 & & & \textbf{1.41} & -0.36 & -0.2 & & & \textbf{2.44} & 0.87& 0.77\\
  & (0.12) & (0.21) & (0.32)& & & (0.16) & (0.22) & (0.63) & & & (0.06) & (0.14) &(0.37)\\
  $\gamma_{m2}$ & -0.07 & \textbf{0.70} & \textbf{0.69} & & & -0.52 & \textbf{1.48}& \textbf{0.16} & & &0.72& \textbf{2.55}&  \textbf{1.43}\\
    & (0.16) & (0.14) & (0.22) & & &(0.12) & (0.13) & (0.25) & & &(0.07)&(0.07)&(0.16)\\
  $\gamma_{m3}$ & -0.20 & 0.26 & \textbf{0.32} & & & -0.66 & -0.17 & \textbf{1.59} & & &0.22& 0.74& \textbf{2.12}\\
  & (0.27) & (0.22) & (0.21) & & & (0.19) & (0.13) & (0.13) & & &(0.12)&(0.09)& (0.12)\\
  age & 0.06 & -0.18 &  0.08 & & & 0.04 & -0.14  & -0.05 & & &-0.08& -0.13 & -0.21 \\
 & (0.05) & (0.06) & (0.09) & & & (0.6) & (0.27)  & (0.32) & & &(0.03)&(0.02)&(0.03)\\
  age $\geq 21$ & 0.04 & & & & & 0.46& & & & & 0.54\\
   & (0.11)& & & & & (0.2) & & & & & (0.07)\\
  $TEDS_{1}$& -0.09 &  &  & & & 0.96 & & & &  &0.67\\
    & (0.09) & & &  & & (0.77) & & & & &(0.50)\\
  $TEDS_{2}$&  & -0.18 &  & &  & & 0.02 & & & & &-0.13\\
    &  & (0.12) & & &  & & (0.48) & & & & &(0.30)\\
   $TEDS_{3}$&  &  & 0.42 &  & &  & & 0.15 & & & & & -0.10\\
    &  &  & (0.32) &  & &  & & (0.44) & & & & & (0.40)\\  
    & \\  
   $N$   & &6317& & & & & 6317 & & & & & 6317\\
   Periods & &2001-2004 & & & & &2001-2004 & & & & &2001-2004 \\
   \# Iterations & & 12 &\\
  \bottomrule 
\end{tabular}

}
}
\vspace{0.5cm}
\textit{\footnotesize \textsc{Notes}: 
The convergence criterion of our iterated GMM procedure is $\norm{\hat{\theta}_{s+1}-\hat{\theta}_{s}}<10^{-4}$. Estimated standard errors are reported in parenthesis.}
\end{table}
\noindent Otherwise, it is noteworthy that the between state dependence estimates can vary quite significantly across specifications.  Again, the naive MLE likely misinterprets spurious correlation from the $A_i$ as state dependence which results in positive and inflated cross effects. Column (IV) and (I) show disagreements of the RE and GMM estimates regarding the strength of the impact of marijuana and hard drugs on alcohol. Overall, this comparative exercise has showed that accounting for unobserved heterogeneity as flexibly as possible can be essential to obtain an accurate picture of the patterns of state dependence in practice.

\section{Conclusion}  \label{Section_7}
Dynamic discrete choice models are widely used to study the determinants of repeated decisions made by individuals or firms over time. 
In this paper, we have introduced a procedure to estimate a family of such models with logistic (or Type I extreme value) errors and potentially many lags while remaining agnostic about the nature of unobserved individual heterogeneity. This type of approach may be attractive when the risk of misspecifying the initial condition and the unit-specific effects are important. We also provided general expressions for average marginal effects in the binary response case which are often the counterfactuals of interest in practice. \\
\indent The list of discrete choice models covered in this paper is of course not exhaustive and it would be interesting to know if our two-step approach could be deployed in other settings with \say{logit} noise. In ongoing work, we have found that this is one avenue to approach estimation of dynamic ordered logit models, potentially of arbitrary lag order.

\nocite{*}
\bibliography{biblio}{}

\newpage
\setcounter{page}{0}
\pagenumbering{arabic}
\setcounter{page}{1}
\appendix

\vskip 3cm
\begin{center}
	{ {\LARGE \textbf{Appendix} }}
\end{center}

\section{Partial Fraction Decomposition}
\begin{lemma}\label{tech_lemma_1}
For any reals $u_1,u_2,\ldots,u_K$, $v_1,v_2,\ldots,v_K$ and $a_1,a_2,\ldots,a_K$, $K\geq 1$ we have
\begin{align*}
    &\frac{1}{1+\cleansum_{k=1}^K e^{v_k+a_k}}+\cleansum_{k=1}^K (1-e^{u_k-v_k})\frac{e^{v_k+a_k}}{\left(1+\cleansum_{k=1}^K e^{v_k+a_k}\right)\left(1+\cleansum_{k=1}^K e^{u_k+a_k}\right)}= \frac{1}{1+\cleansum_{k=1}^K e^{u_k+a_k}}
\end{align*}
and 
\begin{align*}
    &\frac{e^{v_j+a_j}}{1+\cleansum_{k=1}^K e^{v_k+a_k}}+(1-e^{-u_j+v_j})\frac{e^{u_j+a_j}}{\left(1+\cleansum_{k=1}^K e^{v_k+a_k}\right)\left(1+\cleansum_{k=1}^K e^{u_k+a_k}\right)}+\\
    &\sum_{\substack{k=1 \\ k\neq j}}^K (1-e^{(u_k-u_j)-(v_k-v_j)})\frac{e^{v_k+a_k+u_j+a_j}}{\left(1+\cleansum_{k=1}^K e^{v_k+a_k}\right)\left(1+\cleansum_{k=1}^K e^{u_k+a_k}\right)}= \frac{e^{u_j+a_j}}{1+\cleansum_{k=1}^K e^{u_k+a_k}}
\end{align*}
\end{lemma}
\begin{proof}
\begin{align*}
    &\frac{1}{1+\sum_{k=1}^K e^{v_k+a_k}}+\sum_{k=1}^K (1-e^{u_k-v_k})\frac{e^{v_k+a_k}}{\left(1+\sum_{k=1}^K e^{v_k+a_k}\right)\left(1+\sum_{k=1}^K e^{u_k+a_k}\right)}= \\
    &\frac{1+\sum_{k=1}^K e^{u_k+a_k}+\sum_{k=1}^K e^{v_k+a_k}-\sum_{k=1}^K e^{u_k+a_k}}{\left(1+\sum_{k=1}^K e^{v_k+a_k}\right)\left(1+\sum_{k=1}^K e^{u_k+a_k}\right)}\\
    &=\frac{1+\sum_{k=1}^K e^{v_k+a_k}}{\left(1+\sum_{k=1}^K e^{v_k+a_k}\right)\left(1+\sum_{k=1}^K e^{u_k+a_k}\right)} \\
    &=\frac{1}{1+\sum_{k=1}^K e^{u_k+a_k}}
\end{align*}
and 
\begin{align*}
    &\frac{e^{v_j+a_j}}{1+\sum_{k=1}^K e^{v_k+a_k}}+(1-e^{-u_j+v_j})\frac{e^{u_j+a_j}}{\left(1+\sum_{k=1}^K e^{v_k+a_k}\right)\left(1+\sum_{k=1}^K e^{u_k+a_k}\right)}+\\
    &\sum_{\substack{k=1 \\ k\neq j}}^K (1-e^{(u_k-u_j)-(v_k-v_j)})\frac{e^{v_k+a_k+u_j+a_j}}{\left(1+\sum_{k=1}^K e^{v_k+a_k}\right)\left(1+\sum_{k=1}^K e^{u_k+a_k}\right)}= \\
    &\frac{e^{v_j+a_j}+\sum_{k=1}^K e^{v_j+a_j+u_k+a_k}+e^{u_j+a_j}-e^{v_j+a_j}+\cleansum_{\substack{k=1 \\ k\neq j}}^K e^{v_k+a_k+u_j+a_j}-\sum_{\substack{k=1 \\ k\neq j}}^K e^{v_j+a_j+u_k+a_k}}{\left(1+\sum_{k=1}^K e^{v_k+a_k}\right)\left(1+\sum_{k=1}^K e^{u_k+a_k}\right)}\\
    &=\frac{e^{u_j+a_j}+e^{v_j+a_j+u_j+a_j}+\cleansum_{\substack{k=1 \\ k\neq j}}^K e^{v_k+a_k+u_j+a_j}}{\left(1+\sum_{k=1}^K e^{v_k+a_k}\right)\left(1+\sum_{k=1}^K e^{u_k+a_k}\right)} \\
    &=\frac{e^{u_j+a_j}\left(1+\cleansum_{k=1}^K e^{v_k+a_k}\right)}{\left(1+\sum_{k=1}^K e^{v_k+a_k}\right)\left(1+\sum_{k=1}^K e^{u_k+a_k}\right)} \\
    & =\frac{e^{u_j+a_j}}{1+\sum_{k=1}^K e^{u_k+a_k}}
\end{align*}
\end{proof}

\begin{lemma} \label{tech_lemma_2}
Fix $M\geq 2$, let $\mathcal{Y}=\{0,1\}^M$. Then, for any $k\in \mathcal{Y}$ and any reals $u_1,u_2,\ldots,u_M$, $v_1,v_2,\ldots,v_M$ and $a_1,a_2,\ldots,a_M$, we have

\begin{align*}
    \prod_{m=1}^M \frac{e^{k_m(v_m+a_{m})}}{1+e^{v_m+a_{m}}} +\cleansum_{l\in \mathcal{Y}\setminus \{k\}} \left[1-e^{\sum_{j=1}^M (l_j-k_j)(u_j-v_j)}\right]\prod_{m=1}^M\frac{e^{k_m(u_m+a_{m})}}{1+e^{u_m+a_{m}}}\frac{e^{l_m(v_m+a_{m})}}{1+e^{v_m+a_{m}}}=\prod_{m=1}^M\frac{e^{k_m(u_m+a_{m})}}{1+e^{u_m+a_{m}}}
\end{align*}
\end{lemma}
\begin{proof}
Let 
\begin{align*}
LHS= \prod_{m=1}^M \frac{e^{k_m(v_m+a_{m})}}{1+e^{v_m+a_{m}}} +\cleansum_{l\in \mathcal{Y}\setminus \{k\}} \left[1-e^{\sum_{j=1}^M (l_j-k_j)(u_j-v_j)}\right]\prod_{m=1}^M\frac{e^{k_m(u_m+a_{m})}}{1+e^{u_m+a_{m}}}\frac{e^{l_m(v_m+a_{m})}}{1+e^{v_m+a_{m}}}
\end{align*}
and let $Num$ denote the numerator of $LHS$. We have:
\begin{align*}
    Num&=Num_1+Num_2 \\
    Num_1&=\prod_{m=1}^M e^{k_m(v_m+a_{m})}(1+e^{u_m+a_{m}}) \\
    Num_2&=\cleansum_{l\in \mathcal{Y}\setminus \{k\}} \left[1-e^{\sum_{j=1}^M (l_j-k_j)(u_j-v_j)}\right]\prod_{m=1}^M e^{k_m(u_m+a_{m})+l_m(v_m+a_{m})} \\
    &=\prod_{m=1}^M e^{k_m(u_m+a_{m})}\cleansum_{l\in \mathcal{Y}\setminus \{k\}}\prod_{m=1}^Me^{l_m(v_m+a_{m})}-\cleansum_{l\in \mathcal{Y}\setminus \{k\}}e^{\sum_{j=1}^M l_j(u_j+a_j)+k_j(v_j+a_{j})} \\
    &=\prod_{m=1}^M e^{k_m(u_m+a_{m})}\cleansum_{l\in \mathcal{Y}\setminus \{k\}}\prod_{m=1}^Me^{l_m(v_m+a_{m})}-\prod_{m=1}^Me^{ k_m(v_m+a_{m})} \cleansum_{l\in \mathcal{Y}\setminus \{k\}} \prod_{m=1}^M e^{ l_m(u_m+a_m)}
\end{align*}
Now, noting that 
\begin{align*}
    &\cleansum_{l\in \mathcal{Y}}\prod_{m=1}^M e^{l_m(v_m+a_{m})}=\prod_{m=1}^M (1+e^{v_m+a_{m}}) \\
    &\cleansum_{l\in \mathcal{Y}}\prod_{m=1}^M e^{l_m(u_m+a_{m})}=\prod_{m=1}^M (1+e^{u_m+a_{m}})
\end{align*}
we get
\begin{align*}
    Num_2
    &=\prod_{m=1}^M e^{k_m(u_m+a_{m})}\cleansum_{l\in \mathcal{Y}\setminus \{k\}}\prod_{m=1}^Me^{l_m(v_m+a_{m})}-\prod_{m=1}^Me^{ k_m(v_m+a_{m})} \cleansum_{l\in \mathcal{Y}\setminus \{k\}} \prod_{m=1}^M e^{ l_m(u_m+a_m)} \\
    &=\prod_{m=1}^M e^{k_m(u_m+a_{m})}\left(\prod_{m=1}^M (1+e^{v_m+a_{m}})-\prod_{m=1}^M e^{k_m(v_m+a_{m})}\right) \\
    &-\prod_{m=1}^Me^{ k_m(v_m+a_{m})}\left(\prod_{m=1}^M (1+e^{u_m+a_{m}})-\prod_{m=1}^M e^{k_m(u_m+a_{m})}\right) \\
    &=\prod_{m=1}^M e^{k_m(u_m+a_{m})}(1+e^{v_m+a_{m}})-\prod_{m=1}^M  e^{k_m(v_m+a_{m})}(1+e^{u_m+a_{m}}) \\
    &=\prod_{m=1}^M e^{k_m(u_m+a_{m})}(1+e^{v_m+a_{m}})-Num_1
\end{align*}
It follows that $Num=\prod_{m=1}^M e^{k_m(u_m+a_{m})}(1+e^{v_m+a_{m}})$ and consequently
\begin{align*}
    LHS = \frac{\prod_{m=1}^M e^{k_m(u_m+a_{m})}(1+e^{v_m+a_{m}})}{\prod_{m=1}^M (1+e^{u_m+a_{m}})(1+e^{v_m+a_{m}})}=\prod_{m=1}^M \frac{e^{k_m(u_m+a_{m})}}{1+e^{u_m+a_{m}}}
\end{align*}
\end{proof}

\section{Connection to Kitazawa and Honoré-Weidner}  \label{mapping_to_kitazawa_honoreweidner}
Recall from Proposition \ref{proposition_2} that when $T\geq 3$, our simplest moment conditions for $t,s$ such that  $T-1\geq t> s\geq 1$ write:
\begin{align*}
    \psi_{\theta}^{0|0}(Y_{it-1}^{t+1},Y_{is-1}^s,X_i)&=\phi_{\theta}^{0|0}(Y_{it+1},Y_{it},Y_{it-1},X_i)-\zeta_{\theta}^{0|0}(Y_{it-1}^{t+1},Y_{is-1}^s,X_i) \\
    &=\phi_{\theta}^{0|0}(Y_{it+1},Y_{it},Y_{it-1},X_i)-(1-Y_{is})-\omega_{t,s}^{0|0}(\theta)Y_{is}\phi_{\theta}^{0|0}(Y_{it+1},Y_{it},Y_{it-1},X_i) \\
     \psi_{\theta}^{1|1}(Y_{it-1}^{t+1},Y_{is-1}^s,X_i)&=\phi_{\theta}^{1|1}(Y_{it+1},Y_{it},Y_{it-1},X_i)-\zeta_{\theta}^{1|1}(Y_{it-1}^{t+1},Y_{is-1}^s,X_i) \\
    &=\phi_{\theta}^{1|1}(Y_{it+1},Y_{it},Y_{it-1},X_i)-Y_{is}-\omega_{t,s}^{1|1}(\theta)(1-Y_{is})\phi_{\theta}^{1|1}(Y_{it+1},Y_{it},Y_{it-1},X_i)
\end{align*}
where we know from Lemma \ref{lemma_3} that
\begin{align*}
\omega_{t,s}^{0|0}(\theta)&=1-e^{(\kappa_{t}^{0|0}(\theta)-\mu_{s}(\theta))} \\
&=1-e^{(X_{it+1}-X_{is})'\beta-\gamma Y_{is-1}} \\
\omega_{t,s}^{1|1}(\theta)&=1-e^{-(\kappa_{t}^{1|1}(\theta)-\mu_{s}(\theta))} \\
&=1-e^{-\gamma(1-Y_{is-1})-(X_{it+1}-X_{is})'\beta)}
\end{align*}

Now, note that:
\begin{align*}
    &\tanh \left( \frac{\gamma(1-Y_{it-2})+(\Delta X_{it}+\Delta X_{it+1})'\beta}{2}\right)=\frac{1-e^{-\left(\gamma(1-Y_{it-2})+(\Delta X_{it}+\Delta X_{it+1})'\beta \right)}}{1+e^{-\left(\gamma(1-Y_{it-2})+(\Delta X_{it}+\Delta X_{it+1})'\beta \right)}}=\frac{\omega_{t,t-1}^{1|1}(\theta)}{2-\omega_{t,t-1}^{1|1}(\theta)} \\
    &\tanh \left( \frac{-\gamma Y_{it-2}+(\Delta X_{it}+\Delta X_{it+1})'\beta}{2}\right)=\frac{e^{-\gamma Y_{it-2}+(\Delta X_{it}+\Delta X_{it+1})'\beta}-1}{e^{-\gamma Y_{it-2}+(\Delta X_{it}+\Delta X_{it+1})'\beta}+1} =-\frac{\omega_{t,t-1}^{0|0}(\theta)}{2-\omega_{t,t-1}^{0|0}(\theta)}
\end{align*}
and $\phi_{\theta}^{1|1}(Y_{it+1},Y_{it},Y_{it-1},X_i)= \Upsilon_{it}$ and $1-\phi_{\theta}^{0|0}(Y_{it+1},Y_{it},Y_{it-1},X_i)=U_{it}$. Thus, we have:

\begin{align*}
        (2-\omega_{t,t-1}^{0|0}(\theta))\hbar U_{it}&=(2-\omega_{t,t-1}^{0|0}(\theta))(U_{it}-Y_{it-1})+\omega_{t,t-1}^{0|0}(\theta)\left(U_{it}+Y_{it-1}-2U_{it}Y_{it-1}\right) \\
        &=2\left[U_{it}-Y_{it-1}+\omega_{t,t-1}^{0|0}(\theta)Y_{it-1}(1-U_{it})\right] \\
        &=2\left[1-\phi_{\theta}^{0|0}(Y_{it+1},Y_{it},Y_{it-1},X_i)-Y_{it-1}+\omega_{t,t-1}^{0|0}(\theta)Y_{it-1}\phi_{\theta}^{0|0}(Y_{it+1},Y_{it},Y_{it-1},X_i)\right] \\
        &=-2\left[\phi_{\theta}^{0|0}(Y_{it+1},Y_{it},Y_{it-1},X_i)-(1-Y_{it-1})-\omega_{t,t-1}^{0|0}(\theta)Y_{it-1}\phi_{\theta}^{0|0}(Y_{it+1},Y_{it},Y_{it-1},X_i)\right] \\
        &=-2\psi_{\theta}^{0|0}(Y_{it-1}^{t+1},Y_{it-2}^{t-1},X_i) \\
        (2-\omega_{t,t-1}^{1|1}(\theta))\hbar \Upsilon_{it}&= (2-\omega_{t,t-1}^{1|1}(\theta))(\Upsilon_{it}-Y_{it-1})-\omega_{t,t-1}^{1|1}(\theta)\left(\Upsilon_{it}+Y_{it-1}-2\Upsilon_{it}Y_{it-1}\right) \\
        &=2\left[\Upsilon_{it}-Y_{it-1}-\omega_{t,t-1}^{1|1}(\theta)\Upsilon_{it}\left(1-Y_{it-1}\right)\right] \\
        &=2\left[\phi_{\theta}^{1|1}(Y_{it+1},Y_{it},Y_{it-1},X_i)-Y_{it-1}-\omega_{t,t-1}^{1|1}(\theta)\phi_{\theta}^{1|1}(Y_{it+1},Y_{it},Y_{it-1},X_i)\left(1-Y_{it-1}\right)\right] \\
        &=2  \psi_{\theta}^{1|1}(Y_{it-1}^{t+1},Y_{it-2}^{t-1},X_i)
\end{align*}

\noindent To establish the connection to the work of \cite{honore2020moment}, it is useful to re-write the moment functions slightly differently. By re-arranging terms, one obtains the following for $T=3$
\begin{align} \label{moment_b_honoreweidner}
\begin{split}
    \psi_{\theta}^{0|0}(Y_{1}^{3},Y_{i0}^1,X_i)&=(1-Y_{i1})\phi_{\theta}^{0|0}(Y_{i1}^{3},X_i)+e^{(X_{i3}-X_{i1})'\beta-\gamma Y_{i0}}Y_{i1}\phi_{\theta}^{0|0}(Y_{i1}^{3},X_i)-(1-Y_{i1}) \\
    &=e^{(X_{i2}-X_{i3})'\beta }(1-Y_{i1})(1-Y_{i2})Y_{i3} +(1-Y_{i1})(1-Y_{i2})(1-Y_{i3}) \\
    &+e^{(X_{i2}-X_{i1})'\beta+\gamma (1-Y_{i0})}Y_{i1}(1-Y_{i2})Y_{i3} \\
   &+e^{(X_{i3}-X_{i1})'\beta-\gamma Y_{i0}}Y_{i1}(1-Y_{i2})(1-Y_{i3}) \\
   &-(1-Y_{i1}) \\
    &=(e^{(X_{i2}-X_{i3})'\beta }-1)(1-Y_{i1})(1-Y_{i2})Y_{i3}  \\
   &+e^{(X_{i2}-X_{i1})'\beta+\gamma (1-Y_{i0})}Y_{i1}(1-Y_{i2})Y_{i3} \\
   &+e^{(X_{i3}-X_{i1})'\beta-\gamma Y_{i0}}Y_{i1}(1-Y_{i2})(1-Y_{i3}) \\
   &-(1-Y_{i1})Y_{i2}
 \end{split}
\end{align}
where the last line uses the fact that: $(1-Y_{i1})=(1-Y_{i1})Y_{i2}+(1-Y_{i1})(1-Y_{i2})Y_{i3}+(1-Y_{i1})(1-Y_{i2})(1-Y_{i3})$ to make some cancellations. For the initial condition, $Y_{i0}=0$, equation (\ref{moment_b_honoreweidner}) corresponds to their moment function $m_{0}^{b}$ which they express in an extensive form. For $Y_{i0}=1$, we get instead $m_{1}^{b}$. Similarly,
\begin{align} \label{moment_a_honoreweidner}
\begin{split}
    \psi_{\theta}^{1|1}(Y_{i1}^{3},Y_{i0}^1,X_i)
    &=Y_{i1}\phi_{\theta}^{1|1}(Y_{i1}^3,X_i)+e^{-\gamma(1-Y_{i0})-(X_{i3}-X_{i1})'\beta)}(1-Y_{i1})\phi_{\theta}^{1|1}(Y_{i1}^3,X_i)-Y_{i1} \\
    &=e^{(X_{i3}-X_{i2})'\beta}Y_{i1}Y_{i2}(1-Y_{i3})+Y_{i1}Y_{i2}Y_{i3} \\
    &+e^{(X_{i1}-X_{i2})'\beta+\gamma Y_{i0}}(1-Y_{i1})Y_{i2}(1-Y_{i3}) \\ &+e^{(X_{i1}-X_{i3})'\beta-\gamma(1-Y_{i0})}(1-Y_{i1})Y_{i2}Y_{i3} \\
    &-Y_{i1} \\
    &=(e^{(X_{i3}-X_{i2})'\beta}-1)Y_{i1}Y_{i2}(1-Y_{i3}) \\
    &+e^{(X_{i1}-X_{i2})'\beta+\gamma Y_{i0}}(1-Y_{i1})Y_{i2}(1-Y_{i3}) \\ &+e^{(X_{i1}-X_{i3})'\beta-\gamma(1-Y_{i0})}(1-Y_{i1})Y_{i2}Y_{i3} \\
    &-Y_{i1}(1-Y_{i2})
\end{split}
\end{align}
where the last line uses the fact that: $Y_{i1}=Y_{i1}(1-Y_{i2})+Y_{i1}Y_{i2}Y_{i3}+Y_{i1}Y_{i2}(1-Y_{i3})$. For the initial condition $Y_{i0}=0$, equation (\ref{moment_a_honoreweidner})  gives their moment function $m_{0}^{a}$ and for $Y_{i0}=1$, we get $m_{1}^{a}$. Our moments are thus identical, at least for the case $T=3$.

\section{The remaining steps for the AR($p$) model with $p>1$} \label{Section_ARp_othersteps}
\indent As indicated in Section \ref{Section_ARp_transifunc} , \textbf{Step 1)} (b) is now analogous to the AR(1) case since the transition probabilities keep an identical structure. As soon as $T\geq p+2$,  we can construct transition functions other than $\phi_{\theta}^{y_1|y_{1}^{p}}(Y_{it+1},Y_{it},Y^{t-1}_{it-(2p-1)},X_i)$ also associated to $\pi^{y_1|y_{1}^{p}}_{t}(A_i,X_{i})$, for $ y_{1}^p \in \mathcal{Y}^{p}$ in periods $t\in\{p+1,\ldots,T-1\}$,. These new transition functions that we denote $\zeta_{\theta}^{y_1|y_{1}^{p}}(.)$ take the form of a weighted combination of past outcome $\mathds{1}(Y_{is}=y_1)$,  $s\in\{1,\ldots,t-p\}$  and the interaction of $\mathds{1}(Y_{is}\neq y_1)$ with any transition function  whose conditioning set encompasses $Y_{is}$ for it to map to $\pi^{y_1|y_{1}^{p}}_{t}(A_i,X_{i})$. The simplest examples which are also the only ones available when $T=p+2$, are given in Lemma  \ref{lemma_4}.
\begin{lemma} \label{lemma_4}
In model (\ref{ARp_logit_general}) with $T\geq p+2$, for all $t\in\{p+1,\ldots,T-1\}$, $s\in\{1,\ldots,t-p\}$ and
$y_{1}^p\in \mathcal{Y}^p$, let
\begin{align*}
    \mu_{s}(\theta)&=\sum_{r=1}^{p}\gamma_{0r} Y_{is-r}+X_{is}'\beta \\
    \kappa_{t}^{y_1|y_{1}^p}(\theta)&=\sum_{r=1}^{p}\gamma_{0r}y_r+X_{it+1}'\beta \\
    \omega_{t,s}^{y_1|y_{1}^p}(\theta)&=\left[1-e^{(\kappa^{y_1|y_{1}^p}_{t}(\theta)-\mu_{s}(\theta))}\right]^{1-y_{1}}\left[1-e^{-(\kappa^{y_1|y_{1}^p}_{t}(\theta)-\mu_{s}(\theta))}\right]^{y_{1}}
\end{align*}
and define the moment functions:
\begin{align*}
 \zeta_{\theta}^{y_1|y_{1}^p}(Y^{t+1}_{it-(2p-1)},Y_{is-p}^{s},X_i)&=\mathds{1}\{Y_{is}=y_1\}+\omega_{t,s}^{y_1|y_{1}^p}(\theta)\mathds{1}\{Y_{is}\neq y_1\}\phi_{\theta}^{y_1|y_{1}^{p}}(Y_{it+1},Y_{it},Y^{t-1}_{it-(2p-1)},X_i) 
\end{align*}
Then,
\begin{align*}
    \mathbb{E}\left[ \zeta_{\theta_0}^{y_1|y_{1}^p}(Y^{t+1}_{it-(2p-1)},Y_{is-p}^{s},X_i)|Y_i^0,Y_{i1}^{s-1},X_i,A_i\right]&= \pi^{y_1|y_{1}^p}_t(A_i,X_{i})
\end{align*}
\end{lemma}
\noindent Unsurprisingly, as in the AR(1) case, it becomes possible to construct iteratively more transition functions from those given in Lemma \ref{lemma_4} when at least $T=p+3$ periods are observed post initial condition. They are given in Corollary \ref{corollary_4}
below.
\begin{corollary} \label{corollary_4}
In model (\ref{ARp_logit_general}) with $T\geq p+3$, for all $t\in\{p+1,\ldots,T-1\}$ and collection of ordered indices $s_{1}^J$ with $J\geq 2$ satisfying 
$ t-p\geq  s_1>\ldots>s_J\geq 1$, and for all $y_{1}^p\in \mathcal{Y}^p$, let
\begin{align*}
    \zeta_{\theta}^{0|0,y_{2}^p}(Y^{t+1}_{it-(2p-1)},Y_{is_1-p}^{s_1},\ldots,Y_{is_{J}-p}^{s_J},X_i)&=(1-Y_{is_J})+\omega_{t,s_J}^{0|0,y_{2}^p}(\theta)Y_{is_J} \zeta_{\theta}^{0|0,y_{2}^p}(Y_{it-1}^{t+1},Y_{is_1-p}^{s_1},\ldots,Y_{is_{J-1}-p}^{s_{J-1}},X_i) \\
    \zeta_{\theta}^{1|1,y_{2}^p}(Y^{t+1}_{it-(2p-1)},Y_{is_1-p}^{s_1},\ldots,Y_{is_{J}-p}^{s_J},X_i)&=Y_{is_j}+\omega_{t,s_J}^{1|1,y_{2}^p}(\theta)(1-Y_{is_J}) \zeta_{\theta}^{1|1,y_{2}^p}(Y_{it-1}^{t+1},Y_{is_1-p}^{s_1},\ldots,Y_{is_{J-1}-p}^{s_{J-1}},X_i)
\end{align*}
with weights $\omega_{t,s_J}^{y_1|y_{1}^p}(\theta)$ defined as in Lemma \ref{lemma_4}. Then,
\begin{align*}
    \mathbb{E}\left[ \zeta_{\theta_0}^{y_1|y_1^p}(Y^{t+1}_{it-(2p-1)},Y_{is_1-p}^{s_1},\ldots,Y_{is_{J}-p}^{s_J},X_i)|Y_i^0,Y_{i1}^{s_J-1},X_i,A_i\right]&= \pi^{y_1|y_{1}^p}_t(A_i,X_{i})
\end{align*}
\end{corollary}

\indent \textbf{Step 2)}. Provided that $T\geq p+2$, it is clear that the difference between any two distinct transition functions associated to the same transition probability in $t\in\{p+1,\ldots,T-1\}$ will yield a valid moment function.  Proposition \ref{proposition_3} hereinbelow presents one set of valid moment functions that generalize those obtained previously for the one lag case.
\begin{proposition}\label{proposition_3}
In model (\ref{ARp_logit_general}) \\
if $T\geq p+2$,
for all $t\in\{p+1,\ldots,T-1\}$, $s\in\{1,\ldots,t-p\}$ and
$y_{1}^p\in \mathcal{Y}^p$, let
\begin{align*}
    \psi_{\theta}^{y_1|y_{1}^p}(Y^{t+1}_{it-(2p-1)},Y_{is-p}^s,X_i)&=\phi_{\theta}^{y_1|y_{1}^{p}}(Y^{t+1}_{it-(2p-1)},X_i)-\zeta_{\theta}^{y_1|y_{1}^p}(Y^{t+1}_{it-(2p-1)},Y_{is-p}^{s},X_i), 
\end{align*}
if $T\geq p+3$, for all $t\in\{p+1,\ldots,T-1\}$ and collection of ordered indices $s_{1}^J$ with $J\geq 2$ satisfying
$ t-p\geq  s_1>\ldots>s_J\geq 1$, and for all $y_{1}^p\in \mathcal{Y}^p$, let
\begin{align*}
    \psi_{\theta}^{y_1|y_{1}^p}(Y^{t+1}_{it-(2p-1)},Y_{is_1-p}^{s_1},\ldots,Y_{is_{J}-p}^{s_J},X_i)&=\phi_{\theta}^{y_1|y_{1}^{p}}(Y^{t+1}_{it-(2p-1)},X_i)-\zeta_{\theta}^{y_1|y_{1}^p}(Y^{t+1}_{it-(2p-1)},Y_{is_1-p}^{s_1},\ldots,Y_{is_{J}-p}^{s_J},X_i)
\end{align*}
Then,
\begin{align*}
   &\mathbb{E}\left[ \psi_{\theta_0}^{y_1|y_{1}^p}(Y^{t+1}_{it-(2p-1)},Y_{is-p}^s,X_i)|Y_i^0,Y_{i1}^{s-1},X_i,A_i\right]=0 \\
    &\mathbb{E}\left[\psi_{\theta_0}^{y_1|y_{1}^p}(Y^{t+1}_{it-(2p-1)},Y_{is_1-p}^{s_1},\ldots,Y_{is_{J}-p}^{s_J},X_i)|Y_i^0,Y_{i1}^{s_J-1},X_i,A_i\right]=0
\end{align*}
\end{proposition}

\indent This family of moment functions features precisely $2^T-(T+1-p)2^{p}$ distinct elements for any initial condition. Indeed, fix $Y_{i}^0$ and a $p$-vector $y_{1}^p\in \{0,1\}^p$. Then, for a given time period $t\in\{p+1,\ldots,T-1\}$, there are $\binom{t-p}{1}$ moments of the form $\psi_{\theta}^{y_1|y_{1}^p}(Y^{t+1}_{it-(2p-1)},Y_{is-p}^{s},X_i)$ corresponding to choices of $s\in \{1,\ldots,t-p\}$. Moreover, by choosing any feasible sequence $s_{1}^J$, $J\geq 2$, verifying $t-p\geq  s_1>\ldots>s_J \geq 1$ we produce another $\sum_{l=2}^{t-p}\binom{t-p}{l}$ moment functions of the form $\psi_{\theta}^{y_1|y_{1}^p}(Y^{t+1}_{it-(2p-1)},Y_{is_1-p}^{s_1},\ldots,Y_{is_{J}-p}^{s_J},X_i)$. In total, for period $t$, we count :
\begin{align*}
    \sum_{l=1}^{t-p}\binom{t-p}{l}=2^{t-p}-1
\end{align*}
valid moments. Now, summing over all possible values for $t\in\{p+1,\ldots,T-1\}$ and multiplying by the number of distinct values for $y_{1}^p$, namely $2^p$, we get:
\begin{align*}
    2^p\sum_{t=p+1}^{T-1}\sum_{l=1}^{t-p}\binom{t-p}{l}=  2^p\sum_{t=p+1}^{T-1}(2^{t-p}-1)=2^p\left(2\frac{1-2^{T-p-1}}{1-2}-(T-p-1)\right)=2^{T}-(T+1-p)2^p
\end{align*}
Numerical experimentation for various values of $T$ in the AR(1) and AR(2) cases  suggest that the moment functions of Proposition \ref{proposition_3} are effectively linearly independent. Therefore, Theorem \ref{theorem_nummoments_ARp} implies that they constitute a complete family of moment functions for AR($p$) models. From a practical standpoint, this shows that functional differencing at least in panel data logit models can be broken down into a series of equivalent simpler subproblems period by period that find all moment equality restrictions. Our procedure can be advantageous in sophisticated models with a few lags where an analysis of the full likelihood, a high dimensional object, can prove difficult.

\section{Simulation Experiments} \label{Section_MonteCarlo}
In this section, we report the results of a small set of simulations designed to assess the finite sample performance of GMM estimators based on our moment conditions.  

\subsection{Monte Carlo for an AR(3) logit model}

\noindent For our first example, we consider an AR(3) logit model with $T=5$ periods (i.e 8 periods in total with the initial condition) and a single exogenous covariate. We set the common parameters to $\gamma_{01}=1.0$, $ \gamma_{02}=0.5, \gamma_{03}=0.25, \beta_0=0.5$ and use the following generative model in the spirit of \cite{honore2000panel}:
\begin{align*}
    Y_{i-2}&=\mathds{1}\{X_{i-2}'\beta_0+A_i -\epsilon_{i-2}\geq 0\} \\
    Y_{i-1}&=\mathds{1}\{\gamma_{01}Y_{i-2}+X_{i-1}'\beta_0+A_i -\epsilon_{i-1}\geq 0\} \\
    Y_{i0}&=\mathds{1}\{\gamma_{01}Y_{i-1}+\gamma_{02}Y_{i-2}+X_{i0}'\beta_0+A_i -\epsilon_{i0}\geq 0\} \\
    Y_{it}&=\mathds{1}\left\{\gamma_{01}Y_{it-1}+\gamma_{02}Y_{it-2}+\gamma_{03}Y_{it-3}+X_{it}'\beta_0+A_i-\epsilon_{it}\geq 0\right\}, \quad t= 1,\ldots, 5
 \end{align*}
The disturbances $\epsilon_{it}$ are iid standard logistic over time, $X_{it}$ is iid $\mathcal{N}(0,1)$ and the fixed effects are computed as  $A_i=\frac{1}{\sqrt{8}}\cleansum_{t=-2}^5 X_{it}$. To evaluate the performance of the estimators described below, we simulate data for four sample sizes : 500, 2000, 8000, 16000, and perform 1000 Monte Carlo replications for each design. \\
\indent For $T=5$, we know from Proposition \ref{proposition_3} that $8$ valid  moment functions are available, each stemming from the 8 possible transition probabilities of the model (there are really 16 transition probabilities in total but 8 are redundant since probabilities sum to one). We consider the interaction of all 8 valid moment functions with a constant, the 3 initial conditions $Y_{i-2},Y_{i-1},Y_{i0}$ and the covariates $X_{it}$ in each period $t\in\{1,\ldots,5\}$ to construct the $72\times 1$ moment vector:
\begin{align*}
    m_{\theta}(Y_{i},Y_{i}^0,X_i)=
    \begin{pmatrix}
    &\psi_{\theta}^{0|0,0,0}(Y^{5}_{i-1},Y_{i-2}^1,X_i) \\
    &\psi_{\theta}^{0|0,0,1}(Y^{5}_{i-1},Y_{i-2}^1,X_i) \\
    &\psi_{\theta}^{0|0,1,0}(Y^{5}_{i-1},Y_{i-2}^1,X_i) \\
    &\psi_{\theta}^{0|0,1,1}(Y^{5}_{i-1},Y_{i-2}^1,X_i) \\
    &\psi_{\theta}^{1|1,0,0}(Y^{5}_{i-1},Y_{i-2}^1,X_i) \\
    &\psi_{\theta}^{1|1,0,1}(Y^{5}_{i-1},Y_{i-2}^1,X_i) \\
    &\psi_{\theta}^{1|1,1,0}(Y^{5}_{i-1},Y_{i-2}^1,X_i) \\
    &\psi_{\theta}^{1|1,1,1}(Y^{5}_{i-1},Y_{i-2}^1,X_i)
    \end{pmatrix} \otimes 
    \begin{pmatrix}
    &1 \\
    &Y_{i-2} \\
    &Y_{i-1} \\
    &Y_{i0} \\
    &X_{i1}^{5'} 
\end{pmatrix}
\end{align*}
where $\otimes$ denotes the standard Kronecker product. The choice of this particular set of instruments is of course arbitrary and only motivated by simplicity. We also consider a rescaled version of $ m_{\theta}(Y_{i},Y_{i}^0,X_i)$ that we denote $ \widetilde{m_{\theta}}(Y_{i},Y_{i}^0,X_i)$ where each of the 8 valid moment functions are appropriately rescaled so that $\forall y_{1}^3\in\{0,1\}^3$, $\sup_{X_i,Y_i,\theta} \abs{\psi_{\theta}^{y_1|y_1,y_2,y_3}(Y^{5}_{i-1},Y_{i-2}^1,X_i)}<\infty$. We do so by normalizing  $\psi_{\theta}^{y_1|y_1,y_2,y_3}(Y^{5}_{i-1},Y_{i-2}^1,X_i)$ by the sum of the absolute values of all unique values it can take as a function over choice histories $Y_{i1}^5$. The rationale for normalizing the moments originates from \cite{honore2020moment} who presented numerical evidence that a rescaling of this kind improved the finite sample performance of their estimators in the one and two lags cases. Given, $m_{\theta}(Y_{i},Y_{i}^0,X_i)$ and $ \widetilde{m_{\theta}}(Y_{i},Y_{i}^0,X_i)$, we study the properties of two simple GMM estimators:
\begin{align*}
    \hat{\theta}^{a}&=\argmax_{\theta \in \mathbb{R}^{4}} \left(\frac{1}{N}\sum_{i=1}^N  m_{\theta}(Y_{i},Y_{i}^0,X_i)\right)'\left(\frac{1}{N}\sum_{i=1}^N  m_{\theta}(Y_{i},Y_{i}^0,X_i)\right) \\
    \hat{\theta}^{b}&=\argmax_{\theta \in \mathbb{R}^{4}} \left(\frac{1}{N}\sum_{i=1}^N  \widetilde{m_{\theta}}(Y_{i},Y_{i}^0,X_i)\right)'\left(\frac{1}{N}\sum_{i=1}^N \widetilde{m_{\theta}}(Y_{i},Y_{i}^0,X_i)\right) 
\end{align*}
which both put equal weight on their individual components (i.e the weight matrix is the identity)\footnote{In a previous version of this paper we also considered a two-step \say{rescaled} estimator that uses a diagonal weight matrix with the inverse variance of each component in the spirit of \cite{honore2020moment}. It performs very similarly to the equally-weighted estimator $\hat{\theta}^{b}$.}. Under standard regularity conditions, $\hat{\theta}^{a},\hat{\theta}^{b}$ should be consistent and asymptotically normal.
\begin{table}[!htb]
\caption{Performance of GMM estimators for the AR(3)}  \label{table_1}
\centerline{
\begin{tabular}{llccccccccccc} \toprule  \toprule  
     &  &    $\hat{\gamma_1}^{a}$&$\hat{\gamma_1}^{b}$ &  & $\hat{\gamma_2}^{a}$&$\hat{\gamma_2}^{b}$ & & $\hat{\gamma_3}^{a}$&$\hat{\gamma_3}^{b}$ & & $\hat{\beta}^{a}$&$\hat{\beta}^{b}$ \\ 
\cline{3-4} \cline{6-7} \cline{9-10} \cline{12-13} 
 $N=500$ & \\ \cline{1-1} 
& Bias & -0.52 & -0.50 &  & -0.51 & -0.50 &  & -0.39 & -0.32 &  & -0.15 & 0.10 \\ 
 & MAE & 0.52 & 0.69 &  & 0.51 & 0.58 &  & 0.39 & 0.51 &  & 0.15 & 0.14 \\ 
 $N=2000$ & \\ \cline{1-1} 
& Bias & -0.37 & -0.10 &  & -0.45 & -0.12 &  & -0.31 & -0.04 &  & -0.08 & 0.02 \\ 
 & MAE & 0.37 & 0.42 &  & 0.45 & 0.34 &  & 0.31 & 0.25 &  & 0.08 & 0.06 \\ 
 $N=8000$ & \\ \cline{1-1} 
& Bias & -0.24 & 0.04 &  & -0.32 & 0.01 &  & -0.21 & 0.01 &  & -0.04 & 0.00 \\ 
 & MAE & 0.24 & 0.17 &  & 0.32 & 0.15 &  & 0.21 & 0.11 &  & 0.04 & 0.03 \\ 
 $N=16000$ & \\ \cline{1-1} 
& Bias & -0.18 & 0.01 &  & -0.25 & 0.00 &  & -0.16 & 0.00 &  & -0.03 & 0.00 \\ 
 & MAE & 0.18 & 0.11 &  & 0.25 & 0.10 &  & 0.16 & 0.07 &  & 0.03 & 0.02 \\ 
 \bottomrule \bottomrule 
\end{tabular}
}
\vspace{0.5cm}
\textit{\footnotesize \textsc{Notes}: 
Bias and MAE stand for median bias and median absolute error respectively. Reported results are based on a 1000 replications of the DGP.}
\end{table}

\indent Table \ref{table_1} presents the median bias and median absolute errors of the two GMM estimators for each design $N\in\{500,2000,8000,16000\}$. Figure \ref{figure_montecarlo_AR3_T5}  plots their densities which as expected resemble gaussian distributions for the larger values of $N$. Interestingly, a first observation is that both estimators appear to suffer from a negative bias on the lag parameters at least up to $N=2000$. And while this bias effectively vanishes for the \say{rescaled} GMM estimators for the larger sample size $N\geq 8000$, it remains quite significant for all lag parameters and also the slope coefficient for the \say{unnormalized} estimator. This is evident from the sign of the bias in Table \ref{table_1} and from the fact that all green densities are to the left of the true parameters in Figure \ref{figure_montecarlo_AR3_T5}. This observation confirms the practical importance of normalizing all valid moment functions in binary response logit models to obtain precise estimates in small samples. Focusing on the \say{rescaled} estimator $\hat{\theta}^b$, we can see that it performs relatively well for $N\geq 8000$ with very little bias. This is corroborated in Figure \ref{figure_montecarlo_AR3_T5}: the blue densities are approximately centered at the true parameter values for $N\geq 8000$ . Estimates for the slope parameter $\beta$ are quite accurate even for $N=500$ but precise estimation of the transition parameters requires a larger sample size. In terms of median absolute bias, it is interesting to note a ranking on the precision of estimates of the transition parameters: the coefficient on the first lag is noisier than the coefficient on the second lag which itself is noisier than the coefficient on the third lag for each $N\in\{500,2000,8000,16000\}$. In an unreported set of simulations, we have found that this empirical pattern is robust to other choices of the population parameters and initial condition and also applies to the AR(2) model with a similar data generating process.
\begin{figure}[tbp] 
    \caption{Densities of GMM estimators for the AR(3) with one regressor} \label{figure_montecarlo_AR3_T5}
    \begin{center}
        \begin{tabular}{lcccc}
             &\multicolumn{1}{c}{$N=500$} & \multicolumn{1}{c}{$N=2000$} & \multicolumn{1}{c}{$N=8000$} & 
            \multicolumn{1}{c}{$N=16000$} \\
        \makebox[0pt][r]{\makebox[40pt]{\raisebox{55pt}{\rotatebox[origin=c]{360}{$\gamma_1$}}}} &\includegraphics[width=28mm, height=28mm]{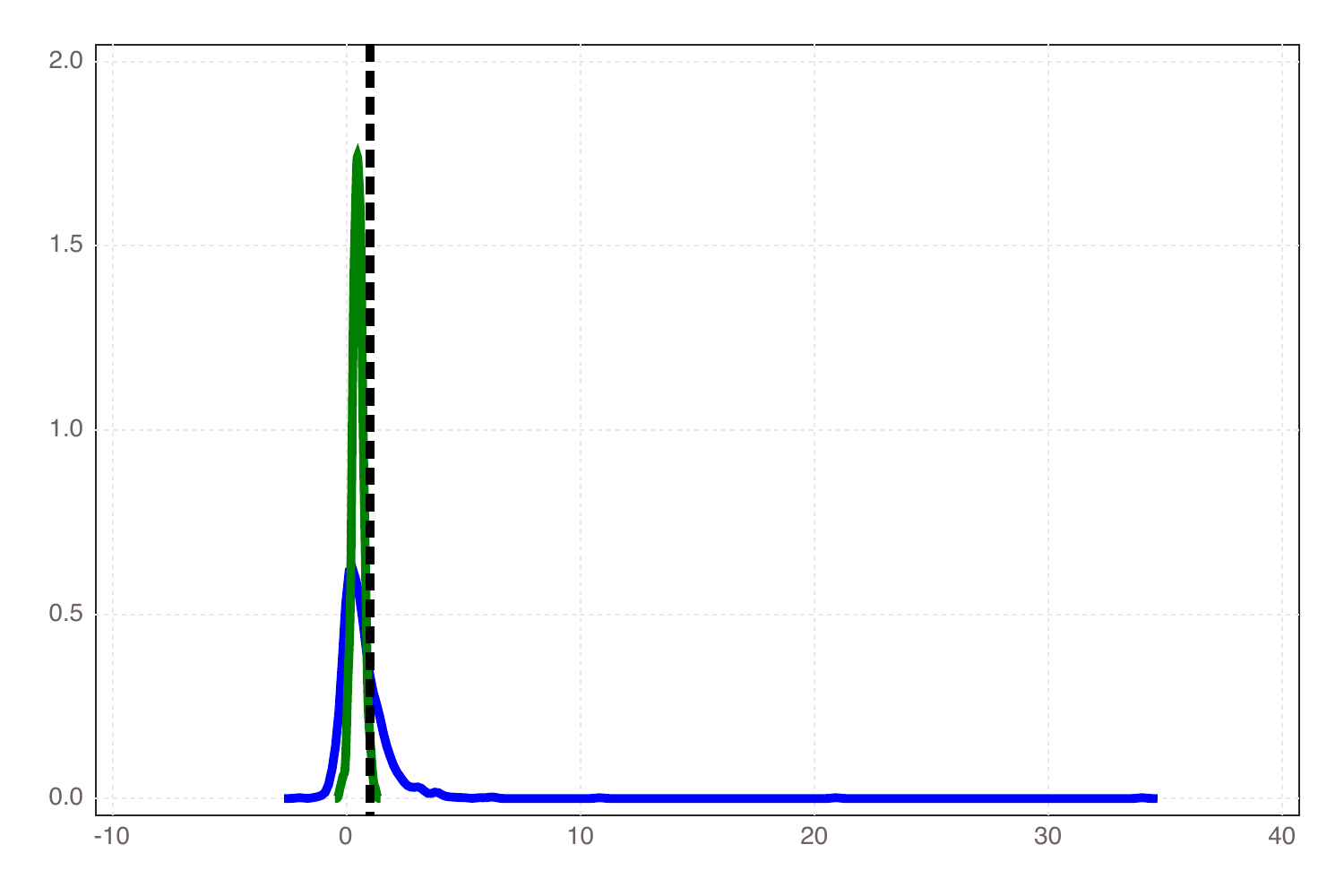} & \includegraphics[width=28mm, height=28mm]{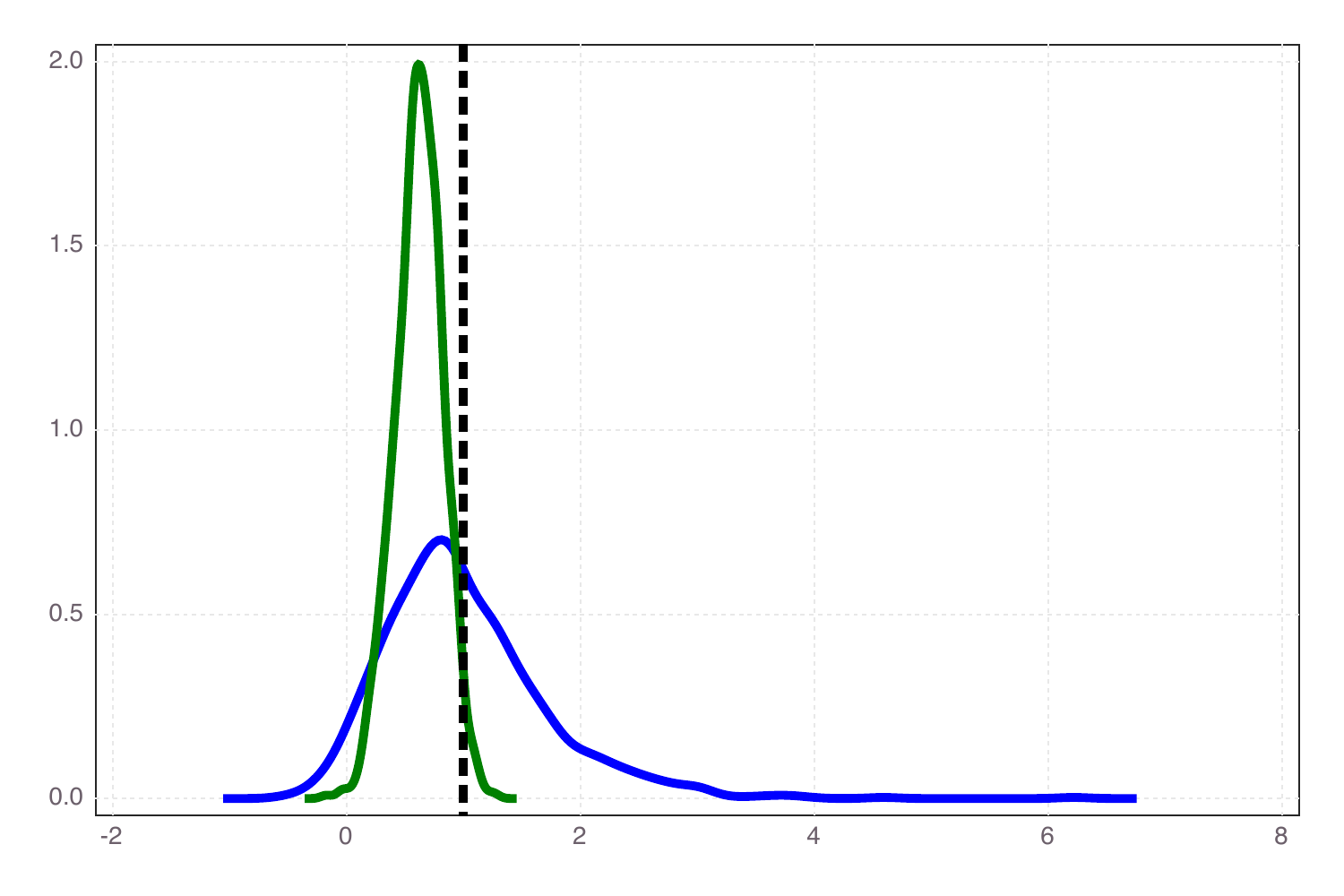} & \includegraphics[width=28mm, height=28mm]{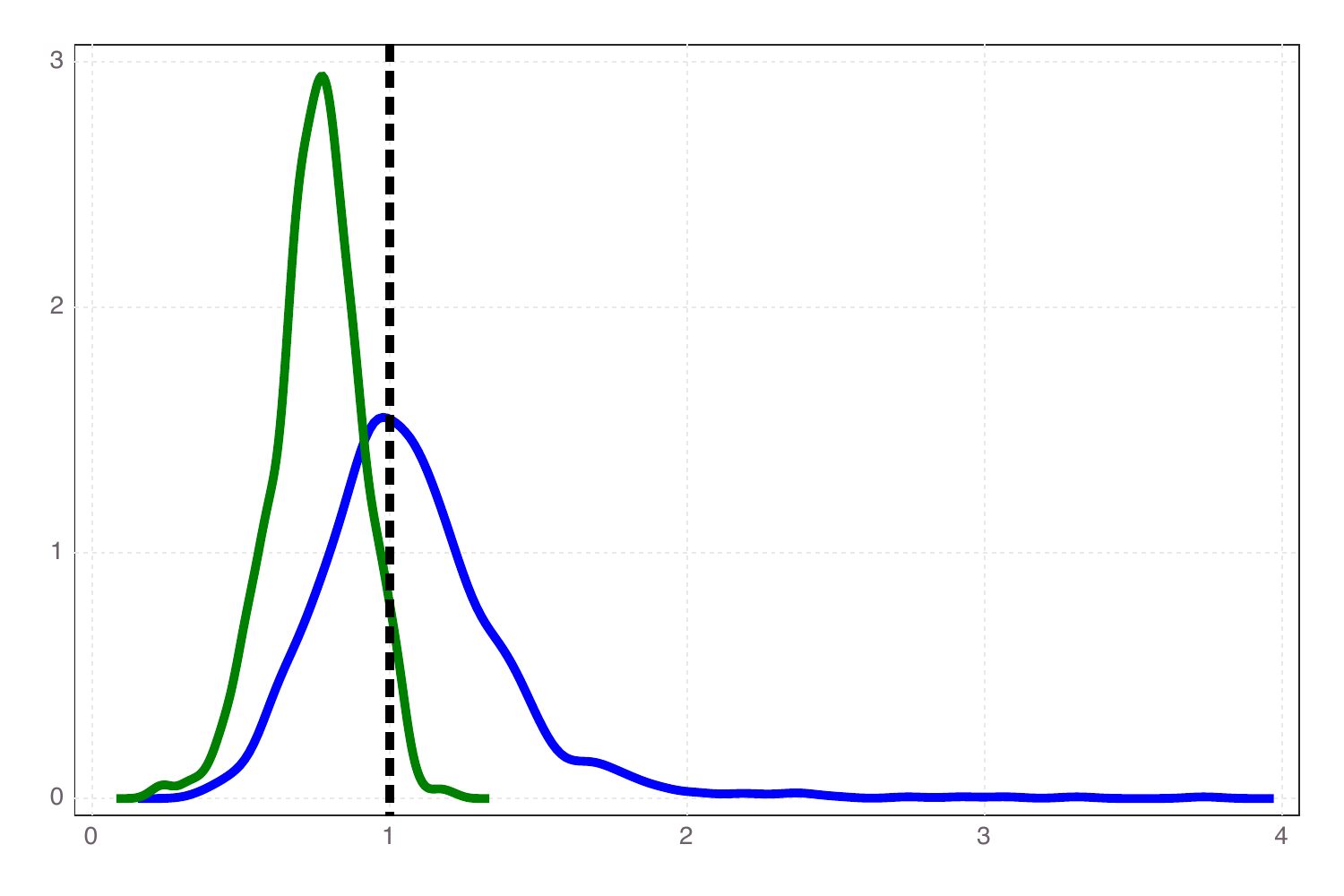} & \includegraphics[width=28mm, height=28mm]{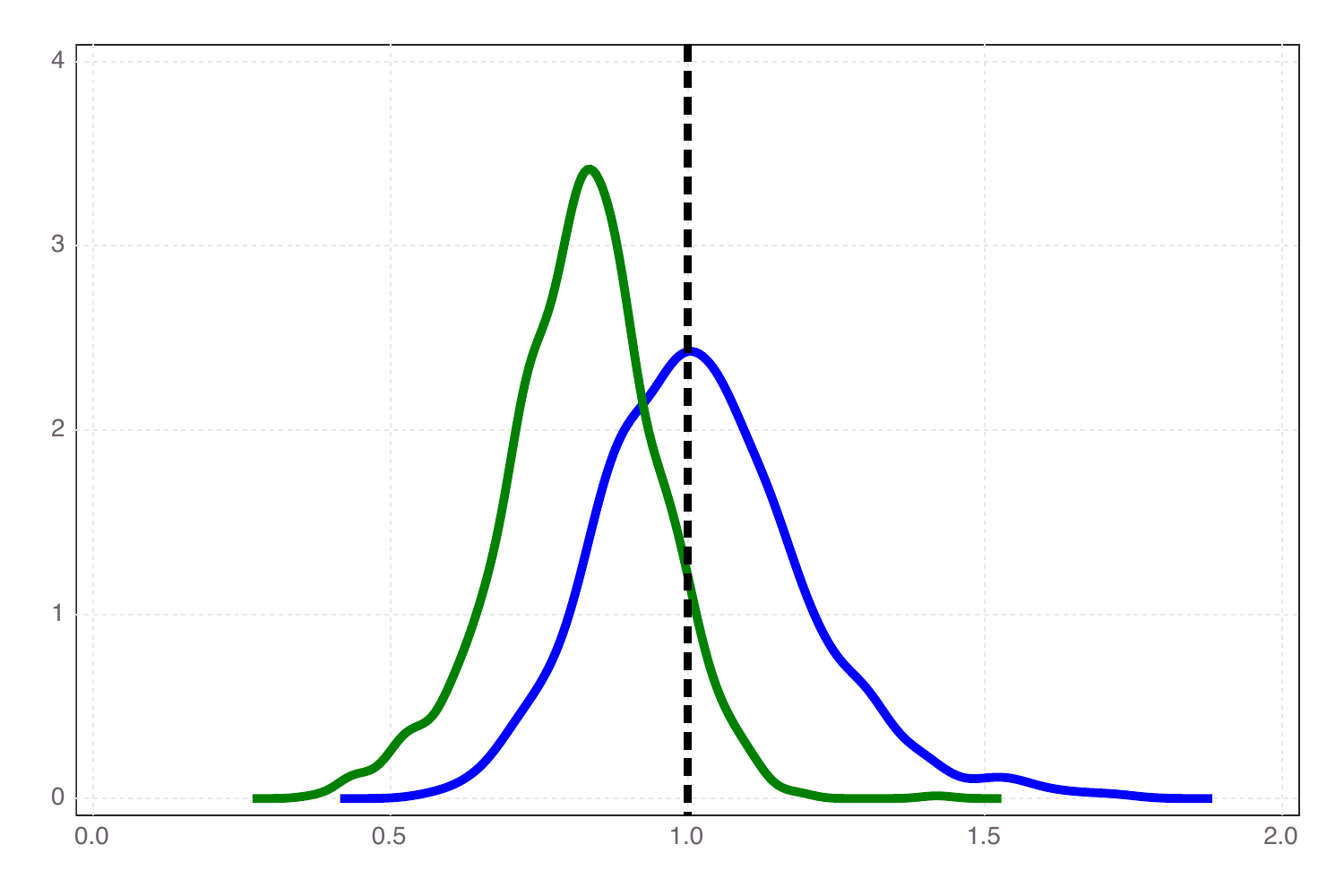} \\
         \makebox[0pt][r]{\makebox[40pt]{\raisebox{55pt}{\rotatebox[origin=c]{360}{$\gamma_2$}}}} &\includegraphics[width=28mm, height=28mm]{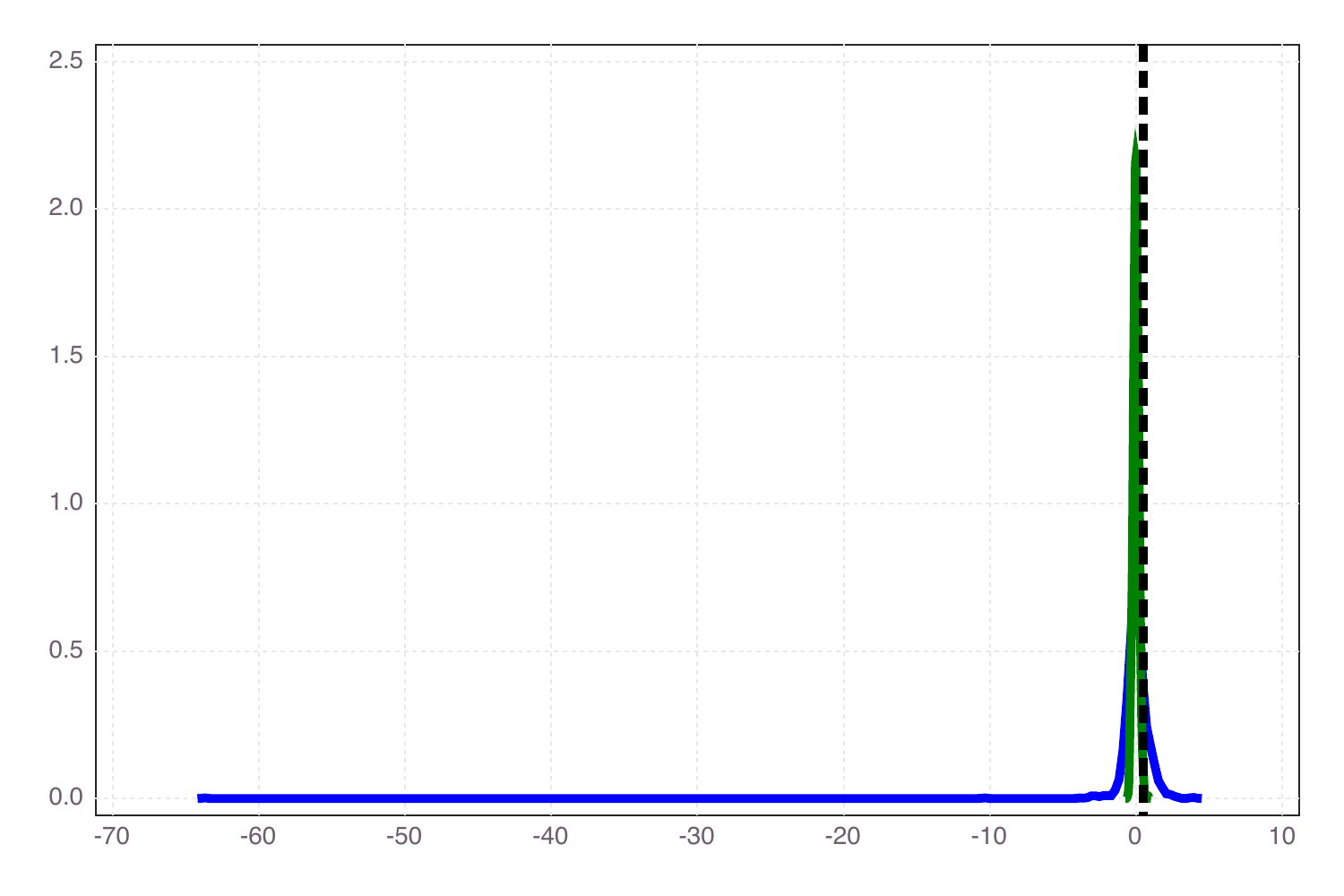} & \includegraphics[width=28mm, height=28mm]{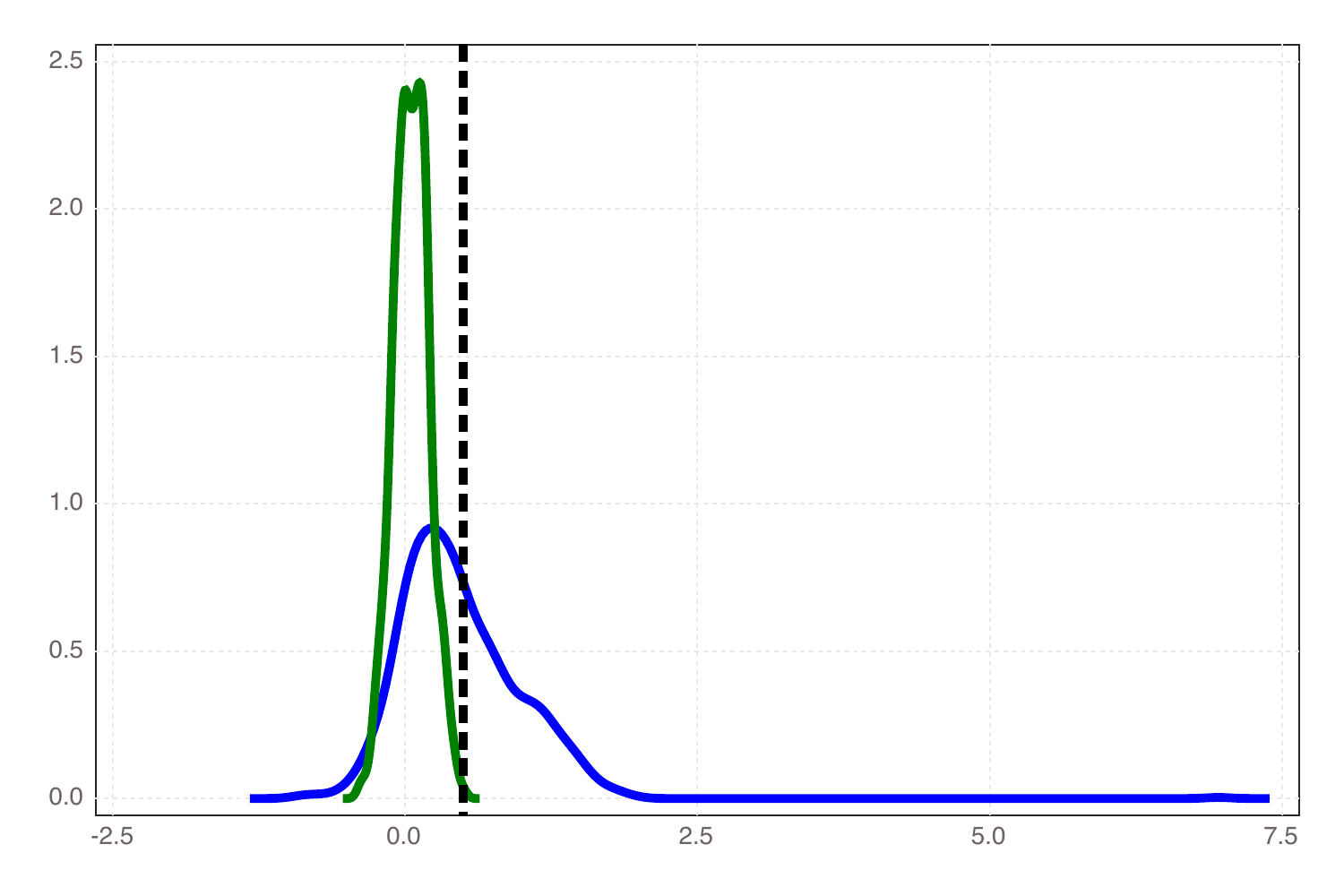} & \includegraphics[width=28mm, height=28mm]{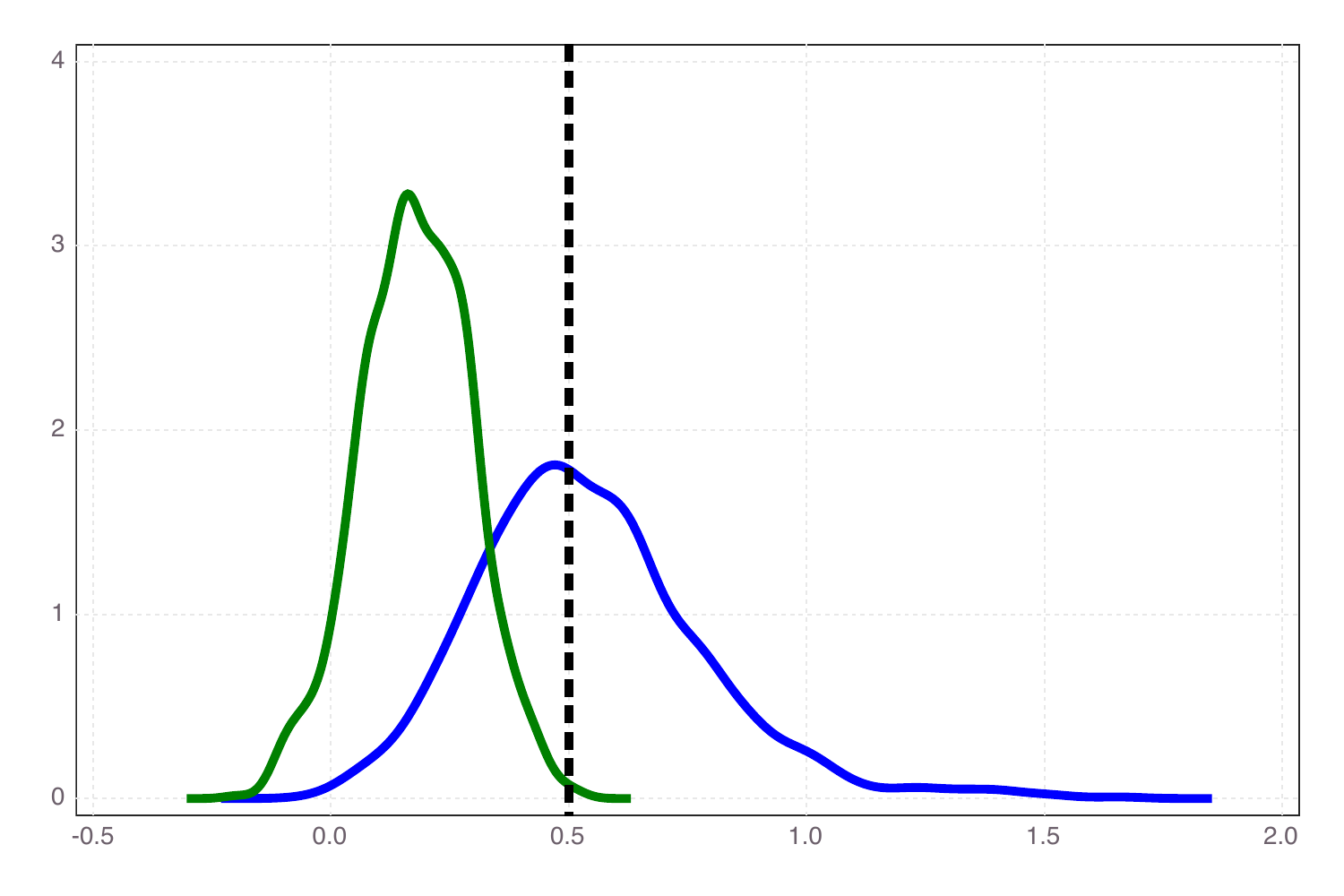} & \includegraphics[width=28mm, height=28mm]{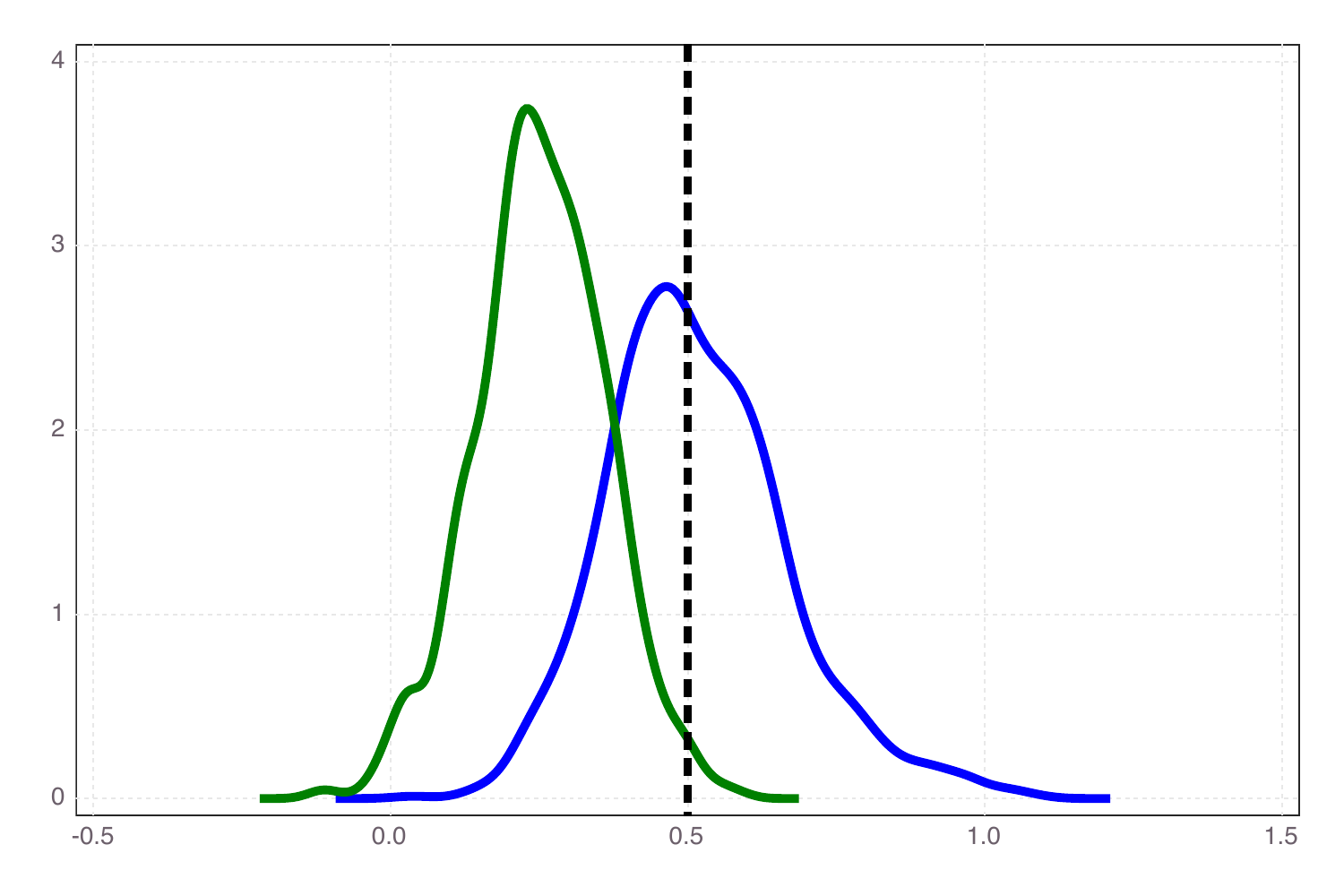} \\
         \makebox[0pt][r]{\makebox[40pt]{\raisebox{55pt}{\rotatebox[origin=c]{360}{$\gamma_3$}}}} &\includegraphics[width=28mm, height=28mm]{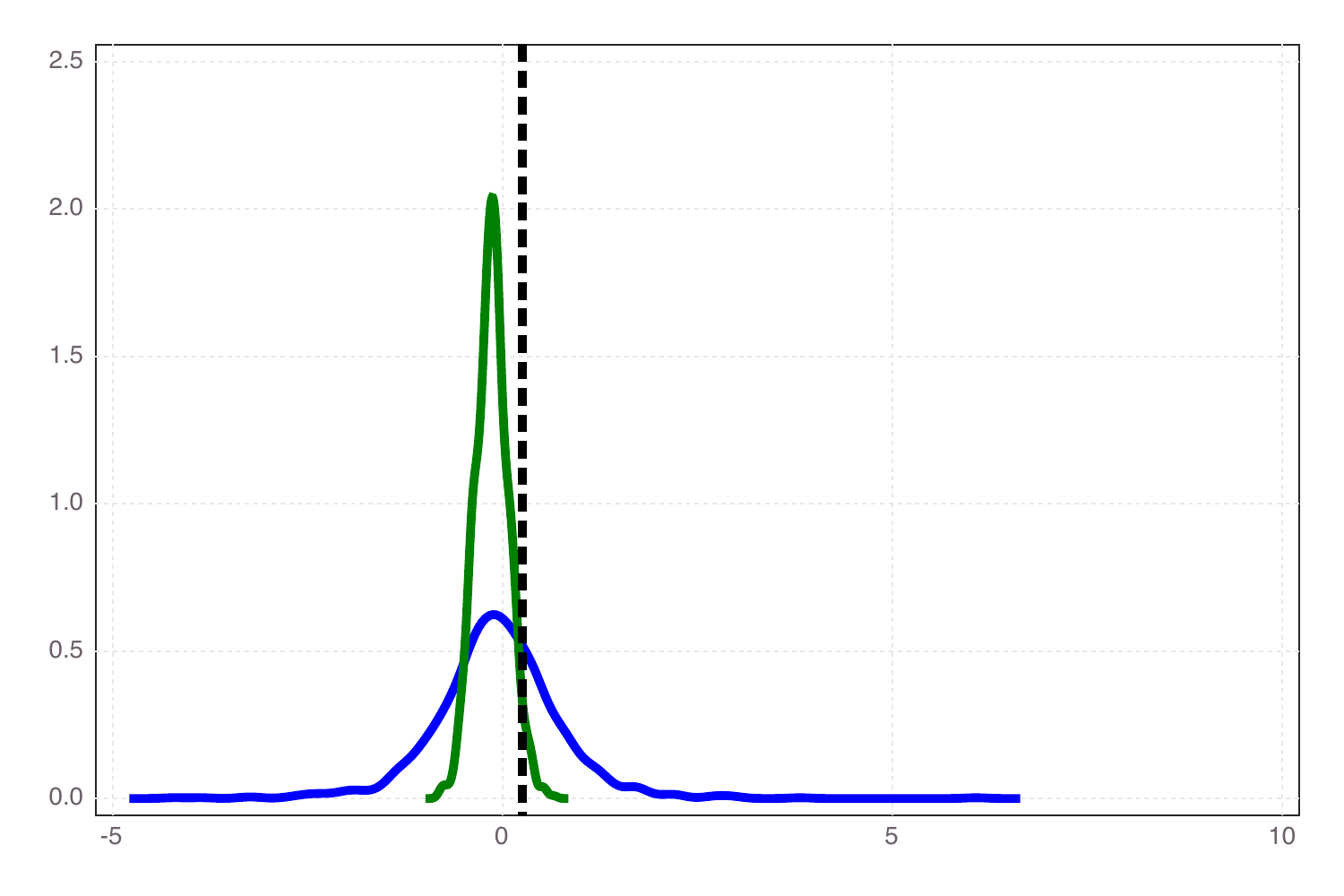} & \includegraphics[width=28mm, height=28mm]{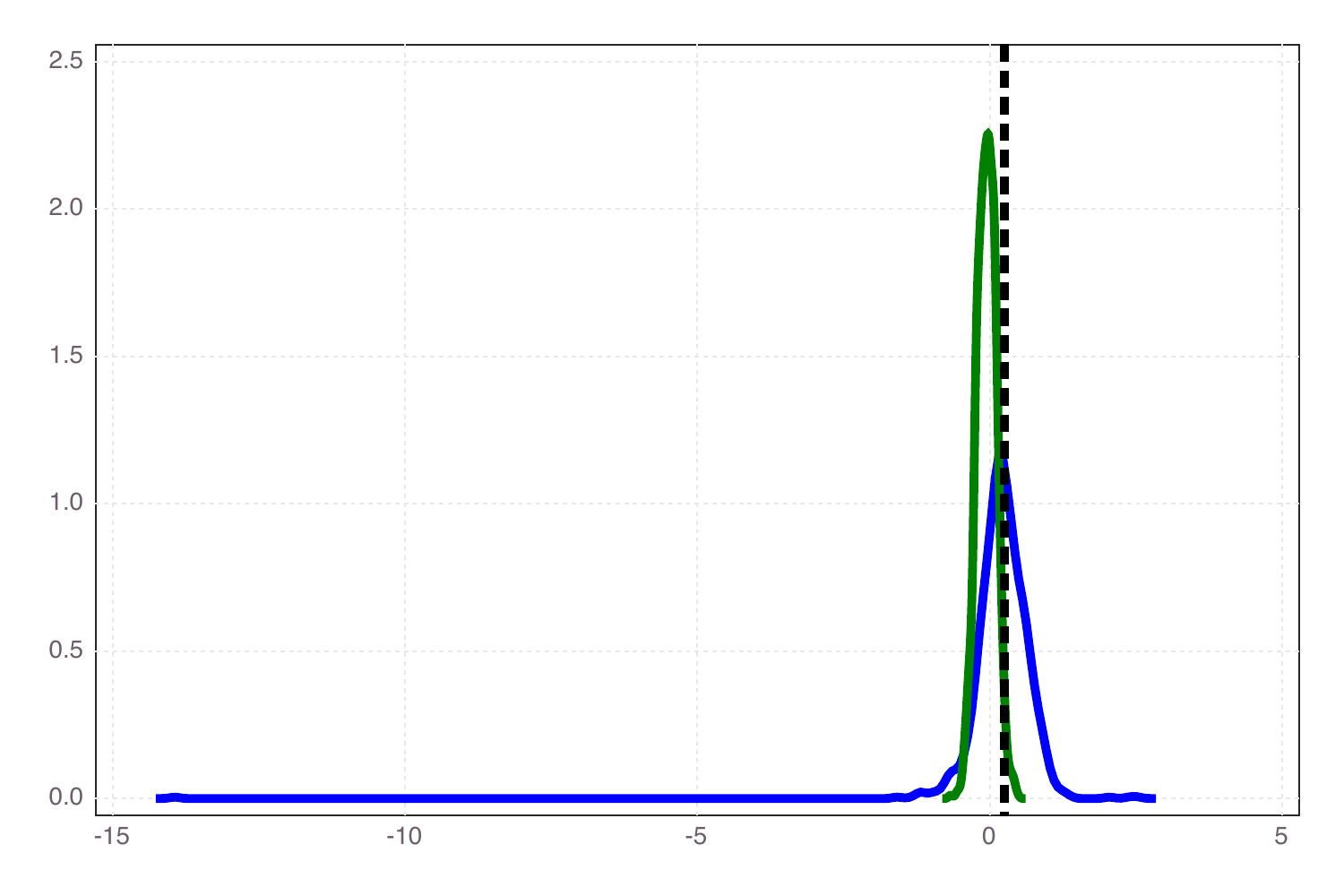} & \includegraphics[width=28mm, height=28mm]{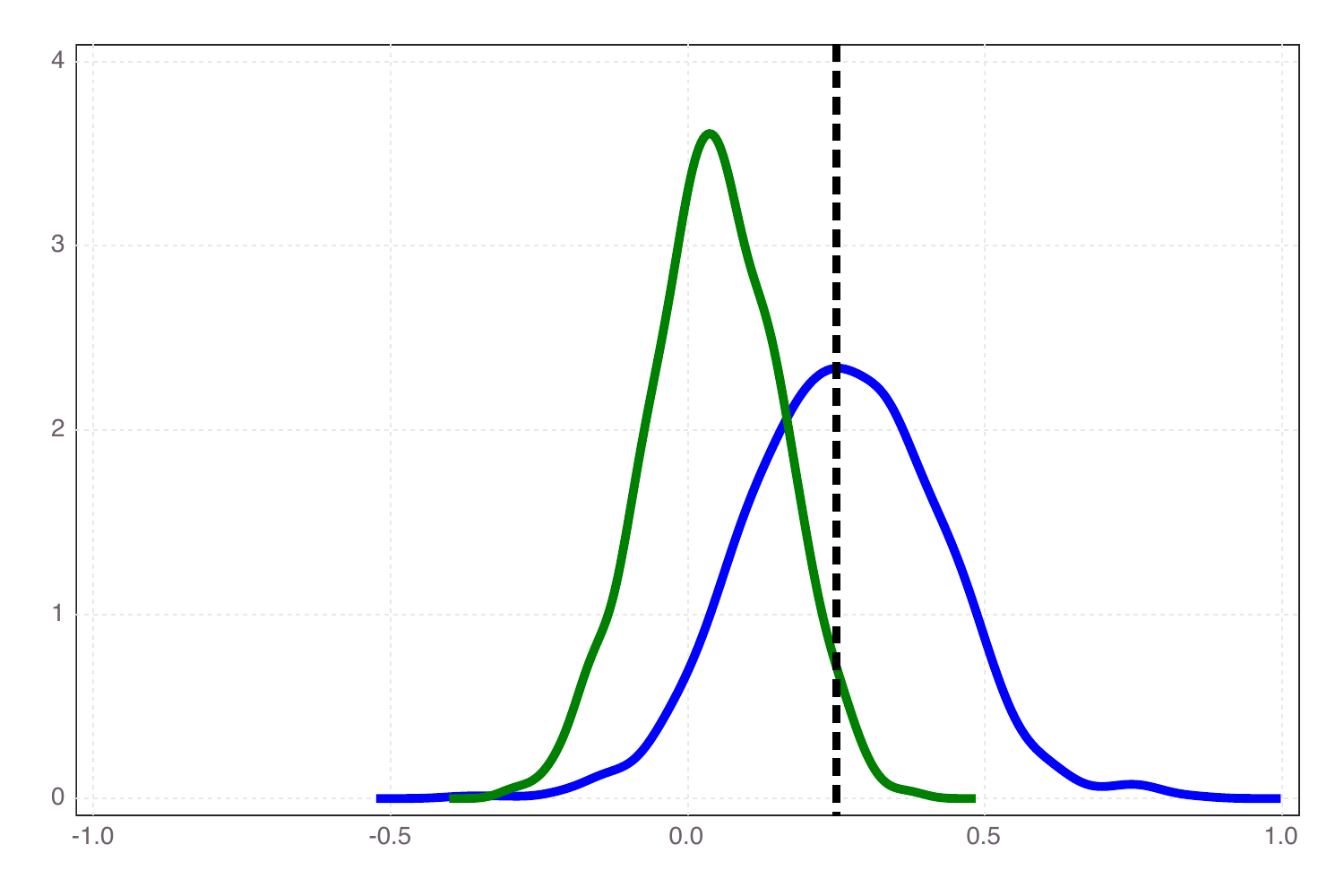} & \includegraphics[width=28mm, height=28mm]{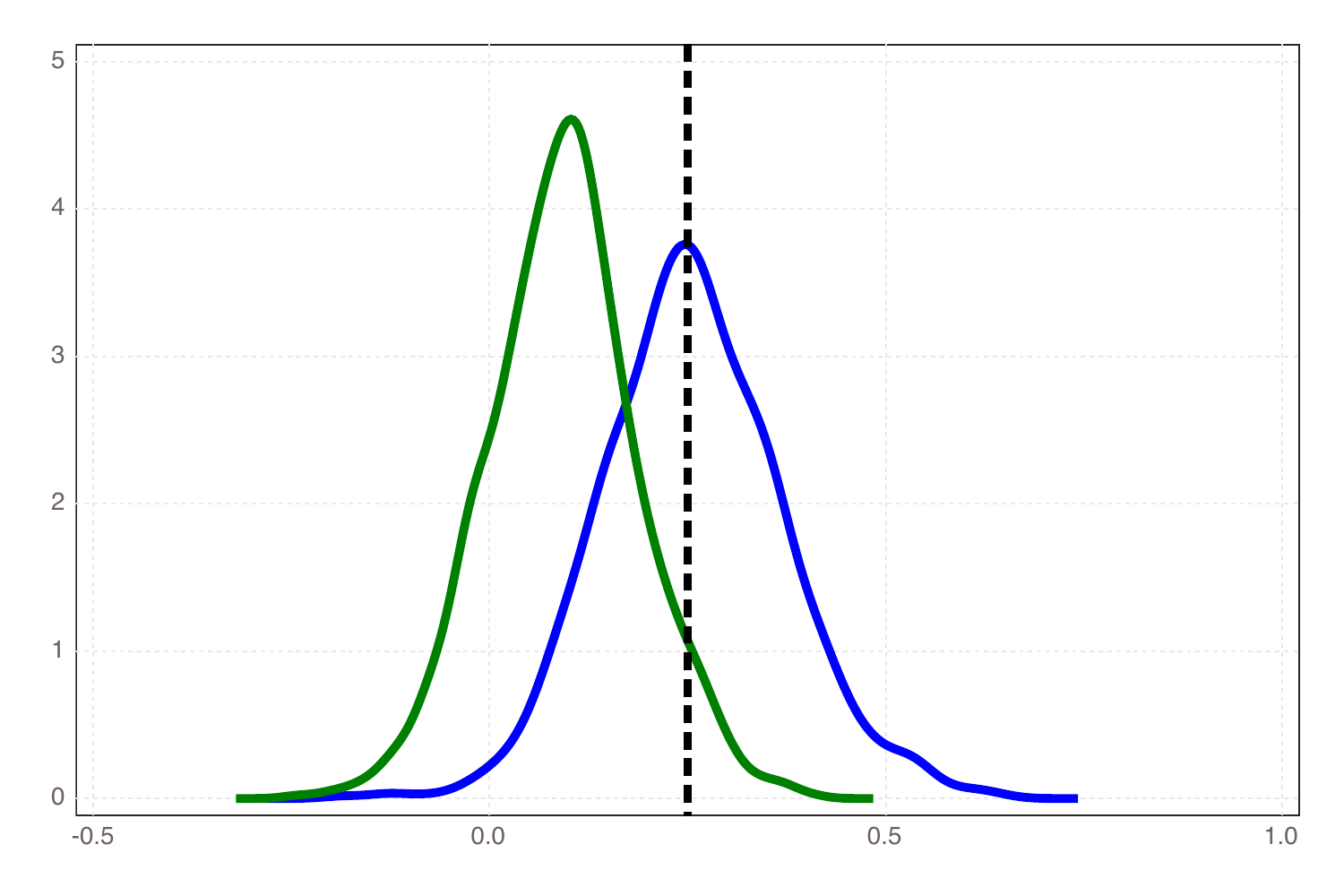} \\
         \makebox[0pt][r]{\makebox[40pt]{\raisebox{55pt}{\rotatebox[origin=c]{360}{$\beta$}}}} &\includegraphics[width=28mm, height=28mm]{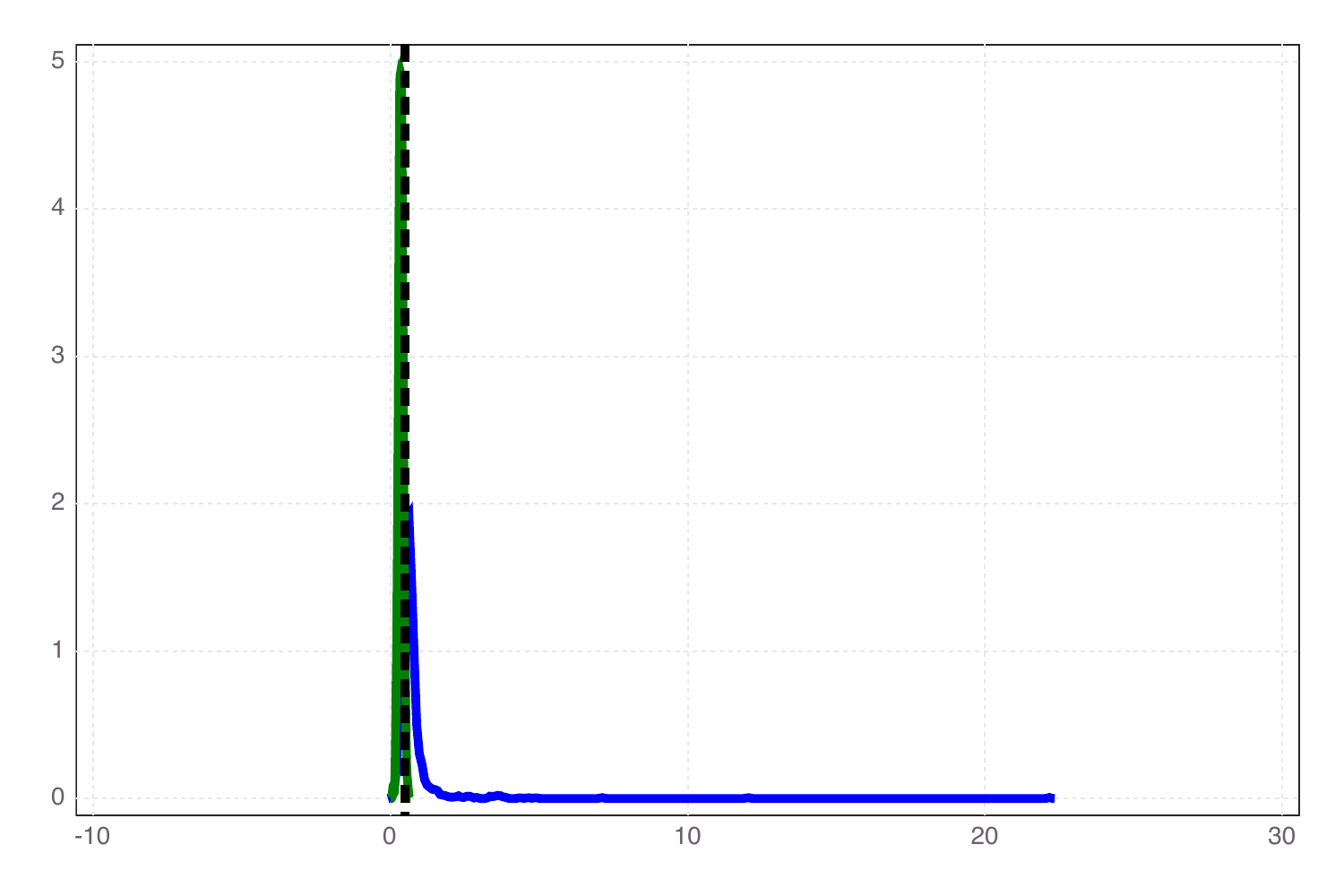} & \includegraphics[width=28mm, height=28mm]{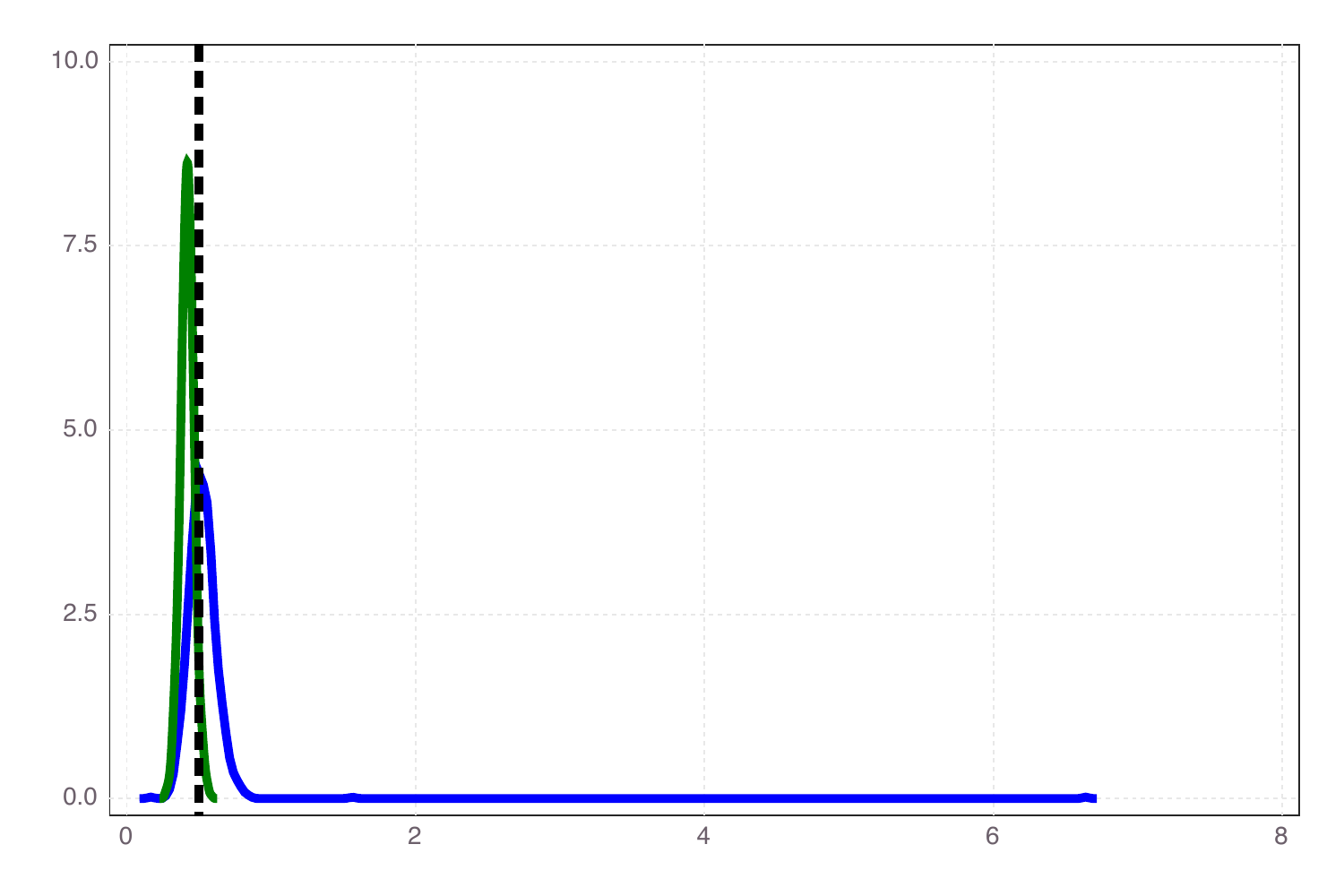} & \includegraphics[width=28mm, height=28mm]{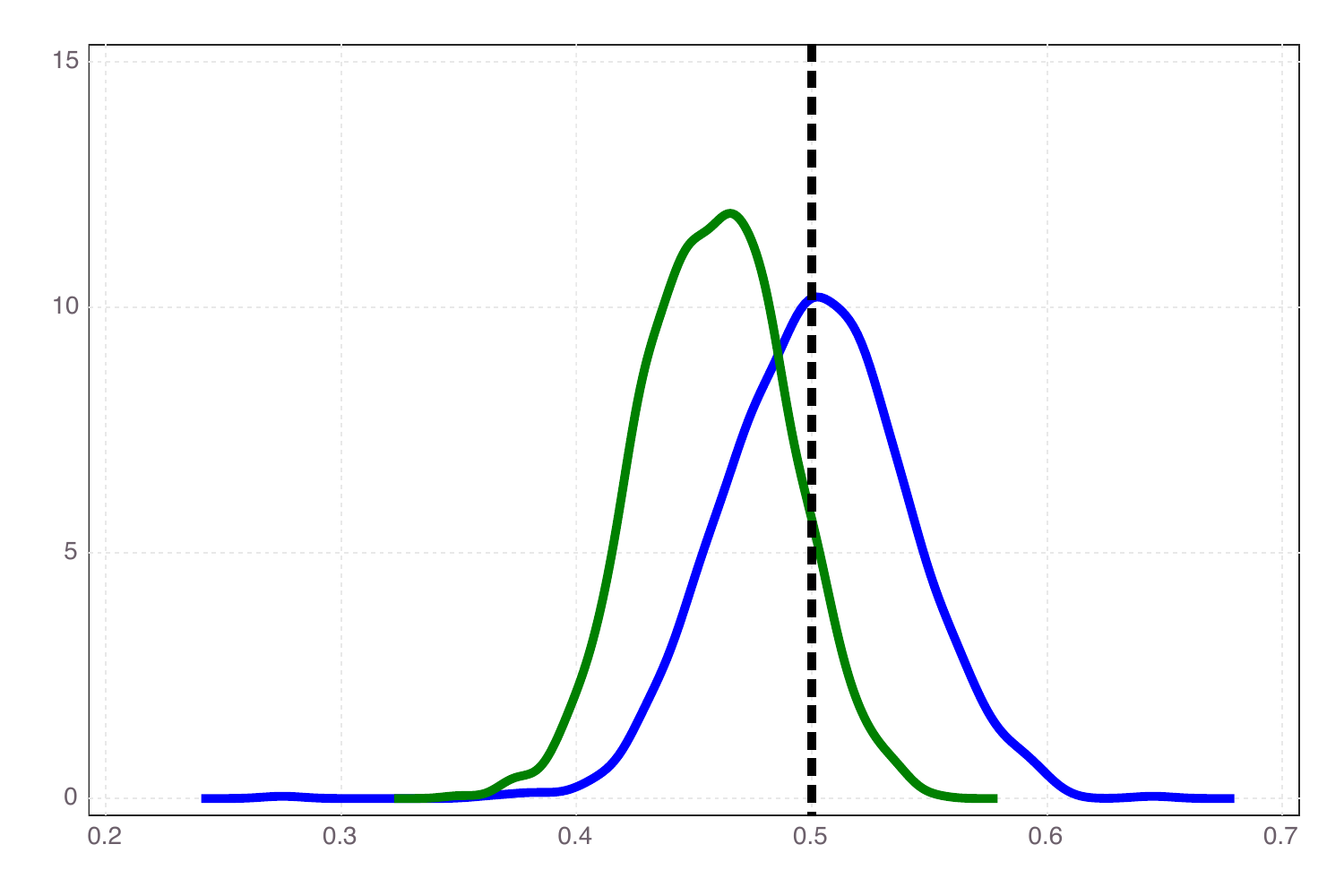} & \includegraphics[width=28mm, height=28mm]{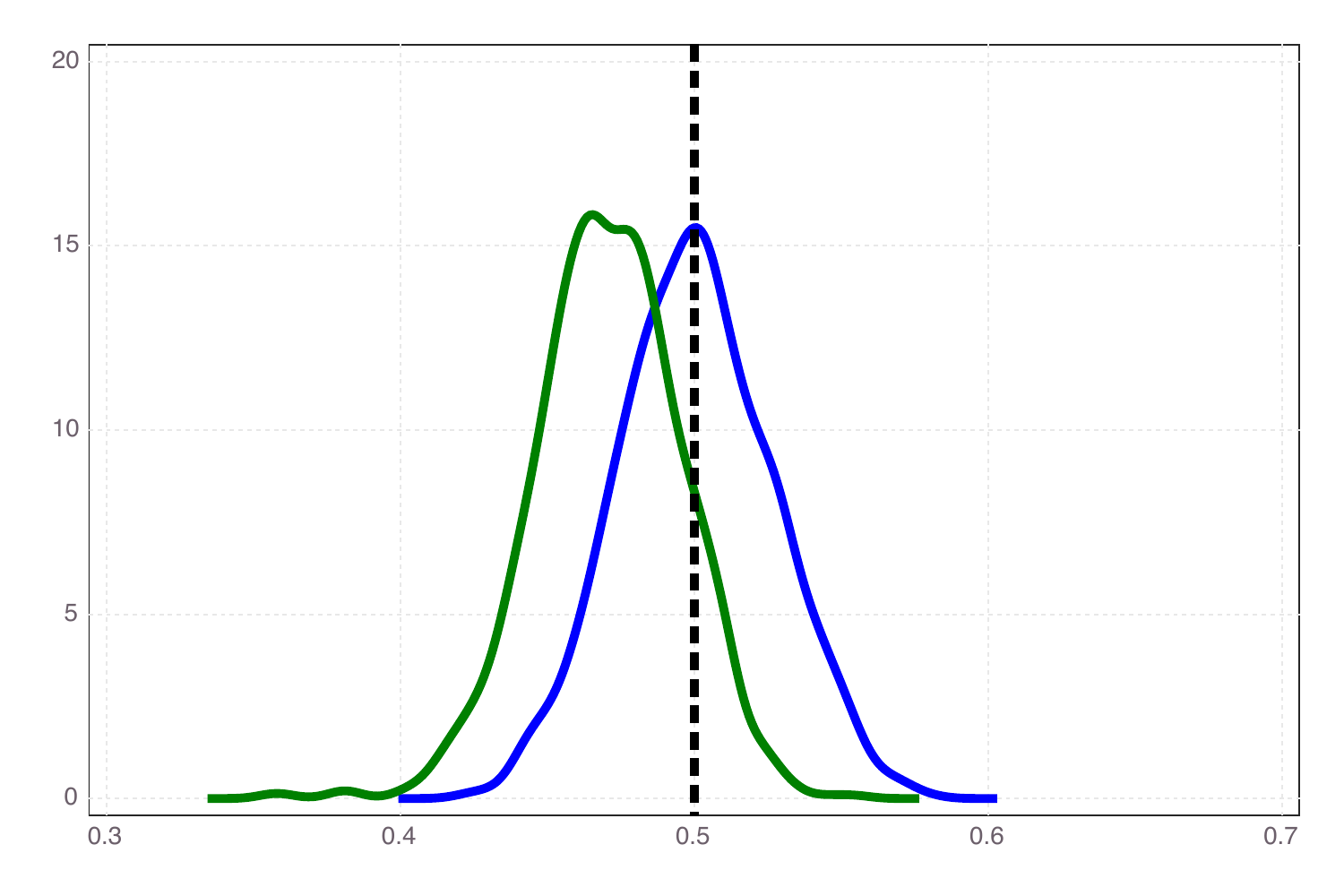}
        \end{tabular}
    \end{center}
    \par
    \textit{{\footnotesize Notes: The densities of estimates based on the first GMM estimator (i.e $\hat{\theta}^{a}$), the second GMM estimator (i.e $\hat{\theta}^{b}$) are indicated in green and blue respectively. Reported results are based on a 1000 replications of the DGP presented above with $\gamma_{01}=1.0, \gamma_{02}=0.5, \gamma_{03}=0.25, \beta_0=0.5$. True parameter values are indicated with a vertical dashed line.}}
\end{figure}
\newpage
\subsection{Monte Carlo for a VAR(1) logit model}
In our next example, we examine a bivariate VAR(1) logit model with $T=3$ and scalar regressors $X_{m,it}$ in each layer $m\in\{1,2\}$. We set the common parameters to $\gamma_{011}=\gamma_{022}=1.0$, $\gamma_{012}=\gamma_{021}=0.5$, $\beta_1=\beta_2=0.5$. The data generating process is:
\begin{align*}
    Y_{m,i0}&=\mathds{1}\left\{X_{m,i0}'\beta_{0m}+A_{m,i}-\epsilon_{m,it}\geq 0\right\}, \quad m=1,2 \\
    Y_{m,it}&=\mathds{1}\left\{\gamma_{0m1}Y_{1,it-1}+\gamma_{0m2}Y_{2,it-1} +X_{m,it}'\beta_{0m}+A_{m,i}-\epsilon_{m,it}\geq 0\right\}, \quad m=1,2, \quad t=1,2,3
\end{align*}
where the disturbances $\epsilon_{m,it}$ are iid standard logistic, the covariates $X_{m,it}$ are iid $\mathcal{N}(0,1)$ and the fixed effects are computed as $A_{m,i}=\frac{1}{\sqrt{4}}\cleansum_{t=0}^3 X_{m,it}$. We consider sample sizes $N\in\{2000,8000,16000\}$ with 1000 Monte Carlo replications per design.\\
We use all four valid moment functions implied by Proposition  \ref{proposition_2} when $T=3$ for the VAR(1) case, viz $\psi_{\theta}^{k|k}(Y^{3}_{i1},Y_{i0}^1,X_i),k\in \{(0,0),(0,1),(1,0),(0,0)\}$ and form the $40\times 1$ moment vector:
\begin{align*}
    m_{\theta}(Y_{i},Y_{i}^0,X_i)=
    \begin{pmatrix}
    &\psi_{\theta}^{(0,0)|(0,0)}(Y^{3}_{i1},Y_{i0}^1,X_i) \\
    &\psi_{\theta}^{(0,1)|(0,1)}(Y^{3}_{i1},Y_{i0}^1,X_i) \\
    &\psi_{\theta}^{(1,0)|(1,0)}(Y^{3}_{i1},Y_{i0}^1,X_i) \\
    &\psi_{\theta}^{(1,1)|(1,1)}(Y^{3}_{i1},Y_{i0}^1,X_i) \\
    \end{pmatrix} \otimes 
    \begin{pmatrix}
    & 1 \\
    &Y_{i}^{0'} \\
    &X_{1,i1}^{3'} \\
    &X_{2,i1}^{3'}
\end{pmatrix}
\end{align*}
Given the importance of rescaling the valid moment functions for better precision of GMM in the context of the AR(3), we also consider a normalized moment vector $ \widetilde{m_{\theta}}(Y_{i},Y_{i}^0,X_i)$ in which each $\psi_{\theta}^{k|k}(Y^{3}_{i1},Y_{i0}^1,X_i)$ is divided by the sum of the absolute values of their unique non-zero entries as a 64-dimensional vector (64 possible choice histories $Y^{3}_{i1}$ per initial condition). With these moment functions in hand, we then compare the finite sample properties of three estimators: i) the VAR(1) analogs of $\hat{\theta}^{a}$ and $\hat{\theta}^{b}$ defined previously for the AR(3), ii) the iterated GMM estimator $\hat{\theta}^{c}$ based on $ m_{\theta}(Y_{i},Y_{i}^0,X_i)$  as in Section \ref{Section_6}. The results of the simulations are summarized in Table \ref{table_2VAR1_transiparam} and Table \ref{table_2VAR1_slopeparam}.

\begin{table}[!htb] 
\caption{Performance of GMM estimators for the bivariate VAR(1): transition parameters}  \label{table_2VAR1_transiparam}
\centerline{
\begin{tabular}{llccccccccccccccc} \toprule  \toprule  
     &  &    $\hat{\gamma_{11}}^{a}$&$\hat{\gamma_{11}}^{b}$&$\hat{\gamma_{11}}^{c}$  &  & $\hat{\gamma_{12}}^{a}$&$\hat{\gamma_{12}}^{b}$&$\hat{\gamma_{12}}^{c}$ & & $\hat{\gamma_{21}}^{a}$&$\hat{\gamma_{21}}^{b}$&$\hat{\gamma_{21}}^{c}$ & & $\hat{\gamma_{22}}^{a}$&$\hat{\gamma_{22}}^{b}$&$\hat{\gamma_{22}}^{c}$ \\ 
\cline{3-5} \cline{7-9} \cline{11-13} \cline{15-17} 
 $N=2000$ & \\ \cline{1-1} 
& Bias & -0.23 & 0.10 & -0.05 &  & -0.21 & -0.04 & -0.04 &  & -0.20 & -0.06 & -0.05 &  & -0.24 & 0.10 & -0.05 \\ 
 & MAE & 0.27 & 0.23 & 0.16 &  & 0.29 & 0.24 & 0.19 &  & 0.27 & 0.23 & 0.19 &  & 0.27 & 0.23 & 0.16 \\ 
 & Iter &  &  & 5 &  &  &  & 5 &  &  &  & 5 &  &  &  & 5 \\ 
 $N=8000$ & \\ \cline{1-1} 
& Bias & -0.07 & 0.03 & -0.00 &  & -0.08 & 0.00 & -0.00 &  & -0.09 & -0.01 & -0.01 &  & -0.06 & 0.03 & -0.00 \\ 
 & MAE & 0.13 & 0.11 & 0.08 &  & 0.14 & 0.12 & 0.09 &  & 0.15 & 0.12 & 0.09 &  & 0.12 & 0.11 & 0.07 \\ 
 & Iter &  &  & 4 &  &  &  & 4 &  &  &  & 4 &  &  &  & 4 \\ 
 $N=16000$ & \\ \cline{1-1} 
& Bias & -0.04 & 0.01 & -0.00 &  & -0.05 & -0.01 & -0.00 &  & -0.07 & -0.01 & -0.00 &  & -0.03 & 0.01 & 0.00 \\ 
 & MAE & 0.09 & 0.08 & 0.05 &  & 0.11 & 0.07 & 0.06 &  & 0.11 & 0.08 & 0.06 &  & 0.08 & 0.08 & 0.06 \\ 
 & Iter &  &  & 3 &  &  &  & 3 &  &  &  & 3 &  &  &  & 3 \\ 
 \bottomrule \bottomrule 
\end{tabular}
}
\vspace{0.5cm}
\textit{\footnotesize \textsc{Notes}: 
Reported results are based on a 1000 replications of the DGP. Bias and MAE stand for median bias and median absolute error respectively. The convergence criterion for the iterated GMM  estimator is $\norm{\hat{\theta}_{s+1}-\hat{\theta}_{s}}<10^{-4}$ and Iter corresponds to the median number of iterations to reach convergence. Bias and MAE for the iterated GMM are reported for replications where convergence is attained which is $\approx 91\%$ for $N=2000$ and $\approx 100\%$ for $N=8000,16000$. }
\end{table}

Similarly to the AR(3) example, both the transition parameters and the slope parameters of $\hat{\theta}^{a}$ are negatively biased for the three sample sizes under consideration. This is particularly true for the \say{between} state-dependence parameters $\hat{\gamma_{12}}^{a},\hat{\gamma_{21}}^{a}$ which maintain a small bias even for $N=8000,16000$. By comparison, the rescaled GMM estimator $\hat{\theta}^{b}$ and the iterated GMM estimator $\hat{\theta}^{c}$ demonstrate better accuracy, especially for $\gamma_{12}$ and $\gamma_{21}$ which are really the key parameters in our empirical application presented in Section \ref{Section_6}. In this specific simulation design, $\hat{\theta}^{c}$  slightly outperforms $\hat{\theta}^{b}$ for all $N=2000,8000,16000$ in terms of median bias and median absolute error for the transition parameters. The comparison is  somewhat less clear for the slope parameters $\beta_1,\beta_2$. \footnote{We also experimented with an iterated GMM estimator based on $\widetilde{m_{\theta}}(Y_{i},Y_{i}^0,X_i)$ and found nearly identical results to $\hat{\theta}^{b}$.}

\begin{table}[!htb]
\caption{Performance of GMM estimators for the bivariate VAR(1): slope parameters}  \label{table_2VAR1_slopeparam}
\centerline{
\begin{tabular}{llccccccc} \toprule  \toprule  
     &  &    $\hat{\beta_{1}}^{a}$&$\hat{\beta_{1}}^{b}$&$\hat{\beta_{1}}^{c}$  &  & $\hat{\beta_{2}}^{a}$&$\hat{\beta_{2}}^{b}$&$\hat{\beta_{2}}^{c}$ \\ 
\cline{3-5} \cline{7-9} 
 $N=2000$ & \\ \cline{1-1} 
& Bias & -0.04 & 0.01 & -0.01 &  & -0.04 & 0.00 & -0.01 \\ 
 & MAE & 0.06 & 0.06 & 0.06 &  & 0.06 & 0.06 & 0.05 \\ 
 & Iter &  &  & 5 &  &  &  & 5 \\ 
 $N=8000$ & \\ \cline{1-1} 
& Bias & -0.01 & -0.00 & 0.00 &  & -0.01 & 0.00 & 0.00 \\ 
 & MAE & 0.03 & 0.03 & 0.03 &  & 0.03 & 0.03 & 0.03 \\ 
 & Iter &  &  & 4 &  &  &  & 4 \\ 
 $N=16000$ & \\ \cline{1-1} 
& Bias & -0.00 & 0.00 & 0.01 &  & -0.00 & 0.00 & 0.01 \\ 
 & MAE & 0.02 & 0.02 & 0.02 &  & 0.02 & 0.02 & 0.02 \\ 
 & Iter &  &  & 3 &  &  &  & 3 \\ 
 \bottomrule \bottomrule 
\end{tabular}
}
\vspace{0.5cm}
\textit{\footnotesize \textsc{Notes}: 
Reported results are based on a 1000 replications of the DGP. Bias and MAE stand for median bias and median absolute error respectively. The convergence criterion for the iterated GMM  estimator is $\norm{\hat{\theta}_{s+1}-\hat{\theta}_{s}}<10^{-4}$ and Iter corresponds to the median number of iterations to reach convergence. Bias and MAE for the iterated GMM are reported for replications where convergence is attained which is $\approx 91\%$ for $N=2000$ and $\approx 100\%$ for $N=8000,16000$. }
\end{table}
Surprisingly, when experimenting with a trivariate logit extension, we found that the analog of $\hat{\theta}^{b}$ performs very poorly for the same simulation design relative to the iterated GMM estimator or even the naive equally-weighted GMM estimator $\hat{\theta}^{a}$. This is perhaps due to the \say{large} rescaling factor applied to each valid moment function in that case which pose problems for the optimization of the GMM objective.  We have not investigated these peculiarities - which could be design specific - further at this moment but a more thorough analysis of the behavior of GMM in future work would be beneficial. The good performance of $\hat{\theta}^c$ and this shortcoming of $\hat{\theta}^{b}$ in the trivariate case was one additional motivation for concentrating on the iterated GMM estimator in our empirical application.

\section{Proofs of Theorem \ref{theorem_nummoments_AR1} and Theorem \ref{theorem_nummoments_ARp}}
We focus our attention on proving Theorem \ref{theorem_nummoments_ARp} since proving Theorem \ref{theorem_nummoments_AR1} would follow nearly identical arguments. At each important step of the proof, we highlight where the arguments for the AR(1) would differ. \\

\noindent Fix a history $y\in \mathcal{Y}^T$ and consider the corresponding basis element $\mathds{1}\{.=y\}$ of $\mathbb{R}^{\mathcal{Y}^T}$. We have:
\begin{align*}
   \mathcal{E}_{y^0,x}^{(p)}\left[\mathds{1}\{.=y\}\right]&=P(Y_i=y|Y_{i}^0=y^0,X_{i}=x,A_{i}=.)
\end{align*}
where by definition, for all $a\in \mathbb{R}$,
\begin{align*}
    P(Y_i=y|Y_{i}^0=y^0,X_{i}=x,A_{i}=a)&=\frac{N^{y|y^0}(e^{a})}{D^{y|y^0}(e^{a})} \\
    N^{y|y^0}(e^{a})&=\prod_{t=1}^Te^{y_t\left(\sum_{r=1}^{p}\gamma_{0r}y_{t-r}+x_t'\beta_0+a\right)} \\
    D^{y|y^0}(e^{a})&=\prod_{t=1}^T\left(1+e^{\sum_{r=1}^{p}\gamma_{0r}y_{t-r}+x_t'\beta_0+a}\right)
\end{align*}
Notice that $N^{y|y^0}(e^{a})$ and $D^{y|y^0}(e^{a})$ are just polynomials of $e^{a}$ - with dependence on $x$ suppressed for conciseness - and that we always have $\deg\left(N^{y|y^0}(e^{a})\right)\leq \deg\left(D^{y|y^0}(e^{a})\right)$ with strict inequality unless $y=1_{T}$. Moreover, since by assumption for any   $t,s\in \{1,\ldots,T-1\}$ and $y,\tilde{y}\in \mathcal{Y}^{p}$, $\gamma_0'y+x_{t}'\beta_0\neq \gamma_0'\tilde{y}+x_{s}'\beta_0$ if $t\neq s$ or $y\neq\tilde{y}$, $D^{y|y^0}(e^{a})$ is a product of distinct irreducible polynomials in $e^{a}$. Therefore, by standard results on \textit{partial fraction decompositions}, we know that there exists a unique set of coefficients $(\lambda_{0}^{y},\lambda_{1}^{y},\ldots, \lambda_{T}^{y})\in \mathbb{R}^{T+1}$ independent of the fixed effect such that:
\begin{align*}
P(Y_i=y|Y_{i}^0=y^0,X_{i}=x,A_{i}=a)&=\lambda_{0}^{y}+\sum_{t=1}^T \lambda_{t}^{y} \frac{1}{1+e^{\sum_{r=1}^{p}\gamma_{0r}y_{t-r}+x_t'\beta_0+a}} \\
&=\lambda_{0}^{y}+T_{0}(a)+T_{1}(a)+T_{2}(a) \\
T_{0}(a)&=\lambda_{1}^{y} \frac{1}{1+e^{\sum_{r=1}^{p}\gamma_{0r}y_{1-r}+x_1'\beta_0+a}} \\
T_{1}(a)&=\sum_{t=2}^{p} \lambda_{t}^{y} \frac{1}{1+e^{\sum_{r=1}^{p}\gamma_{0r}y_{t-r}+x_t'\beta_0+a}} \\
T_{3}(a)&=\sum_{t=p+1}^{T} \lambda_{t}^{y} \frac{1}{1+e^{\sum_{r=1}^{p}\gamma_{0r}y_{t-r}+x_t'\beta_0+a}}
\end{align*}
with $\lambda_{0}^{y}=0$ unless $y=1_T$. This decomposition breaks down the conditional probability $P(Y_i=y|Y_{i}^0=y^0,X_{i}=x,A_{i}=a)$ into components that depend on the initial condition, namely $T_{0}(a),T_{1}(a)$, and components that do not, i.e $T_{2}(a)$. Notice that $T_{1}(a)$ would not appear in the AR(1) case. Starting with the first group, we can write:
\begin{align*}
    T_{0}(a)&=\lambda_{1}^{y}\pi_{0}^{0|y^0}(a,x)\\
    &=\lambda_{1}^{y}\mathds{1}\{y_0=0\}\pi_{0}^{y_0|y^0}(x,a)+\lambda_{1}^{y}\mathds{1}\{y_0=1\}\left(1-\pi_{0}^{y_0|y^0}(x,a)\right) \\
    &=\lambda_{1}^{y}\mathds{1}\{y_0=1\}+\lambda_{1}^{y}\mathds{1}\{y_0=0\}\pi_{0}^{y_0|y^0}(x,a)-\lambda_{1}^{y}\mathds{1}\{y_0=1\}\pi_{0}^{y_0|y^0}(x,a)
\end{align*}
and
\begin{align*}
    T_{1}(a)&=\sum_{t=2}^{p} \lambda_{t}^{y} \sum_{\tilde{y}_{1}^{t-1}\in \mathcal{Y}^{t-1}} \mathds{1}\{y_{t-1}=\tilde{y}_{1},\ldots,y_{1}=\tilde{y}_{t-1}\} \pi_{t-1}^{0|\tilde{y}_{1}^{t-1},y_{0},\ldots,y_{-(p-t)}}(a,x) \\
    &=\sum_{t=2}^{p} \lambda_{t}^{y} \sum_{\tilde{y}_{2}^{t-2}\in \mathcal{Y}^{t-2}} \mathds{1}\{y_{t-1}=0,y_{t-2}=\tilde{y}_2,\ldots,y_{1}=\tilde{y}_{t-1}\} \pi_{t-1}^{0|0,\tilde{y}_{2}^{t-1},y_{0},\ldots,y_{-(p-t)}}(a,x)\\
    &+\sum_{t=2}^{p} \lambda_{t}^{y} \sum_{\tilde{y}_{2}^{t-2}\in \mathcal{Y}^{t-2}} \mathds{1}\{y_{t-1}=1,y_{t-2}=\tilde{y}_2,\ldots,y_{1}=\tilde{y}_{t-1}\}\left(1-\pi_{t-1}^{1|1,\tilde{y}_{2}^{t-1},y_{0},\ldots,y_{-(p-t)}}(a,x)\right)\\
    &=\sum_{t=2}^{p} \lambda_{t}^{y} \sum_{\tilde{y}_{2}^{t-2}\in \mathcal{Y}^{t-2}} \mathds{1}\{y_{t-1}=1,y_{t-2}=\tilde{y}_2,\ldots,y_{1}=\tilde{y}_{t-1}\} \\
    &+\sum_{t=2}^{p} \lambda_{t}^{y} \sum_{\tilde{y}_{2}^{t-2}\in \mathcal{Y}^{t-2}} \mathds{1}\{y_{t-1}=0,y_{t-2}=\tilde{y}_2,\ldots,y_{1}=\tilde{y}_{t-1}\} \pi_{t-1}^{0|0,\tilde{y}_{2}^{t-1},y_{0},\ldots,y_{-(p-t)}}(a,x)\\
    &-\sum_{t=2}^{p} \lambda_{t}^{y} \sum_{\tilde{y}_{2}^{t-2}\in \mathcal{Y}^{t-2}} \mathds{1}\{y_{t-1}=1,y_{t-2}=\tilde{y}_2,\ldots,y_{1}=\tilde{y}_{t-1}\} \pi_{t-1}^{1|1,\tilde{y}_{2}^{t-1},y_{0},\ldots,y_{-(p-t)}}(a,x)\\
\end{align*}
Then, for the second group,
\begin{align*}
    T_{3}(a)&=\sum_{t=p+1}^{T} \lambda_{t}^{y,y^0} \sum_{\tilde{y}_{1}^{p}\in \mathcal{Y}^{p}} \mathds{1}\{y_{t-1}=\tilde{y}_1,\ldots,y_{t-p}=\tilde{y}_p\} \pi_{t-1}^{0|\tilde{y}_{1}^{p}}(a,x) \\
&=\sum_{t=p+1}^{T} \lambda_{t}^{y,y^0} \sum_{\tilde{y}_{2}^{p}\in \mathcal{Y}^{p-1}} \mathds{1}\{y_{t-1}=0,y_{t-2}=y_2,\ldots,y_{t-p}=\tilde{y}_p\}\pi_{t-1}^{0|0,\tilde{y}_{2}^{p}}(a,x) \\
&+\sum_{t=p+1}^{T} \lambda_{t}^{y,y^0} \sum_{\tilde{y}_{2}^{p-1}\in \mathcal{Y}^{p-1}} \mathds{1}\{y_{t-1}=1,y_{t-2}=y_2,\ldots,y_{t-p}=\tilde{y}_p\}\left(1-\pi_{t-1}^{1|1,\tilde{y}_{2}^{p}}(a,x)\right) \\
&=+\sum_{t=p+1}^{T} \lambda_{t}^{y,y^0} \sum_{\tilde{y}_{2}^{p-1}\in \mathcal{Y}^{p-1}} \mathds{1}\{y_{t-1}=1,y_{t-2}=y_2,\ldots,y_{t-p}=\tilde{y}_p\} \\
&+\sum_{t=p+1}^{T} \lambda_{t}^{y,y^0} \sum_{\tilde{y}_{2}^{p}\in \mathcal{Y}^{p-1}} \mathds{1}\{y_{t-1}=0,y_{t-2}=y_2,\ldots,y_{t-p}=\tilde{y}_p\}\pi_{t-1}^{0|0,\tilde{y}_{2}^{p}}(a,x) \\
&-\sum_{t=p+1}^{T} \lambda_{t}^{y,y^0} \sum_{\tilde{y}_{2}^{p-1}\in \mathcal{Y}^{p-1}} \mathds{1}\{y_{t-1}=1,y_{t-2}=y_2,\ldots,y_{t-p}=\tilde{y}_p\}\pi_{t-1}^{1|1,\tilde{y}_{2}^{p}}(a,x)
\end{align*}

\noindent The unique decompositions for each term make it clear that
\begin{align*}
    \mathcal{F}_{y^0,p,T}=\left\{1,\pi_{0}^{y_0|y^0}(.,x),\left\{\left(\pi_{t-1}^{y_{1}|y_{1}^{t-1},y_{0},\ldots,y_{-(p-t)}}(.,x))\right)_{y_{1}^{t-1}\in \mathcal{Y}^{t-1}}\right\}_{t=2}^p,\left\{\left(\pi_{t-1}^{y_1|y_{1}^{p}}(.,x)\right)_{y_{1}^p\in \mathcal{Y}^p}\right\}_{t=p+1}^T\right\}
\end{align*}
forms a basis of $\Ima\left(\mathcal{E}_{y^0,x}^{(p)}\right)$ if we can show that the transition probabilities are elements of $\Ima\left(\mathcal{E}_{y^0,x}^{(p)}\right)$. We now argue that it is indeed the case:
\begin{itemize}
    \item First, $\pi_{0}^{y_0|y^0}(.,x)\in \Ima\left(\mathcal{E}_{y^0,x}^{(p)}\right)$ since
    \begin{align*}
        \mathbb{E}[(1-Y_{i1})|Y_{i}^0=y^0,X_{i}=x,A_{i}=a]=\frac{1}{1+e^{\sum_{r=1}^{p}\gamma_{0r}y_{1-r}+x_1'\beta_0+a}}=\pi_{0}^{y_0|y^0}(a,x), \quad \text{if } y_0=0 \\
        \mathbb{E}[Y_{i1}|Y_{i}^0=y^0,X_{i}=x,A_{i}=a]=\frac{e^{\sum_{r=1}^{p}\gamma_{0r}y_{1-r}+x_1'\beta_0+a}}{1+e^{\sum_{r=1}^{p}\gamma_{0r}y_{1-r}+x_1'\beta_0+a}}=\pi_{0}^{y_0|y^0}(a,x), \quad \text{if } y_0=1
    \end{align*}
    \item Second, $\left\{\left(\pi_{t-1}^{y_1|y_{1}^{p}}(.,x)\right)_{y_{1}^p\in \mathcal{Y}^p}\right\}_{t=p+1}^T\in \Ima\left(\mathcal{E}_{y^0,x}^{(p)}\right)$ by Theorem \ref{theorem_ARp_transit}. For the AR(1) model, one would appeal to Lemma \ref{lemma_2}.\\
    \item Finally, one can easily adapt the proof of  Theorem \ref{theorem_ARp_transit} to show that $\left\{\left(\pi_{t-1}^{y_{1}|y_{1}^{t-1},y_{0},\ldots,y_{-(p-t)}}(.,x))\right)_{y_{1}^{t-1}\in \mathcal{Y}^{t-1}}\right\}_{t=2}^p\in \Ima\left(\mathcal{E}_{y^0,x}^{(p)}\right)$. First, it follows immediately from Lemma  \ref{lemma_theorem_ARp_transit} that:
    \begin{align*}
        \left(\pi_{1}^{y_{1}|y_{1},y_{0},\ldots,y_{-(p-2)}}(.,x))\right)_{y_{1}\in \mathcal{Y}^{t-1}}\in \Ima\left(\mathcal{E}_{y^0,x}^{(p)}\right)
    \end{align*}
    Then, by inspecting the induction argument of Theorem \ref{theorem_ARp_transit}, it is easily seen that the result that for $T\geq p+1$ and $t\in\{p,\ldots,T-1\}$
\begin{align*}
    \mathbb{E}\left[\phi_{\theta_0}^{y_1|y_{1}^{k+1}}(Y_{it+1},Y_{it},Y^{t-1}_{it-(p+k)},X_i)|Y_i^0,Y_{i1}^{t-(k+1)},X_i,A_i\right]&=\pi^{y_1|y_{1}^{k+1},Y_{it-(k+1)},\ldots,Y_{it-(p-1)}}_{t}(A_i,X_i)
\end{align*}
for $k=0,\ldots,p-2$ can be generalized. It actually holds for $t=k+1$ when $k=0,\ldots,p-2$,  yielding
\begin{align*}
    \mathbb{E}\left[\phi_{\theta_0}^{y_1|y_{1}^{t}}(Y_{it+1},Y_{it},Y^{t-1}_{i1-p},X_i)|Y_i^0,X_i,A_i\right]&=\pi^{y_1|y_{1}^{t},Y_{i0},\ldots,Y_{it-(p-1)}}_{t}(A_i,X_i)
\end{align*}
This is the desired result. The terms $\left\{\left(\pi_{t-1}^{y_{1}|y_{1}^{t-1},y_{0},\ldots,y_{-(p-t)}}(.,x))\right)_{y_{1}^{t-1}\in \mathcal{Y}^{t-1}}\right\}_{t=2}^p$ are not present in the AR(1) case which simplifies the argument.
\end{itemize}
Thus, we have shown that $ \mathcal{F}_{y^0,p,T}$ is a basis of $\Ima\left(\mathcal{E}_{y^0,x}^{(p)}\right)$. Next, since $\mathcal{E}_{y^0,x}^{(p)}$ is a linear mapping, we know by the \textit{rank nullity theorem} that:
\begin{align*}
\dim\left(\ker(\mathcal{E}_{y^0,x}^{(p)})\right)=\dim\left(\mathbb{R}^{\{0,1\}^T}\right)-\rank\left(\mathcal{E}_{y^0,x}^{(p)}\right)
\end{align*} 
Therefore, we have the following implications:
\begin{enumerate}
    \item If $T\leq p, \quad |\mathcal{F}_{y^0,p,T}|=1+1+\sum\limits_{t=2}^T2^{t-1}=2+\sum\limits_{t=1}^{T-1}2^{t}=2+2\frac{1-2^{T-1}}{1-2}=2^{T}$. Hence, $\rank\left(\mathcal{E}_{y^0,x}^{(p)}\right)=2^{T}$ and the rank nullity theorem implies $\dim\left(\ker(\mathcal{E}_{y^0,x}^{(p)})\right)=0$
    \item If $T=p+1,\quad |\mathcal{F}_{y^0,p,T}|=1+1+\sum\limits_{t=2}^p2^{t-1}+2^p=2\times2^p=2^{p+1}$. Then, $\rank\left(\mathcal{E}_{y^0,x}^{(p)}\right)=2^{T}$ and the rank nullity theorem implies $\dim\left(\ker(\mathcal{E}_{y^0,x}^{(p)})\right)=0$
    \item If $T\geq p+2$, $|\mathcal{F}_{y^0,p,T}|=1+1+\sum\limits_{t=2}^p2^{t-1}+2^{p}(T-p)=2^{p}+2^{p}(T-p)=(T-p+1)2^{p}$. It follows that $\rank\left(\mathcal{E}_{y^0,x}^{(p)}\right)=(T-p+1)2^{p}$ and $\dim\left(\ker(\mathcal{E}_{y^0,x}^{(p)})\right)=2^T-(T-p+1)2^{p}$
\end{enumerate}

\section{Proofs of Propositions \ref{proposition_1}, \ref{proposition_2}, \ref{proposition_3}}
Propositions \ref{proposition_1}, \ref{proposition_2} and \ref{proposition_3} all follow from the same strategy proof based on the the law of iterated expectations. We focus on Proposition \ref{proposition_1} here and leave the other cases to the reader. \\

\noindent Take any $t,s$ verifying $T-1\geq t>s\geq 1$. For any $k \in \mathcal{Y}$, we have
\begin{align*}
   \mathbb{E}\left[\psi_{\theta_0}^{k|k}(Y_{it-1}^{t+1},Y_{is-1}^{s+1})|Y_{i0},Y_{i1}^{s-1},A_i\right]&= \mathbb{E}\left[\phi_{\theta_0}^{k|k}(Y_{it+1},Y_{it},Y_{it-1})-\phi_{\theta_0}^{k|k}(Y_{is+1},Y_{is},Y_{is-1})|Y_{i0},Y_{i1}^{s-1},A_i\right]\\ 
   &=\mathbb{E}\left[\mathbb{E}\left[\phi_{\theta_0}^{k|k}(Y_{it+1},Y_{it},Y_{it-1})|Y_{i0},Y_{i1}^{t-1},A_i\right]|Y_{i0},Y_{i1}^{s-1},A_i\right]-
\pi^{k|k}(A_i) \\
&=\mathbb{E}\left[\pi^{k|k}(A_i)|Y_{i0},Y_{i1}^{s-1},A_i\right]-
\pi^{k|k}(A_i) \\
&=\pi^{k|k}(A_i)-\pi^{k|k}(A_i) \\
&=0
\end{align*}
The second and third equalities follow from the law of iterated expectation and Lemma \ref{lemma_1}. 

\section{Proofs of Lemma \ref{lemma_1} and Lemma \ref{lemma_2}}
Without loss of generality, we will consider the case with covariates. The proposed functional form for the transition function $\phi_{\theta}^{0|0}(Y_{it+1},Y_{it},Y_{it-1},X_i)$ implies that it is null when $Y_{it}\neq 0$. Hence
\begin{align*}
    &\mathbb{E}\left[\phi_{\theta}^{0|0}(Y_{it+1},Y_{it},Y_{it-1},X_i)|Y_{i0},Y_{i1}^{t-1},X_i,A_i\right]=  \frac{1}{1+e^{\gamma_0 Y_{it-1}+X_{it}'\beta_0+A_i}}\times  \\
    & \left(\frac{e^{X_{it+1}'\beta_0+A_i}}{1+e^{X_{it+1}'\beta_0+A_i}}\phi_{\theta}^{0|0}(1,0,Y_{it-1},X_i)+\frac{1}{1+e^{X_{it+1}'\beta_0+A_i}}\phi_{\theta}^{0|0}(0,0,Y_{it-1},X_i)\right)
\end{align*}
Thus, to obtain the transition probability $\pi^{0|0}_{t}(A_i,X_{i})=\frac{1}{1+e^{X_{it+1}'\beta_0+A_i}}$ at $\theta=\theta_0$, we must set:
\begin{align*}
    \phi_{\theta}^{0|0}(1,0,Y_{it-1},X_i)&=e^{\gamma Y_{it-1}+(X_{it}-X_{it+1})'\beta} \\
    \phi_{\theta}^{0|0}(0,0,Y_{it-1},X_i)&=1 \\
    \phi_{\theta}^{0|0}(k,1,Y_{it-1},X_i)&=0, \quad \forall k \in \mathcal{Y}
\end{align*}
This can be expressed compactly as: $\phi_{\theta}^{0|0}(Y_{it+1},Y_{it},Y_{it-1},X_i)=(1-Y_{it})e^{Y_{it+1}\left(\gamma Y_{it-1}-\Delta X_{it+1}'\beta \right)}$

\noindent Likewise, for  $\phi_{\theta}^{1|1}(Y_{it+1},Y_{it},Y_{it-1},X_i)$ we have:
\begin{align*}
    &\mathbb{E}\left[\phi_{\theta}^{1|1}(Y_{it+1},Y_{it},Y_{it-1},X_i)|Y_{i0},Y_{i1}^{t-1},X_i,A_i\right]=  \frac{e^{\gamma_0 Y_{it-1}+X_{it}'\beta_0+A_i}}{1+e^{\gamma_0 Y_{it-1}+X_{it}'\beta_0+A_i}}\times  \\
    &\left(\frac{e^{\gamma_0+X_{it+1}'\beta_0+A_i}}{1+e^{\gamma_0+X_{it+1}'\beta_0+A_i}}\phi_{\theta}^{1|1}(1,1,Y_{it-1},X_i)+\frac{1}{1+e^{\gamma_0+X_{it+1}'\beta_0+A_i}}\phi_{\theta}^{1|1}(0,1,Y_{it-1},X_i)\right)
\end{align*}
Hence, to get $\pi^{1|1}_{t}(A_i,X_{i})=\frac{e^{\gamma_0+X_{it+1}'\beta_0+A_i}}{1+e^{\gamma_0+X_{it+1}'\beta_0+A_i}}$  at $\theta=\theta_0$, we must set:
\begin{align*}
    \phi_{\theta}^{1|1}(1,1,Y_{it-1},X_i)&=1 \\
    \phi_{\theta}^{1|1}(0,1,Y_{it-1},X_i)&=e^{\gamma (1-Y_{it-1})+(X_{it+1}-X_{it})'\beta} \\
    \phi_{\theta}^{1|1}(k,0,Y_{it-1},X_i)&=0, \quad \forall k \in \mathcal{Y}
\end{align*}
This can be written succinctly as: $\phi_{\theta}^{1|1}(Y_{it+1},Y_{it},Y_{it-1},X_i)=Y_{it}e^{ (1-Y_{it+1})\left(\gamma(1-Y_{it-1})+\beta \Delta X_{it+1}\right)}$

\section{Proofs of Lemmas \ref{lemma_3},\ref{lemma_4} and Corollaries \ref{corollary_2}, \ref{corollary_4}}
The proofs of Lemma \ref{lemma_3}, Lemma \ref{lemma_4}, Corollary \ref{corollary_2}, Corollary \ref{corollary_4} all follow the same logic based on the use of a \textit{partial fraction expansion}. We prove Lemma \ref{lemma_3} here and leave the other cases to the reader.

The result hinges on the simple rational fraction identity provided in Lemma \ref{tech_lemma_1} that for any three reals $v,u,a$, we have:
\begin{align*}
    &\frac{1}{1+e^{v+a}}+(1-e^{u-v})\frac{e^{v+a}}{(1+e^{v+a})(1+e^{u+a})}=\frac{1}{(1+e^{u+a})} \\
    &\frac{e^{v+a}}{1+e^{v+a}}+(1-e^{-(u-v)})\frac{e^{u+a}}{(1+e^{v+a})(1+e^{u+a})}=\frac{e^{u+a}}{(1+e^{u+a})}
\end{align*}
By construction for $T\geq 3$, and $t,s$ such that $T-1\geq t> s\geq 1$:
\begin{align*}
    &\mathbb{E}\left[\zeta_{\theta_{0}}^{0|0}(Y_{it-1}^{t+1},Y_{is-1}^s,X_i)|Y_{i0},Y_{i1}^{s-1},X_i,A_i\right]\\
    &=\mathbb{E}\left[(1-Y_{is})+\omega_{t,s}^{0|0}(\theta_{0})Y_{is}\phi_{\theta_{0}}^{0|0}(Y_{it+1},Y_{it},Y_{it-1},X_i)|Y_{i0},Y_{i1}^{s-1},X_i,A_i\right] \\
    &=\frac{1}{1+e^{\mu_{s}(\theta_{0})+A_i}}+\omega_{t,s}^{0|0}(\theta_{0}) \mathbb{E}\left[Y_{is} \mathbb{E}\left[\phi_{\theta_{0}}^{0|0}(Y_{it+1},Y_{it},Y_{it-1},X_i)|Y_{i0},Y_{i1}^{t-1},X_i,A_i\right]|Y_{i0},Y_{i1}^{s-1},X_i,A_i\right] \\
    &=\frac{1}{1+e^{\mu_{s}(\theta_{0})+A_i}}+\omega_{t,s}^{0|0}(\theta_{0}) \mathbb{E}\left[Y_{is} |Y_{i0},Y_{i1}^{s-1},X_i,A_i\right]\frac{1}{1+e^{\kappa_{t}^{0|0}(\theta_{0})+A_i}} \\
    &=\frac{1}{1+e^{\mu_{s}(\theta_{0})+A_i}}+(1-e^{\kappa_{t}^{0|0}(\theta_{0})-\mu_{s}(\theta_{0})})\frac{e^{\mu_{s}(\theta_{0})+A_i}}{(1+e^{\mu_{s}(\theta_{0})+A_i})(1+e^{\kappa_{t}^{0|0}(\theta_{0})+A_i})} \\
    &=\frac{1}{1+e^{\kappa_{t}^{0|0}(\theta_{0})+A_i}} \\
    &=\pi^{0|0}_{t}(A_i,X_{i})
\end{align*}
The second equality follows from the measureability of the weight $\omega_{t,s}^{0|0}(\theta_{0})$ with respect to the conditioning set. The third equality follows from the law of iterated expectations and Lemma \ref{lemma_2}. The penultimate equality uses the first mathematical identity presented above. \\
Similarly, 
\begin{align*}
     &\mathbb{E}\left[\zeta_{\theta_{0}}^{1|1}(Y_{it-1}^{t+1},Y_{is-1}^s,X_i)|Y_{i0},Y_{i1}^{s-1},X_i,A_i\right]\\
     &=\mathbb{E}\left[Y_{is}+\omega_{t,s}^{1|1}(\theta_{0})(1-Y_{is})\phi_{\theta_{0}}^{1|1}(Y_{it+1},Y_{it},Y_{it-1},X_i)|Y_{i0},Y_{i1}^{s-1},X_i,A_i\right] \\
     &= \frac{e^{\mu_{s}(\theta_{0})+A_i}}{1+e^{\mu_{s}(\theta_{0})+A_i}}+\omega_{t,s}^{1|1}(\theta_{0})\mathbb{E}\left[(1-Y_{is})\mathbb{E}\left[\phi_{\theta_{0}}^{1|1}(Y_{it+1},Y_{it},Y_{it-1},X_i)|Y_{i0},Y_{i1}^{t-1},X_i,A_i\right]|Y_{i0},Y_{i1}^{s-1},X_i,A_i\right] \\
     &= \frac{e^{\mu_{s}(\theta_{0})+A_i}}{1+e^{\mu_{s}(\theta_{0})+A_i}}+\omega_{t,s}^{1|1}(\theta_{0}) \mathbb{E}\left[(1-Y_{is})|Y_{i0},Y_{i1}^{s-1},X_i,A_i\right]\frac{e^{\kappa_{t}^{1|1}(\theta_{0})+A_i}}{1+e^{\kappa_{t}^{1|1}(\theta_{0})+A_i}} \\
     &=\frac{e^{\mu_{s}(\theta_{0})+A_i}}{1+e^{\mu_{s}(\theta_{0})+A_i}}+\left(1-e^{-(\kappa_{t}^{1|1}(\theta_{0})-\mu_{s}(\theta_{0}))}\right)\frac{e^{\kappa_{t}^{1|1}(\theta_{0})+A_i}}{(1+e^{\mu_{s}(\theta_{0})+A_i})(1+e^{\kappa_{t}^{1|1}(\theta_{0})+A_i})} \\
     &=\frac{e^{\kappa_{t}^{1|1}(\theta_{0})+A_i}}{1+e^{\kappa_{t}^{1|1}(\theta_{0})+A_i}} \\
     &=\pi^{1|1}_{t}(A_i,X_{i})
\end{align*}
The second equality follows from the measurability of the weight $\omega_{t,s}^{0|0}(\theta_{0})$ with respect to the conditioning set. The third equality follows from the law of iterated expectations and Lemma \ref{lemma_2}. The penultimate equality uses the second mathematical identity presented above. 

\section{Proof of Theorem \ref{theorem_ARp_transit}}
We start by proving the following Lemma
\begin{lemma} \label{lemma_theorem_ARp_transit}
In model (\ref{ARp_logit_general}), with $T\geq 2$ and $t\in\{1,\ldots,T-1\}$, let
\begin{align*}
    \phi_{\theta}^{0|0}(Y_{it+1},Y_{it},Y_{it-p}^{t-1},X_i)&=(1-Y_{it})e^{Y_{it+1}(\gamma_{1}Y_{it-1}-\sum_{l=2}^{p} \gamma_{l}\Delta Y_{it+1-l}-\Delta X_{it+1}'\beta)} \\
    \phi_{\theta}^{1|1}(Y_{it+1},Y_{t},Y^{t-1}_{it-p},X_i)&=Y_{it}e^{ (1-Y_{it+1})\left(\gamma_{1}(1-Y_{it-1})+\sum_{l=2}^{p} \gamma_{l}\Delta Y_{it+1-l}+\Delta X_{it+1}'\beta\right)}
\end{align*}
Then, 
\begin{align*}
    \mathbb{E}\left[\phi_{\theta_0}^{0|0}(Y_{it+1},Y_{it},Y^{t-1}_{it-p},X_i)|Y_i^0,Y_{i1}^{t-1},X_i,A_i\right]&= \pi^{0|0,Y_{it-1},\ldots,Y_{it-(p-1)}}_{t}(A_i,X_{i})\\
    &=\frac{1}{1+e^{\sum_{l=2}^{p} \gamma_{0l} Y_{it+1-l}+X_{it+1}'\beta_0+A_i}} \\
    \mathbb{E}\left[\phi_{\theta_0}^{1|1}(Y_{it+1},Y_{it},Y^{t-1}_{it-p},X_i)|Y_i^0,Y_{i1}^{t-1},X_i,A_i \right]&=\pi^{1|1,Y_{it-1},\ldots,Y_{it-(p-1)}}_{t}(A_i,X_{i})\\
    &=\frac{e^{\gamma_{01}+\sum_{l=2}^{p} \gamma_{0l} Y_{it+1-l}+X_{it+1}'\beta_0+A_i}}{1+e^{\gamma_{01}+\sum_{l=2}^{p} \gamma_{0l} Y_{it+1-l}+X_{it+1}'\beta_0+A_i}}
\end{align*}
\end{lemma}
\noindent Instead of verifying the result directly from the expression given in the Lemma, it is easier to start from the heuristic idea, emphasized throughout the text, that we look for two functions such that:
\begin{align*}
    &\phi_{\theta}^{0|0}(Y_{it+1},Y_{it},Y^{t-1}_{it-p},X_i)=(1-Y_{it})\phi_{\theta}^{0|0}(Y_{it+1},0,Y^{t-1}_{it-p},X_i) \\
     &\phi_{\theta}^{1|1}(Y_{it+1},Y_{it},Y_{it-1},X_i)=Y_{it}\phi_{\theta}^{1|1}(Y_{it+1},1,Y^{t-1}_{it-p},X_i) \\
     &\mathbb{E}\left[\phi_{\theta_0}^{k|k}(Y_{it+1},Y_{it},Y^{t-1}_{it-p},X_i)|Y_i^0,Y_{i1}^{t-1},X_i,A_i\right]= \pi^{k|k,Y_{it-1},\ldots,Y_{it-(p-1)}}_{t}(A_i,X_{i}), \quad  \forall k \in \mathcal{Y}
\end{align*}
 By definition, $\phi_{\theta}^{0|0}(Y_{it+1},Y_{it},Y^{t-1}_{it-p},X_i)$ is null when $Y_{it}\neq 0$. Hence
\begin{align*}
    &\mathbb{E}\left[\phi_{\theta}^{0|0}(Y_{it+1},Y_{it},Y^{t-1}_{it-p},X_i)|Y_i^0,Y_{i1}^{t-1},X,A\right]=  \frac{1}{1+e^{\sum_{l=1}^{p} \gamma_{0l} Y_{it-l}+X_{it}'\beta_0+A_i}}\times (  \\
    & \frac{e^{\sum_{l=2}^{p} \gamma_{0l} Y_{it+1-l}+X_{it+1}'\beta_0+A_i}}{1+e^{\sum_{l=2}^{p} \gamma_{0l} Y_{it+1-l}+X_{it+1}'\beta_0+A_i}}\phi_{\theta}^{0|0}(1,0,Y^{t-1}_{it-p},X_i)+\frac{1}{1+e^{\gamma_{02} Y_{it-1}+X_{it+1}'\beta_0+A_i}}\phi_{\theta}^{0|0}(0,0,Y^{t-1}_{it-p},X_i) )
\end{align*}
Thus, to obtain $\pi^{0|0,Y_{it-1},\ldots,Y_{it-(p-1)}}_{t}(A_i,X_{i})=\frac{1}{1+e^{\sum_{l=2}^{p} \gamma_{0l} Y_{it+1-l}+X_{it+1}'\beta_0+A_i}}$ at $\theta=\theta_0$, we must set:
\begin{align*}
   \phi_{\theta}^{0|0}(1,0,Y^{t-1}_{it-p},X_i)&=e^{\gamma_{1}Y_{it-1}-\sum_{l=2}^{p} \gamma_{l}\Delta Y_{it+1-l}-\Delta X_{it+1}'\beta} \\
    \phi_{\theta}^{0|0}(0,0,Y^{t-1}_{it-p},X_i)&=1 \\
    \phi_{\theta}^{0|0}(k,1,Y^{t-1}_{it-p},X_i)&=0, \forall k \in \mathcal{Y}
\end{align*}
more compactly this writes, $$\phi_{\theta}^{0|0}(Y_{it+1},Y_{it},Y^{t-1}_{it-p},X_i)=(1-Y_{it})e^{Y_{it+1}(\gamma_{1}Y_{it-1}-\sum_{l=2}^{p} \gamma_{l}\Delta Y_{it+1-l}-\Delta X_{it+1}'\beta)}$$
Analogously, $\phi_{\theta}^{1|1}(Y_{it+1},Y_{it},Y^{t-1}_{it-p},X_i)$ is null when $Y_{it}\neq 1$. Hence
\begin{align*}
    &\mathbb{E}\left[\phi_{\theta}^{1|1}(Y_{it+1},Y_{it},Y^{t-1}_{it-p},X_i))|Y_i^0,Y_{1}^{t-1},X,A\right]=  \frac{e^{\sum_{l=1}^{p} \gamma_{0l} Y_{it-l}+X_{it}'\beta_0+A_i}}{1+e^{\sum_{l=1}^{p} \gamma_{0l} Y_{it-l}+X_{it}'\beta_0+A_i}}\times (  \\
    & \frac{e^{\gamma_{01}+\sum_{l=2}^{p} \gamma_{0l} Y_{it+1-l}+X_{it+1}'\beta_0+A_i}}{1+e^{\gamma_{01}+\sum_{l=2}^{p} \gamma_{0l} Y_{it+1-l}+X_{it+1}'\beta_0+A_i}}\phi_{\theta}^{1|1}(1,1,Y^{t-1}_{it-p},X_i)+\frac{1}{1+e^{\gamma_{01}+\gamma_{02} Y_{it-1}+X_{it+1}'\beta_0+A_i}}\phi_{\theta}^{1|1}(0,1,Y^{t-1}_{it-p},X_i) )
\end{align*}
Consequently, to get $\pi^{1|1,Y_{it-1},\ldots,Y_{it-(p-1)}}_{t}(A_i,X_{i})=\frac{e^{\gamma_{01}+\sum_{l=2}^{p} \gamma_{0l} Y_{it+1-l}+X_{it+1}'\beta_0+A_i}}{1+e^{\gamma_{01}+\sum_{l=2}^{p} \gamma_{0l} Y_{it+1-l}+X_{it+1}'\beta_0+A_i}}$ at $\theta=\theta_0$, we must set:
\begin{align*}
    \phi_{\theta}^{1|1}(1,1,Y^{t-1}_{it-p},X_i)&=1 \\
    \phi_{\theta}^{1|1}(0,1,Y^{t-1}_{it-p},X_i)&=e^{\gamma_{1}(1-Y_{it-1})+\sum_{l=2}^{p} \gamma_{l}\Delta Y_{it+1-l}+\Delta X_{it+1}'\beta} \\
    \phi_{\theta}^{1|1}(k,0,Y^{t-1}_{it-p},X_i)&=0, \forall k \in \mathcal{Y}
\end{align*}
This can be written succinctly as: $$ \phi_{\theta}^{1|1}(Y_{it+1},Y_{t},Y^{t-1}_{it-p},X_i)=Y_{it}e^{ (1-Y_{it+1})\left(\gamma_{1}(1-Y_{it-1})+\sum_{l=2}^{p} \gamma_{l}\Delta Y_{it+1-l}+\Delta X_{it+1}'\beta\right)}$$ which completes the proof of the Lemma. \\

\indent Now, for $T\geq p+1$ fix $t\in\{p,\ldots,T-1\}$ and $y=(y_1,\ldots,y_p)=y_{1}^p\in \{0,1\}^p$.
We will prove by finite induction the statement $\mathcal{P}(k)$:
\begin{align*}
    \mathbb{E}\left[\phi_{\theta_0}^{y_1|y_{1}^{k+1}}(Y_{it+1},Y_{it},Y^{t-1}_{it-(p+k)},X_i)|Y_i^0,Y_{i1}^{t-(k+1)},X_i,A_i\right]&=\pi^{y_1|y_{1}^{k+1},Y_{it-(k+1)},\ldots,Y_{it-(p-1)}}_{t}(A_i,X_i)
\end{align*}
for $k=0,\ldots,p-2$ for $p\geq 2$. 
\newpage
\noindent \textbf{Base step:} \\
$\mathcal{P}(0)$ is true by Lemma \ref{lemma_theorem_ARp_transit} which also deals with the edge case $p=2$. Thus, let us assume $p\geq 3$ in the remainder of the induction argument.\\

\noindent \textbf{Induction Step:} \\
Suppose $\mathcal{P}(k-1)$ is true for some $k\in \{1,\ldots,p-2\}$, we show that $\mathcal{P}(k)$ is true.
Using the law of iterated expectations, the induction hypothesis $\mathcal{P}(k-1)$ and the identities of Lemma \ref{tech_lemma_1}, we have: \\
\noindent If $y_1=0,y_{k+1}=1$
\begin{align*}
  & \mathbb{E}\left[\phi_{\theta_0}^{0|0,y_{2}^{k},1}(Y_{it+1},Y_{it},Y^{t-1}_{it-(p+k)},X_i)|Y_i^0,Y_{i1}^{t-(k+1)},X_i,A_i\right] \\
  &=\mathbb{E}\left[(1-Y_{it-k})+w^{0|0,y_{2}^k,1}_t(\theta_0)\phi_{\theta_0}^{0|0,y_{2}^{k}}(Y_{it+1},Y_{it},Y^{t-1}_{it-(p+k-1)},X_i)Y_{it-k}|Y_i^0,Y_{i1}^{t-(k+1)},X_i,A_i\right]\\
  &=\frac{1}{1+e^{u_{t-k}(\theta_0)+A_i}}\\
 &+w^{0|0,y_{2}^k,1}_t(\theta_0)\mathbb{E}\left[\mathbb{E}\left[\phi_{\theta_0}^{0|0,y_{2}^{k}}(Y_{it+1},Y_{it},Y^{t-1}_{it-(p+k-1)},X_i)|Y_i^0,Y_{i1}^{t-k},X_i,A_i\right]Y_{it-k}|Y_i^0,Y_{i1}^{t-(k+1)},X_i,A_i\right] \\
  &=\frac{1}{1+e^{u_{t-k}(\theta_0)+A_i}}w^{0|0,y_{2}^k,1}_t(\theta_0)\mathbb{E}\left[\pi^{0|0,y_{2}^{k},Y_{it-k},\ldots,Y_{it-(p-1)}}_{t}(A_i,X_i)Y_{it-k}|Y_i^0,Y_{i1}^{t-(k+1)},X_i,A_i\right] \\
  &=\frac{1}{1+e^{u_{t-k}(\theta_0)+A_i}}+w^{0|0,y_{2}^k,1}_t(\theta_0)\mathbb{E}\left[\frac{1}{1+e^{\sum_{r=2}^{k}\gamma_{0r}y_r+\sum_{r=k+1}^{p} \gamma_{0r}Y_{it-(r-1)}+X_{it+1}'\beta_0+A_i}}Y_{it-k}|Y_i^0,Y_{i1}^{t-(k+1)},X_i,A_i\right] \\
  &=\frac{1}{1+e^{u_{t-k}(\theta_0)+A_i}}+(1-e^{(k^{0|0,y_{2}^k,1}_{t}(\theta_0)-u_{t-k}(\theta_0))})\frac{1}{1+e^{k^{0|0,y_{2}^k,1}_{t}(\theta_0)+A_i}}\frac{e^{u_{t-k}(\theta_0)+A_i}}{1+e^{u_{t-k}(\theta_0)+A_i}} \\
  &=\frac{1}{1+e^{k^{0|0,y_{2}^k,1}_{t}(\theta_0)+A_i}} \\
  &=\pi^{0|0,y_{2}^{k},1,Y_{it-(k+1)},\ldots,Y_{it-(p-1)}}_{t}(A_i,X_i)
\end{align*}
\noindent If $y_1=0,y_{k+1}=0$
\begin{align*}
    & \mathbb{E}\left[\phi_{\theta_0}^{0|0,y_{2}^{k},0}(Y_{it+1},Y_{it},Y^{t-1}_{it-(p+k)},X_i)|Y_i^0,Y_{i1}^{t-(k+1)},X_i,A_i\right] \\
    &=\mathbb{E}\left[1-Y_{it-k}-w^{0|0,y_{2}^k,0}_t(\theta_0)\left(1-\phi_{\theta_0}^{0|0,y_{2}^{k}}(Y_{it+1},Y_{it},Y^{t-1}_{it-(p+k-1)},X_i)\right)(1-Y_{it-k})|Y_i^0,Y_{i1}^{t-(k+1)},X_i,A_i\right] \\
    &=1-\frac{e^{u_{t-k}(\theta_0)+A_i}}{1+e^{u_{t-k}(\theta_0)+A_i}}\\
    &-w^{0|0,y_{2}^k,0}_t(\theta_0)\times \\
    &\mathbb{E}\left[\mathbb{E}\left[\left(1-\phi_{\theta_0}^{0|0,y_{2}^{k}}(Y_{it+1},Y_{it},Y^{t-1}_{it-(p+k-1)},X_i)\right)|Y_i^0,Y_{i1}^{t-k},X_i,A_i\right](1-Y_{it-k})|Y_i^0,Y_{i1}^{t-(k+1)},X_i,A_i\right] \\
    &=1-\frac{e^{u_{t-k}(\theta_0)+A_i}}{1+e^{u_{t-k}(\theta_0)+A_i}} 
    -w^{0|0,y_{2}^k,0}_t(\theta_0)\mathbb{E}\left[(1-\pi^{0|0,y_{2}^{k},Y_{it-k},\ldots,Y_{it-(p-1)}}_{t}(A_i,X_i))(1-Y_{it-k})|Y_i^0,Y_{i1}^{t-(k+1)},X_i,A_i\right] \\
    &=1-\frac{e^{u_{t-k}(\theta_0)+A_i}}{1+e^{u_{t-k}(\theta_0)+A_i}} \\
    &-w^{0|0,y_{2}^k,0}_t(\theta_0)\mathbb{E}\left[\frac{e^{\sum_{r=2}^{k}\gamma_{0r}y_r+\sum_{r=k+1}^{p} \gamma_{0r}Y_{it-(r-1)}+X_{it+1}'\beta_0+A_i}}{1+e^{\sum_{r=2}^{k}\gamma_{0r}y_r+\sum_{r=k+1}^{p} \gamma_{0r}Y_{it-(r-1)}+X_{it+1}'\beta_0+A_i}}(1-Y_{it-k})|Y_i^0,Y_{i1}^{t-(k+1)},X_i,A_i\right] \\
    &=1-\left(\frac{e^{u_{t-k}(\theta_0)+A_i}}{1+e^{u_{t-k}(\theta_0)+A_i}}
    +(1-e^{-(k^{0|0,y_{2}^k,0}_{t}(\theta_0)-u_{t-k}(\theta_0))})\frac{e^{k^{0|0,y_{2}^k,0}_{t}(\theta_0)+A_i}}{1+e^{k^{0|0,y_{2}^k,0}_{t}(\theta_0)+A_i}}\frac{1}{1+e^{u_{t-k}(\theta_0)+A_i}}\right) \\
    &=1-\frac{e^{k^{0|0,y_{2}^k,0}_{t}(\theta_0)+A_i}}{1+e^{k^{0|0,y_{2}^k,0}_{t}(\theta_0)+A_i}} \\
    &=\frac{1}{1+e^{k^{0|0,y_{2}^k,0}_{t}(\theta_0)+A_i}} \\
    &=\pi^{0|0,y_{2}^{k},0,Y_{it-(k+1)},\ldots,Y_{it-(p-1)}}_{t}(A_i,X_i)
\end{align*}
\noindent If $y_1=1,y_{k+1}=0$
\begin{align*}
     & \mathbb{E}\left[\phi_{\theta_0}^{1|1,y_{2}^{k},0}(Y_{it+1},Y_{it},Y^{t-1}_{it-(p+k)},X_i)|Y_i^0,Y_{i1}^{t-(k+1)},X_i,A_i\right] \\
     &=\mathbb{E}\left[Y_{it-k}+w^{1|1,y_{2}^k,0}_t(\theta_0)\phi_{\theta_0}^{1|1,y_{2}^{k}}(Y_{it+1},Y_{it},Y^{t-1}_{it-(p+k-1)},X_i)(1-Y_{it-k})|Y_i^0,Y_{i1}^{t-(k+1)},X_i,A_i\right]\\
     &=\frac{e^{u_{t-k}(\theta_0)+A_i}}{1+e^{u_{t-k}(\theta_0)+A_i}}+w^{1|1,y_{2}^k,0}_t(\theta_0)\times \\
     &\mathbb{E}\left[\mathbb{E}\left[\phi_{\theta_0}^{1|1,y_{2}^{k}}(Y_{it+1},Y_{it},Y^{t-1}_{it-(p+k-1)},X_i)|Y_i^0,Y_{i1}^{t-k},X_i,A_i\right](1-Y_{it-k})|Y_i^0,Y_{i1}^{t-(k+1)},X_i,A_i\right] \\
     &=\frac{e^{u_{t-k}(\theta_0)+A_i}}{1+e^{u_{t-k}(\theta_0)+A_i}}+w^{1|1,y_{2}^k,0}_t(\theta_0)\mathbb{E}\left[\pi^{1|1,y_{2}^{k},Y_{it-k},\ldots,Y_{it-(p-1)}}_{t}(A_i,X_i)(1-Y_{it-k})|Y_i^0,Y_{i1}^{t-(k+1)},X_i,A_i\right] \\
     &=\frac{e^{u_{t-k}(\theta_0)+A_i}}{1+e^{u_{t-k}(\theta_0)+A_i}}\\
     &+w^{1|1,y_{2}^k,0}_t(\theta_0)\mathbb{E}\left[\frac{e^{\gamma_{01}+\sum_{r=2}^{k}\gamma_{0r}y_r+\sum_{r=k+1}^{p} \gamma_{0r}Y_{it-(r-1)}+X_{it+1}'\beta_0+A_i}}{1+e^{\gamma_{01}+\sum_{r=2}^{k}\gamma_{0r}y_r+\sum_{r=k+1}^{p} \gamma_{0r}Y_{it-(r-1)}+X_{it+1}'\beta_0+A_i}}(1-Y_{it-k})|Y_i^0,Y_{i1}^{t-(k+1)},X_i,A_i\right] \\
     &=\frac{e^{u_{t-k}(\theta_0)+A_i}}{1+e^{u_{t-k}(\theta_0)+A_i}}+(1-e^{-(k^{1|1,y_{2}^k,0}_{t}(\theta_0)-u_{t-k}(\theta_0))})\frac{e^{k^{1|1,y_{2}^k,0}_{t}(\theta_0)+A_i}}{1+e^{k^{1|1,y_{2}^k,0}_{t}(\theta_0)+A_i}}\frac{1}{1+e^{u_{t-k}(\theta_0)+A_i}} \\
     &=\frac{e^{k^{1|1,y_{2}^k,0}_{t}(\theta_0)+A_i}}{1+e^{k^{1|1,y_{2}^k,0}_{t}(\theta_0)+A_i}} \\
     &=\pi^{1|1,y_{2}^{k},0,Y_{it-(k+1)},\ldots,Y_{it-(p-1)}}_{t}(A_i,X_i)
\end{align*} 
\noindent If $y_1=1,y_{k+1}=1$
\begin{align*}
        & \mathbb{E}\left[\phi_{\theta_0}^{1|1,y_{2}^{k},1}(Y_{it+1},Y_{it},Y^{t-1}_{it-(p+k)},X_i)|Y_i^0,Y_{i1}^{t-(k+1)},X_i,A_i\right] \\
         &=\mathbb{E}\left[1-(1-Y_{it-k})-w^{1|1,y_{2}^k,1}_t(\theta_0)\left(1-\phi_{\theta_0}^{1|1,y_{2}^{k}}(Y_{it+1},Y_{it},Y^{t-1}_{it-(p+k-1)},X_i)\right)Y_{it-k}|Y_i^0,Y_{i1}^{t-(k+1)},X_i,A_i\right]\\
         &=1-\frac{1}{1+e^{u_{t-k}(\theta_0)+A_i}}\\
         &-w^{1|1,y_{2}^k,1}_t(\theta_0)\mathbb{E}\left[\mathbb{E}\left[\left(1-\pi^{1|1,y_{2}^{k},Y_{it-k},\ldots,Y_{it-(p-1)}}_{t}(A_i,X_i)\right)|Y_i^0,Y_{i1}^{t-k},X_i,A_i\right]Y_{it-k}|Y_i^0,Y_{i1}^{t-(k+1)},X_i,A_i\right] \\
         &=1-\frac{1}{1+e^{u_{t-k}(\theta_0)+A_i}}\\
         &-w^{1|1,y_{2}^k,1}_t(\theta_0)\mathbb{E}\left[\frac{1}{1+e^{\gamma_{01}+\sum_{r=2}^{k}\gamma_{0r}y_r+\sum_{r=k+1}^{p} \gamma_{0r}Y_{it-(r-1)}+X_{it+1}'\beta_0+A_i}}Y_{it-k}|Y_i^0,Y_{i1}^{t-(k+1)},X_i,A_i\right] \\
         &=1-\left(\frac{1}{1+e^{u_{t-k}(\theta_0)+A_i}}+(1-e^{(k^{1|1,y_{2}^k,1}_{t}(\theta_0)-u_{t-k}(\theta_0))})\frac{1}{1+e^{k^{1|1,y_{2}^k,1}_{t}(\theta_0)+A_i}}\frac{e^{u_{t-k}(\theta_0)+A_i}}{1+e^{u_{t-k}(\theta_0)+A_i}}\right) \\
         &=1-\frac{1}{1+e^{k^{1|1,y_{2}^k,1}_{t}(\theta_0)+A_i}} \\
         &=\frac{e^{k^{1|1,y_{2}^k,1}_{t}(\theta_0)+A_i}}{1+e^{k^{1|1,y_{2}^k,1}_{t}(\theta_0)+A_i}} \\
        &=\pi^{1|1,y_{2}^{k},1,Y_{it-k},\ldots,Y_{it-(p-1)}}_{t}(A_i,X_i)
\end{align*}

\noindent Putting these intermediate results together, we have effectively proved that
\begin{align*}
    &\mathbb{E}\left[\phi_{\theta_0}^{y_1|y_{1}^{k+1}}(Y_{it+1},Y_{it},Y^{t-1}_{it-(p+k)},X_i)|Y_i^0,Y_{i1}^{t-(k+1)},X_i,A_i\right]=\pi^{y_1|y_{1}^{k+1},Y_{it-(k+1)},\ldots,Y_{it-(p-1)}}_{t}(A_i,X_i)
\end{align*}
which shows that $\mathcal{P}(k)$ is true and completes the induction argument. \\
\indent Now, it only remains to show that
\begin{align*}
    \mathbb{E}\left[\phi_{\theta_0}^{y_1|y_{1}^{p}}(Y_{it+1},Y_{it},Y^{t-1}_{it-(2p-1)},X_i)|Y_i^0,Y_{i1}^{t-p},X_i,A_i\right]=\pi^{y_1|y_{1}^{p}}_{t}(A_i,X_i)
\end{align*}
To this end, it suffices to perform calculations identical to those used in the induction argument but using this time
\begin{align*}
    &\mathbb{E}\left[\phi_{\theta_0}^{y_1|y_{1}^{p-1}}(Y_{it+1},Y_{it},Y^{t-1}_{it-(2p-2)},X_i)|Y_i^0,Y_{i1}^{t-(p-1)},X_i,A_i\right]=\pi^{y_1|y_{1}^{p-1},Y_{it-(p-1)}}_{t}(A_i,X_i) \\
    &k^{y_1|y_{1}^{p}}_{t}(\theta)=\sum_{r=1}^{p}\gamma_{r}y_r+X_{it+1}'\beta \\
    &u_{t-(p-1)}(\theta)=\sum_{r=1}^{p} \gamma_{r}Y_{it-(r+p-1)}+X_{it-(p-1)}'\beta \\
    &w^{y_1|y_{1}^{p}}_t(\theta)=\left[1-e^{(k^{y_1|y_{1}^{p}}_{t}(\theta)-u_{t-(p-1)}(\theta))}\right]^{y_{p}}\left[1-e^{-(k^{y_1|y_{1}^{p}}_{t}(\theta)-u_{t-(p-1)}(\theta))}\right]^{1-y_{p}}
\end{align*}
This concludes the proof of the theorem.

\section{Identification of the AR(2) with strictly exogenous regressors and $T=3$} \label{appendix_identif_AR2_logit_general}
\subsection{Identification for \texorpdfstring{$T=3$}{} with variability in the initial condition} \label{appendix_identif_AR2_logit_general_subsec1}
\noindent By Theorem \ref{theorem_ARp_transit}, the transition functions associated to: $\pi_{2}^{0|0,0}(A_i,X_i),\pi_{2}^{0|0,1}(A_i,X_i),\pi_{2}^{1|1,0}(A_i,X_i), \pi_{2}^{1|1,1}(A_i,X_i)$ are given by:
\begin{align*}
    &\phi_{\theta}^{0|0,0}(Y_{i3},Y_{i2},Y^{1}_{i-1},X_i)=e^{\gamma_1Y_{i0}+\gamma_2Y_{i-1}-X_{i31}'\beta}(1-Y_{i1}) \\
    &+\left(1-e^{\gamma_1Y_{i0}+\gamma_2Y_{i-1}-X_{i31}'\beta}\right)(1-Y_{i1})(1-Y_{i2})e^{Y_{i3}(\gamma_{2}Y_{i0}- X_{i32}'\beta)} \\
     &\phi_{\theta}^{0|0,1}(Y_{i3},Y_{i2},Y^{1}_{i-1},X_i)=(1-Y_{i1})+\left(1-e^{-\gamma_1Y_{i0}+\gamma_2(1-Y_{i-1})+X_{i31}'\beta}\right)Y_{i1}(1-Y_{i2})e^{Y_{i3}(\gamma_{1}-\gamma_{2}(1-Y_{i0})- X_{i32}'\beta)} \\
    &\phi_{\theta}^{1|1,1}(Y_{i3},Y_{i2},Y^{1}_{i-1},X_i)=e^{\gamma_1(1-Y_{i0})+\gamma_2(1-Y_{i-1})+X_{i31}'\beta}Y_{i1} \\
    &+\left(1-e^{\gamma_1(1-Y_{i0})+\gamma_2(1-Y_{i-1})+X_{i31}'\beta}\right)Y_{i1}Y_{i2}e^{(1-Y_{i3})(\gamma_{2}(1-Y_{i0})+X_{i32}'\beta)} \\
    &\phi_{\theta}^{1|1,0}(Y_{i3},Y_{i2},Y^{1}_{i-1},X_i)=Y_{i1}+\left(1-e^{-\gamma_1(1-Y_{i0})+\gamma_2Y_{i-1}-X_{i31}'\beta}\right)(1-Y_{i1})Y_{i2}e^{(1-Y_{i3})(\gamma_{1}-\gamma_{2}Y_{i0}+ X_{i32}'\beta)}
\end{align*}
Moreover, an application of Lemma \ref{lemma_theorem_ARp_transit} gives 
\begin{align*}
     &\phi_{\theta}^{0|0}(Y_{i2},Y_{i1},Y_{i-1}^{0},X_i)=(1-Y_{i1})e^{Y_{i2}(\gamma_{1}Y_{i0}-\gamma_{2}(Y_{i0}-Y_{i-1})-X_{i21}'\beta)} \\
    &\phi_{\theta}^{1|1}(Y_{i2},Y_{i1},Y^{0}_{i-1},X_i)=Y_{i1}e^{ (1-Y_{i2})\left(\gamma_{1}(1-Y_{i0})+ \gamma_{2}(Y_{i0}-Y_{i-1})+X_{i21}'\beta\right)}
\end{align*}
such that:
\begin{align*}
    \mathbb{E}\left[\phi_{\theta}^{0|0}(Y_{i2},Y_{i1},Y_{i-1}^{0},X_i)|Y_{i-1},Y_{i0},A_i\right]&=\pi_{1}^{0|0,Y_{i0}}(A_i,X_i)=\frac{1}{1+e^{\gamma_{2} Y_{i0}+X_{i2}'\beta+A_i}} \\
    \mathbb{E}\left[\phi_{\theta}^{1|1}(Y_{i2},Y_{i1},Y_{i-1}^{0},X_i)|Y_{i-1},Y_{i0},A_i\right]&=\pi_{1}^{1|1,Y_{i0}}(A_i,X_i)=\frac{e^{\gamma_{1}+\gamma_{2} Y_{i0}+X_{i2}'\beta+A_i}}{1+e^{\gamma_{1}+\gamma_{2} Y_{i0}++X_{i2}'\beta+A_i}}
\end{align*}

\noindent For $\pi_{2}^{0|0,0}(A_i,X_i)$ and $\pi_{1}^{0|0,Y_{i0}}(A_i,X_i)$ to match, we require  both $Y_{i0}=0$ and $X_{i3}=X_{i2}$ in which case:
\begin{align*}
    \phi_{\theta}^{0|0,0}(Y_{i1}^3,0,Y_{i-1},X_i)&=e^{\gamma_2Y_{i-1}-X_{i31}'\beta}(1-Y_{i1})+ \left(1-e^{\gamma_2Y_{i-1}-X_{i31}'\beta}\right)(1-Y_{i1})(1-Y_{i2}) \\
    \phi_{\theta}^{0|0}(Y_{i1}^2,0,Y_{i-1},X_i)&=(1-Y_{i1})e^{Y_{i2}(\gamma_{2}Y_{i-1}-X_{i31}'\beta)}  \\
    &=(1-Y_{i1})Y_{i2}e^{\gamma_{2}Y_{i-1}-X_{i31}'\beta}+(1-Y_{i1})(1-Y_{i2})
\end{align*}
Therefore,
\begin{align*}
    \psi_{\theta}^{0|0,0}(Y_{i1}^3,0,Y_{i-1},X_i)&=\phi_{\theta}^{0|0,0}(Y_{i1}^3,0,Y_{i-1},X_i)-\phi_{\theta}^{0|0}(Y_{i1}^2,0,Y_{i-1},X_i)=0
\end{align*}
So there is no information about the model parameters in this moment function. \\

\noindent For $\pi_{2}^{0|0,1}(A_i,X_i)$ and $\pi_{1}^{0|0,Y_{i0}}(A_i,X_i)$ to match, we require both $Y_{i0}=1$ and $X_{i3}=X_{i2}$ in which case:
\begin{align*}
    \phi_{\theta}^{0|0,1}(Y_{i1}^3,1,Y_{i-1},X_i)&=(1-Y_{i1})+\left(1-e^{-\gamma_1+\gamma_2(1-Y_{i-1})+X_{i31}'\beta}\right)Y_{i1}(1-Y_{i2})e^{\gamma_{1}Y_{i3}} \\
    \phi_{\theta}^{0|0}(Y_{i1}^2,1,Y_{i-1},X_i)&=(1-Y_{i1})e^{Y_{i2}(\gamma_{1}-\gamma_{2}(1-Y_{i-1})-X_{i31}'\beta)}
\end{align*}
Then, a valid moment condition that depends on all model parameters is:
\begin{align*}
     &\psi_{\theta}^{0|0,1}(Y_{i1}^3,1,Y_{i-1},X_i)=\phi_{\theta}^{0|0,1}(Y_{i1}^3,1,Y_{i-1},X_i)-  \phi_{\theta}^{0|0}(Y_{i1}^2,1,Y_{i-1},X_i) \\
     &=
     \left(1-e^{-\gamma_1+\gamma_2(1-Y_{i-1})+X_{i31}'\beta}\right)e^{\gamma_{1}}Y_{i1}(1-Y_{i2})Y_{i3}\\
     &+\left(1-e^{-\gamma_1+\gamma_2(1-Y_{i-1})+X_{i31}'\beta}\right)Y_{i1}(1-Y_{i2})(1-Y_{i3}) \\
     &-e^{\gamma_{1}-\gamma_{2}(1-Y_{i-1})-X_{i31}'\beta}(1-e^{-\gamma_{1}+\gamma_{2}(1-Y_{i-1})+X_{i31}'\beta})(1-Y_{i1})Y_{i2}
\end{align*}
Rescaling this moment function by the factor $\left(e^{\gamma_{1}-\gamma_{2}(1-Y_{i-1})-X_{i31}'\beta}(1-e^{-\gamma_{1}+\gamma_{2}(1-Y_{i-1})+X_{i31}'\beta})\right)^{-1}$, one obtains \begin{align*}
     \widetilde{\psi_{\theta}^{0|0,1}}(Y_{i1}^3,1,Y_{i-1},X_i)
     &=
     e^{\gamma_{2}(1-Y_{i-1})+X_{i31}'\beta}Y_{i1}(1-Y_{i2})Y_{i3}+e^{-\gamma_{1}+\gamma_{2}(1-Y_{i-1})+X_{i31}'\beta}Y_{i1}(1-Y_{i2})(1-Y_{i3})\\
     &-(1-Y_{i1})Y_{i2}
\end{align*}
Thus, for for the initial condition $Y_{i0}=1, Y_{i-1}=1$, we have
\begin{align*}
     \widetilde{\psi_{\theta}^{0|0,1}}(Y_{i1}^3,1,1,X_i)
     &=
     e^{X_{i31}'\beta}Y_{i1}(1-Y_{i2})Y_{i3}+e^{-\gamma_{1}+X_{i31}'\beta}Y_{i1}(1-Y_{i2})(1-Y_{i3})-(1-Y_{i1})Y_{i2}
\end{align*}
which only depends on $\gamma_1$ and $\beta$. In the notation of \cite{honore2020moment}, this coincides with their moment function $m_{(1,1)}$. Clearly, it is strictly decreasing in $\gamma_1$. Furthermore, this moment function is either increasing or decreasing in $\beta_k$ depending on the sign of $X_{i3k}-X_{i1k}$. \cite{honore2020moment} show that these monotonocity properties can be exploited to uniquely identifies $\gamma_1,\beta$. Instead, for the initial condition $Y_{i0}=1, Y_{i-1}=0$, we have 
\begin{align*}
     \widetilde{\psi_{\theta}^{0|0,1}}(Y_{i1}^3,1,0,X_i)
     &=
     e^{\gamma_{2}+X_{i31}'\beta}Y_{i1}(1-Y_{i2})Y_{i3}+e^{-\gamma_{1}+\gamma_{2}+X_{i31}'\beta}Y_{i1}(1-Y_{i2})(1-Y_{i3})-(1-Y_{i1})Y_{i2}
\end{align*}
which \cite{honore2020moment} denote as $m_{(1,0)}$. Provided that $\gamma_1,\beta$ are identified, the strict monotonicity of the moment functions in $\gamma_2$  ensure that $\gamma_2$ is identified. \\
\noindent Analogously, for $\pi_{2}^{1|1,0}(A_i,X_i)$ and $\pi_{1}^{0|0,Y_{i0}}(A_i)$ to match, we require both $Y_{i0}=0$ and $X_{i3}=X_{i2}$ in which case:
\begin{align*}
    \phi_{\theta}^{1|1,0}(Y_{i1}^3,0,Y_{i-1},X_i)&=Y_{i1}+\left(1-e^{-\gamma_1+\gamma_2Y_{i-1}-X_{i31}'\beta}\right)(1-Y_{i1})Y_{i2}e^{\gamma_{1}(1-Y_{i3})} \\
    \phi_{\theta}^{1|1}(Y_{i1}^2,0,Y_{i-1},X_i)&=Y_{i1}e^{ (1-Y_{i2})\left(\gamma_{1}- \gamma_{2}Y_{i-1}+X_{i31}'\beta\right)}
\end{align*}
Then, a valid moment function that depends on all model parameters is:
\begin{align*}
     \psi_{\theta}^{1|1,0}(Y_{i1}^3,0,Y_{i-1},X_i)&= \phi_{\theta}^{1|1,0}(Y_{i1}^3,0,Y_{i-1},X_i)- \phi_{\theta}^{1|1}(Y_{i1}^2,0,Y_{i-1},X_i) \\
     &=\left(1-e^{-\gamma_1+\gamma_2Y_{i-1}-X_{i31}'\beta}\right)e^{\gamma_{1}}(1-Y_{i1})Y_{i2}(1-Y_{i3}) \\
     &+\left(1-e^{-\gamma_1+\gamma_2Y_{i-1}-X_{i31}'\beta}\right)(1-Y_{i1})Y_{i2}Y_{i3} \\
     &-e^{\gamma_{1}- \gamma_{2}Y_{i-1}+X_{i31}'\beta}\left(1-e^{-\gamma_{1}+ \gamma_{2}Y_{i-1}-X_{i31}'\beta}\right)Y_{i1}(1-Y_{i2})
\end{align*}
Rescaling this moment function by the factor $\left(e^{\gamma_{1}- \gamma_{2}Y_{i-1}+X_{i31}'\beta}\left(1-e^{-\gamma_{1}+ \gamma_{2}Y_{i-1}-X_{i31}'\beta}\right)\right)^{-1}$, one obtains
\begin{align*}
     \widetilde{\psi_{\theta}^{1|1,0}}(Y_{i1}^3,0,Y_{i-1},X_i)     &=e^{\gamma_2Y_{i-1}-X_{i31}'\beta}(1-Y_{i1})Y_{i2}(1-Y_{i3}) +e^{-\gamma_1+\gamma_2Y_{i-1}-X_{i31}'\beta}(1-Y_{i1})Y_{i2}Y_{i3} -Y_{i1}(1-Y_{i2})
\end{align*}
 For the initial condition $Y_{i0}=0,Y_{i-1}=0$, we have 
\begin{align*}
     \widetilde{\psi_{\theta}^{1|1,0}}(Y_{i1}^3,0,0,X_i)     &=e^{-X_{i31}'\beta}(1-Y_{i1})Y_{i2}(1-Y_{i3}) +e^{-\gamma_1-X_{i31}'\beta}(1-Y_{i1})Y_{i2}Y_{i3} -Y_{i1}(1-Y_{i2})
\end{align*}
This moment function also only depends on $\gamma_1, \beta$ and coincides with the moment function $m_{(0,0)}$ in \cite{honore2020moment}. Similarly to $ \widetilde{\psi_{\theta}^{0|0,1}}(Y_{i1}^3,1,1,X_i)$, the monotonicity properties of $\widetilde{\psi_{\theta}^{1|1,0}}(Y_{i1}^3,0,0,X_i)$ can be exploited to uniquely identifies $\gamma_1,\beta$ (see \cite{honore2020moment}). Instead, for the initial condition $Y_{i0}=0,Y_{i-1}=1$, we obtain 
\begin{align*}
     \widetilde{\psi_{\theta}^{1|1,0}}(Y_{i1}^3,0,1,X_i)   &=e^{\gamma_2-X_{i31}'\beta}(1-Y_{i1})Y_{i2}(1-Y_{i3}) +e^{-\gamma_1+\gamma_2-X_{i31}'\beta}(1-Y_{i1})Y_{i2}Y_{i3} -Y_{i1}(1-Y_{i2})
\end{align*}
Provided that $\gamma_1,\beta$ is identified, the strict monotonicity of this moment function in $\gamma_2$ implies that it identifies $\gamma_2$ uniquely.  This is $m_{(0,1)}$ in \cite{honore2020moment}.\\
\indent Lastly, for $\pi_{2}^{1|1,1}(A_i)$ and $\pi_{1}^{1|1,Y_{i0}}(A_i)$ to match, we require both $Y_{i0}=1$ and $X_{i3}=X_{i2}$ in which case:
\begin{align*}
    \phi_{\theta}^{1|1,1}(Y_{i1}^3,1,Y_{i-1},X_i)&=e^{\gamma_2(1-Y_{i-1})+X_{i31}'\beta}Y_{i1}+ \left(1-e^{\gamma_2(1-Y_{i-1})+X_{i31}'\beta}\right)Y_{i1}Y_{i2} \\
    \phi_{\theta}^{1|1}(Y_{i1}^2,1,Y_{i-1},X_i)&=Y_{i1}e^{ (1-Y_{i2})\left( \gamma_{2}(1-Y_{i-1})+X_{i21}'\beta\right)} \\
    &=Y_{i1}(1-Y_{i2})e^{ \gamma_{2}(1-Y_{i-1})+X_{i21}'\beta} +Y_{i1}Y_{i2}
\end{align*}
Then, a valid moment function
\begin{align*}
    \psi_{\theta}^{1|1,1}(Y_{i1}^3,1,Y_{i-1},X_i)&= \phi_{\theta}^{1|1,1}(Y_{i1}^3,1,Y_{i-1},X_i)- \phi_{\theta}^{1|1}(Y_{i1}^2,1,Y_{i-1},X_i) \\
    &=0
\end{align*}
is identically zero and hence contains no information about the model parameters.

\subsection{Proof of Theorem \ref{theorem_AR2_identif}}
\label{proof_theorem_AR2_identif}
We recall from the discussion of Section \ref{Section_AR2_identif} that $T=4$ and $K_{x}\geq 2$ so that there are at least $2$ exogenous explanatory variables. We have $X_{it}=(W_{it},R_{it}')'\in \mathbb{R}^{K_{x}}$, $\beta=(\beta_W,\beta_R')'\in \mathbb{R}^{K_{x}}$ and $Z_{i}=(R_{i}',W_{i1},W_{i3},W_{i4})'\in \mathbb{R}^{4K_{x}-1}$ . Our goal is to prove Theorem \ref{theorem_AR2_identif} under Assumptions \ref{assumption_1} and \ref{assumption_2}. \\
\noindent Specializing Proposition \ref{proposition_3} to the AR(2) with $T=4$ yields the valid moment function:
\begin{align*}
     &\psi_{\theta}^{0|0,0}(Y_{i4},Y_{i3},Y^{2}_{i-1},X_i)=
    \left(e^{\gamma_2Y_{i0}-X_{i42}'\beta}-1\right)(1-Y_{i1})(1-Y_{i2})Y_{i3}\\
    &+\left[e^{\gamma_2Y_{i0}-X_{i42}'\beta}+\left(1-e^{\gamma_2Y_{i0}-X_{i42}'\beta}\right)e^{- X_{i43}'\beta}-1\right](1-Y_{i1})(1-Y_{i2})(1-Y_{i3})Y_{i4} \\
    &+e^{\gamma_1(1-Y_{i0})+\gamma_2(Y_{i0}-Y_{i-1})+X_{i21}'\beta}Y_{i1}(1-Y_{i2})Y_{i3}\\
    &+e^{-\gamma_1Y_{i0}-\gamma_{2}Y_{i-1}+X_{i41}'\beta}\left[e^{\gamma_1+\gamma_2Y_{i0}-X_{i42}'\beta}+\left(1-e^{\gamma_1+\gamma_2Y_{i0}-X_{i42}'\beta}\right)e^{\gamma_{2}- X_{i43}'\beta}\right]Y_{i1}(1-Y_{i2})(1-Y_{i3})Y_{i4} \\
    &+e^{-\gamma_1Y_{i0}-\gamma_{2}Y_{i-1}+X_{i41}'\beta}Y_{i1}(1-Y_{i2})(1-Y_{i3})(1-Y_{i4}) \\
    &-(1-Y_{i1})Y_{i2}
\end{align*}
 Define, the \say{limiting} moment function, where we have taken $W_{i2}$ to $+\infty$
\begin{align} \label{limit_valid_moment_function}
\begin{split}
     \psi_{\theta,\infty}^{0|0,0}(Y_{i4},Y_{i3},Y^{2}_{i-1},Z_i)&=-(1-Y_{i1})(1-Y_{i2})Y_{i3} \\
    &+\left[e^{X_{i34}'\beta}-1\right](1-Y_{i1})(1-Y_{i2})(1-Y_{i3})Y_{i4} \\
    &+e^{-\gamma_1Y_{i0}+\gamma_{2}(1-Y_{i-1})+X_{i31}'\beta}Y_{i1}(1-Y_{i2})(1-Y_{i3})Y_{i4} \\
    &+e^{-\gamma_1Y_{i0}-\gamma_{2}Y_{i-1}+X_{i41}'\beta}Y_{i1}(1-Y_{i2})(1-Y_{i3})(1-Y_{i4})
\end{split}
\end{align} 
For $s\in\{-,+\}^{K_x}$, consider the moment objective
\begin{align*}
    \Psi_{s,y^0}^{0|0,0}(\theta)&=\lim_{w_{2}\to \infty} \mathbb{E}\left[ \psi_{\theta}^{0|0,0}(Y_{i4},Y_{i3},Y^{2}_{i-1},X_i)|Y_{i}^0=y^0,X_i\in \mathcal{X}_s,W_{i2}=w_2 \right]
\end{align*}
We will show in two successive steps (a) and (b) that
\begin{align*}
    \Psi_{s,y^0}^{0|0,0}(\theta)&=\lim_{w_{2}\to \infty} \mathbb{E}\left[ \psi_{\theta,\infty}^{0|0,0}(Y_{i4},Y_{i3},Y^{2}_{i-1},Z_i)|Y_{i}^0=y^0,X_i\in \mathcal{X}_s,W_{i2}=w_2 \right] \quad  \text{(a)}\\
    &= \mathbb{E}\left[ \psi_{\theta,\infty}^{0|0,0}(Y_{i4},Y_{i3},Y^{2}_{i-1},Z_i)|Y_{i}^0=y^0,X_i\in \mathcal{X}_s,W_{i2}=\infty \right] \quad \text{(b)}
\end{align*}
To establish (a), we start by observing that the history sequence $(1-Y_{i1})Y_{i2}$ featuring in $\psi_{\theta}^{0|0,0}$ has expectation zero. To see this, note that by iterated expectations
\begin{align*}
    &\lim_{w_2\to\infty} \mathbb{E}\left[(1-Y_{i1})Y_{i2}|Y_{i}^0=y^0,X_i\in \mathcal{X}_s,W_{i2}=w_2\right] \\
    &=\lim_{w_2\to\infty} \int \frac{e^{\gamma_{02}y_0+x_2'\beta_0+a}}{1+e^{\gamma_{02}y_0+x_2'\beta_0+a}}\frac{1}{1+e^{\gamma_{01}y_0+\gamma_{02}y_{i-1}+x_1'\beta_0+a}} p(a,z|y_{0},\mathcal{X}_s,w_2)dadz
\end{align*}
Now, $p(a,z|y_{0},\mathcal{X}_s,w_2)=p(a|y_{0},z,w_2)p(z|y_{0},\mathcal{X}_s,w_2)=p(a|y_{0},z,w_2)\frac{p(z|y_{0},w_2)\mathds{1}\{X_i\in \mathcal{X}_s\}}{\int_{\mathcal{X}_s}p(z|y_{0},w_2)dz}$. Hence, by part (iii) of Assumption \ref{assumption_2}, an integrable dominating function of the integrand is
\begin{align*}
   \frac{e^{\gamma_{02}y_0+x_2'\beta_0+a}}{1+e^{\gamma_{02}y_0+x_2'\beta_0+A_i}}\frac{1}{1+e^{\gamma_{01}y_0+\gamma_{02}y_{i-1}+x_1'\beta_0+a}} p(a,z|y_{0},\mathcal{X}_s,w_2)\leq d_0(a)\frac{d_{2}(z)}{\int_{\mathcal{X}_s}d_{1}(z)dz}
\end{align*}
Moreover, by parts (ii)-(iii) of Assumption \ref{assumption_2} and the Dominated Convergence Theorem, 
\begin{align*}
    \lim_{w_2\to\infty} p(a,z|y_{0},\mathcal{X}_s,w_2)&=q(a|y_{0},z)\frac{q(z|y_{0})\mathds{1}\{X_i\in \mathcal{X}_s\}}{\int_{\mathcal{X}_s}q(z|y_{0})dz}\equiv q(a,z|y_{0},\mathcal{X}_s)
\end{align*}
Hence another application of the Dominated Convergence Theorem gives
\begin{align*}
    &\lim_{w_2\to\infty} \mathbb{E}\left[(1-Y_{i1})Y_{i2}|Y_{i}^0=y^0,X_i\in \mathcal{X}_s,W_{i2}=w_2\right] \\
    &= \int \lim_{w_2\to\infty} \frac{e^{\gamma_{02}y_0+x_2'\beta_0+a}}{1+e^{\gamma_{02}y_0+x_2'\beta_0+a}}\frac{1}{1+e^{\gamma_{01}y_0+\gamma_{02}y_{i-1}+x_1'\beta_0+a}} p(a,z|y_{0},\mathcal{X}_s,w_2)dadz \\
    &= \int 0\times q(a,z|y_{0},\mathcal{X}_s)dadz \\
    &=0
\end{align*}
where the third line follows from the fact that $\lim_{w_2\to\infty} e^{w_2\beta_W}=0$ by Assumption \ref{assumption_1}. Applying the same arguments to each remaining summand of $\psi_{\theta}^{0|0,0}$ and collecting terms delivers (a). To obtain (b), we note that by part (iv) of Assumption \ref{assumption_1}, $w_2\mapsto \mathbb{E}\left[ \psi_{\theta,\infty}^{0|0,0}(Y_{i4},Y_{i3},Y^{2}_{i-1},Z_i)|Y_{i}^0=y^0,X_i\in \mathcal{X}_s,W_{i2}=w_2\right]$ is continuous with a well defined limit at infinity in light of (a). As a result, we can work directly with its continuous extension at infinity. \\
\indent Let us focus on the initial condition $y_0=y_{-1}=0$. It is clear from Equation (\ref{limit_valid_moment_funcition}) that $\Psi_{s,0,0}^{0|0,0}(\theta)$ does not depend on $\gamma_1$. Furthermore,  by parts (i) of Assumption  \ref{assumption_2} we note that we have the following integrable dominating functions for the derivative:
\begin{align*}
    \abs{\pdv{\psi_{\theta,\infty}^{0|0,0}(Y_{i4},Y_{i3},Y^{2}_{i-1},Z_i)}{\gamma_2}}&=e^{\gamma_{2}+X_{i31}'\beta}Y_{i1}(1-Y_{i2})(1-Y_{i3})Y_{i4}\leq \sup_{g_2\in \mathbb{G}_2,b\in \mathbb{B}} e^{g_{2}+2 \max(\abs{\bar{x}},\abs{\underline{x}})\norm{b}_1}  
    \\
    \abs{\pdv{\psi_{\theta,\infty}^{0|0,0}(Y_{i4},Y_{i3},Y^{2}_{i-1},Z_i)}{\beta_k}}&=\bigg\vert X_{ik ,34}e^{X_{i34}'\beta}(1-Y_{i1})(1-Y_{i2})(1-Y_{i3})Y_{i4} \\
    &+X_{ik,31}e^{\gamma_2+X_{i31}'\beta}Y_{i1}(1-Y_{i2})(1-Y_{i3})Y_{i4} \\
    &+X_{ik,41}e^{\gamma_2+X_{i31}'\beta}Y_{i1}(1-Y_{i2})(1-Y_{i3})(1-Y_{i4})\bigg\vert \\
    &  \leq \abs{X_{ik ,34}}e^{X_{i34}'\beta}+\abs{X_{ik,31}}e^{\gamma_2+X_{i31}'\beta} 
    +\abs{X_{ik,41}}e^{\gamma_2+X_{i31}'\beta} \\
    &\leq 2\max(\abs{\bar{x}},\abs{\underline{x}}) \sup_{b\in \mathbb{B}}e^{2 \max(\abs{\bar{x}},\abs{\underline{x}})\norm{b}_1}(1+2\sup_{g_2\in \mathbb{G}_2} e^{g_2})
\end{align*}
Hence, by Leibniz integral rule, we get
\begin{align*}
    &\pdv{\Psi_{s,0,0}^{0|0,0}(\theta)}{\gamma_2}\\
    &=\mathbb{E}\left[ \pdv{\psi_{\theta,\infty}^{0|0,0}(Y_{i4},Y_{i3},Y^{2}_{i-1},Z_i)}{\gamma_2}|Y_{i}^0=(0,0),X_i\in \mathcal{X}_s,W_{i2}=\infty \right] \\
    &=\mathbb{E}\left[e^{\gamma_{2}+X_{i31}'\beta}Y_{i1}(1-Y_{i2})(1-Y_{i3})Y_{i4}|Y_{i}^0=(0,0),X_i\in \mathcal{X}_s,W_{i2}=\infty \right] \\
    &=\mathbb{E}\left[e^{\gamma_{2}+X_{i31}'\beta} \underbrace{\mathbb{E}\left[Y_{i1}(1-Y_{i2})(1-Y_{i3})Y_{i4}|Y_{i}^0=(0,0),Z_i,W_{i2}=\infty,A_i\right]}_{>0}|Y_{i}^0=(0,0),X_i\in \mathcal{X}_s,W_{i2}=\infty \right] \\
    &>0
\end{align*}
Similarly, 
\begin{align*}
    &\pdv{\Psi_{s,0,0}^{0|0,0}(\theta)}{\beta_k}\\
    &=\mathbb{E}\left[ \pdv{\psi_{\theta,-\infty}^{0|0,0}(Y_{i4},Y_{i3},Y^{2}_{i-1},Z_i)}{\beta_k}|Y_{i}^0=(0,0),X_i\in \mathcal{X}_s,W_{i2}=\infty \right] \\
    &=\mathbb{E}\left[X_{ik ,34}e^{X_{i34}'\beta}\times \right. \\
    &\left. \underbrace{\mathbb{E}\left[(1-Y_{i1})(1-Y_{i2})(1-Y_{i3})Y_{i4}|Y_{i}^0=(0,0),Z_i,W_{i2}=\infty,A_i\right]}_{>0}|Y_{i}^0=(0,0),X_i\in \mathcal{X}_s,W_{i2}=\infty \right] \\
    &+\mathbb{E}\left[X_{ik,31}e^{\gamma_2+X_{i31}'\beta}\times \right.\\
    &\left. \underbrace{\mathbb{E}\left[Y_{i1}(1-Y_{i2})(1-Y_{i3})Y_{i4}|Y_{i}^0=(0,0),Z_i,W_{i2}=\infty,A_i\right]}_{>0}|Y_{i}^0=(0,0),X_i\in \mathcal{X}_s,W_{i2}=\infty \right] \\
    &+\mathbb{E}\left[X_{ik,41}e^{\gamma_2+X_{i31}'\beta}\times \right. \\
    &\left. \underbrace{\mathbb{E}\left[Y_{i1}(1-Y_{i2})(1-Y_{i3})(1-Y_{i4})|Y_{i}^0=(0,0),Z_i,W_{i2}=\infty,A_i\right]}_{>0}|Y_{i}^0=(0,0),X_i\in \mathcal{X}_s,W_{i2}=\infty \right]
\end{align*}
The last display shows that $\pdv{\Psi_{s,0,0}^{0|0,0}(\theta)}{\beta_k}>0$ if $s_k=+$ and $\pdv{\Psi_{s,0,0}^{0|0,0}(\theta)}{\beta_k}<0$ if $s_k=-$. Therefore, appealing to Lemma 2 in \cite{honore2020moment}, we conclude that the $2^{K_x}$ system of equations in $K_{x}+1$ unkowns given by:
\begin{align*}
    \Psi_{s,0,0}^{0|0,0}(\theta)=0, \quad \forall s\in \{-,+\}^{K_x}
\end{align*}
has at most one solution. It is precisely $(\gamma_{02},\beta_{0})$, since the validity of $\psi_{\theta}^{0|0,0}(Y_{i4},Y_{i3},Y^{2}_{i-1},X_i)$ for arbitrary $X_i$ directly implies the validity of the limiting moment $\psi_{\theta,\infty}^{0|0,0}(Y_{i4},Y_{i3},Y^{2}_{i-1},Z_i)$ at \say{$W_{i2}=\infty$}. Then, notice that for any other initial condition $y^0\in\{(0,1),(1,0),(1,1)\}$, the objective 
$\Psi_{s,y^0}^{0|0,0}(\theta)$ is strictly monotonic in $\gamma_1$. Hence, given $(\gamma_{02},\beta_{0})$, it point identifies $\gamma_{01}$. This concludes the proof of Theorem \ref{theorem_AR2_identif}.

\section{Proof of Proposition \ref{proposition_4}}

\noindent We recall that by definition,
\begin{align*}
&\Pi_{t}^{k_{1}^s|l_1^{p}}(y^0,x_{1}^{t+s})=\\
&\mathbb{E}\left[P(Y_{it+s}=k_s,\ldots, Y_{it+1}=k_1\,|\,Y_{it}=l_1,\ldots,Y_{it-(p-1)}=l_p,X_{i1}^{t+s}=x_{1}^{t+s},A_i)\,|\,Y_{i}^0=y^0,X_{i1}^{t+s}=x_{1}^{t+s}\right] 
\end{align*}
We have 
\begin{align*}
    P(Y_{it+s}=k_s,\ldots, Y_{it+1}=k_1\,|\,Y_{it}=l_1,\ldots,Y_{it-(p-1)}=l_p,X_{i1}^{t+s}=x_{1}^{t+s},A_i)=\frac{N^{k_{1}^s|l_1^{p}}(e^{a})}{D^{k_{1}^s|l_1^{p}}(e^{a})}
\end{align*}
where $N^{k_{1}^s|l_1^{p}}(e^{a}),D^{k_{1}^s|l_1^{p}}(e^{a})$ are polynomials in $e^{a}$.
There are two cases to consider. \\

\noindent \underline{\textbf{Case 1:}} $s< p$ \\
Then,
\begin{align*}
    N^{k_{1}^s|l_1^{p}}(e^{a})&=e^{k_1\left(\sum_{r=1}^{p}\gamma_{0r}l_{r}+x_{t+1}'\beta_0+a\right)}\prod_{j=1}^{s-1}e^{k_{j+1}\left(\sum_{r=1}^{j}\gamma_{0r}k_{j+1-r}+\sum_{r=j+1}^{p}\gamma_{0r}l_{r-j}+x_{t+1+j}'\beta_0+a\right)} \\
    D^{k_{1}^s|l_1^{p}}(e^{a})&=\left(1+e^{\sum_{r=1}^{p}\gamma_{0r}l_{r}+x_{t+1}'\beta_0+a}\right)\prod_{j=1}^{s-1}\left(1+e^{\sum_{r=1}^{j}\gamma_{0r}k_{j+1-r}+\sum_{r=j+1}^{p}\gamma_{0r}l_{r-j}+x_{t+1+j}'\beta_0+a}\right)
\end{align*}
We note that $\deg( N^{k_{1}^s|l_1^{p}}(e^{a}))\leq \deg(D^{k_{1}^s|l_1^{p}}(e^{a}))$ with strict inequality unless $k_{1}^s=1_{s}$. Furthermore, since by assumption for any   $t\in\{p,\ldots,T-2\}$, $s\in\{1,\ldots,T-1-t\}$ and $y,\tilde{y}\in \mathcal{Y}^{p}$, $\gamma_0'y+x_{t}'\beta_0\neq \gamma_0'\tilde{y}+x_{t+s}'\beta_0$, $D^{k_{1}^s|l_1^{p}}(e^{a})$ is a product of distinct irreducible polynomials in $e^{a}$. Consequently, standard results on \textit{partial fraction decompositions} entail that there exists a unique set of known coefficients $(\mu,\lambda_0,\lambda_1,\ldots,\lambda_{s-1})\in \mathbb{R}^{s+1}$ such that:

\begin{align*}
    \frac{N^{k_{1}^s|l_1^{p}}(e^{a})}{D^{k_{1}^s|l_1^{p}}(e^{a})}=\mu+\lambda_{0}\frac{1}{\left(1+e^{\sum_{r=1}^{p}\gamma_{0r}l_{r}+x_{t+1}'\beta_0+a}\right)}+\sum_{j=1}^{s-1}\lambda_j\frac{1}{1+e^{\sum_{r=1}^{j}\gamma_{0r}k_{j+1-r}+\sum_{r=j+1}^{p}\gamma_{0r}l_{r-j}+x_{t+1+j}'\beta_0+a}}
\end{align*}
with $\mu=0$ unless $k_{1}^s=1_{s}$. We can rewrite this in terms of transition probabilities as:
\begin{align*}
    \frac{N^{k_{1}^s|l_1^{p}}(e^{a})}{D^{k_{1}^s|l_1^{p}}(e^{a})}&=\mu+\lambda_{0}\pi_{t}^{0|l_{1}^p}(a,x_{t+1})+\sum_{j=1}^{s-1}\lambda_j\pi_{t+j}^{0|k_j,\ldots,k_1,l_{1}^{p-j}}(a,x_{t+1+j}) \\
    &=\mu+\lambda_{0}(1-l_1)\pi_{t}^{l_1|l_{1}^p}(a,x_{t+1})+\lambda_{0}l_1(1-\pi_{t}^{l_1|l_{1}^p}(a,x_{t+1}))+ \\
    &\sum_{j=1}^{s-1}\lambda_j(1-k_j)\pi_{t+j}^{k_j|k_j,\ldots,k_1,l_{1}^{p-j}}(a,x_{t+1+j})+\sum_{j=1}^{s-1}\lambda_jk_j(1-\pi_{t+j}^{k_j|k_j,\ldots,k_1,l_{1}^{p-j}}(a,x_{t+1+j}))
\end{align*}
This last result in conjunction with Theorem \ref{theorem_ARp_transit}, implies that:
\begin{align*}
    \Pi_{t}^{k_{1}^s|l_1^{p}}(y^0,x_{1}^{t+s})&=\mu\\
    &+\mathbb{E}\left[\lambda_{0}(1-l_1)\phi_{\theta_0}^{l_1|l_{1}^{p}}(Y^{t+1}_{it-(2p-1)},x_{1}^{t+s})+\lambda_{0}l_1\left(1-\phi_{\theta_0}^{l_1|l_{1}^{p}}(Y^{t+1}_{it-(2p-1)},x_{1}^{t+s})\right) \right. \\
    & +\left. \sum_{j=1}^{s-1}\lambda_j(1-k_j)\phi_{\theta_0}^{k_j|k_j,\ldots,k_1,l_{1}^{p-j}}(Y^{t+j+1}_{it+j-(2p-1)},x_{1}^{t+s}) \right. \\
    &\left. +\sum_{j=1}^{s-1}\lambda_jk_j\left(1-\phi_{\theta_0}^{k_j|k_j,\ldots,k_1,l_{1}^{p-j}}(Y^{t+j+1}_{it+j-(2p-1)},x_{1}^{t+s})\right) \,|\,Y_{i}^0=y^0,X_{i1}^{t+s}=x_{1}^{t+s}\right]
\end{align*}
which shows that $\Pi_{t}^{k_{1}^s|l_1^{p}}(y^0,x_{1}^{t+s})$ is identified given that $\theta_0$ is identified by assumption. \\

\noindent \underline{\textbf{Case 2:}} $s\geq p$ \\
Then,
\begin{align*}
    D^{k_{1}^s|l_1^{p}}(e^{a})&=\left(1+e^{\sum_{r=1}^{p}\gamma_{0r}l_{r}+x_{t+1}'\beta_0+a}\right)\prod_{j=1}^{p-1}\left(1+e^{\sum_{r=1}^{j}\gamma_{0r}k_{j+1-r}+\sum_{r=j+1}^{p}\gamma_{0r}l_{r-j}+x_{t+1+j}'\beta_0+a}\right)\\
    &\times \prod_{j=p}^{s-1}\left(1+e^{\sum_{r=1}^{p}\gamma_{0r}k_{j+1-r}+x_{t+1+j}'\beta_0+a}\right)
\end{align*}

\begin{align*}
    N^{k_{1}^s|l_1^{p}}(e^{a})&=e^{k_1\left(\sum_{r=1}^{p}\gamma_{0r}l_{r}+x_{t+1}'\beta_0+a\right)}\prod_{j=1}^{p-1}e^{k_{j+1}\left(\sum_{r=1}^{j}\gamma_{0r}k_{j+1-r}+\sum_{r=j+1}^{p}\gamma_{0r}l_{r-j}+x_{t+1+j}'\beta_0+a\right)}\\
    &\times \prod_{j=p}^{s-1}e^{k_{j+1}\left(\sum_{r=1}^{p}\gamma_{0r}k_{j+1-r}+x_{t+1+j}'\beta_0+a\right)}
\end{align*}
Invoking identical arguments as in the case $s<p$, there exists a unique set of known  coefficients $(\mu,\lambda_0,\lambda_1,\ldots,\lambda_{s-1})\in \mathbb{R}^{s+1}$ such that:
\begin{align*}
    \Pi_{t}^{k_{1}^s|l_1^{p}}(y^0,x_{1}^{t+s})&=\mu\\
    &+\mathbb{E}\left[\lambda_{0}(1-l_1)\phi_{\theta_0}^{l_1|l_{1}^{p}}(Y^{t+1}_{it-(2p-1)},x_{1}^{t+s})+\lambda_{0}l_1\left(1-\phi_{\theta_0}^{l_1|l_{1}^{p}}(Y^{t+1}_{it-(2p-1)},x_{1}^{t+s})\right) \right. \\
    & +\left. \sum_{j=1}^{p-1}\lambda_j(1-k_j)\phi_{\theta_0}^{k_j|k_j,\ldots,k_1,l_{1}^{p-j}}(Y^{t+j+1}_{it+j-(2p-1)},x_{1}^{t+s}) \right. \\
    &\left. +\sum_{j=1}^{p-1}\lambda_jk_j\left(1-\phi_{\theta_0}^{k_j|k_j,\ldots,k_1,l_{1}^{p-j}}(Y^{t+j+1}_{it+j-(2p-1)},x_{1}^{t+s})\right) \right. \\
     & \left. + \sum_{j=p}^{s-1}\lambda_j(1-k_j)\phi_{\theta_0}^{k_j|k_j,\ldots,k_{j+1-p}}(Y^{t+j+1}_{it+j-(2p-1)},x_{1}^{t+s})  \right. \\
    & +\left. \sum_{j=p}^{s-1}\lambda_jk_j\left(1-\phi_{\theta_0}^{k_j|k_j,\ldots,k_{j+1-p}}(Y^{t+j+1}_{it+j-(2p-1)},x_{1}^{t+s})\right)\,|\,Y_{i}^0=y^0,X_{i1}^{t+s}=x_{1}^{t+s}\right]
\end{align*}
which again shows that $\Pi_{t}^{k_{1}^s|l_1^{p}}(y^0,x_{1}^{t+s})$ is identified given that $\theta_0$ is identified by assumption. This concludes the proof.

\section{Proof of Lemma \ref{lemma_5}} \label{proof_lemma_5}
Let
  \begin{align*}
     \phi_{\theta}^{k|k}(Y_{it+1},Y_{it},Y_{it-1},X_i)&=\mathds{1}\{Y_{it}=k\}e^{\sum_{m=1}^M (Y_{m,it+1}-k_m)\left(\sum_{j=1}^M \gamma_{mj}(Y_{j,it-1}-k_j)-\Delta X_{m,it+1}'\beta_m\right)}
 \end{align*}
 We verify the claim by direct calculation. 
 {\allowdisplaybreaks
 \begin{align*}
     &\mathbb{E}\left[\phi_{\theta}^{k|k}(Y_{it+1},Y_{it},Y_{it-1},X_i)|Y_{i0},Y_{i1}^{t-1},X_i,A_i\right]=P(Y_{it}=k|Y_{i0},Y_{i1}^{t-1},X_i,A_i) \\
     &\times\sum_{l\in \mathcal{Y}}P(Y_{it+1}=l|Y_{i0},Y_{i1}^{t-1},Y_{it}=k,X_i,A_i)\phi_{\theta}^{k|k}(l,k,Y_{it-1},X_i) \\
     &=\prod_{m=1}^M \frac{e^{k_m(\sum_{j=1}^{M} \gamma_{mj}Y_{j,it-1} +X_{m,it}'\beta_m+A_{m,i})}}{1+e^{\sum_{j=1}^{M} \gamma_{mj}Y_{j,it-1} +X_{m,it}'\beta_m+A_{m,i}}}  \\
     &\times\sum_{l\in \mathcal{Y}}\prod_{m=1}^M \frac{e^{l_m(\sum_{j=1}^{M} \gamma_{mj}k_{j} +X_{m,it+1}'\beta_m+A_{m,i})}}{1+e^{\sum_{j=1}^{M} \gamma_{mj}k_{j} +X_{m,it+1}'\beta_m+A_{m,i}}}e^{\sum_{m=1}^M (l_{m}-k_m)\left(\sum_{j=1}^M \gamma_{mj}(Y_{j,it-1}-k_j)-\Delta X_{m,it+1}'\beta_m\right)} \\
     &=\sum_{l\in \mathcal{Y}}\prod_{m=1}^M\frac{e^{l_m(\sum_{j=1}^{M} \gamma_{mj}Y_{j,it-1} +X_{m,it}'\beta_m+A_{m,i})}}{1+e^{\sum_{j=1}^{M} \gamma_{mj}k_{j} +X_{m,it+1}'\beta_m+A_{m,i}}}\frac{e^{k_m(\sum_{j=1}^{M} \gamma_{mj}k_{j} +X_{m,it+1}'\beta_m+A_{m,i})}}{1+e^{\sum_{j=1}^{M} \gamma_{mj}Y_{j,it-1} +X_{m,it}'\beta_m+A_{m,i}}} \\
     &=\prod_{m=1}^M \frac{e^{k_m(\sum_{j=1}^{M} \gamma_{mj}k_{j} +X_{m,it+1}'\beta_m+A_{m,i})}}{1+e^{\sum_{j=1}^{M} \gamma_{mj}Y_{j,it-1} +X_{m,it}'\beta_m+A_{m,i}}} \frac{1}{1+e^{\sum_{j=1}^{M} \gamma_{mj}k_{j} +X_{m,it+1}'\beta_m+A_{m,i}}}\sum_{l\in \mathcal{Y}} \prod_{m=1}^M e^{l_m(\sum_{j=1}^{M} \gamma_{mj}Y_{j,it-1} +X_{m,it}'\beta_m+A_{m,i})}
 \end{align*}
 }
Now, noting that
\begin{align*}
    \sum_{l\in \mathcal{Y}} \prod_{m=1}^M e^{l_m(\sum_{j=1}^{M} \gamma_{mj}Y_{j,it-1} +X_{m,it}'\beta_m+A_{m,i})}= \prod_{m=1}^M (1+e^{\sum_{j=1}^{M} \gamma_{mj}Y_{j,it-1} +X_{m,it}'\beta_m+A_{m,i}})
\end{align*}
we finally get
 \begin{align*}
     &\mathbb{E}\left[\phi_{\theta}^{k|k}(Y_{it+1},Y_{it},Y_{it-1},X_i)|Y_{i0},Y_{i1}^{t-1},X_i,A_i\right] \\
     &=\prod_{m=1}^M \frac{e^{k_m(\sum_{j=1}^{M} \gamma_{mj}k_{j} +X_{m,it+1}'\beta_m+A_{m,i})}}{1+e^{\sum_{j=1}^{M} \gamma_{mj}Y_{j,it-1} +X_{m,it}'\beta_m+A_{m,i}}} \frac{1}{1+e^{\sum_{j=1}^{M} \gamma_{mj}k_{j} +X_{m,it+1}'\beta_m+A_{m,i}}}\prod_{m=1}^M (1+e^{\sum_{j=1}^{M} \gamma_{mj}Y_{j,it-1} +X_{m,it}'\beta_m+A_{m,i}}) \\
     &=\prod_{m=1}^M \frac{e^{k_m(\sum_{j=1}^{M} \gamma_{mj}k_{j} +X_{m,it+1}'\beta_m+A_{m,i})}}{1+e^{\sum_{j=1}^{M} \gamma_{mj}k_{j} +X_{m,it+1}'\beta_m+A_{m,i}}} \\
     &=\pi^{k|k}_{t}(A_i,X_i)
 \end{align*}
which concludes the proof.
 
\section{Proof of Lemma \ref{lemma_6}}

By definition, for $T\geq 3$, and for $t,s$ such that $T-1\geq t> s\geq 1$:
 {\allowdisplaybreaks
 {\allowdisplaybreaks
\begin{align*}
    &\mathbb{E}\left[\zeta_{\theta}^{k|k}(Y_{it-1}^{t+1},Y_{is-1}^s,X_i)|Y_{i0},Y_{i1}^{s-1},X_i,A_i\right]=P(Y_{is}=k|Y_{i0},Y_{i1}^{s-1},X_i,A_i)+ \\
    &\cleansum_{l\in \mathcal{Y}\setminus \{k\}} \omega_{t,s,l}^{k|k}(\theta) \mathbb{E}\left[\mathds{1}\{Y_{is}=l\} \phi_{\theta}^{k|k}(Y_{it-1}^{t+1},X_i)|Y_{i0},Y_{i1}^{s-1},X_i,A_i\right] \\
    &=\prod_{m=1}^M \frac{e^{k_m(\mu_{m,s}(\theta)+A_{m,i})}}{1+e^{\mu_{m,s}(\theta)+A_{m,i}}}+\cleansum_{l\in \mathcal{Y}\setminus \{k\}} \omega_{t,s,l}^{k|k}(\theta) \pi^{k|k}_{t}(A_i,X_{i}) P(Y_{is}=l |Y_{i0},Y_{i1}^{s-1},X_i,A_i) \\
     &=\prod_{m=1}^M \frac{e^{k_m(\mu_{m,s}(\theta)+A_{m,i})}}{1+e^{\mu_{m,s}(\theta)+A_{m,i}}}+\cleansum_{l\in \mathcal{Y}\setminus \{k\}} \left[1-e^{\sum_{j=1}^M (l_j-k_j)\left[\kappa_{j,t}^{ k|k}(\theta)-\mu_{j,s}(\theta)\right]}\right]\prod_{m=1}^M\frac{e^{k_m(\kappa_{m,t}^{ k|k}(\theta)+A_{m,i})}}{1+e^{\kappa_{m,t}^{ k|k}(\theta)+A_{m,i}}}\frac{e^{l_m(\mu_{m,s}(\theta)+A_{m,i})}}{1+e^{\mu_{m,s}(\theta)+A_{m,i}}} \\
     &=\prod_{m=1}^M\frac{e^{k_m(\kappa_{m,t}^{ k|k}(\theta)+A_{m,i})}}{1+e^{\kappa_{m,t}^{ k|k}(\theta)+A_{m,i}}} \\
     &=\pi_{t}^{k|k}(A_i,X_i)
\end{align*}
}
The first line follows from the measurability of the weight $\omega_{t,s,l}^{k|k}(\theta)$ with respect to the conditioning set and the linearity of conditional expectations. The second line uses the definition of $\mu_{j,s}(\theta)$ and follows from the law of iterated expectations and Lemma \ref{lemma_6}. The third line makes use of the definition of $\kappa_{m,t}^{ k|k}(\theta)$ and $\omega_{t,s,l}^{k|k}(\theta)$ and the penultime line uses Appendix Lemma \ref{tech_lemma_2}. 

\section{Dynamic network formation with transitivity} \label{network_model}

\cite{graham2013comment} studies a variant of model (\ref{VAR1_logit_general}) to describe network formation amongst groups of 3 individuals. This is a panel data setting where a large sample of many such groups and the evolution of their social ties are observed over $T=3$ periods (4 counting the initial condition). Interactions are assumed undirected and modelled at the dyad level as:
\begin{align} \label{network_specification}
\begin{split}
   D_{ijt}&=\mathds{1}\left\{\gamma_0 D_{ijt-1}+ \delta_0 R_{ijt-1}+A_{ij}-\epsilon_{ijt}\geq 0\right\} \quad t=1,\ldots,T \\
   R_{ijt-1}&=D_{ikt-1}D_{jkt-1}
\end{split}
\end{align}
where $i,j,k$ denote the 3 different agents and $D_{ijt}\in\{0,1\}$ encodes the presence or absence of a link between agent $i$ and agent $j$ at time $t$. The network $D_0\in \{0,1\}^3$ forms the initial condition. The parameter $\gamma_0$ captures state dependence while $\delta_0$ captures transitivity in relationships, i.e the effect of sharing friends in common on the propensity to establish friendships. Finally, $A_{ij}$ is an unrestricted dyad level fixed effect that could potentiall capture unobserved homophily and $\epsilon_{ijt}$ is a standard logistic shock, iid over time and individuals. While \cite{graham2013comment} establishes identification of $(\gamma_0,\delta_0)$ for $T=3$ via a conditional likelihood approach in the spirit of \cite{chamberlain_1985}, one limitation of the model is the absence of other covariates, in particular time-specific effects. Controlling for such effects can be essential to adequately capture important variation in social dynamics: think about the persistent impact of Covid-19 on all types of social interactions. A relevant extension is thus:
\begin{align}\label{network_specification_withX}
\begin{split}
   D_{ijt}&=\mathds{1}\left\{\gamma_0 D_{ijt-1}+ \delta_0 D_{ikt-1}D_{jkt-1}+X_{ijt}'\beta_0+A_{ij}-\epsilon_{ijt}\geq 0\right\} \quad t=1,\ldots,T \\
   R_{ijt-1}&=D_{ikt-1}D_{jkt-1}
\end{split}
\end{align}
Letting $\mathbb{D}=\{0,1\}^{3}$ denote the support of the network $D_t=(D_{ijt},D_{ikt},D_{jkt})$, it is straightforward to see that the results developed for the VAR(1) case can be repurposed to suit model (\ref{network_specification_withX}) . For $T=3$, an adaptation of Lemma \ref{lemma_5} yields 8 possible transition functions given by:
\begin{align*}
    \phi_{\theta}^{d|d}(D_{3},D_{2},D_{1},X)=\mathds{1}\{D_{2}=d\} \exp\left(\cleansum_{i<j} (D_{ij3}-d_{ij2})[\gamma (D_{ij1}-d_{ij2})- \Delta R_{ij1}\delta -\Delta X_{ij2}'\beta]\right), \quad d \in \mathbb{D}
\end{align*}
An adaptation of Lemma \ref{lemma_6} implies that we can construct another $8$ transition functions given by
\begin{align*}
     \zeta_{\theta}^{d|d}(D_3,D_2,D_1,D_0,X)=\mathds{1}\{D_{1}=d\}+ \cleansum_{d'\in \mathbb{D}\setminus \{d\}} \omega_{2,1,d'}^{d|d}(\theta) \mathds{1}\{D_{1}=l\} \phi_{\theta}^{d|d}(D_3,D_2,D_2,X), \quad d \in \mathbb{D}
\end{align*}
where 
\begin{align*}
    &\mu_{ij,1}(\theta)= \gamma D_{ij0}+\delta R_{ij0}+X_{ij1}'\beta \\
    &\kappa_{ij,2}^{d|d}(\theta)=\gamma d_{ij}+\delta r_{ij}+X_{ij3}'\beta \\
    &\omega_{2,1,d'}^{d|d}(\theta)=1-e^{\sum_{i<j} (d_{ij}'-d_{ij})\left[\kappa_{ij,2}^{d|d}(\theta)-\mu_{ij,1}(\theta)\right]}
\end{align*}
Therefore, for $T=3$, 8 moment functions that all meaningfully depend on the model parameter are:
\begin{align*}
    \psi_{\theta}^{d|d}(D_3,D_2,D_1,D_0,X)=\phi_{\theta}^{d|d}(D_{3},D_{2},D_{1},X)-\zeta_{\theta}^{d|d}(D_3,D_2,D_1,D_0,X), \quad d \in \mathbb{D}
\end{align*}
Their validity, in the sense of verifying equation (\ref{moment_function_final}), follows from the law of iterated expectations. 

\section{Proof of Lemma \ref{lemma_7}}
Let
\begin{align*}
    \phi_{\theta}^{k|k}(Y_{it-1}^{t+1},X_i)&=\mathds{1}\{Y_{it}=k\} e^{\sum_{c\in \mathcal{Y}\setminus\{k\}} \mathds{1}\{Y_{it+1}=c\} \left(\sum_{j\in \mathcal{Y}} (\gamma_{cj}-\gamma_{kj})\mathds{1}(Y_{it-1}=j)+\gamma_{kk}-\gamma_{ck}+\Delta X_{ikt+1}'\beta_k-\Delta X_{ict+1}'\beta_c\right)}
\end{align*}
We proceed to verify that the claim by direct computation. We have:
\begin{align*}
	&\mathbb{E}\left[\phi_{\theta}^{k|k}(Y_{it+1},Y_{it},Y_{it-1},X_i)|Y_{i0},Y_{i1}^{t-1},X_i\right]=P(Y_{it}=k|Y_i^0,Y_{i1}^{t-1},X_i,A_i)\times \\
     &\sum_{l\in \mathcal{Y}}P(Y_{it+1}=l|Y_i^0,Y_{i1}^{t-1},Y_{it}=k,X_i,A_i)\phi_{\theta}^{k|k}(l,k,Y_{it-1},X_i) \\
     &=\frac{e^{\sum_{c=0}^C \gamma_{kc}\mathds{1}(Y_{it-1}=c)+X_{ikt}'\beta_k+A_{ik}}}{\sum \limits_{j=0}^C e^{\sum_{c=0}^C \gamma_{jc}\mathds{1}(Y_{it-1}=c)+X_{ijt}'\beta_j+A_{ij}}} \times \\
     &\sum_{l \in \mathcal{Y}} \frac{e^{\gamma_{lk}+X_{ilt+1}'\beta_l+A_{il}}}{\sum \limits_{j=0}^C e^{\gamma_{jk}+X_{ijt+1}'\beta_j+A_{ij}}}\phi_{\theta}^{k|k}(l,k,Y_{it-1},X_i)  \\
      &=\frac{e^{\sum_{c=0}^C \gamma_{kc}\mathds{1}(Y_{it-1}=c)+X_{ikt}'\beta_k+A_{ik}}}{\sum \limits_{j=0}^C e^{\sum_{c=0}^C \gamma_{jc}\mathds{1}(Y_{it-1}=c)+X_{ijt}'\beta_j+A_{ij}}} \times \\
     &\left(\frac{e^{\gamma_{kk}+X_{ikt+1}'\beta_k+A_{ik}}}{\sum \limits_{j=0}^C e^{\gamma_{jk}+X_{ijt+1}'\beta_j+A_{ij}}} + \sum_{l \in \mathcal{Y}\setminus\{k\}} \frac{e^{\gamma_{lk}+X_{ilt+1}'\beta_l+A_{il}}}{\sum \limits_{j=0}^C e^{\gamma_{jk}+X_{ijt+1}'\beta_j+A_{ij}}}
     e^{\left(\sum_{j=0}^C (\gamma_{lj}-\gamma_{kj})\mathds{1}(Y_{it-1}=j)+\gamma_{kk}-\gamma_{lk}+\Delta X_{ikt+1}'\beta_k-\Delta X_{ilt+1}'\beta_l\right)}\right) \\
     &=\frac{e^{\sum_{c=0}^C \gamma_{kc}\mathds{1}(Y_{it-1}=c)+X_{ikt}'\beta_k+A_{ik}}}{\sum \limits_{j=0}^C e^{\sum_{c=0}^C \gamma_{jc}\mathds{1}(Y_{it-1}=c)+X_{ijt}'\beta_j+A_{ij}}} \times\frac{e^{\gamma_{kk}+X_{ikt+1}'\beta_k+A_{ik}}}{\sum \limits_{j=0}^C e^{\gamma_{jk}+X_{ijt+1}'\beta_j+A_{ij}}}  \\
     &+\frac{e^{\gamma_{kk}+X_{ikt+1}'\beta_k+A_{ik}}}{\sum \limits_{j=0}^C e^{\sum_{c=0}^C \gamma_{jc}\mathds{1}(Y_{it-1}=c)+X_{ijt}'\beta_j+A_{ij}}} \times\sum_{l \in \mathcal{Y}\setminus\{k\}} \frac{1}{\sum \limits_{j=0}^C e^{\gamma_{jk}+X_{ijt+1}'\beta_j+A_{ij}}}
     e^{\sum_{j=0}^C \gamma_{lj}\mathds{1}(Y_{it-1}=j)+ X_{ilt}'\beta_l+A_{il}} \\
     &=\frac{e^{\gamma_{kk}+X_{ikt+1}'\beta_k+A_{ik}}}{\sum \limits_{j=0}^C e^{\sum_{c=0}^C \gamma_{jc}\mathds{1}(Y_{it-1}=c)+X_{ijt}'\beta_j+A_{ij}}}\frac{1}{\sum \limits_{j=0}^C e^{\gamma_{jk}+X_{ijt+1}'\beta_j+A_{ij}}} \sum_{l \in \mathcal{Y}} 
     e^{\sum_{j=0}^C \gamma_{lj}\mathds{1}(Y_{it-1}=j)+X_{ilt}'\beta_l+A_{il}} \\
     &=\frac{e^{\gamma_{kk}+X_{ikt+1}'\beta_k+A_{ik}}}{\sum \limits_{j=0}^C e^{\gamma_{jk}+X_{ijt+1}'\beta_j+A_{ij}}} \\
     &=\pi_{t}^{k|k}(A_i,X_i)
\end{align*}
which concludes the proof.

\section{Proof of Lemma \ref{lemma_8}}
By construction for $T\geq 3$, and $t,s$ such that $T-1\geq t> s\geq 1$,
\begin{align*}
    &\mathbb{E}\left[\zeta_{\theta_0}^{0|0}(Y_{it-1}^{t+1},Y_{is-1}^s,X_i)|Y_{i0},Y_{i1}^{s-1},X_i,A_i\right] \\
    &=P(Y_{is}=0|Y_{i0},Y_{i1}^{s-1},X_i,A_i)\\
    &+\cleansum_{l\in \mathcal{Y}\setminus \{0\}} \omega_{t,s,l}^{0|0}(\theta)\mathbb{E}\left[\mathds{1}\{Y_{is}=l\} \mathbb{E}\left[\phi_{\theta}^{0|0}(Y_{it-1}^{t+1},X_i)|Y_{i0},Y_{i1}^{t-1},X_i,A_i\right]|Y_{i0},Y_{i1}^{s-1},X_i,A_i\right] \\
    &=\frac{1}{1+\sum_{c=1}^C e^{\mu_{c,s}(\theta)+A_{ic}}}+\cleansum_{l=1}^C \omega_{t,s,l}^{0|0}(\theta)\mathbb{E}\left[\mathds{1}\{Y_{is}=l\} |Y_{i0},Y_{i1}^{s-1},X_i,A_i\right]\pi_{t}^{0|0}(A_i,X_i) \\
    &=\frac{1}{1+\sum_{c=1}^C e^{\mu_{c,s}(\theta)+A_{ic}}}+\cleansum_{l=1}^C \left(1-e^{(\kappa_{l,t}^{0|0}(\theta)-\mu_{l,s}(\theta))}\right)\frac{e^{\mu_{l,s}(\theta)+A_{il}}}{1+\sum_{c=1}^C e^{\mu_{c,s}(\theta)+A_{ic}}}\frac{1}{1+\sum_{c=1}^C e^{\kappa_{c,t}^{0|0}(\theta)+A_{ic}}} \\
    &=\frac{1}{1+\sum_{c=1}^C e^{\kappa_{c,t}^{0|0}(\theta)+A_{ic}}} \\
    &=\pi_{t}^{0|0}(A_i,X_i)
\end{align*}

The first line follows from the measurability of the weight $\omega_{t,s,l}^{0|0}(\theta)$ with respect to the conditioning set and the linearity of conditional expectations. The second line uses the definition of $\mu_{c,s}(\theta)$ and follows from the law of iterated expectations and Lemma \ref{lemma_7}. The third line makes use of the definition of $\kappa_{c,t}^{0|0}(\theta)$, $\omega_{t,s,l}^{0|0}(\theta)$ and the normalization $\gamma_{c0}=\gamma_{0c}=0, A_{0c}=0$ for all $c\in \mathcal{Y}$. The penultime line uses Appendix Lemma \ref{tech_lemma_1}. 

\noindent Likewise, for all $k\in \mathcal{Y}\setminus\{0\}$,
\begin{align*}
    &\mathbb{E}\left[\zeta_{\theta_0}^{k|k}(Y_{it-1}^{t+1},Y_{is-1}^s,X_i)|Y_{i0},Y_{i1}^{s-1},X_i,A_i\right] \\
    &=P(Y_{is}=k|Y_{i0},Y_{i1}^{s-1},X_i,A_i)\\
    &+\cleansum_{l\in \mathcal{Y}\setminus \{k\}} \omega_{t,s,l}^{k|k}(\theta)\mathbb{E}\left[\mathds{1}\{Y_{is}=l\} \mathbb{E}\left[\phi_{\theta}^{k|k}(Y_{it-1}^{t+1},X_i)|Y_{i0},Y_{i1}^{t-1},X_i,A_i\right]|Y_{i0},Y_{i1}^{s-1},X_i,A_i\right] \\
    &=\frac{e^{\mu_{k,s}(\theta)+A_{ik}}}{1+\sum_{c=1}^C e^{\mu_{c,s}(\theta)+A_{ic}}}+\cleansum_{l\in \mathcal{Y}\setminus \{k\}} \omega_{t,s,l}^{k|k}(\theta)\mathbb{E}\left[\mathds{1}\{Y_{is}=l\} |Y_{i0},Y_{i1}^{s-1},X_i,A_i\right]\pi_{t}^{k|k}(A_i,X_i) \\
    &=\frac{e^{\mu_{k,s}(\theta)+A_{ik}}}{1+\sum_{c=1}^C e^{\mu_{c,s}(\theta)+A_{ic}}}\\
    &+\cleansum_{l\in \mathcal{Y}\setminus \{k\}} \left(1-e^{(\kappa_{l,t}^{k|k}(\theta)-\mu_{l,s}(\theta))-( \kappa_{k,t}^{k|k}(\theta)-\mu_{k,s}(\theta))}\right)\frac{e^{\mu_{l,s}(\theta)+A_{il}}}{1+\sum_{c=1}^C e^{\mu_{c,s}(\theta)+A_{ic}}}\frac{e^{\kappa_{k,t}^{k|k}(\theta)+A_{ik}}}{1+\sum_{c=1}^C e^{\kappa_{c,t}^{k|k}(\theta)+A_{ic}}} \\
    &=\frac{e^{\mu_{k,s}(\theta)+A_{ik}}}{1+\sum_{c=1}^C e^{\mu_{c,s}(\theta)+A_{ic}}}+\left(1-e^{-\kappa_{k,t}^{k|k}(\theta)+\mu_{k,s}(\theta)}\right)\frac{1}{1+\sum_{c=1}^C e^{\mu_{c,s}(\theta)+A_{ic}}}\frac{e^{\kappa_{k,t}^{k|k}(\theta)+A_{ik}}}{1+\sum_{c=1}^C e^{\kappa_{c,t}^{k|k}(\theta)+A_{ic}}} \\
    &+\cleansum_{\substack{l=1 \\ l\neq k } }^C \left(1-e^{(\kappa_{l,t}^{k|k}(\theta)-\mu_{l,s}(\theta))-( \kappa_{k,t}^{k|k}(\theta)-\mu_{k,s}(\theta))}\right)\frac{e^{\mu_{l,s}(\theta)+A_{il}}}{1+\sum_{c=1}^C e^{\mu_{c,s}(\theta)+A_{ic}}}\frac{e^{\kappa_{k,t}^{k|k}(\theta)+A_{ik}}}{1+\sum_{c=1}^C e^{\kappa_{c,t}^{k|k}(\theta)+A_{ic}}} \\
    &=\frac{e^{\kappa_{k,t}^{k|k}(\theta)+A_{ik}}}{1+\sum_{c=1}^C e^{\kappa_{c,t}^{k|k}(\theta)+A_{ic}}} \\
    &=\pi_{t}^{k|k}(A_i,X_i)
\end{align*}

The first line follows from the measurability of the weight $\omega_{t,s,l}^{k|k}(\theta)$ with respect to the conditioning set and the linearity of conditional expectations. The second line uses the definition of $\mu_{k,s}(\theta)$ and follows from the law of iterated expectations and Lemma \ref{lemma_7}. The third line makes use of the definition of $\kappa_{c,t}^{ k|k}(\theta)$ and $\omega_{t,s,l}^{k|k}(\theta)$. The fourth line uses the fact that $\kappa_{0,t}^{k|k}(\theta)=\mu_{0,s}(\theta)=0$ due to the normalization $\gamma_{c0}=\gamma_{0c}=0, A_{0c}=0$ for all $c\in \mathcal{Y}$. The penultime line uses Appendix Lemma \ref{tech_lemma_1}. 
\section{Proof of Theorem \ref{thm_effbound}}
In what follows, we will drop the cross-sectional subscript $i$ to economize on space. To avoid excessive repetition, we will detail the argument for the initial condition $Y_0=0$. A set of completely symmetric arguments will deliver the result for $Y_0=1$ and can be provided upon request. For conciseness, we will further omit the conditioning on the initial condition $Y_0=0$ in conditional expectations. \\

\noindent \underline{\textbf{A) Preliminary calculations}} \\
The conditional density of history $(Y_1,Y_2,Y_3)$ of the AR(1) model given initial condition $Y_0$, regressors $X$ and fixed effect $A$ is $f(Y_{1},Y_{2},Y_{3}|Y_0,X,A;\theta)=\prod \limits_{t=1}^3\frac{e^{Y_{t}(\gamma Y_{t-1}+X_{t}'\beta+A)}}{\left(1+e^{\gamma Y_{t-1}+X_{t}'\beta+A}\right)}$. This implies 
\begin{align*}
    \ln f(Y_{1},Y_{2},Y_{3}|Y_0,X,A;\theta)&=
    \sum_{t=1}^3 Y_{t}(\gamma Y_{t-1}+X_{t}'\beta+A)-\sum_{t=1}^3Y_{t-1} \ln \left(1+e^{\gamma+X_{t}'\beta+A}\right) \\
    &-\sum_{t=1}^3(1-Y_{t-1}) \ln \left(1+e^{X_{t}'\beta+A}\right)
\end{align*}
and hence 
\begin{align*}
    \pdv{ \ln f(Y_{1},Y_{2},Y_{3}|Y_0,X,A;\theta)}{\gamma}&=\sum_{t=1}^3 Y_{t}\left(Y_{t-1}-\frac{e^{\gamma+X_{t}'\beta+A}}{1+e^{\gamma+X_{t}'\beta+A}} \right)\\
    \pdv{ \ln f(Y_{1},Y_{2},Y_{3}|Y_0,X,A;\theta)}{\beta}
    &=\sum_{t=1}^3X_{t}\left( Y_{t}-Y_{t-1}\frac{e^{\gamma+X_{t}'\beta+A}}{1+e^{\gamma+X_{t}'\beta+A}}-(1-Y_{t-1})\frac{e^{X_{t}'\beta+A}}{1+e^{X_{t}'\beta+A}}\right)
\end{align*}

Our candidate for the efficient score is the efficient moment based on the conditional moment restriction: $\mathbb{E}\left[\psi_{\theta}(Y_{1}^{3},Y_{0}^1,X)|Y_{0}=0,X\right]=0$. By \cite{chamberlain1987asymptotic}, it is given by,
\begin{align*}
    \psi_{\theta}^{eff}(Y_{1}^{3},X)=-\Omega(X)\psi_{\theta}(Y_{1}^{3},Y_{0}^1,X) 
\end{align*}
where $\Omega(X)=D(X)'\Sigma(X)^{-1}$ (recall that we are omitting the dependence on the initial condition $Y_0=0$ here).  The following expressions for $D(X),\Sigma(X),\Omega(X)$ are useful for the derivations ahead:
\begin{align*}
    D_{11}(X)&= e^{X_{21}'\beta+\gamma}P_{101}(X) \\
    D_{21}(X)&=-e^{X_{13}'\beta-\gamma}P_{011}(X) \\
    D_{1j}(X)&=X_{23,j-1}e^{X_{23}'\beta}P_{001}(X) +X_{21,j-1}e^{X_{21}'\beta+\gamma}P_{101}(X) +X_{31,j-1}e^{X_{31}'\beta}P_{100}(X), \quad j=2,\ldots,K+1 \\
    D_{2j}(X)&=X_{32,j-1}e^{X_{32}'\beta}P_{110}(X) +X_{12,j-1}e^{X_{12}'\beta}P_{010}(X) +X_{13,j-1}e^{X_{13}'\beta-\gamma}P_{011}(X), \quad j=2,\ldots,K+1 \\
    \Sigma_{11}(X)&=(e^{X_{23}'\beta }-1)^2P_{001}(X)+e^{2X_{21}'\beta+2\gamma}P_{101}(X)+e^{2X_{31}'\beta}P_{100}(X)+P_{01}(X) \\
    \Sigma_{22}(X)&=(e^{X_{32}'\beta}-1)^2P_{110}(X) 
    +e^{2X_{12}'\beta}P_{010}(X)+e^{2X_{13}'\beta-2\gamma}P_{011}(X)+P_{10}(X) \\
    \Sigma_{12}(X)&=\Sigma_{21}(X)=-\left(e^{X_{21}'\beta+\gamma}P_{101}(X)+e^{X_{31}'\beta}P_{100}(X)+e^{X_{12}'\beta}P_{010}(X)+e^{X_{13}'\beta-\gamma}P_{011}(X)\right) \\
    det\left(\Sigma(X)\right)&= \Sigma_{11}(X) \Sigma_{22}(X)- \Sigma_{12}(X)^2 \\
    \Omega_{j1}(X)&=\frac{1}{ det\left(\Sigma(X)\right)}\left( D_{1j}(X)\Sigma_{22}(X)-D_{2j}(X)\Sigma_{12}(X)\right), \quad j=1,\ldots,K+1 \\
    \Omega_{j2}(X)&=\frac{1}{ det\left(\Sigma(X)\right)}\left(- D_{1j}(X)\Sigma_{12}(X)+D_{2j}(X)\Sigma_{11}(X)\right), \quad j=1,\ldots,K+1
\end{align*}
were I use the shorthand $P_{y_1\ldots y_n}(X)=P(Y_1=y_1,\ldots,Y_n=y_n|Y_0=0,X)$ \\

\noindent \underline{\textbf{B) Scores and nonparametric tangent set}} \\
With $T=3$, the conditional likelihood of history ($Y_1,Y_2,Y_3$) given $X=x,Y_0=y_0$  writes:
\begin{align*}
\mathcal{L}(\theta)&= \int f(Y_{1},Y_{2},Y_{3}|y_0,x,a;\theta) \pi(a|y_0,x) da 
\end{align*}
where $\pi(.|y_0,x)$ denotes the conditional density of $A$ given $X=x,Y_0=y_0$. Consider a scalar parametric submodel for the heterogeneity distribution $\pi(.|y_0,x;\eta)$ such that $\pi(.|y_0,x)=\pi(.|y_0,x;\eta_0)$. Then, the conditional likelihood of the parametric submodel is
\begin{align*}
\mathcal{L}(\theta,\eta)&=\int f(Y_{1},Y_{2},Y_{3}|y_0,x,a;\theta) \pi(a|y_0,x;\eta) da
\end{align*}
 Define 
\begin{align*}
    C_{y_1y_2y_3}(x_t)&=\mathbb{E}\left[\frac{e^{\gamma+x_{t}'\beta+A}}{1+e^{\gamma+x_{t}'\beta+A}}|Y_1=y_1,Y_2=y_2,Y_3=y_3,X=x\right] \\
    B_{y_1y_2y_3}(x_t)&=\mathbb{E}\left[\frac{e^{x_{t}'\beta+A}}{1+e^{x_{t}'\beta+A}}|Y_1=y_1,Y_2=y_2,Y_3=y_3,X=x\right]
\end{align*}
Careful bookkeeping yield the following scores for $\gamma$ and $\beta$
\begin{align}
\begin{split}
         S_{\gamma}&=\pdv{\ln \mathcal{L}(\theta,\eta)}{\gamma} =\mathbb{E}\left[\pdv{ \ln f(Y_{1},Y_{2},Y_{3}|Y_0,X,A;\theta)}{\gamma}|Y_1,Y_2,Y_3,X=x\right]\\
         &=\left(1-C_{111}(x_2)+1-C_{111}(x_3)\right)Y_1Y_2Y_3+(1-C_{110}(x_2)-C_{110}(x_3))Y_1Y_2(1-Y_3) \\
     &-C_{101}(x_2)Y_1(1-Y_2)Y_3-C_{100}(x_2)Y_1(1-Y_2)(1-Y_3) \\
     &+(1-C_{011}(x_3))(1-Y_1)Y_2Y_3-C_{010}(x_3)(1-Y_1)Y_2(1-Y_3)
     \end{split}
\end{align}
and 
\begin{align*}
    S_{\beta}&=\pdv{\ln \mathcal{L}(\theta,\eta)}{\beta} =\mathbb{E}\left[\pdv{ \ln f(Y_{1},Y_{2},Y_{3}|Y_0,X,A;\theta)}{\beta}|Y_1,Y_2,Y_3,X=x\right]\\
    &=\left(x_1(1-B_{111}(x_1))+x_2(1-C_{111}(x_2))+x_3(1-C_{111}(x_3))\right)Y_1Y_2Y_3 \\
    &+\left(x_1(1-B_{110}(x_1))+x_2(1-C_{110}(x_2))-x_3C_{110}(x_3)\right)Y_1Y_2(1-Y_3) \\
    &+\left(x_1(1-B_{101}(x_1))-x_2C_{101}(x_2)+x_3(1-B_{101}(x_3))\right)Y_1(1-Y_2)Y_3 \\
    &+\left(x_1(1-B_{100}(x_1))-x_2C_{100}(x_2)-x_3B_{100}(x_3)\right)Y_1(1-Y_2)(1-Y_3) \\
    &+\left(-x_1B_{011}(x_1)+x_2(1-B_{011}(x_2))+x_3(1-C_{011}(x_3))\right)(1-Y_1)Y_2Y_3 \\
    &+\left(-x_1B_{010}(x_1)+x_2(1-B_{010}(x_2))-x_3C_{010}(x_3)\right)(1-Y_1)Y_2(1-Y_3) \\
    &+\left(-x_1B_{001}(x_1)-x_2B_{001}(x_2)+x_3(1-B_{001}(x_3))\right)(1-Y_1)(1-Y_2)Y_3 \\
    &+\left(-x_1B_{000}(x_1)-x_2B_{000}(x_2)-x_3B_{000}(x_3)\right)(1-Y_1)(1-Y_2)(1-Y_3)
\end{align*}
The score for the nuisance parameter  is
\begin{align*}
    S_\eta = \pdv{\ln \mathcal{L}(\theta,\eta_0)}{\eta}
    =\mathbb{E}\left[\pdv{\ln \pi(A|y_0,x;\eta_0)}{\eta}|Y_1,Y_2,Y_3,X=x\right]
\end{align*}
Following \cite{hahn2001information}, this implies that the \textit{nonparametric tangent set} is given by
\begin{align*}
    \mathcal{T}=\left\{\mathbb{E}[K(A,x)|Y_1,Y_2,Y_3,x] \text{ such that } \mathbb{E}[K(A,x)|x]=0  \right\}
\end{align*}
 To prove that $\psi_{\theta}^{eff}$ is semiparametrically efficient, we will verify the conditions for an application of Theorem 3.2 in \cite{newey1990semiparametric}. Noting that $\mathcal{L}(\theta,\eta)$ is differentiable in $\theta$, that $ \mathcal{T}$ is linear, and that by Assumption \ref{assumption_effbound}, $\mathbb{E}\left[\psi_{\theta}^{eff}(Y_{1}^{3},X)\psi_{\theta}^{eff}(Y_{1}^{3},X)'\right]=\mathbb{E}\left[D(X)\Sigma(X)^{-1}D(X)'\right]$ is non singular, all that remains to check are: 
i) $\psi_{\theta}^{eff}(Y_{1}^{3},X) \in \mathcal{T}^{\perp}$ and ii) $S_{\theta}-\psi_{\theta}^{eff}(Y_{1}^{3},X) \in  \mathcal{T}$. \\

\noindent \underline{\textbf{C) Verification of condition i) $\mathbf{\psi_{\theta}^{eff}(Y_{1}^{3},X) \in \mathcal{T}^{\perp}}$}} \\
To verify condition i), let us characterize the orthocomplement of $\mathcal{T}$ which will also be useful to verify condition ii). By definition, any $g(Y_1,Y_2,Y_3,x)\in \mathcal{T}^{\perp}$ is such that for any element of $\mathcal{T}$, $\mathbb{E}[K(A,x)|Y_1,Y_2,Y_3,x]$, we have 
\begin{align*}
   0&= \mathbb{E}\left[g(Y_1,Y_2,Y_3,x)\mathbb{E}[K(A,x)|Y_1,Y_2,Y_3,x]|x\right] \\
   &=\int K(a,x)\mathbb{E}\left[g(Y_1,Y_2,Y_3,x)|x,a\right]\pi(a|x)da
\end{align*}
because this equality must be valid for any $K(a,x)$ verifying $\mathbb{E}[K(A,x)|x]=0$, it must be the case that $\mathbb{V}\left(\mathbb{E}\left[g(Y_1,Y_2,Y_3,x)|x,A\right]|x\right)=0$ or equivalently that $\mathbb{E}\left[g(Y_1,Y_2,Y_3,x)|x,A\right]=\mathbb{E}\left[g(Y_1,Y_2,Y_3,x)|x\right]$.  Conversely, this short calculation makes it clear that any $g$ function such that $E\left[g(Y_1,Y_2,Y_3,x)|x,A\right]$ is constant will be an element of $\mathcal{T}^{\perp}$. We conclude that,
\begin{align*}
    \mathcal{T}^{\perp}&=\{g(Y_1,Y_2,Y_3,x)\,|\,\mathbb{E}\left[g(Y_1,Y_2,Y_3,x)-\mathbb{E}\left[g(Y_1,Y_2,Y_3,x)|x\right]|x,A\right]=0\} =\mathbb{R}+\mathcal{T}_{*}^{\perp} \\
    \mathcal{T}_{*}^{\perp}&=\{g_{*}(Y_1,Y_2,Y_3,x)\,|\,\mathbb{E}\left[g_{*}(Y_1,Y_2,Y_3,x)|x,A\right]=0\}
\end{align*}
At this stage, an important observation is that $ \mathcal{T}_{*}^{\perp}$ coincides with the set of valid moment functions in the AR(1) model with $T=3$. By Theorem \ref{theorem_nummoments_AR1}, this is a 2-dimensional space when $T=3$ with basis elements $\psi_{\theta}^{0|0}(Y_{i1}^{3},Y_{i0}^1,X_i),\psi_{\theta}^{1|1}(Y_{i1}^{3},Y_{i0}^1,X_i)$. As a result, we further conclude that  $ \mathcal{T}_{*}^{\perp}=\spn\left(\{\psi_{\theta}^{0|0}(Y_{i1}^{3},Y_{i0}^1,X_i),\psi_{\theta}^{1|1}(Y_{i1}^{3},Y_{i0}^1,X_i)\}\right)$. Hence, 
$\psi_{\theta}^{eff}(Y_{1}^{3},X)\in \mathcal{T}_{*}^{\perp}$ since it is a linear combination of $\psi_{\theta}^{0|0}(Y_{i1}^{3},Y_{i0}^1,X_i)$ and $\psi_{\theta}^{1|1}(Y_{i1}^{3},Y_{i0}^1,X_i)$. Finally since $\mathcal{T}_{*}^{\perp}\subset \mathcal{T}^{\perp}$, $\psi_{\theta}^{eff}(Y_{1}^{3},X)\in \mathcal{T}^{\perp}$. \\

\noindent \underline{\textbf{D) Verification of condition ii) $S_{\theta}-\psi_{\theta}^{eff}(Y_{1}^{3},x) \in  \mathcal{T}$}} \\
To check condition ii) $S_{\theta}-\psi_{\theta}^{eff}(Y_{1}^{3},x) \in  \mathcal{T}$, we will verify the equivalent condition that for any element $g\in \mathcal{T}^{\perp}$, $\mathbb{E}\left[\left(S_{\theta}-\psi_{\theta}^{eff}(Y_{1}^{3},x)\right)g(Y_1,Y_2,Y_3,x)|x\right]=0$. Given our characterization of $\mathcal{T}^{\perp}$, it is equivalent to verify that $\forall k \in \{0,1\}$, $\mathbb{E}\left[\left(S_{\theta}-\psi_{\theta}^{eff}(Y_{1}^{3},x)\right)\psi_{\theta}^{k|k}(Y_{1}^{3},Y_{0}^1,x)|x\right]=0$ \\

\noindent \underline{\textbf{D)1) $S_{\gamma}-\psi_{\gamma}^{eff}(Y_{1}^{3},x)\perp \psi_{\theta}^{0|0}(Y_{1}^{3},Y_{0}^1,x)$}} \\
Let $\Delta_{\gamma}^{0|0}=(S_{\gamma} -\psi_{\gamma}^{eff}(Y_{1}^{3},Y_{i0}^1,x)) \psi_{\theta}^{0|0}(Y_{1}^{3},Y_{0}^1,x)$. It is tedious but straightforward to show that
\begin{align*}
\Delta_{\gamma}^{0|0}&=\Delta_{\gamma,1}^{0|0}+\Delta_{\gamma,2}^{0|0}+\Delta_{\gamma,3}^{0|0}+\Delta_{\gamma,4}^{0|0}+ \Delta_{\gamma,5}^{0|0} \\
\Delta_{\gamma,1}^{0|0}&=(1-C_{101}(x_2))e^{x_{21}'\beta+\gamma}Y_1(1-Y_2)Y_3-C_{100}(x_2)e^{x_{31}'\beta}Y_1(1-Y_2)(1-Y_3) \\
\Delta_{\gamma,2}^{0|0}&=-(1-C_{011}(x_3))(1-Y_1)Y_2Y_3+C_{010}(x_3)(1-Y_1)Y_2(1-Y_3) \\
\Delta_{\gamma,3}^{0|0}&= \Omega_{11}(x)(e^{x_{23}'\beta}-1)^2(1-Y_{1})(1-Y_{2})Y_{3}+\Omega_{11}(x)e^{2x_{21}'\beta+2\gamma}Y_{1}(1-Y_{2})Y_{3}\\
    &+\Omega_{11}(x)e^{2X_{31}'\beta}Y_{1}(1-Y_{2})(1-Y_{3})+\Omega_{11}(x)(1-Y_{1})Y_{2} \\
 \Delta_{\gamma,4}^{0|0}&= -\Omega_{12}(x)e^{x_{21}'\beta+\gamma}Y_1(1-Y_2)Y_3-\Omega_{12}(x)e^{x_{31}'\beta}Y_1(1-Y_2)(1-Y_3)\\
   &-\Omega_{12}(x)e^{x_{12}'\beta}(1-Y_1)Y_2(1-Y_3)-\Omega_{12}(x)e^{x_{13}'\beta-\gamma}(1-Y_1)Y_2Y_3 \\
   \Delta_{\gamma,5}^{0|0}&=-e^{X_{21}'\beta+\gamma}Y_1(1-Y_2)Y_3
\end{align*}
We then note that
\begin{align*}
    \mathbb{E}\left[\Delta_{\gamma,1}^{0|0}|x\right]&= \int \frac{1}{1+e^{\gamma+x_2'\beta+a}}\frac{e^{x_1'\beta+a}}{1+e^{x_1'\beta+a}}\frac{1}{1+e^{\gamma+x_2'\beta+a}}\frac{e^{x_3'\beta+a}}{1+e^{x_3'\beta+a}}e^{x_{21}'\beta+\gamma}\pi(a|x)da \\
    &-\int \frac{e^{\gamma+x_2'\beta+a}}{1+e^{\gamma+x_2'\beta+a}}\frac{e^{x_1'\beta+a}}{1+e^{x_1'\beta+a}}\frac{1}{1+e^{\gamma+x_2'\beta+a}}\frac{1}{1+e^{x_3'\beta+a}}e^{x_{31}'\beta}\pi(a|x)da \\
    &=\int \frac{e^{\gamma+x_2'\beta+a}}{1+e^{\gamma+x_2'\beta+a}}\frac{1}{1+e^{x_1'\beta+a}}\frac{1}{1+e^{\gamma+x_2'\beta+a}}\frac{e^{x_3'\beta+a}}{1+e^{x_3'\beta+a}}\pi(a|x)da \\
    &-\int \frac{e^{\gamma+x_2'\beta+a}}{1+e^{\gamma+x_2'\beta+a}}\frac{1}{1+e^{x_1'\beta+a}}\frac{1}{1+e^{\gamma+x_2'\beta+a}}\frac{e^{x_3'\beta+a}}{1+e^{x_3'\beta+a}}\pi(a|x)da \\
    &=0
\end{align*}
and by a similar calculation $ \mathbb{E}\left[\Delta_{\gamma,2}^{0|0}|x\right]=0$. Next, we immediately have
\begin{align*}
    \mathbb{E}\left[\Delta_{\gamma,3}^{0|0}|x\right]
    &=\Omega_{11}(x)\Sigma_{11}(x) \\
     \mathbb{E}\left[\Delta_{\gamma,4}^{0|0}|x\right]&=\Omega_{12}(x)\Sigma_{12}(x) \\
    \mathbb{E}\left[\Delta_{\gamma,5}^{0|0}|x\right]&=-e^{X_{21}'\beta+\gamma}P_{101}(x)
\end{align*}
 and hence,
\begin{align*}
    \Delta_{\gamma}^{0|0}=\Omega_{11}(x)\Sigma_{11}(X)+\Omega_{12}(x)\Sigma_{12}(x)-e^{x_{21}'\beta+\gamma}P_{101}(x)=D_{11}(x)-D_{11}(x)=0
\end{align*}

\noindent \underline{\textbf{D)2) $S_{\gamma}-\psi_{\gamma}^{eff}(Y_{1}^{3},x)\perp \psi_{\theta}^{1|1}(Y_{1}^{3},Y_{0}^1,x)$}} \\
Let $\Delta_{\gamma}^{1|1}=(S_{\gamma} -\psi_{\gamma}^{eff}(Y_{1}^{3},Y_{i0}^1,x)) \psi_{\theta}^{1|1}(Y_{1}^{3},Y_{0}^1,x)$. It can be decomposed as follows
\begin{align*}
    \Delta_{\gamma}^{1|1}&= \Delta_{\gamma,1}^{1|1}+\Delta_{\gamma,2}^{1|1}+\Delta_{\gamma,3}^{1|1}+\Delta_{\gamma,4}^{1|1}+\Delta_{\gamma,5}^{1|1} \\
    \Delta_{\gamma,1}^{1|1}&=-(e^{X_{32}'\beta}-1)C_{1,1,0}(x_2)Y_1Y_2(1-Y_3)+C_{1,0,1}(x_2)Y_1(1-Y_2)Y_3 +C_{1,0,0}(x_2)Y_1(1-Y_2)(1-Y_3)  \\
    \Delta_{\gamma,2}^{1|1}&=+(e^{X_{32}'\beta}-1)(1-C_{1,1,0}(x_3))Y_1Y_2(1-Y_3)-e^{x_{12}'\beta}C_{0,1,0}(x_3)(1-Y_1)Y_2(1-Y_3)\\
    &-e^{x_{13}'\beta-\gamma}C_{0,1,1}(x_3)(1-Y_1)Y_2Y_3 \\
    \Delta_{\gamma,3}^{1|1}&= -\Omega_{11}(x)e^{x_{21}'\beta+\gamma}Y_1(1-Y_2)Y_3-\Omega_{11}(x)e^{x_{31}'\beta}Y_1(1-Y_2)(1-Y_3)\\
    &-\Omega_{11}(x)e^{X_{12}'\beta}(1-Y_1)Y_2(1-Y_3)-\Omega_{11}(x)e^{x_{13}'\beta-\gamma}(1-Y_1)Y_2Y_3 \\
    \Delta_{\gamma,4}^{1|1}&=+\Omega_{12}(x)(e^{X_{32}'\beta}-1)^2Y_1Y_2(1-Y_3)+\Omega_{12}(x)e^{2x_{12}'\beta}(1-Y_1)Y_2(1-Y_3)\\
    &+\Omega_{12}(x)e^{2x_{13}'\beta-2\gamma}(1-Y_1)Y_2Y_3+\Omega_{12}(x)Y_1(1-Y_2) \\
     \Delta_{\gamma,5}^{1|1}&=e^{x_{13}'\beta-\gamma}(1-Y_1)Y_2Y_3
\end{align*}
First, we have
\begin{align*}
    \mathbb{E}\left[\Delta_{\gamma,1}^{1|1}|x\right]&=-\int \frac{e^{\gamma+x_2'\beta+a}}{1+e^{\gamma+x_2'\beta+a}}\frac{e^{x_1'\beta+a}}{1+e^{x_1'\beta+a}}\frac{e^{\gamma+x_2'\beta+a}}{1+e^{\gamma+x_2'\beta+a}} \frac{1}{1+e^{\gamma+x_3'\beta+a}}(e^{x_{32}'\beta}-1)\pi(a|x)da \\
    &+\int \frac{e^{\gamma+x_2'\beta+a}}{1+e^{\gamma+x_2'\beta+a}}\frac{e^{x_1'\beta+a}}{1+e^{x_1'\beta+a}}\frac{1}{1+e^{\gamma+x_2'\beta+a}} \pi(a|x)da \\
     &=-\int \frac{1}{1+e^{\gamma+x_2'\beta+a}}\frac{e^{x_1'\beta+a}}{1+e^{x_1'\beta+a}}\frac{e^{\gamma+x_2'\beta+a}}{1+e^{\gamma+x_2'\beta+a}} \frac{e^{\gamma+x_3'\beta+a}}{1+e^{\gamma+x_3'\beta+a}}\pi(a|x)da \\
     &+\int \frac{e^{\gamma+x_2'\beta+a}}{1+e^{\gamma+x_2'\beta+a}}\frac{e^{x_1'\beta+a}}{1+e^{x_1'\beta+a}}\frac{e^{\gamma+x_2'\beta+a}}{1+e^{\gamma+x_2'\beta+a}} \frac{1}{1+e^{\gamma+x_3'\beta+a}}\pi(a|x)da \\
    &+\int \frac{e^{\gamma+x_2'\beta+a}}{1+e^{\gamma+x_2'\beta+a}}\frac{e^{x_1'\beta+a}}{1+e^{x_1'\beta+a}}\frac{1}{1+e^{\gamma+x_2'\beta+a}} \pi(a|x)da \\
    &=+\int \frac{e^{\gamma+x_2'\beta+a}}{1+e^{\gamma+x_2'\beta+a}}\frac{e^{x_1'\beta+a}}{1+e^{x_1'\beta+a}}\frac{e^{\gamma+x_2'\beta+a}}{1+e^{\gamma+x_2'\beta+a}} \frac{1}{1+e^{\gamma+x_3'\beta+a}}\pi(a|x)da \\
    &+\int \frac{e^{\gamma+x_2'\beta+a}}{1+e^{\gamma+x_2'\beta+a}}\frac{e^{x_1'\beta+a}}{1+e^{x_1'\beta+a}}\frac{1}{1+e^{\gamma+x_2'\beta+a}}\frac{1}{1+e^{\gamma+x_3'\beta+a}} \pi(a|x)da \\
    &=\mathbb{E}[Y_1Y_2(1-Y_3)|Y_0=0,x]
\end{align*}
By a very similar calculation, $ \mathbb{E}\left[\Delta_{\gamma,2}^{1|1}|x\right]=-\mathbb{E}[Y_1Y_2(1-Y_3)|Y_0=0,x]$. Then, 
\begin{align*}
    \mathbb{E}\left[\Delta_{\gamma,3}^{1|1}|x\right]&=\Omega_{11}(x)\Sigma_{12}(x) \\
    \mathbb{E}\left[\Delta_{\gamma,4}^{1|1}|x\right]&=\Omega_{12}(x)\Sigma_{22}(x) \\
\mathbb{E}\left[\Delta_{\gamma,5}^{1|1}|x\right]&=+e^{x_{13}'\beta-\gamma}P_{011}(x)
\end{align*}
It follows that 
\begin{align*}
    \mathbb{E}\left[\Delta_{\gamma}^{1|1}|x\right]=\Omega_{11}(x)\Sigma_{12}(x)+\Omega_{12}(x)\Sigma_{22}(x)+e^{x_{13}'\beta-\gamma}P_{011}(x)=D_{21}(x)-D_{21}(x)=0
\end{align*}

\noindent \underline{\textbf{D)3) $S_{\beta}-\psi_{\beta}^{eff}(Y_{1}^{3},x)\perp \psi_{\theta}^{0|0}(Y_{1}^{3},Y_{0}^1,x)$}} \\
Fix $j\in \{2,\ldots,K+1\}$.
Let $\Delta_{\beta_{j-1}}^{0|0}=(S_{\beta_{j-1}} -\psi_{\beta_{j-1}}^{eff}(Y_{1}^{3},Y_{i0}^1,x)) \psi_{\theta}^{0|0}(Y_{1}^{3},Y_{0}^1,x)$. Tedious calculations and rearrangements lead to the following decomposition:

\begin{align*}
    \Delta_{\beta_{j-1}}^{0|0}&= \Delta_{\beta_{j-1},1}^{0|0}+\Delta_{\beta_{j-1},2}^{0|0}+\Delta_{\beta_{j-1}}^{0|0}(x_1)+\Delta_{\beta_{j-1}}^{0|0}(x_2)+\Delta_{\beta_{j-1}}^{0|0}(x_3)
\end{align*}
where
\begin{align*}
    \Delta_{\beta_{j-1}}^{0|0}(x_1)&= \Delta_{\beta_{j-1},1}^{0|0}(x_1)+ \Delta_{\beta_{j-1},2}^{0|0}(x_1)\\
    \Delta_{\beta_{j-1},1}^{0|0}(x_1)&=-(e^{x_{23}'\beta }-1)x_{1,j-1}B_{001}(x_{1})(1-Y_1)(1-Y_2)Y_3 -e^{x_{21}'\beta+\gamma}x_{1,j-1}B_{101}(x_{1})Y_1(1-Y_2)Y_3 \\
    &-e^{x_{31}'\beta}x_{1,j-1}B_{100}(x_{1})Y_1(1-Y_2)(1-Y_3) \\
     &+x_{1,j-1}B_{011}(x_{1})(1-Y_1)Y_2Y_3+x_{1,j-1}B_{010}(x_{1})(1-Y_1)Y_2(1-Y_3) \\
    \Delta_{\beta_{j-1},2}^{0|0}(x_1)&=e^{x_{21}'\beta+\gamma}x_{1,j-1}(Y_1(1-Y_2)Y_3+e^{x_{31}'\beta}x_{1,j-1}Y_1(1-Y_2)(1-Y_3)
\end{align*}
and
\begin{align*}
    \Delta_{\beta_{j-1}}^{0|0}(x_2)&= \Delta_{\beta_{j-1},1}^{0|0}(x_2)+\Delta_{\beta_{j-1},2}^{0|0}(x_2)+\Delta_{\beta_{j-1},3}^{0|0}(x_2)\\
     \Delta_{\beta_{j-1},1}^{0|0}(x_2)&=e^{x_{23}'\beta}x_{2,j-1}(1-B_{001}(x_{2}))(1-Y_1)(1-Y_2)Y_3 +x_{2,j-1}B_{001}(x_{2})(1-Y_1)(1-Y_2)Y_3  \\
     &-x_{2,j-1}(1-B_{011}(x_{2}))(1-Y_1)Y_2Y_3 -x_{2,j-1}(1-B_{010}(x_{2}))(1-Y_1)Y_2(1-Y_3) \\
     \Delta_{\beta_{j-1},2}^{0|0}(x_2)&=  +e^{x_{21}'\beta+\gamma}x_{2,j-1}(1-C_{101}(x_{2}))Y_1(1-Y_2)Y_3-e^{x_{31}'\beta}x_{2,j-1}C_{100}(x_{2})Y_1(1-Y_2)(1-Y_3) \\
    \Delta_{\beta_{j-1},3}^{0|0}(x_2)&=-e^{x_{23}'\beta}x_{2,j-1}(1-Y_1)(1-Y_2)Y_3-e^{x_{21}'\beta+\gamma}x_{2,j-1}Y_1(1-Y_2)Y_3
\end{align*}
and
\begin{align*}
    \Delta_{\beta_{j-1}}^{0|0}(x_3)&=\Delta_{\beta_{j-1},1}^{0|0}(x_3)+\Delta_{\beta_{j-1},2}^{0|0}(x_3)+\Delta_{\beta_{j-1},3}^{0|0}(x_3) \\
    \Delta_{\beta_{j-1},1}^{0|0}(x_3)&=-(e^{x_{23}'\beta}-1)x_{3,j-1}B_{001}(x_{3})(1-Y_1)(1-Y_2)Y_3 -x_{3,j-1}(1-Y_1)(1-Y_2)Y_3  \\
    &+e^{x_{21}'\beta+\gamma}x_{3,j-1}(1-B_{101}(x_{3}))Y_1(1-Y_2)Y_3 +e^{x_{31}'\beta}x_{3,j-1}(1-B_{100}(x_{3}))Y_1(1-Y_2)(1-Y_3) \\
    \Delta_{\beta_{j-1},2}^{0|0}(x_3)&=-x_{3,j-1}(1-C_{011}(x_{3}))(1-Y_1)Y_2Y_3+x_{3,j-1}C_{010}(x_{3})(1-Y_1)Y_2(1-Y_3) \\
     \Delta_{\beta_{j-1},3}^{0|0}(x_3)&=e^{x_{23}'\beta }x_{3,j-1}(1-Y_1)(1-Y_2)Y_3 -e^{x_{31}'\beta}x_{3,j-1}Y_1(1-Y_2)(1-Y_3)
\end{align*}
and last
\begin{align*}
    \Delta_{\beta_{j-1},1}^{0|0}&=+\Omega_{j1}(x)(e^{x_{23}'\beta }-1)^2(1-Y_{1})(1-Y_{2})Y_{3} +\Omega_{j1}(x)e^{2x_{21}'\beta+2\gamma}Y_{1}(1-Y_{2})Y_{3} \\
   &+\Omega_{j1}(x)e^{2x_{31}'\beta}Y_{1}(1-Y_{2})(1-Y_{3}) +\Omega_{j1}(x)(1-Y_{1})Y_{2} \\
     \Delta_{\beta_{j-1},2}^{0|0}&=-\Omega_{j2}(x)e^{x_{21}'\beta+\gamma}Y_1(1-Y_2)Y_3 -\Omega_{j2}(x)e^{x_{31}'\beta}Y_1(1-Y_2)(1-Y_3)\\
   &-\Omega_{j2}(x)e^{x_{12}'\beta}(1-Y_1)Y_2(1-Y_3)-\Omega_{j2}(x)e^{x_{13}'\beta-\gamma}(1-Y_1)Y_2Y_3
\end{align*}
Starting first with the terms in \say{$x_1$}, we have:
\begin{align*}
    &\frac{1}{x_{1,j-1}}\mathbb{E}[\Delta_{\beta_{j-1},1}^{0|0}(x_1)|x]=\mathbb{E}\left[\frac{e^{x_1'\beta+A}}{1+e^{x_1'\beta+A}}\mathbb{E}\left[-\psi_{\theta}^{0|0}(Y_{1}^{3},Y_{0}^1,x)|x,A\right]|x\right] =0 \\
    & \mathbb{E}[\Delta_{\beta_{j-1},2}^{0|0}(x_1)|x]=e^{x_{21}'\beta+\gamma}x_{1,j-1}P_{101}(x)+e^{x_{31}'\beta}x_{1,j-1}P_{100}(x)
\end{align*}
Next, for the terms in \say{$x_2$}, we have:
\begin{align*}
    \frac{1}{x_{2,j-1}}\mathbb{E}\left[ \Delta_{\beta,1}(x_2)|x\right]&=
    \int  \frac{1}{1+e^{x_2'\beta+a}} \frac{1}{1+e^{x_1'\beta+a}} \frac{1}{1+e^{x_2'\beta+a}} \frac{e^{x_3'\beta+a}}{1+e^{x_3'\beta+a}}e^{x_{23}'\beta}\pi(a|x)da \\
    &+\int  \frac{e^{x_2'\beta+a}}{1+e^{x_2'\beta+a}} \frac{1}{1+e^{x_1'\beta+a}} \frac{1}{1+e^{x_2'\beta+a}} \frac{e^{x_3'\beta+a}}{1+e^{x_3'\beta+a}}\pi(a|x)da \\
    &-\int  \frac{1}{1+e^{x_2'\beta+a}} \frac{1}{1+e^{x_1'\beta+a}} \frac{e^{x_2'\beta+a}}{1+e^{x_2'\beta+a}} \pi(a|x)da \\
     &= \int  \frac{1}{1+e^{x_2'\beta+a}} \frac{1}{1+e^{x_1'\beta+a}} \frac{e^{x_2'\beta+a}}{1+e^{x_2'\beta+a}} \pi(a|x)da \\
       &-\int  \frac{1}{1+e^{x_2'\beta+a}} \frac{1}{1+e^{x_1'\beta+a}} \frac{e^{x_2'\beta+a}}{1+e^{x_2'\beta+a}} \pi(a|x)da \\
       &=0  \\
\end{align*}
\begin{align*}
        \frac{1}{x_{2,j-1}}E\left[ \Delta_{\beta,2}(x_2)|x\right]&=\int  \frac{1}{1+e^{\gamma+x_2'\beta+a}} \frac{e^{x_1'\beta+a}}{1+e^{x_1'\beta+a}} \frac{1}{1+e^{\gamma+x_2'\beta+a}} \frac{e^{x_3'\beta+a}}{1+e^{x_3'\beta+a}}e^{x_{21}'\beta+\gamma}\pi(a|x)da \\
        &-\int  \frac{e^{\gamma+x_2'\beta+a}}{1+e^{\gamma+x_2'\beta+a}} \frac{e^{x_1'\beta+a}}{1+e^{x_1'\beta+a}} \frac{1}{1+e^{\gamma+x_2'\beta+a}} \frac{1}{1+e^{x_3'\beta+a}}e^{x_{31}'\beta}\pi(a|x)da \\
        &=\int  \frac{1}{1+e^{\gamma+x_2'\beta+a}} \frac{1}{1+e^{x_1'\beta+a}} \frac{e^{\gamma+x_2'\beta+a}}{1+e^{\gamma+x_2'\beta+a}} \frac{e^{x_3'\beta+a}}{1+e^{x_3'\beta+a}}\pi(a|x)da \\
        &-\int  \frac{e^{\gamma+x_2'\beta+a}}{1+e^{\gamma+x_2'\beta+a}} \frac{1}{1+e^{x_1'\beta+a}} \frac{1}{1+e^{\gamma+x_2'\beta+a}} \frac{e^{x_3'\beta+a}}{1+e^{x_3'\beta+a}}\pi(a|x)da  \\
        &=0 \\
        \mathbb{E}\left[\Delta_{\beta_{j-1},3}^{0|0}(x_2)|x\right]&=-e^{x_{23}'\beta}x_{2,j-1}P_{001}(x)-e^{x_{21}'\beta+\gamma}x_{2,j-1}P_{101}(x)
\end{align*}
By the same token, for the terms in \say{$x_3$}, one arrives at $\mathbb{E}\left[\Delta_{\beta_{j-1},1}^{0|0}(x_3)|x\right]=\mathbb{E}\left[\Delta_{\beta_{j-1},2}^{0|0}(x_3)|x\right]=0$ and 
\begin{align*}
    \mathbb{E}\left[\Delta_{\beta_{j-1},1}^{0|0}(x_3)|x\right]&=\mathbb{E}\left[\Delta_{\beta_{j-1},2}^{0|0}(x_3)|x\right]=0\\
      \mathbb{E}\left[\Delta_{\beta_{j-1},3}^{0|0}(x_3)|x\right]&=e^{x_{23}'\beta }x_{3,j-1}P_{001}(x)-e^{x_{31}'\beta}x_{3,j-1}P_{100}(x)
\end{align*}
Finally, $\mathbb{E}[\Delta_{\beta_{j-1},1}^{0|0}|x]=\Omega_{j,1}(x)\Sigma_{11}(x),\mathbb{E}[\Delta_{\beta_{j-1},2}^{0|0}|x]=\Omega_{j,2}(x)\Sigma_{12}(x)$. Collecting terms, we get
\begin{align*}
      \mathbb{E}\left[\Delta_{\beta_{j-1}}^{0|0}|x\right]&=e^{x_{21}'\beta+\gamma}x_{1,j-1}P_{101}(x)+e^{x_{31}'\beta}x_{1,j-1}P_{100}(x)-e^{x_{23}'\beta}x_{2,j-1}P_{001}(x)-e^{x_{21}'\beta+\gamma}x_{2,j-1}P_{101}(x) \\
      &e^{x_{23}'\beta }x_{3,j-1}P_{001}(x)-e^{x_{31}'\beta}x_{3,j-1}P_{100}(x)+\Omega_{j1}(x)\Sigma_{11}(x)+\Omega_{j2}(x)\Sigma_{12}(x) \\
      &=-D_{1j}(x)+D_{1j}(x) \\
      &=0
\end{align*}
This is of course valid for all slope parameters $\beta_j$ and hence  $S_{\beta}-\psi_{\beta}^{eff}(Y_{1}^{3},x)\perp \psi_{\theta}^{0|0}(Y_{1}^{3},Y_{0}^1,x)$ \\

\noindent \underline{\textbf{D)4) $S_{\beta}-\psi_{\beta}^{eff}(Y_{1}^{3},x)\perp \psi_{\theta}^{1|1}(Y_{1}^{3},Y_{0}^1,x)$}} \\
\noindent Fix $j\in \{2,\ldots,K+1\}$.
Let $\Delta_{\beta_{j-1}}^{1|1}=(S_{\beta_{j-1}} -\psi_{\beta_{j-1}}^{eff}(Y_{1}^{3},Y_{i0}^1,x)) \psi_{\theta}^{1|1}(Y_{1}^{3},Y_{0}^1,x)$. A last set of lengthy calculations and rearrangements lead to the following decomposition:
\begin{align*}
    \Delta_{\beta_{j-1}}^{1|1}&= \Delta_{\beta_{j-1},1}^{1|1}+ \Delta_{\beta_{j-1},2}^{1|1} + \Delta_{\beta_{j-1}}^{1|1}(x_1)+\Delta_{\beta_{j-1}}^{1|1}(x_2)+\Delta_{\beta_{j-1}}^{1|1}(x_3)
\end{align*}
where
\begin{align*}
    \Delta_{\beta_{j-1}}^{1|1}(x_1)&= \Delta_{\beta_{j-1},1}^{1|1}(x_1)+\Delta_{\beta_{j-1},2}^{1|1}(x_1) \\
    \Delta_{\beta_{j-1},1}^{1|1}(x_1)&= +(e^{x_{32}'\beta}-1)x_{1,j-1}(1-B_{110}(x_1))Y_1Y_2(1-Y_3) +e^{x_{12}'\beta}x_{1,j-1}(1-B_{010}(x_1))(1-Y_1)Y_2(1-Y_3) \\
    &+e^{x_{13}'\beta-\gamma}x_{1,j-1}(1-B_{011}(x_1))(1-Y_1)Y_2Y_3 \\
     &-x_{1,j-1}(1-B_{101}(x_1))Y_1(1-Y_2)Y_3 -x_{1,j-1}(1-B_{100}(x_1))Y_1(1-Y_2)(1-Y_3)  \\
    \Delta_{\beta_{j-1},2}^{1|1}(x_1)&=-e^{x_{12}'\beta}x_{1,j-1}(1-Y_1)Y_2(1-Y_3)-e^{x_{13}'\beta-\gamma}x_{1,j-1}(1-Y_1)Y_2Y_3
\end{align*}
and
\begin{align*}
    \Delta_{\beta_{j-1}}^{1|1}(x_2)&=\Delta_{\beta_{j-1},1}^{1|1}(x_2)+\Delta_{\beta_{j-1},2}^{1|1}(x_2)+\Delta_{\beta_{j-1},3}^{1|1}(x_2) \\
    \Delta_{\beta_{j-1},1}^{1|1}(x_2)&= -e^{x_{32}'\beta}x_{2,j-1}C_{110}(x_2)Y_1Y_2(1-Y_3)-x_{2,j-1}(1-C_{110}(x_2))Y_1Y_2(1-Y_3) \\
     &+x_{2,j-1}C_{101}(x_2)Y_1(1-Y_2)Y_3 +x_{2,j-1}C_{100}(x_2)Y_1(1-Y_2)(1-Y_3) \\
    \Delta_{\beta_{j-1},2}^{1|1}(x_2)&=-e^{x_{12}'\beta}x_{2,j-1}B_{010}(x_2)(1-Y_1)Y_2(1-Y_3)+e^{x_{13}'\beta-\gamma}x_{2,j-1}(1-B_{011}(x_2))(1-Y_1)Y_2Y_3 \\
    \Delta_{\beta_{j-1},3}^{1|1}(x_2)&=e^{x_{32}'\beta}x_{2,j-1}Y_1Y_2(1-Y_3)+e^{x_{12}'\beta}x_{2,j-1}(1-Y_1)Y_2(1-Y_3)
\end{align*}
and
\begin{align*}
\Delta_{\beta_{j-1}}^{1|1}(x_3)&=\Delta_{\beta_{j-1},1}^{1|1}(x_3)+\Delta_{\beta_{j-1},2}^{1|1}(x_3)+\Delta_{\beta_{j-1},3}^{1|1}(x_3)\\
    \Delta_{\beta_{j-1},1}^{1|1}(x_3)&= +e^{x_{32}'\beta}x_{3,j-1}(1-C_{110}(x_3))Y_1Y_2(1-Y_3)+x_{3,j-1}C_{110}(x_3)Y_1Y_2(1-Y_3) \\
    &-e^{x_{12}'\beta}x_{3,j-1}C_{010}(x_3)(1-Y_1)Y_2(1-Y_3)-e^{x_{13}'\beta-\gamma}x_{3,j-1}C_{011}(x_3)(1-Y_1)Y_2Y_3 \\
   \Delta_{\beta_{j-1},2}^{1|1}(x_3)&=-x_{3,j-1}(1-B_{101}(x_3))Y_1(1-Y_2)Y_3 +x_{3,j-1}B_{100}(x_3)Y_1(1-Y_2)(1-Y_3)  \\
     \Delta_{\beta_{j-1},3}^{1|1}(x_3)&=e^{x_{13}'\beta-\gamma}x_{3,j-1}(1-Y_1)Y_2Y_3-e^{x_{32}'\beta}x_{3,j-1}Y_1Y_2(1-Y_3)
\end{align*}
and last
\begin{align*}
     \Delta_{\beta_{j-1},1}^{1|1}&=-\Omega_{j1}(x)e^{x_{21}'\beta+\gamma}Y_1(1-Y_2)Y_3-\Omega_{j1}(x)e^{x_{31}'\beta}Y_1(1-Y_2)(1-Y_3)\\
    &-\Omega_{j1}(x)e^{x_{12}'\beta}(1-Y_1)Y_2(1-Y_3)-\Omega_{j1}(x)e^{x_{13}'\beta-\gamma}(1-Y_1)Y_2Y_3 \\
     \Delta_{\beta_{j-1},2}^{1|1}&=+\Omega_{j2}(x)(e^{x_{32}'\beta}-1)^2Y_1Y_2(1-Y_3)
    +\Omega_{j2}(x)e^{2x_{12}'\beta}(1-Y_1)Y_2(1-Y_3)\\
    &+\Omega_{j2}(x)e^{2x_{13}'\beta-2\gamma}(1-Y_1)Y_2Y_3+\Omega_{j2}(x)Y_1(1-Y_2) 
\end{align*}
Starting first with the terms in \say{$x_1$}, we have:
\begin{align*}
     \frac{1}{x_{1,j-1}}\mathbb{E}\left[\Delta_{\beta_{j-1},1}^{1|1}(x_1)|x\right]&= \mathbb{E}\left[\frac{1}{1+e^{x_1'\beta+A}}\mathbb{E}\left[\psi_{\theta}^{1|1}(Y_{i1}^{3},Y_{i0}^1,X_i)|x,A\right]|x\right]=0 \\
      \mathbb{E}\left[\Delta_{\beta_{j-1},2}^{1|1}(x_1)|x\right]&=-e^{x_{12}'\beta}x_{1,j-1}P_{010}(x)-e^{x_{13}'\beta-\gamma}x_{1,j-1}P_{011}(x)
\end{align*}
For the terms in \say{$x_2$}
\begin{align*}
      \frac{1}{x_{2,j-1}}\mathbb{E}\left[\Delta_{\beta_{j-1},1}^{1|1}(x_2)|x\right]&= -\int \frac{e^{\gamma+x_2'\beta+a}}{1+e^{\gamma+x_2'\beta+a}} \frac{e^{x_1'\beta+a}}{1+e^{x_1'\beta+a}} \frac{e^{\gamma+x_2'\beta+a}}{1+e^{\gamma+x_2'\beta+a}} \frac{1}{1+e^{\gamma+x_3'\beta+a}}e^{x_{32}'\beta}\pi(a|x)da \\
    &-\int \frac{1}{1+e^{\gamma+x_2'\beta+a}} \frac{e^{x_1'\beta+a}}{1+e^{x_1'\beta+a}} \frac{e^{\gamma+x_2'\beta+a}}{1+e^{\gamma+x_2'\beta+a}} \frac{1}{1+e^{\gamma+x_3'\beta+a}}\pi(a|x)da\\
    &+\int  \frac{e^{\gamma+x_2'\beta+a}}{1+e^{\gamma+x_2'\beta+a}}\frac{e^{x_1'\beta+a}}{1+e^{x_1'\beta+a}}  \frac{1}{1+e^{\gamma+x_2'\beta+a}} \pi(a|x)da \\
    &=-\int \frac{e^{\gamma+x_2'\beta+a}}{1+e^{\gamma+x_2'\beta+a}} \frac{e^{x_1'\beta+a}}{1+e^{x_1'\beta+a}} \frac{1}{1+e^{\gamma+x_2'\beta+a}} \frac{e^{\gamma+x_3'\beta+a}}{1+e^{\gamma+x_3'\beta+a}}\pi(a|x)da \\
    &+\int  \frac{e^{\gamma+x_2'\beta+a}}{1+e^{\gamma+x_2'\beta+a}}\frac{e^{x_1'\beta+a}}{1+e^{x_1'\beta+a}}  \frac{1}{1+e^{\gamma+x_2'\beta+a}}\frac{e^{\gamma+x_3'\beta+a}}{1+e^{\gamma+x_3'\beta+a}} \pi(a|x)da \\
    &=0 \\
    \frac{1}{x_{2,j-1}}\mathbb{E}\left[\Delta_{\beta_{j-1},2}^{1|1}(x_2)|x\right]&=-\int \frac{e^{x_2'\beta+a}}{1+e^{x_2'\beta+a}}\frac{1}{1+e^{x_1'\beta+a}}\frac{e^{\gamma+x_2'\beta+a}}{1+e^{\gamma+x_2'\beta+a}}\frac{1}{1+e^{\gamma+x_3'\beta+a}}e^{x_{12}'\beta}\pi(a|x)da \\
    &+\int \frac{1}{1+e^{x_2'\beta+a}}\frac{1}{1+e^{x_1'\beta+a}}\frac{e^{\gamma+x_2'\beta+a}}{1+e^{\gamma+x_2'\beta+a}}\frac{e^{\gamma+x_3'\beta+a}}{1+e^{\gamma+x_3'\beta+a}}e^{x_{13}'\beta-\gamma}\pi(a|x)da \\
    &=-\int \frac{1}{1+e^{x_2'\beta+a}}\frac{e^{x_1'\beta+a}}{1+e^{x_1'\beta+a}}\frac{e^{\gamma+x_2'\beta+a}}{1+e^{\gamma+x_2'\beta+a}}\frac{1}{1+e^{\gamma+x_3'\beta+a}}\pi(a|x)da \\
    &+\int \frac{1}{1+e^{x_2'\beta+a}}\frac{e^{x_1'\beta+a}}{1+e^{x_1'\beta+a}}\frac{e^{\gamma+x_2'\beta+a}}{1+e^{\gamma+x_2'\beta+a}}\frac{1}{1+e^{\gamma+x_3'\beta+a}}\pi(a|x)da \\
    &=0 \\
    \mathbb{E}\left[\Delta_{\beta_{j-1},3}^{1|1}(x_2)|x\right]&=e^{x_{32}'\beta}x_{2,j-1}P_{110}(x)+e^{x_{12}'\beta}x_{2,j-1}P_{010}(x)
\end{align*}
Similar calculations for the terms in \say{$x_3$} yield $  \mathbb{E}\left[\Delta_{\beta_{j-1},1}^{1|1}(x_3)|x\right]=\mathbb{E}\left[\Delta_{\beta_{j-1},2}^{1|1}(x_3)|x\right]=0$ and $\mathbb{E}\left[\Delta_{\beta_{j-1},3}^{1|1}(x_3)|x\right]=e^{x_{13}'\beta-\gamma}x_{3,j-1}P_{011}(x)-e^{x_{32}'\beta}x_{3,j-1}P_{110}(x)$. Finally, 
\begin{align*}
    \mathbb{E}\left[\Delta_{\beta_{j-1},1}^{1|1}|x\right]&=-\Omega_{j1}(x)e^{x_{21}'\beta+\gamma}P_{101}(x)-\Omega_{j1}(x)e^{x_{31}'\beta}P_{100}(x)\\
    &-\Omega_{j1}(x)e^{x_{12}'\beta}P_{010}(x)-\Omega_{j1}(x)e^{x_{13}'\beta-\gamma}P_{011}(x) \\
    &=\Omega_{j1}(x)\Sigma_{12}(x) \\
     \mathbb{E}\left[\Delta_{\beta_{j-1},2}^{1|1}|x\right]&=+\Omega_{j2}(x)(e^{x_{32}'\beta}-1)^2P_{110}(x)
    +\Omega_{j2}(x)e^{2x_{12}'\beta}P_{010}(x)\\
    &+\Omega_{j2}(x)e^{2x_{13}'\beta-2\gamma}P_{011}(x)+\Omega_{j2}(x)P_{10}(x) \\
    &=\Omega_{j2}(x)\Sigma_{22}(x)
\end{align*}
Putting the different pieces together, we ultimately obtain
\begin{align*}
    \mathbb{E}\left[ \Delta_{\beta_{j-1}}^{1|1}|x\right]&=-e^{x_{12}'\beta}x_{1,j-1}P_{010}(x)-e^{x_{13}'\beta-\gamma}x_{1,j-1}P_{011}(x) \\
    &+e^{x_{32}'\beta}x_{2,j-1}P_{110}(x)+e^{x_{12}'\beta}x_{2,j-1}P_{010}(x) \\
    &+e^{x_{13}'\beta-\gamma}x_{3,j-1}P_{011}(x)-e^{x_{32}'\beta}x_{3,j-1}P_{110}(x) \\
    &+\Omega_{j1}(x)\Sigma_{12}(x)+\Omega_{j2}(x)\Sigma_{22}(x) \\
    &=-D_{2j}(x)+D_{2j}(x) \\
    &=0
\end{align*}
This is of course valid for all slope parameters $\beta_j$ and hence  $S_{\beta}-\psi_{\beta}^{eff}(Y_{1}^{3},x)\perp \psi_{\theta}^{1|1}(Y_{1}^{3},Y_{0}^1,x)$ \\

\noindent \underline{\textbf{E) Conclusion}} \\
Having verified all the conditions of Theorem 3.2 in \cite{newey1990semiparametric} for the initial condition $Y_0=0$, we conclude that in that case $\psi_{\theta}^{eff}(Y_{1}^{3},X)$ is the efficient score of the AR(1) model. The semiparametric efficiency bound is given by $\mathbb{E}\left[D(X)'\Sigma(X)^{-1}D(X)\right]^{-1}$. Symmetric results can be shown to hold for the case $Y_0=1$.

\end{document}